\documentclass[english,11pt, letter]{article}
\usepackage[latin9]{inputenc}
\usepackage{geometry}
\usepackage{array}
\usepackage{float}
\usepackage{multirow}
\usepackage{amsthm}
\usepackage{amsmath}
\usepackage{amssymb}
\usepackage{esint}

\makeatletter

\providecommand{\tabularnewline}{\\}

\numberwithin{equation}{section}
\numberwithin{figure}{section}
\theoremstyle{plain}
\newtheorem{thm}{\protect\theoremname}
  \theoremstyle{definition}
  \newtheorem{defn}[thm]{\protect\definitionname}
  \theoremstyle{plain}
  \newtheorem{lem}[thm]{\protect\lemmaname}
  \theoremstyle{remark}
  \newtheorem{rem}[thm]{\protect\remarkname}
  \theoremstyle{plain}
  \newtheorem{fact}[thm]{\protect\factname}
  \theoremstyle{plain}
  \newtheorem{cor}[thm]{\protect\corollaryname}

\pdfoutput=1

\usepackage{amsmath, amsthm, amssymb, amsfonts, bbold}
\usepackage{bbm}
\usepackage[basic]{complexity}
\usepackage[pdftex,bookmarks,colorlinks]{hyperref}
\usepackage{enumitem}
\usepackage{tikz}

\usepackage[lined,boxed,ruled,norelsize,algo2e]{algorithm2e}

\DeclareMathOperator*{\argmaxTex}{arg\,max}
\DeclareMathOperator*{\argminTex}{arg\,min}
\DeclareMathOperator*{\signTex}{sign}
\DeclareMathOperator*{\rankTex}{rank}
\DeclareMathOperator*{\diagTex}{diag}
\DeclareMathOperator*{\imTex}{im}

\hypersetup{colorlinks=true,citecolor=red}

\renewcommand{\varepsilon}{\epsilon}
\renewcommand{\subsetneq}{\not\subseteq}

\makeatother

\usepackage{babel}
  \providecommand{\corollaryname}{Corollary}
  \providecommand{\definitionname}{Definition}
  \providecommand{\factname}{Fact}
  \providecommand{\lemmaname}{Lemma}
  \providecommand{\remarkname}{Remark}
\providecommand{\theoremname}{Theorem}

\begin{document}

\global\long\def\R{\mathbb{R}}
 \global\long\def\Rn{\mathbb{R}^{n}}
 \global\long\def\Rm{\mathbb{R}^{m}}
 \global\long\def\Rmn{\mathbb{R}^{m \times n}}
 \global\long\def\Rnm{\mathbb{R}^{n \times m}}
 \global\long\def\Rmm{\mathbb{R}^{m \times m}}
 \global\long\def\Rnn{\mathbb{R}^{n \times n}}
 \global\long\def\Z{\mathbb{Z}}
 \global\long\def\rPos{\R_{> 0}}
 \global\long\def\dom{\mathrm{dom}}


\global\long\def\ellOne{\ell_{1}}
 \global\long\def\ellTwo{\ell_{2}}
 \global\long\def\ellInf{\ell_{\infty}}
 \global\long\def\ellP{\ell_{p}}

\global\long\def\otilde{\widetilde{O}}

\global\long\def\argmax{\argmaxTex}

\global\long\def\argmin{\argminTex}

\global\long\def\sign{\signTex}

\global\long\def\rank{\rankTex}

\global\long\def\diag{\diagTex}

\global\long\def\im{\imTex}

\global\long\def\enspace{\quad}

\global\long\def\boldVar#1{\mathbf{#1}}

\global\long\def\mvar#1{\boldVar{#1}}

\global\long\def\vvar#1{\vec{#1}}



\global\long\def\defeq{\stackrel{\mathrm{{\scriptscriptstyle def}}}{=}}

\global\long\def\diag{\mathrm{diag}}

\global\long\def\mDiag{\mvar{diag}}
 \global\long\def\ceil#1{\left\lceil #1 \right\rceil }

\global\long\def\E{\mathbb{E}}
 \global\long\def\abs#1{\left|#1\right|}
\global\long\def\var{\mathrm{Var}}


\global\long\def\onesVec{\vec{\mathbb{1}}}
 \global\long\def\indicVec#1{\onesVec_{#1}}
\global\long\def\indic{1}

\global\long\def\va{\vvar a}
 \global\long\def\vb{\vvar b}
 \global\long\def\vc{\vvar c}
 \global\long\def\vd{\vvar d}
 \global\long\def\ve{\vvar e}
 \global\long\def\vf{\vvar f}
 \global\long\def\vg{\vvar g}
 \global\long\def\vh{\vvar h}
 \global\long\def\vl{\vvar l}
\global\long\def\vq{\vvar q}
 \global\long\def\vm{\vvar m}
 \global\long\def\vn{\vvar n}
 \global\long\def\vo{\vvar o}
\global\long\def\vr{\vvar r}
 \global\long\def\vp{\vvar p}
 \global\long\def\vs{\vvar s}
\global\long\def\vt{\vvar t}
 \global\long\def\vu{\vvar u}
 \global\long\def\vv{\vvar v}
\global\long\def\vw{\vvar w}
 \global\long\def\vx{\vvar x}
 \global\long\def\vy{\vvar y}
 \global\long\def\vz{\vvar z}
 \global\long\def\vxi{\vvar{\xi}}
 \global\long\def\valpha{\vvar{\alpha}}
 \global\long\def\veta{\vvar{\eta}}
 \global\long\def\vphi{\vvar{\phi}}
\global\long\def\vpsi{\vvar{\psi}}
 \global\long\def\vsigma{\vvar{\sigma}}
 \global\long\def\vtau{\vvar{\tau}}
 \global\long\def\vgamma{\vvar{\gamma}}
 \global\long\def\vphi{\vvar{\phi}}
\global\long\def\vdelta{\vvar{\delta}}
\global\long\def\vDelta{\vvar{\Delta}}
\global\long\def\vzero{\vvar 0}
 \global\long\def\vones{\vvar 1}
\global\long\def\vupsilon{\vvar{\upsilon}}
\global\long\def\vtheta{\vec{\theta}}

\global\long\def\ma{\mvar A}
 \global\long\def\mb{\mvar B}
 \global\long\def\mc{\mvar C}
 \global\long\def\md{\mvar D}
\global\long\def\mE{\mvar E}
 \global\long\def\mf{\mvar F}
 \global\long\def\mg{\mvar G}
 \global\long\def\mh{\mvar H}
 \global\long\def\mj{\mvar J}
 \global\long\def\mk{\mvar K}
 \global\long\def\mm{\mvar M}
 \global\long\def\mn{\mvar N}
 \global\long\def\mq{\mvar Q}
 \global\long\def\mr{\mvar R}
 \global\long\def\ms{\mvar S}
 \global\long\def\mt{\mvar T}
 \global\long\def\mU{\mvar U}
 \global\long\def\mv{\mvar V}
 \global\long\def\mx{\mvar X}
 \global\long\def\my{\mvar Y}
 \global\long\def\mz{\mvar Z}
 \global\long\def\mSigma{\mvar{\Sigma}}
 \global\long\def\mLambda{\mvar{\Lambda}}
\global\long\def\mPhi{\mvar{\Phi}}
\global\long\def\mpi{\mvar{\Pi}}
 \global\long\def\mZero{\mvar 0}
 \global\long\def\iMatrix{\mvar I}
\global\long\def\mDelta{\mvar{\Delta}}
\global\long\def\mPsi{\mvar{\Psi}}

\global\long\def\oracle{\mathcal{O}}

\global\long\def\mpr{\mvar P}
 \global\long\def\mptwo{\mvar P^{(2)}}

\global\long\def\vLever{\vsigma}
 \global\long\def\mLever{\mSigma}
 \global\long\def\vWeight{\vec{w}}
 \global\long\def\mWeight{\mvar W}
 \global\long\def\width{\mathrm{MinWidth}}
 \global\long\def\OPT{\mathrm{OPT}}
 \global\long\def\mwidth{{\color{red}TODO}}


\global\long\def\norm#1{\big\|#1\big\|}
 \global\long\def\normFull#1{\left\Vert #1\right\Vert }
 \global\long\def\normFullInf#1{\normFull{#1}_{\infty}}
 \global\long\def\normInf#1{\norm{#1}_{\infty}}
 \global\long\def\normOne#1{\norm{#1}_{1}}
 \global\long\def\normTwo#1{\norm{#1}_{2}}
 \global\long\def\innerproduct#1#2{\left\langle #1,#2\right\rangle }

\global\long\def\TODO#1{{\color{red}TODO:\text{#1}}}

\global\long\def\next#1{#1^{\text{(new)}}}

\global\long\def\code#1{\texttt{#1}}
\global\long\def\conv{\mathrm{conv}}

\global\long\def\nnz{\mathrm{nnz}}
\global\long\def\solver{\mathrm{\mathtt{S}}}
\global\long\def\time{{\color{red}TODO}}
\global\long\def\tr{\mathrm{Tr}}
\global\long\def\grad{\nabla}
\global\long\def\hess{\nabla^{2}}
\global\long\def\stdhess{\mathcal{H}}

\newcommand{\bracket}[1]{[#1]}

\global\long\def\EO{\mathrm{EO}}
 \global\long\def\LO{\mathrm{LO}}
\global\long\def\SO{\mathrm{SO}}
 \global\long\def\OO{\mathrm{OO}}
 \global\long\def\GO{\mathrm{GO}}
 \global\long\def\XO{\mathrm{XO}}
\global\long\def\oracleFont#1{\mathrm{#1}}

\title{A Faster Cutting Plane Method and its\\
Implications for Combinatorial and Convex Optimization }

\author{Yin Tat Lee\\
MIT\\
yintat@mit.edu\and  Aaron Sidford\\
MIT\\
sidford@mit.edu\and  Sam Chiu-wai Wong\\
UC Berkeley\\
samcwong@berkeley.edu}

\date{}
\maketitle
\begin{abstract}
In this paper we improve upon the running time for finding a point
in a convex set given a separation oracle. In particular, given a
separation oracle for a convex set $K\subset\R^{n}$ that is contained
in a box of radius $R$ we show how to either compute a point in $K$
or prove that $K$ does not contain a ball of radius $\epsilon$ using
an expected $O(n\log(nR/\epsilon))$ evaluations of the oracle and
additional time $O(n^{3}\log^{O(1)}(nR/\epsilon))$. This matches
the oracle complexity and improves upon the $O(n^{\omega+1}\log(nR/\epsilon))$
additional time of the previous fastest algorithm achieved over 25
years ago by Vaidya \cite{vaidya89convexSet} for the current value
of the matrix multiplication constant $\omega<2.373$ \cite{williams2012matrixmult,gall2014powers}
when $R/\epsilon=O(\poly(n))$.

Using a mix of standard reductions and new techniques we show how
our algorithm can be used to improve the running time for solving
classic problems in continuous and combinatorial optimization. In
particular we provide the following running time improvements:
\begin{itemize}
\item \textbf{Submodular Function Minimization}: let $n$ be the size of
the ground set, $M$ be the maximum absolute value of function values,
and $\text{EO}$ be the time for function evaluation.{\small \par}

Our weakly and strongly polynomial time algorithms have a running
time of $O(n^{2}\log nM\cdot\text{EO}+n^{3}\log^{O(1)}nM)$ and $O(n^{3}\log^{2}n\cdot\text{EO}+n^{4}\log^{O(1)}n)$,
improving upon the previous best of $O((n^{4}\cdot\text{EO}+n^{5})\log M)$
and $O(n^{5}\cdot\text{EO}+n^{6})$ respectively.{\small \par}

\item \textbf{Matroid Intersection}: let $n$ be the size of the ground
set, $r$ be the maximum size of independent sets, $M$ be the maximum
absolute value of element weight, and $\mathcal{T_{\text{rank}}}$
and $\mathcal{T_{\text{ind}}}$ be the time for each rank and independence
oracle query.{\small \par}

We obtain a running time of $O(nr\mathcal{T_{\text{rank}}}\log n\log(nM)+n^{3}\log^{O(1)}nM)$
and $O(n^{2}\mathcal{T_{\text{ind}}}\log(nM)+n^{3}\log^{O(1)}nM)$,
achieving the first quadratic bound on the query complexity for the
independence and rank oracles. In the unweighted case, this is the
first improvement since 1986 for independence oracle.{\small \par}

\item \textbf{Submodular Flow}: let $n$ and $m$ be the number of vertices
and edges, $C$ be the maximum edge cost in absolute value, and $U$
be the maximum edge capacity in absolute value.{\small \par}

We obtain a faster weakly polynomial running time of $O(n^{2}\log nCU\cdot\EO+n^{3}\log^{O(1)}nCU)$,
improving upon the previous best of $O(mn^{5}\log nU\cdot\EO)$ and
$O\left(n^{4}h\min\left\{ \log C,\log U\right\} \right)$ from 15
years ago by a factor of $\tilde{O}(n^{4})$. We also achieve faster
strongly polynomial time algorithms as a consequence of our result
on submodular minimization.{\small \par}

\item \textbf{Semidefinite Programming}: let $n$ be the number of constraints,
$m$ be the number of dimensions and $S$ be the total number of non-zeros
in the constraint matrix.{\small \par}

We obtain a running time of $\tilde{O}(n(n^{2}+m^{\omega}+S))$, improving
upon the previous best of $\tilde{O}(n(n^{\omega}+m^{\omega}+S))$
for the regime $S$ is small.{\small \par}

\end{itemize}
\end{abstract}
\pagebreak{}

\tableofcontents{}\pagebreak{}

\setcounter{part}{-1}

\part{Overview}

\section{Introduction}

The ellipsoid method and more generally, \emph{cutting plane methods},%
\footnote{Throughout this paper our focus is on algorithms for polynomial time
solvable convex optimization problems given access to a linear separation
oracle. Our usage of the term \emph{cutting plane methods}, should
not be confused with work on integer programming, an NP-hard problem.%
} that is optimization algorithms which iteratively call a separation
oracle, have long been central to theoretical computer science. In
combinatorial optimization, since Khachiyan's seminal result in 1980
\cite{khachiyan1980polynomial} proving that the ellipsoid method
solves linear programs in polynomial time, the ellipsoid method has
been crucial to solving discrete problems in polynomial time \cite{grotschel1981ellipsoid}.
In continuous optimization, cutting plane methods have long played
a critical role in convex optimization, where they are fundamental
to the theory of non-smooth optimization \cite{goffin2002convex}. 

Despite the key role that cutting plane methods have played historically
in both combinatorial and convex optimization, over the past two decades
progress on improving both the theoretical running time of cutting
plane methods as well as the complexity of using cutting plane methods
for combinatorial optimization has stagnated.%
\footnote{There are exceptions to this trend. For example, \cite{krishnan2003properties}
showed how to apply cutting plane methods to yield running time improvements
for semidefinite programming, and recently \cite{bubeck2015geometric}
showed how to use cutting plane methods to obtain an optimal result
for smooth optimization problems.%
} The theoretical running time of cutting plane methods for convex
optimization has not been improved since the breakthrough result by
Vaidya in 1989 \cite{vaidya89convexSet,vaidya1996new}. Moreover,
for many of the key combinatorial applications of ellipsoid method,
such as submodular minimization, matroid intersection and submodular
flow, the running time improvements over the past two decades have
been primarily combinatorial; that is they have been achieved by discrete
algorithms that do not use numerical machinery such as cutting plane
methods. 

In this paper we make progress on these classic optimization problems
on two fronts. First we show how to improve on the running time of
cutting plane methods for a broad range of parameters that arise frequently
in both combinatorial applications and convex programming (Part~\ref{part:Ellipsoid}).
Second, we provide several frameworks for applying the cutting plane
method and illustrate the efficacy of these frameworks by obtaining
faster running times for semidefinite programming, matroid intersection,
and submodular flow (Part~\ref{part:app}). Finally, we show how
to couple our approach with the problem specific structure and obtain
faster weakly and strongly polynomial running times for submodular
function minimization, a problem of tremendous importance in combinatorial
optimization (Part \ref{part:Submodular-Function-Minimization}).
In both cases our algorithms are faster than previous best by a factor
of roughly $\Omega(n^{2})$.

We remark that many of our running time improvements come both from
our faster cutting method and from new careful analysis of how to
apply these cutting plane methods. In fact, simply using our reductions
to cutting plane methods and a seminal result of Vaidya \cite{vaidya89convexSet,vaidya1996new}
on cutting plane methods we provide running times for solving many
of these problems that improves upon the previous best stated. As
such, we organized our presentation to hopefully make it easy to apply
cutting plane methods to optimization problems and obtain provable
guarantees in the future. 

Our results demonstrate the power of cutting plane methods in theory
and possibly pave the way for new cutting plane methods in practice.
We show how cutting plane methods can continue to improve running
times for classic optimization problems and we hope that these methods
may find further use. As cutting plane methods such as analytic cutting
plane method \cite{goffin1996complexity,atkinson1995cutting,goffin1999shallow,nesterov1995complexity,ye1996complexity,goffin2002convex}
are frequently used in practice \cite{gondzio1996accpm,goffin1997solving},
these techniques may have further implications.

\subsection{Paper Organization}

After providing an overview of our results (Section~\ref{sec:intro:results})
and preliminary information and notation used throughout the paper
(Section~\ref{sec:intro:preliminaries}), we split the remainder
of the paper into three parts:
\begin{itemize}
\item In Part~\ref{part:Ellipsoid} we provide our new cutting plane method. 
\item In Part~\ref{part:app} we provide several general frameworks for
using this cutting plane method and illustrate these frameworks with
applications in combinatorial and convex optimization.
\item In Part~\ref{part:Submodular-Function-Minimization} we then consider
the more specific problem of submodular function minimization and
show how our methods can be used to improve the running time for both
strongly and weakly polynomial time algorithms.
\end{itemize}
We aim to make each part relatively self contained. While each part
builds upon the previous and the problems considered in each part
are increasingly specific, we present the key results in each section
in a modular way so that they may be read in any order. The dependencies
between the different parts of our paper are characterized by the
following:
\begin{itemize}
\item Part~\ref{part:Ellipsoid} presents our faster cutting plane method
as Theorem~\ref{thm:main_result}.
\item Part~\ref{part:app} depends only on Theorem~\ref{thm:main_result}
of Part~\ref{part:Ellipsoid} and presents a general running time
guarantee for convex optimization problems as Theorem \ref{thm:conv_opt}.
\item The faster weakly polynomial time algorithm in Part~\ref{part:Submodular-Function-Minimization}
depends only on Theorem~\ref{thm:conv_opt}, Part~\ref{part:app}.
\item The faster strongly polynomial time algorithm in Part~\ref{part:Submodular-Function-Minimization}
depends only on Theorem~\ref{thm:main_result}, Part~\ref{part:Ellipsoid}.
\end{itemize}

\section{Overview of Our Results\label{sec:intro:results}}

Here we briefly summarize the contributions of our paper. For each
of Part~\ref{part:Ellipsoid}, Part~\ref{part:app}, and Part~\ref{part:Submodular-Function-Minimization}
we describe the key technical contributions and present the running
time improvements achieved.

\subsection{Cutting Plane Methods \label{sec:intro:results:cutting}}

The central problem we consider in Part~\ref{part:Ellipsoid} is
as follows. We are promised that a set $K$ is contained a box of
radius $R$ and a separation oracle that given a point $\vx$ in time
$\SO$ either outputs that $\vx$ is in $K$ or outputs a separating
hyperplane. We wish to either find a point in $K$ or prove that $K$
does not contain an ball of radius $\epsilon$. The running times
for this problem are given in Table~\ref{tab:ellip_run-1}.

\begin{table}[h]
\centering{}%
\begin{tabular}{|c|l|c|}
\hline 
Year  & Algorithm & Complexity\tabularnewline
\hline 
\hline 
1979 & Ellipsoid Method \cite{shor1977cut,yudin1976evaluation,khachiyan1980polynomial} & $O(n^{2}\SO\log\kappa+n^{4}\log\kappa)$ \tabularnewline
\hline 
1988 & Inscribed Ellipsoid \cite{khachiyan1988method,nesterov1989self} & $O(n\SO\log\kappa+\left(n\log\kappa\right)^{4.5})$ \tabularnewline
\hline 
1989 & Volumetric Center \cite{vaidya89convexSet} & $O(n\SO\log\kappa+n^{1+\omega}\log\kappa)$ \tabularnewline
\hline 
1995 & Analytic Center \cite{atkinson1995cutting} & $O(n\SO\log^{2}\kappa+n^{\omega+1}\log^{2}\kappa+\left(n\log\kappa\right)^{2+\omega/2})$ \tabularnewline
\hline 
2004 & Random Walk \cite{bertsimas2004solving} & $\rightarrow O(n\SO\log\kappa+n^{7}\log\kappa)$\tabularnewline
\hline 
2013 & This paper  & $O(n\SO\log\kappa+n^{3}\log^{O(1)}\kappa)$ \tabularnewline
\hline 
\end{tabular}\protect\caption{\label{tab:ellip_run-1}Algorithms for the Feasibility Problem. $\kappa$
indicates $nR/\epsilon$. The arrow, $\rightarrow$, indicates that
it solves a more general problem where only a membership oracle is
given.}
\end{table}

In Part~\ref{part:Ellipsoid} we show how to solve this problem in
$O(n\SO\log(nR/\epsilon)+n^{3}\log^{O(1)}(nR/\epsilon))$ time. This
is an improvement over the previous best running time of $\otilde(n\SO\log(nR/\epsilon)+n^{\omega+1}\log(nR/\epsilon))$
for the current best known bound of $\omega<2.37$ \cite{gall2014powers}
assuming that $R/\epsilon=O(\poly(n))$, a common assumption for many
problems in combinatorial optimization and numerical analysis as we
find in Part~\ref{part:app} and Part~\ref{part:Submodular-Function-Minimization}.
(See Table~\ref{tab:ellip_run-1} for a summary of previous running
times.)

Our key idea for achieving this running time improvement is a new
straightforward technique for providing low variance unbiased estimates
for changes in leverage scores that we hope will be of independent
interest (See Section~\ref{sec:est_leveragescore}). We show how
to use this technique along with ideas from \cite{atkinson1995cutting,vaidya1989speeding,lee2015efficient}
to decrease the $\otilde(n^{\omega+1}\log(D/\epsilon))$ overhead
in the previous fastest algorithm \cite{vaidya89convexSet}.

\subsection{Convex Optimization\label{sec:intro:results:opt}}

In Part~\ref{part:app} we provide two techniques for applying our
cutting plane method (and cutting plane methods in general) to optimization
problems and provide several applications of these techniques.

The first technique concerns reducing the number of dimensions through
duality. For many problems, their dual is significantly simpler than
itself (primal). We use semidefinite programming as a concrete example
to show how to improve upon the running time for finding both primal
and dual solution by using the cutting planes maintained by our cutting
plane method. (See Table~\ref{tab:Algorithms-for-solving}.)

The second technique concerns how to minimize a linear function over
the intersection of convex sets using optimization oracle. We analyze
a simple potential function which allows us to bypass the typical
reduction between separation and optimization to achieve faster running
times. This reduction provides an improvement over the reductions
used previously in \cite{grotschel1981ellipsoid}. Moroever, we show
how this technique allows us to achieve improved running times for
matroid intersection and minimum cost submodular flow. (See Tables~\ref{tab:Algorithms-for-solving},
\ref{tab:Algorithms-for-(unweighted)}, \ref{tab:Previous-algorithms-for},
and \ref{tab:Algorithms-for-minimum} for running time summaries.)

\begin{table}[h]
\begin{centering}
\begin{tabular}{|c|c|c|}
\hline 
Authors & Years & Running times\tabularnewline
\hline 
\hline 
Nesterov, Nemirovsky\cite{nesterov1992conic} & 1992 & $\tilde{O}(\sqrt{m}(nm^{\omega}+n^{\omega-1}m^{2}))$\tabularnewline
\hline 
Anstreicher \cite{anstreicher2000volumetric} & 2000 & $\tilde{O}((mn)^{1/4}(nm^{\omega}+n^{\omega-1}m^{2}))$ \tabularnewline
\hline 
Krishnan, Mitchell \cite{krishnan2003properties} & 2003 & $\tilde{O}(n(n^{\omega}+m^{\omega}+S))$ (dual SDP)\tabularnewline
\hline 
\textbf{This paper} & 2015 & $\tilde{O}(n(n^{2}+m^{\omega}+S))$\tabularnewline
\hline 
\end{tabular}
\par\end{centering}

\protect\caption{\label{tab:Algorithms-for-solving}Algorithms for solving a $m\times m$
SDP with $n$ constraints and $S$ non-zero entries}
\end{table}

\begin{table}[h]
\begin{centering}
\begin{tabular}{|c|c|c|}
\hline 
Authors & Years & Complexity\tabularnewline
\hline 
\hline 
Edmonds \cite{edmonds1968matroid} & 1968 & not stated\tabularnewline
\hline 
Aigner, Dowling \cite{aigner1971matching} & 1971 & $O(nr^{2}\mathcal{T_{\text{ind}}})$\tabularnewline
\hline 
Tomizawa, Iri \cite{tomizawa1974algorithm} & 1974 & not stated\tabularnewline
\hline 
Lawler \cite{lawler1975matroid} & 1975 & $O(nr^{2}\mathcal{T_{\text{ind}}})$\tabularnewline
\hline 
Edmonds \cite{edmonds1979matroid} & 1979 & not stated\tabularnewline
\hline 
Cunningham \cite{cunningham1986improved} & 1986 & $O(nr^{1.5}\mathcal{T_{\text{ind}}})$\tabularnewline
\hline 
\textbf{This paper} & 2015 & $\begin{array}{c}
O(n^{2}\log n\mathcal{T_{\text{ind}}}+n^{3}\log^{O(1)}n)\\
O(nr\log^{2}n\mathcal{T_{\text{rank}}}+n^{3}\log^{O(1)}n)
\end{array}$\tabularnewline
\hline 
\end{tabular}
\par\end{centering}

\protect\caption{\label{tab:Algorithms-for-(unweighted)}Algorithms for (unweighted)
matroid intersection. $n$ is the size of the ground set, $r$ is
the maximum rank of the two matroids, $\mathcal{T_{\text{ind}}}$
is the time to check if a set is independent (membership oracle),
and $\mathcal{T_{\text{rank}}}$ is the time to compute the rank of
a given set (rank oracle).}
\end{table}

\begin{table}[h]
\begin{centering}
\begin{tabular}{|c|c|c|}
\hline 
Authors & Years & Running times\tabularnewline
\hline 
\hline 
Edmonds \cite{edmonds1968matroid} & 1968 & not stated\tabularnewline
\hline 
Tomizawa, Iri \cite{tomizawa1974algorithm} & 1974 & not stated\tabularnewline
\hline 
Lawler \cite{lawler1975matroid} & 1975 & $O(nr^{2}\mathcal{T_{\text{ind}}}+nr^{3})$\tabularnewline
\hline 
Edmonds \cite{edmonds1979matroid} & 1979 & not stated\tabularnewline
\hline 
Frank \cite{frank1981weighted} & 1981 & $O(n^{2}r(\mathcal{T_{\text{circuit}}}+n))$\tabularnewline
\hline 
Orlin, Ahuja \cite{orlin1983primal} & 1983 & not stated\tabularnewline
\hline 
Brezovec, Cornuéjols, Glover\cite{brezovec1986two} & 1986 & $O(nr(\mathcal{T_{\text{circuit}}}+r+\log n))$\tabularnewline
\hline 
Fujishige, Zhang \cite{fujishige1995efficient} & 1995 & $O(n^{2}r^{0.5}\log rM\cdot\mathcal{T_{\text{ind}}})$\tabularnewline
\hline 
Shigeno, Iwata \cite{shigeno1995dual} & 1995 & $O((n+\mathcal{T_{\text{circuit}}})nr^{0.5}\log rM)$\tabularnewline
\hline 
\textbf{This paper} & 2015 & $\begin{array}{c}
O(n^{2}\log nM\mathcal{T_{\text{ind}}}+n^{3}\log^{O(1)}nM)\\
O(nr\log n\log nM\mathcal{T_{\text{rank}}}+n^{3}\log^{O(1)}nM)
\end{array}$\tabularnewline
\hline 
\end{tabular}
\par\end{centering}

\protect\caption{\label{tab:Previous-algorithms-for} Algorithms for weighted matroid
intersection. In addition to the notation in Table~\ref{tab:Algorithms-for-(unweighted)}
$\mathcal{T_{\text{circuit}}}$ is the time needed to find a fundamental
circuit and $M$ is the bit complexity of the weights.}
\end{table}

\begin{table}[H]
\begin{centering}
\begin{tabular}{|c|c|c|}
\hline 
Authors & Years & Running times\tabularnewline
\hline 
\hline 
Fujishige \cite{1978Fujishige} & 1978 & not stated\tabularnewline
\hline 
Grotschel, Lovasz, Schrijver{\small{}\cite{grotschel1981ellipsoid}} & 1981 & weakly polynomial\tabularnewline
\hline 
Zimmermann \cite{zimmermann1982minimization} & 1982 & not stated\tabularnewline
\hline 
Barahona, Cunningham \cite{barahona1984submodular} & 1984 & not stated\tabularnewline
\hline 
Cunningham, Frank \cite{cunningham1985primal} & 1985 & $\rightarrow O(n^{4}h\log C)$\tabularnewline
\hline 
Fujishige \cite{fujishige1987out} & 1987 & not stated\tabularnewline
\hline 
Frank, Tardos \cite{frank1987application} & 1987 & strongly polynomial\tabularnewline
\hline 
Cui, Fujishige \cite{1988Fujishige} & 1988 & not stated\tabularnewline
\hline 
Fujishige, Röck, Zimmermann\cite{fujishige1989strongly} & 1989 & $\rightarrow O(n^{6}h\log n)$\tabularnewline
\hline 
Chung, Tcha \cite{chung1991dual} & 1991 & not stated\tabularnewline
\hline 
Zimmermann \cite{zimmermann1992negative} & 1992 & not stated\tabularnewline
\hline 
McCormick, Ervolina \cite{mccormick1993canceling} & 1993 & $O(n^{7}h^{*}\log nCU)$\tabularnewline
\hline 
Wallacher, Zimmermann \cite{wallacher1999polynomial} & 1994 & $O(n^{8}h\log nCU)$\tabularnewline
\hline 
Iwata \cite{iwata1997capacity} & 1997 & $O(n^{7}h\log U)$\tabularnewline
\hline 
Iwata, McCormick, Shigeno \cite{iwata1998faster} & 1998 & $O\left(n^{4}h\min\left\{ \log nC,n^{2}\log n\right\} \right)$\tabularnewline
\hline 
Iwata, McCormick, Shigeno \cite{iwata1999strongly} & 1999 & $O\left(n^{6}h\min\left\{ \log nU,n^{2}\log n\right\} \right)$\tabularnewline
\hline 
Fleischer, Iwata, McCormick\cite{fleischer2002faster} & 1999 & $O\left(n^{4}h\min\left\{ \log U,n^{2}\log n\right\} \right)$\tabularnewline
\hline 
Iwata, McCormick, Shigeno \cite{iwata2000fast} & 1999 & $O\left(n^{4}h\min\left\{ \log C,n^{2}\log n\right\} \right)$\tabularnewline
\hline 
Fleischer, Iwata \cite{fleischer2000improved} & 2000 & $O(mn^{5}\log nU\cdot\EO)$\tabularnewline
\hline 
\textbf{This paper} & 2015 & $O(n^{2}\log nCU\cdot\EO+n^{3}\log^{O(1)}nCU)$\tabularnewline
\hline 
\end{tabular}
\par\end{centering}

\protect\caption{\label{tab:Algorithms-for-minimum}Algorithms for minimum cost submodular
flow with $n$ vertices, maximum cost $C$ and maximum capacity $U$.
The factor $h$ is the time for an exchange capacity oracle, $h^{*}$
is the time for a ``more complicated exchange capacity oracle,''
and $\protect\EO$ is the time for evaluation oracle of the submodular
function. The arrow, $\rightarrow$, indicates that it uses the current
best submodular flow algorithm as subroutine which was non-existent
at the time of the publication.}
\end{table}

\subsection{Submodular Function Minimization\label{sec:intro:results:submodular}}

In Part~\ref{part:Submodular-Function-Minimization} we consider
the problem of submodular minimization, a fundamental problem in combinatorial
optimization with many diverse applications in theoretical computer
science, operations research, machine learning and economics. We show
that by considering the interplay between the guarantees of our cutting
plane algorithm and the primal-dual structure of submodular minimization
we can achieve improved running times in various settings. 

First, we show that a direct application of our method yields an improved
weakly polynomial time algorithm for submodular minimization. Then,
we present a simple geometric argument that submodular function can
be solved with $O(n^{3}\log n\cdot\text{EO})$ oracle calls but with
exponential running time. Finally, we show that by further studying
the combinatorial structure of submodular minimization and a modification
to our cutting plane algorithm we can obtained a fully improved strongly
polynomial time algorithm for submodular minimization. We summarize
the improvements in Table~\ref{fig:SFM_table_1}.

\begin{center}
\begin{table}[h]
\begin{centering}
\begin{tabular}{|c|c|c|c|}
\hline 
Authors & Years & Running times & Remarks\tabularnewline
\hline 
\hline 
Grötschel, Lovász,  & \multirow{2}{*}{1981,1988} & \multirow{2}{*}{$\widetilde{O}(n^{5}\cdot\text{EO}+n^{7})$ \cite{mccormicksurvey}} & first weakly\tabularnewline
Schrijver \cite{grotschel1981ellipsoid,grotschel1988ellipsoid} &  &  &  and strongly\tabularnewline
\hline 
Cunningham \cite{cunningham1985submodular} & 1985 & $O(Mn^{6}\log nM\cdot\text{EO})$ & first combin. pseudopoly\tabularnewline
\hline 
Schrijver \cite{schrijver2000combinatorial} & 2000 & $O(n^{8}\cdot\text{EO}+n^{9})$ & first combin. strongly\tabularnewline
\hline 
Iwata, Fleischer,  & \multirow{2}{*}{2000} & \multirow{2}{*}{$\begin{array}{c}
O(n^{5}\cdot\text{EO}\log M)\\
O(n^{7}\log n\cdot\text{EO})
\end{array}$} & \multirow{2}{*}{first combin. strongly}\tabularnewline
Fujishige\cite{iwata2001combinatorial} &  &  & \tabularnewline
\hline 
Iwata, Fleischer \cite{fleischer2003push} & 2000 & $O(n^{7}\cdot\text{EO}+n^{8})$ & \tabularnewline
\hline 
Iwata \cite{iwata2003faster} & 2003 & $\begin{array}{c}
O((n^{4}\cdot\text{EO}+n^{5})\log M)\\
O((n^{6}\cdot\text{EO}+n^{7})\log n)
\end{array}$ & current best weakly\tabularnewline
\hline 
Vygen \cite{vygen2003note} & 2003 & $O(n^{7}\cdot\text{EO}+n^{8})$ & \tabularnewline
\hline 
Orlin \cite{orlin2009faster} & 2007 & $O(n^{5}\cdot\text{EO}+n^{6})$ & current best strongly\tabularnewline
\hline 
Iwata, Orlin \cite{iwata2009simple} & 2009 & $\begin{array}{c}
O((n^{4}\cdot\text{EO}+n^{5})\log nM)\\
O((n^{5}\cdot\text{EO}+n^{6})\log n)
\end{array}$ & \tabularnewline
\hline 
\textbf{Our algorithms} & 2015 & $\begin{array}{c}
O(n^{2}\log nM\cdot\text{EO}+n^{3}\log^{O(1)}nM)\\
O(n^{3}\log^{2}n\cdot\text{EO}+n^{4}\log^{O(1)}n)
\end{array}$ & \tabularnewline
\hline 
\end{tabular}\,
\par\end{centering}

\protect\caption{\label{fig:SFM_table_1}Algorithms for submodular function minimization.}
\end{table}

\par\end{center}

\section{Preliminaries \label{sec:intro:preliminaries}}

Here we introduce notation and concepts we use throughout the paper.

\subsection{Notation\label{sec:intro:preliminaries:notation}}

$\textbf{Basics:}$ Throughout this paper, we use vector notation,
e.g $\vec{x}=(x_{1},\dots,x_{n})$, to denote a vector and bold, e.g.
$\ma$, to denote a matrix. We use $\nnz(\vx)$ or $\nnz(\ma)$ to
denote the number of nonzero entries in a vector or a matrix respectively.
Frequently, for $\vx\in\R^{d}$ we let $\mx\in\R^{d\times d}$ denote
$\mDiag(\vx)$, the diagonal matrix such that $\mx_{ii}=x_{i}$. For
a symmetric matrix, $\mm$, we let $\diag(\mm)$ denote the vector
corresponding to the diagonal entries of $\mm$, and for a vector,
$\vx$, we let $\|\vx\|_{\mm}\defeq\sqrt{\vx^{T}\mm\vx}$. \\\\$\textbf{Running Times:}$
We typically use $\XO$ to denote the running time for invoking the
oracle, where $\oracleFont X$ depends on the type of oracle, e.g.,
$\SO$ typically denotes the running time of a separation oracle,
$\EO$ denotes the running time of an evaluation oracle, etc. Furthermore,
we use $\otilde(f)\defeq O(f\log^{O(1)}f)$. \\\\\textbf{ Spectral
Approximations:} For symmetric matrices $\mn,\mm\in\R^{n\times n}$,
we write $\mn\preceq\mm$ to denote that $\vx^{T}\mn\vx\leq\vx^{T}\mm\vx$
for all $\vx\in\R^{n}$ and we define $\mn\succeq\mm$, $\mn\prec\mm$
and $\mn\succ\mm$ analogously. \\\\\textbf{ Standard Convex Sets:}
We let $B_{p}(r)\defeq\{\vx\,:\,\norm{\vx}_{p}\leq r\}$ denote a
ball of radius $r$ in the $\ellP$ norm. For brevity we refer to
$B_{2}(r)$ as a a \emph{ball of radius $r$ }and $B_{\infty}(r)$
as a \emph{box of radius $r$}. \\\\ \textbf{Misc:} We let $\omega<2.373$
\cite{williams2012matrixmult} denote the matrix multiplication constant.

\subsection{Separation Oracles\label{sec:intro:preliminaries:separation-oracles}}

Throughout this paper we frequently make assumptions about the existence
of separation oracles for sets and functions. Here we formally define
these objects as we use them throughout the paper. Our definitions
are possibly non-standard and chosen to handle the different settings
that occur in this paper.
\begin{defn}[Separation Oracle for a Set]
\label{def:weak_sep}  Given a set $K\subset\Rn$ and $\delta\geq0$,
a\emph{ $\delta$-separation oracle }for $K$ is a function on $\Rn$
such that for any input $\vx\in\Rn$, it either outputs ``successful''
or a half space of the form $H=\{\vz:\vc^{T}\vz\leq\vc^{T}\vx+b\}\supseteq K$
with $b\leq\delta\norm{\vc}_{2}$ and $\vc\neq\vzero$. We let $\SO_{\delta}(K)$
be the time complexity of this oracle.
\end{defn}
For brevity we refer to a $0$-separation oracle for a set as just
a \emph{separation oracle. }We refer to the hyperplanes defining the
halfspaces returned by a $\delta$-separation oracle as \emph{oracle
hyperplanes}.

Note that in Definition~\ref{def:weak_sep} we do not assume that
$K$ is convex. However, we remark that it is well known that there
is a separation oracle for a set if and only if it is convex and that
there is a $\delta$ separation oracle if and only if the set is close
to convex in some sense.
\begin{defn}[Separation Oracle for a Function]
\label{def:weak_sep2} For any convex function $f$, $\eta\geq0$
and $\delta\geq0$, a $(\eta,\delta)$-separation oracle on a convex
set $\Gamma$ for $f$ is a function on $\Rn$ such that for any input
$\vx\in\Gamma$, it either asserts $f(\vx)\leq\min_{\vy\in\Gamma}f(\vy)+\eta$
or outputs a half space $H$ such that 
\begin{equation}
\{\vz\in\Gamma:f(\vz)\leq f(\vx)\}\subset H\defeq\{\vz:\vc^{T}\vz\leq\vc^{T}\vx+b\}\label{eq:weak_sep2_guarantee}
\end{equation}
with $b\leq\delta\norm{\vc}$ and $\vc\neq\vzero$. We let $\SO_{\eta,\delta}(f)$
be the time complexity of this oracle. \end{defn}

\pagebreak{}
\part{A Faster Cutting Plane Method\label{part:Ellipsoid}}

\section{Introduction}

Throughout Part~\ref{part:Ellipsoid} we study the following \emph{feasibility}
\emph{problem}: 
\begin{defn}[Feasibility Problem]
 Given a separation oracle for a set $K\subseteq\R^{n}$ contained
in a box of radius $R$ either find a point $\vx\in K$ or prove that
$K$ does not contain a ball of radius $\epsilon$.
\end{defn}
This feasibility problem is one of the most fundamental and classic
problems in optimization. Since the celebrated result of Yudin and
Nemirovski \cite{yudin1976evaluation} in 1976 and Khachiyan \cite{khachiyan1980polynomial}
in 1979 essentially proving that it can be solved in time $O(\poly(n)\cdot\SO\cdot\log(R/\epsilon))$,
this problem has served as one of the key primitives for solving numerous
problems in both combinatorial and convex optimization. 

Despite the prevalence of this feasibility problem, the best known
running time for solving this problem has not been improved in over
25 years. In a seminal paper of Vaidya in 1989 \cite{vaidya89convexSet},
he showed how to solve the problem in $\otilde(n\cdot\SO\cdot\log(nR/\epsilon)+n^{\omega+1}\log(nR/\epsilon))$
time. Despite interesting generalizations and practical improvements
\cite{anstreicher1997vaidya,ramaswamy1995long,goffin1996complexity,atkinson1995cutting,goffin1999shallow,nesterov1995complexity,ye1996complexity,goffin2002convex,bubeck2015geometric},
the best theoretical guarantees for solving this problem have not
been improved since.

In Part~\ref{part:Ellipsoid} we show how to improve upon Vaidya's
running time in certain regimes. We provide a cutting plane algorithm
which achieves an expected running time of $O(n\cdot\SO\cdot\log(nR/\epsilon)+n^{3}\log^{O(1)}(nR/\epsilon))$,
improving upon the previous best known running time for the current
known value of $\omega<2.373$ \cite{williams2012matrixmult,gall2014powers}
when $R/\epsilon=O(\poly(n))$. 

We achieve our results by the combination of multiple techniques.
First we show how to use techniques from the work of Vaidya and Atkinson
to modify Vaidya's scheme so that it is able to tolerate random noise
in the computation in each iteration. We then show how to use known
numerical machinery \cite{vaidya1989speeding,spielmanS08sparsRes,lee2015efficient}
in combination with some new techniques (Section~\ref{sec:est_leveragescore}
and Section~\ref{sec:chasing_0}) to implement each of these relaxed
iterations efficiently. We hope that both these numerical techniques
as well as our scheme for approximating complicated methods, such
as Vaidya's, may find further applications.

While our paper focuses on theoretical aspects of cutting plane methods,
we achieve our results via the careful application of practical techniques
such as dimension reduction and sampling. As such we hope that ideas
in this paper may lead to improved practical%
\footnote{Although cutting plane methods are often criticized for their empirical
performance, recently, Bubeck, Lee and Singh \cite{bubeck2015geometric}
provided a variant of the ellipsoid method that achieves the same
convergence rate as Nesterov\textquoteright s accelerated gradient
descent. Moreover, they provided numerical evidence that this method
can be superior to Nesterov's accelerated gradient descent, thereby
suggesting that cutting plane methods can be as aggressive as first
order methods if designed properly.%
} algorithms for non-smooth optimization.

\subsection{Previous Work}

\label{sub:previous_work}

Throughout this paper, we restrict our attention to algorithms for
the feasibility problem that have a polynomial dependence on $\SO$,
$n$, and $\log(R/\epsilon)$. Such ``efficient'' algorithms typically
follow the following iterative framework. First, they compute some
trivial region $\Omega$ that contains $K$. Then, they call the separation
oracle at some point $\vx\in\Omega$. If $\vx\in K$ the algorithm
terminates having successfully solved the problem. If $\vx\notin K$
then the separation oracle must return a half-space containing $K$.
The algorithm then uses this half-space to shrink the region $\Omega$
while maintaining the invariant that $K\subseteq\Omega$. The algorithm
then repeats this process until it finds a point $\vx\in K$ or the
region $\Omega$ becomes too small to contain a ball with radius $\epsilon$.

Previous works on efficient algorithms for the feasibility problem
all follow this iterative framework. They vary in terms of what set
$\Omega$ they maintain, how they compute the center to query the
separation oracle, and how they update the set. In Table~\ref{tab:ellip_run},
we list the previous running times for solving the feasibility problem.
As usual $\SO$ indicates the cost of the separation oracle. To simplify
the running times we let $\kappa\defeq nR/\epsilon$. The running
times of some algorithms in the table depend on $R/\epsilon$ instead
of $nR/\epsilon$. However, for many situations, we have $\log(R/\epsilon)=\Theta(\log(nR/\epsilon))$
and hence we believe this is still a fair comparison. 

The first efficient algorithm for the feasibility problem is the \emph{ellipsoid
}method, due to Shor \cite{shor1977cut}, Nemirovksii and Yudin \cite{yudin1976evaluation},
and Khachiyan \cite{khachiyan1980polynomial}. The ellipsoid method
maintains an ellipsoid as $\Omega$ and uses the center of the ellipsoid
as the next query point. It takes $\Theta(n^{2}\log\kappa)$ calls
of oracle which is far from the lower bound $\Omega(n\log\kappa)$
calls \cite{NemirovskyA.S.&Yudin1983}.

To alleviate the problem, the algorithm could maintain all the information
from the oracle, i.e., the polytope created from the intersection
of all half-spaces obtained. The center of gravity method \cite{levin}
achieves the optimal oracle complexity using this polytope and the
center of gravity of this polytope as the next point. However, computing
center of gravity is computationally expensive and hence we do not
list its running time in Table~\ref{tab:ellip_run}. The Inscribed
Ellipsoid Method \cite{khachiyan1988method} also achieved an optimal
oracle complexity using this polytope as $\Omega$ but instead using
the center of the maximal inscribed ellipsoid in the polytope to query
the separation oracle. We listed it as occurring in year 1988 in Table~\ref{tab:ellip_run}
because it was \cite{nesterov1989self} that yielded the first polynomial
time algorithm to actually compute this maximal inscribed ellipsoid
for polytope. 

Vaidya{\footnotesize{}~\cite{vaidya89convexSet}} obtained a faster
algorithm by maintaining an approximation of this polytope and using
a different center, namely the volumetric center. Although the oracle
complexity of this volumetric center method is very good, the algorithm
is not extremely efficient as each iteration involves matrix inversion.
Atkinson and Vaidya \cite{atkinson1995cutting} showed how to avoid
this computation in certain settings. However, they were unable to
achieve the desired convergence rate from their method. 

Bertsimas and Vempala \cite{bertsimas2004solving} also gives an algorithm
that avoids these expensive linear algebra operations while maintaining
the optimal convergence rate by using techniques in sampling convex
sets. Even better, this result works for a much weaker oracle, the
membership oracle. However, the additional cost of this algorithm
is relatively high in theory. We remark that while there are considerable
improvemenst on the sampling techniques \cite{lovaszV06,kannan2012random,lee2015efficient},
the additional cost is still quite high compared to standard linear
algebra.

\begin{table}
\centering{}%
\begin{tabular}{|c|l|c|}
\hline 
Year  & Algorithm & Complexity\tabularnewline
\hline 
\hline 
1979 & Ellipsoid Method \cite{shor1977cut,yudin1976evaluation,khachiyan1980polynomial} & $O(n^{2}\SO\log\kappa+n^{4}\log\kappa)$ \tabularnewline
\hline 
1988 & Inscribed Ellipsoid \cite{khachiyan1988method,nesterov1989self} & $O(n\SO\log\kappa+\left(n\log\kappa\right)^{4.5})$ \tabularnewline
\hline 
1989 & Volumetric Center \cite{vaidya89convexSet} & $O(n\SO\log\kappa+n^{1+\omega}\log\kappa)$ \tabularnewline
\hline 
1995 & Analytic Center \cite{atkinson1995cutting} & $O\left(\begin{array}{c}
n\SO\log^{2}\kappa+n^{\omega+1}\log^{2}\kappa\\
+\left(n\log\kappa\right)^{2+\omega/2}
\end{array}\right)$ \tabularnewline
\hline 
2004 & Random Walk \cite{bertsimas2004solving} & $\rightarrow O(n\SO\log\kappa+n^{7}\log\kappa)$\tabularnewline
\hline 
2013 & This paper  & $O(n\SO\log\kappa+n^{3}\log^{O(1)}\kappa)$ \tabularnewline
\hline 
\end{tabular}\protect\caption{\label{tab:ellip_run}Algorithms for the Feasibility Problem. $\kappa$
indicates $nR/\epsilon$. The arrow, $\rightarrow$, indicates that
it solves a more general problem where only a membership oracle is
given.}
\end{table}

\subsection{Challenges in Improving Previous Work}

Our algorithm builds upon the previous fastest algorithm of Vaidya
\cite{vaidya1996new}. Ignoring implementation details and analysis,
Vaidya's algorithm is quite simple. This algorithm simply maintains
a polytope $P^{(k)}=\{x\in\R^{n}\,:\,\ma\vx-\vb\geq\vzero\}$ as the
current $\Omega$ and uses the \emph{volumetric center}, the minimizer
of the following \emph{volumetric barrier function}
\begin{equation}
\argmin_{\vx}\frac{1}{2}\log\det\left(\ma^{T}\ms_{\vx}^{-2}\ma\right)\enspace\text{ where }\enspace\ms_{\vx}\defeq\mDiag(\ma\vx-\vb)\label{eq:volumetric_center}
\end{equation}
as the point at which to query the separation oracle. The polytope
is then updated by adding shifts of the half-spaces returned by the
separation oracle and dropping unimportant constraints. By choosing
the appropriate shift, picking the right rule for dropping constraints,
and using Newton's method to compute the volumetric center he achieved
a running time of $O(n\cdot SO\cdot\log\kappa+n^{1+\omega}\log\kappa)$. 

While Vaidya's algorithm's dependence on $\SO$ is essentially optimal,
the additional per-iteration costs of his algorithm could possibly
be improved. The computational bottleneck in each iteration of Vaidya's
algorithm is computing the gradient of $\log\det$ which in turn involves
computing the leverage scores $\vsigma(\vx)\defeq\diag(\ms_{x}^{-1}\ma\left(\ma^{T}\ms_{x}^{-2}\ma\right)^{-1}\ma^{T}\ms_{x}^{-1})$,
a commonly occurring quantity in numerical analysis and convex optimization
\cite{spielmanS08sparsRes,Cohen2014,li2012iterative,lee2015efficient,leeS14}.
As the best known algorithms for computing leverage scores exactly
in this setting take time $O(n^{\omega})$, directly improving the
running time of Vaidya's algorithm seems challenging.

However, since an intriguing result of Spielman and Srivastava in
2008 \cite{spielmanS08sparsRes}, it has been well known that using
Johnson-Lindenstrauss transform these leverage scores can be computed
up to a multiplicative $(1\pm\epsilon)$ error by solving $O(\epsilon^{-2}\log n)$
linear systems involving $\ma^{T}\ms_{x}^{-2}\ma$. While in general
this still takes time $O(\epsilon^{-2}n^{\omega})$, there are known
techniques for efficiently maintaining the inverse of a matrix so
that solving linear systems take amortized $O(n^{2})$ time \cite{vaidya1989speeding,leeS14,lee2015efficient}.
Consequently if it could be shown that computing \textit{approximate}
leverage scores sufficed, this would potentially decrease the amortized
cost per iteration of Vaidya's method. 

Unfortunately, Vaidya's method does not seem to tolerate this type
of multiplicative error. If leverage scores were computed this crudely
then in using them to compute approximate gradients for \eqref{eq:volumetric_center},
it seems that any point computed would be far from the true center.
Moreover, without being fairly close to the true volumetric center,
it is difficult to argue that such a cutting plane method would make
sufficient progress.

To overcome this issue, it is tempting to directly use recent work
on improving the running time of linear program \cite{leeS14}. In
this work, the authors faced a similar issue where a volumetric, i.e.
$\log\det$, potential function had the right analytic and geometric
properties, however was computational expensive to minimize. To overcome
this issue the authors instead computed a weighted analytic center:
\[
\argmin_{\vx}-\sum_{i\in[m]}w_{i}\log s_{i}(\vx)\enspace\text{ where }\enspace\vs(\vx)\defeq\ma\vx-\vb\enspace.
\]
For carefully chosen weights this center provides the same convergence
guarantees as the volumetric potential function, while each step can
be computed by solving few linear systems (rather than forming the
matrix inverse).

Unfortunately, it is unclear how to directly extend the work in \cite{leeS14}
on solving an explicit linear program to the feasibility problem specified
by a separation oracle. While it is possible to approximate the volumetric
barrier by a weighted analytic center in many respects, proving that
this approximation suffices for fast convergence remains open. In
fact, the volumetric barrier function as used in Vaidya's algorithm
is well approximated simply by the standard analytic center
\[
\argmin_{\vx}-\sum_{i\in[m]}\log s_{i}(\vx)\enspace\text{ where }\enspace\vs(\vx)\defeq\ma\vx-\vb\enspace.
\]
as all the unimportant constraints are dropped during the algorithm.
However, despite decades of research, the best running times known
for solving the feasibility problem using the analytic center are
Vaidya and Atkinson algorithm from 1995 \cite{atkinson1995cutting}.
While the running time of this algorithm could possibly be improved
using approximate leverage score computations and amortized efficient
linear system solvers, unfortunately at best, without further insight
this would yield an algorithm which requires a suboptimal $O(n\log^{O(1)}\kappa)$
queries to the separation oracle.

As pointed out in \cite{atkinson1995cutting}, the primary difficulty
in using any sort of analytic center is quantifying the amount of
progress made in each step. We still believe providing direct near-optimal
analysis of weighted analytic center is a tantalizing open question
warranting further investigation. However, rather than directly address
the question of the performance of weighted analytic centers for the
feasibility problem, we take a slightly different approach that side-steps
this issue. We provide a partial answer that still sheds some light
on the performance of the weighted analytic center while still providing
our desired running time improvements.

\subsection{Our Approach\label{sec:ellipsoid:intro:our_approach}}

To overcome the shortcoming of the volumetric and analytic centers
we instead consider a hybrid barrier function
\[
\argmin_{\vx}-\sum_{i\in[m]}w_{i}\log s_{i}(\vx)+\log\det(\ma^{T}\ms_{x}^{-1}\ma)\enspace\text{ where }\enspace\vs(\vx)\defeq\ma\vx-\vb\enspace.
\]
for careful chosen weights. Our key observation is that for correct
choice of weights, we can compute the gradient of this potential function.
In particular if we let $\vw=\vtau-\vsigma(\vx)$ then the gradient
of this potential function is the same as the gradients of $\sum_{i\in[m]}\tau_{i}\log s_{i}(\vx)$,
which we can compute efficiently. Moreover, since we are using $\log\det$,
we can use analysis similar to Vaidya's algorithm \cite{vaidya89convexSet}
to analyze the convergence rate of this algorithm. 

Unfortunately, this is a simple observation and does not immediately
change the problem substantially. It simply pushes the problem of
computing gradients of $\log\det$ to computing $\vWeight$. Therefore,
for this scheme to work, we would need to ensure that the weights
do not change too much and that when they change, they do not significantly
hurt the progress of our algorithm. In other words, for this scheme
to work, we would still need very precise estimates of leverage scores.

However, we note that the leverage scores $\vsigma(\vx)$ do not change
too much between iterations. Moreover, we provide what we believe
is an interesting technical result that an unbiased estimates to the
changes in leverage scores can be computed using linear system solvers
such that the \emph{total error} of the estimate is bounded by the
total change of the leverage scores (See Section~\ref{sec:lever_change}).
Using this result our scheme simply follows Vaidya's basic scheme
in \cite{vaidya89convexSet}, however instead of minimizing the hybrid
barrier function directly we alternate between taking Newton steps
we can compute, changing the weights so that we can still compute
Newton steps, and computing accurate unbiased estimates of the changes
in the leverage scores so that the weights do not change adversarially
by too much. 

To make this scheme work, there are two additional details that need
to be dealt with. First, we cannot let the weights vary too much as
this might ultimately hurt the rate of progress of our algorithm.
Therefore, in every iteration we compute a single leverage score to
high precision to control the value of $w_{i}$ and we show that by
careful choice of the index we can ensure that no weight gets too
large (See Section~\ref{sec:chasing_0}).

Second, we need to show that changing weights does not affect our
progress by much more than the progress we make with respect to $\log\det$.
To do this, we need to show the slacks are bounded above and below.
We enforce this by adding regularization terms and instead consider
the potential function
\[
p_{\ve}(\vx)=-\sum_{i\in[m]}w_{i}\log s_{i}(\vx)+\frac{1}{2}\log\det\left(\ma^{T}\ms_{x}^{-2}\ma+\lambda\iMatrix\right)+\frac{\lambda}{2}\norm x_{2}^{2}
\]
This allows us to ensure that the entries of $\vs(\vx)$ do not get
too large or too small and therefore changing the weighting of the
analytic center cannot affect the function value too much. 

Third, we need to make sure our potential function is convex. If we
simply take $\vWeight=\vtau-\vsigma(\vx)$ with $\vtau$ as an estimator
of $\vsigma(\vx)$, $\vWeight$ can be negative and the potential
function could be non-convex. To circumvent this issue, we use $\vWeight=c_{e}+\vtau-\vsigma(\vx)$
and make sure $\norm{\vtau-\vsigma(\vx)}_{\infty}<c_{e}$.

Combining these insights, using efficient algorithms for solving a
sequence of slowly changing linear systems \cite{vaidya1989speeding,leeS14,lee2015efficient},
and providing careful analysis ultimately allows us to achieve a running
time of $O(n\SO\log\kappa+n^{3}\log^{O(1)}\kappa)$ for the feasibility
problem. Furthermore, in the case that $K$ does not contain a ball
of radius $\epsilon$, our algorithm provides a proof that the polytope
does not contain a ball of radius $\epsilon$. This proof ultimately
allows us to achieve running time improvements for strongly polynomial
submodular minimization in Part~\ref{part:Submodular-Function-Minimization}.

\subsection{Organization}

The rest of Part~\ref{part:Ellipsoid} is organized as follows. In
Section~\ref{sec:ellipsoid_preliminaries} we provide some preliminary
information and notation we use throughout Part~\ref{part:Ellipsoid}.
In Section~\ref{sec:ellipsoid:method} we then provide and analyze
our cutting plane method. In Section~\ref{sec:ellipsoid:tools} we
provide key technical tools which may be of independent interest.

\section{Preliminaries\label{sec:ellipsoid_preliminaries}}

Here we introduce some notation and concepts we use throughout Part~\ref{part:Ellipsoid}.

\subsection{Leverage Scores\label{sub:leverage_scores}}

Our algorithms in this section make extensive use of \emph{leverage
scores}, a common measure of the importance of rows of a matrix. We
denote the leverage scores of a matrix $\ma\in\R^{n\times d}$ by
$\vsigma\in\R^{n}$ and say the \emph{leverage score of row $i\in[n]$}
is $\sigma_{i}\defeq[\ma\left(\ma^{T}\ma\right)^{-1}\ma^{T}]_{ii}$.
For $\ma\in\R^{n\times d}$, $\vd\in\rPos^{n}$, and $\md\defeq\mDiag(\vd)$
we use the shorthand $\vsigma_{\ma}(\vd)$ to denote the leverage
scores of the matrix $\md^{1/2}\ma$. We frequently use well known
facts regarding leverage scores, such as $\sigma_{i}\in[0,1]$ and
$\normOne{\vsigma}\leq d$. (See \cite{spielmanS08sparsRes,mahoney11survey,li2012iterative,Cohen2014}
for a more in-depth discussion of leverage scores, their properties,
and their many applications.) In addition, we make use of the fact
that given an efficient linear system solver of $\ma^{T}\ma$ we can
efficiently compute multiplicative approximations to leverage scores
(See Definition~\ref{def:linear_system_solver} and Lemma~\ref{lem:computing_leverage_scores}
below).
\begin{defn}[Linear System Solver]
\label{def:linear_system_solver} An algorithm $\solver$ is a $\mathcal{\LO}$-time
solver of a PD matrix $\mm\in\R^{n\times n}$ if for all $\vb\in\R^{n}$
and $\epsilon\in(0,1/2]$, the algorithm outputs a vector $\solver(\vb,\epsilon)\in\R^{n}$
in time $O(\LO\cdot\log(\epsilon^{-1}))$ such that with high probability
in $n$, $\norm{\mathcal{\solver}(\vb,\epsilon)-\mm^{-1}\vb}_{\mm}^{2}\leq\epsilon\norm{\mm^{-1}\vb}_{\mm}^{2}$.
\end{defn}
\begin{lem}[Computing Leverage Scores \cite{spielmanS08sparsRes}]
\label{lem:computing_leverage_scores} Let $\ma\in\R^{n\times d}$,
let $\vsigma$ denote the leverage scores of $\ma$, and let $\epsilon>0$.
If we have a $\LO$-time solver for $\ma^{T}\ma$ then in time $\otilde((\nnz(\ma)+\LO)\epsilon^{-2}\log(\epsilon^{-1}))$
we can compute $\vec{\tau}\in\R^{n}$ such that with high probability
in $d$, $(1-\epsilon)\sigma_{i}\leq\tau_{i}\leq(1+\epsilon)\sigma_{i}$
for all $i\in[n]$.
\end{lem}

\subsection{Hybrid Barrier Function}

As explained in Section~\ref{sec:ellipsoid:intro:our_approach} our
cutting plane method maintains a polytope $P=\{\vx\in\R^{n}\,:\,\ma\vx\geq\vb\}$
for $\ma\in\R^{m\times n}$ and $\vb\in\R^{n}$ that contains some
target set $K$. We then maintain a minimizer of the following hybrid
barrier function:
\[
p_{\ve}(\vx)\defeq-\sum_{i\in[m]}\left(c_{e}+e_{i}\right)\log s_{i}(\vx)+\frac{1}{2}\log\det\left(\ma^{T}\ms_{x}^{-2}\ma+\lambda\iMatrix\right)+\frac{\lambda}{2}\norm x_{2}^{2}
\]
where $\ve\in\R^{m}$ is a variable we maintain, $c_{e}\geq0$ and
$\lambda\geq0$ are constants we fix later, $\vs(\vx)\defeq\ma\vx-\vb$,
and $\ms_{x}=\mDiag(\vs(\vx))$. When the meaning is clear from context
we often use the shorthand $\ma_{x}\defeq\ms_{x}^{-1}\ma$.

Rather than maintaining $\ve$ explicitly, we instead maintain a vector
$\vtau\in\R^{m}$ that approximates the leverage score
\[
\vpsi(\vx)\defeq\diag\left(\ma_{x}\left(\ma_{x}^{T}\ma_{x}+\lambda\iMatrix\right)^{-1}\ma_{x}^{T}\right)\enspace.
\]
Note that $\vpsi(\vx)$ is simply the leverage scores of certain rows
of the matrix 
\[
\left[\begin{array}{c}
\ma_{x}\\
\sqrt{\lambda}\iMatrix
\end{array}\right].
\]
and therefore the usual properties of leverage scores hold, i.e. $\psi_{i}(\vx)\in(0,1)$
and $\normOne{\psi_{i}(\vx)}\leq n$. We write $\vpsi(\vx)$ equivalently
as $\vpsi_{\ma_{x}}$ or $\vpsi_{P}$ when we want the matrix to be
clear. Furthermore, we let $\mPsi_{x}\defeq\mDiag(\vpsi(\vx))$ and
$\mu(\vx)\defeq\min_{i}\psi_{i}(\vx)$. Finally, we typically pick
$\ve$ using the function $\ve_{P}(\vtau,\vx)\defeq\vtau-\vpsi(\vx)$.
Again, we use the subscripts of $\ma_{x}$ and $P$ interchangeably
and often drop them when the meaning is clear from context.

We remark that the last term $\frac{\lambda}{2}\norm x_{2}^{2}$ ensures
that our point is always within a certain region (Lemma \ref{lem:boundedness})
and hence the term $\left(c_{e}+e_{i}\right)\log s_{i}(\vx)_{i}$
never gets too large. However, this $\ell^{2}$ term changes the Hessian
of the potential function and hence we need to put a $\lambda\iMatrix$
term inside both the $\log\det$ and the leverage score to reflect
this. This is the reason why we use $\vpsi$ instead of the standard
leverage score.

\section{Our Cutting Plane Method\label{sec:ellipsoid:method}}

In this section we develop and prove the correctness of our cutting
plane method. We use the notation introduced in Section~\ref{sec:intro:preliminaries}
and Section~\ref{sec:ellipsoid_preliminaries} as well as the technical
tools we introduce in Section~\ref{sec:ellipsoid:tools}.

We break the presentation and proof of correctness of our cutting
plane methods into multiple parts. First in Section~\ref{sec:ellipsoid:centering}
we describe how we maintain a center of the hybrid barrier function
$p_{\ve}$ and analyze this procedure. Then, in Section~\ref{sub:Changing-Constraints}
we carefully analyze the effect of changing constraints on the hybrid
barrier function and in Section~\ref{sec:ellipsoid:dikin_prop} we
prove properties of an approximate center of hybrid barrier function,
which we call a hybrid center. In Section~\ref{sub:ellip_algo} we
then provide our cutting plane method and in Section~\ref{sub:ellip_guarantee}
we prove that the cutting plane method solves the feasibility problem
as desired.

\subsection{Centering\label{sec:ellipsoid:centering} }

In this section we show how to compute approximate \emph{centers }or
minimizers of the hybrid barrier function for the current polytope
$P=\{\vx\,:\,\ma\vx\geq\vb\}$. We split this proof up into multiple
parts. First we simply bound the gradient and Hessian of the hybrid
barrier function, $p_{\ve}$, as follows. 
\begin{lem}
\label{lem:grad_Hessian_vol} For $f(\vx)\defeq\frac{1}{2}\log\det\left(\ma^{T}\ms_{x}^{-2}\ma+\lambda\iMatrix\right)$,
we have that
\[
\nabla f(\vx)=-\ma_{x}^{T}\vpsi(\vx)\enspace\text{ and }\enspace\ma_{x}^{T}\mPsi(\vx)\ma_{x}\preceq\hess f(\vx)\preceq3\ma_{x}^{T}\mPsi(\vx)\ma_{x}\enspace.
\]
\end{lem}
\begin{proof}
Our proof is similar to \cite[Appendix]{anstreicher1996large} which
proved the statement when $\lambda=0$. This case does not change
the derivation significantly, however for completeness we include
the proof below. 

We take derivatives on $\vs$ first and then apply chain rule. Let
$f(\vs)=\frac{1}{2}\log\det\left(\ma^{T}\ms^{-2}\ma+\lambda\iMatrix\right)$.
We use the notation $Df(\vx)[\vh]$ to denote the directional derivative
of $f$ along the direction $\vh$ at the point $\vx$. Using the
standard formula for the derivative of $\log\det$, i.e. $\frac{d}{dt}\log\det\mb_{t}=\tr((\mb_{t})^{-1}(\frac{d\mb_{t}}{dt}))$,
we have 
\begin{eqnarray}
Df(\vs)[\vh] & = & \frac{1}{2}\tr((\ma^{T}\ms^{-2}\ma+\lambda\iMatrix)^{-1}(\ma^{T}(-2)\ms^{-3}\mh\ma))\label{eq:deriative_log_det}\\
 & = & -\sum_{i}\frac{h_{i}}{s_{i}}\onesVec_{i}^{T}\ms^{-1}\ma\left(\ma^{T}\ms^{-2}\ma+\lambda\iMatrix\right)^{-1}\ma\ms^{-1}\onesVec_{i}=-\sum_{i}\frac{\psi_{i}h_{i}}{s_{i}}\enspace.\nonumber 
\end{eqnarray}
Applying chain rules, we have $\nabla f(\vx)=-\ma_{x}^{T}\vpsi.$
Now let $\mvar P\defeq\ms^{-1}\ma\left(\ma^{T}\ms^{-2}\ma+\lambda\iMatrix\right)^{-1}\ma^{T}\ms^{-1}$.
Taking the derivative of \eqref{eq:deriative_log_det} again and using
the cyclic property of trace, we have
\begin{eqnarray*}
D^{2}f(\vs)[\vh_{1},\vh_{2}] & = & \tr\left(\left(\ma^{T}\ms^{-2}\ma+\lambda\iMatrix\right)^{-1}\left(\ma^{T}(-2)\ms^{-3}\mh_{2}\ma\right)\left(\ma^{T}\ms^{-2}\ma+\lambda\iMatrix\right)^{-1}\left(\ma^{T}\ms^{-3}\mh_{1}\ma\right)\right)\\
 &  & -\tr\left(\left(\ma^{T}\ms^{-2}\ma+\lambda\iMatrix\right)^{-1}\left(\ma^{T}(-3)\ms^{-4}\mh_{2}\mh_{1}\ma\right)\right)\\
 & = & 3\tr\left(\mvar P\ms^{-2}\mh_{2}\mh_{1}\right)-2\tr\left(\mvar P\ms^{-1}\mh_{2}\mvar P\ms^{-1}\mh_{1}\right)\\
 & = & 3\sum_{i}P_{ii}\frac{\vh_{1}(i)\vh_{2}(i)}{s_{i}^{2}}-2\sum_{ij}P_{ij}\frac{\vh_{2}(j)}{s_{j}}P_{ji}\frac{\vh_{2}(i)}{s_{i}}\\
 & = & 3\sum_{i}\psi_{i}\frac{\vh_{1}(i)\vh_{2}(i)}{s_{i}^{2}}-2\sum_{ij}P_{ij}^{2}\frac{\vh_{2}(j)}{s_{j}}\frac{\vh_{2}(i)}{s_{i}}\enspace.
\end{eqnarray*}
Consequently, $D^{2}f(\vx)[\indicVec i,\indicVec j]=[\ms^{-1}\left(3\mPsi-2\mvar P^{(2)}\right)\ms^{-1}]_{ij}$
where $\mvar P^{(2)}$ is the Schur product of $\mvar P$ with itself. 

Now note that 
\begin{eqnarray*}
\sum_{i}P_{ij}^{2} & = & \onesVec_{j}\ms^{-1}\ma\left(\ma^{T}\ms^{-2}\ma+\lambda\iMatrix\right)^{-1}\ma^{T}\ms^{-2}\ma\left(\ma^{T}\ms^{-2}\ma+\lambda\iMatrix\right)^{-1}\ma^{T}\ms^{-1}\onesVec_{j}\\
 & \leq & \onesVec_{j}\ms^{-1}\ma\left(\ma^{T}\ms^{-2}\ma+\lambda\iMatrix\right)^{-1}\ma^{T}\ms^{-1}\onesVec_{j}=P_{jj}=\mPsi_{jj}\enspace.
\end{eqnarray*}
Hence, the Gershgorin circle theorem shows that the eigenvalues of
$\mPsi-\mvar P^{(2)}$ are lies in union of the interval $[0,2\psi_{j}]$
over all $j$. Hence, $\mPsi-\mvar P^{(2)}\succeq\mZero$. On the
other hand, Schur product theorem shows that $\mvar P^{(2)}\succeq\mZero$
as $\mvar P\succeq\mZero$. Hence, the result follows by chain rule.
\end{proof}
Lemma~\ref{lem:grad_Hessian_vol} immediately shows that under our
choice of $\ve=\ve_{P}(\vx,\vtau)$ we can compute the gradient of
the hybrid barrier function, $p_{\ve}(\vx)$ efficiently. Formally,
Lemma~\ref{lem:grad_Hessian_vol} immediately implies the following:
\begin{lem}[Gradient]
\label{lem:ellipsoid:gradient} For $\vx\in P=\{\vy\in\R^{n}\,:\,\ma\vy\geq\vb\}$
and $\ve\in\R^{m}$ we have
\[
\grad p_{\ve}(\vx)=-\ma_{x}^{T}(c_{e}\onesVec+\ve+\vpsi_{P}(\vx))+\lambda\vx
\]
and therefore for all $\vtau\in\R^{m}$, we have
\[
\grad p_{\ve(\vtau,\vx)}(\vx)=-\ma_{x}^{T}\left(c_{e}\onesVec+\vtau\right)+\lambda\vx.
\]
\end{lem}
\begin{rem}
To be clear, the vector $\grad p_{\ve(\vtau,\vx)}(\vx)$ is defined
as the vector such that
\[
[\grad p_{\ve(\vtau,\vx)}(\vx)]_{i}=\lim_{t\rightarrow0}\frac{1}{t}\left(p_{\ve(\vtau,\vx)}(\vx+t\onesVec_{i})-p_{\ve(\vtau,\vx)}(\vx)\right)\enspace.
\]
In other words, we treat the parameter $\ve(\vtau,\vx)$ as fixed.
This is the reason we denote it by subscript to emphasize that $p_{\ve}(\vx)$
is a family of functions, $p_{\ve(\vtau,\vx)}$ is one particular
function, and $\nabla p_{\ve(\vtau,\vx)}$ means taking gradient on
that particular function. 
\end{rem}
Consequently, we can always compute $\grad p_{\ve(\vtau,\vx)}(\vx)$
efficiently. Now, we measure \emph{centrality} or how close we are
to the hybrid center as follows.
\begin{defn}[Centrality]
\label{def:ellipsoid:centrality} For $\vx\in P=\left\{ \vy\in\R^{n}\,:\,\ma\vy\geq\vb\right\} $
and $\ve\in\R^{m}$, we define the \emph{centrality }of $\vx$ by
\[
\delta_{\ve}(\vx)\defeq\norm{\grad p_{\ve}(\vx)}_{\mbox{\ensuremath{\mh}(\ensuremath{\vx})}^{-1}}
\]
where $\mh(\vx)\defeq\ma_{x}^{T}\left(c_{e}\iMatrix+\mPsi(\vx)\right)\ma_{x}+\lambda\iMatrix$.
Often, we use \emph{weights} $\vw\in\R_{>0}^{m}$ to approximate this
Hessian and consider $\mq(\vx,\vw)\defeq\ma_{x}^{T}\left(c_{e}\iMatrix+\mWeight\right)\ma_{x}+\lambda\iMatrix$.

Next, we bound how much slacks can change in a region close to a nearly
central point.

\end{defn}
\begin{lem}
\label{lem:ellipsoid:mult_change} Let $\vx\in P=\left\{ \vy\in\R^{n}\,:\,\ma\vy\geq\vb\right\} $
and $\vy\in\R^{n}$ such that $\norm{\vx-\vy}_{\mh(\vx)}\leq\epsilon\sqrt{c_{e}+\mu(\vx)}$
for $\epsilon<1$. Then $\vy\in P$ and $(1-\epsilon)\ms_{x}\preceq\ms_{y}\preceq(1+\epsilon)\ms_{x}$
.\end{lem}
\begin{proof}
Direct calculation reveals the following:
\begin{align*}
\norm{\ms_{x}^{-1}(\vs_{\vy}-\vs_{x})}_{\infty} & \leq\normTwo{\ma_{x}(\vy-\vx)}\leq\frac{1}{\sqrt{c_{e}+\mu(\vx)}}\norm{\ma_{x}(\vy-\vx)}_{c_{e}\iMatrix+\mPsi(\vx)}\\
 & \leq\frac{1}{\sqrt{c_{e}+\mu(\vx)}}\norm{\vy-\vx}_{\mh(\vx)}\leq\epsilon\enspace.
\end{align*}
Consequently, $(1-\epsilon)\ms_{x}\preceq\ms_{y}\preceq(1+\epsilon)\ms_{x}$.
Since $y\in P$ if and only if $\ms_{y}\succeq\mZero$ the result
follows.
\end{proof}
Combining the previous lemmas we obtain the following.
\begin{lem}
\label{lem:ellipsoid:hess_bound} Let $\vx\in P=\{\vy\in\R^{n}\,:\,\ma\vy\geq\vb\}$
and $\ve,\vw\in\R^{m}$ such that $\normInf{\ve}\leq\frac{1}{2}c_{e}\leq1$
and $\mPsi(\vx)\preceq\mWeight\preceq\frac{4}{3}\mPsi(\vx)$. If $\vy\in\R^{n}$
satisfies $\norm{\vx-\vy}_{\mq(\vx,\vw)}\leq\frac{1}{10}\sqrt{c_{e}+\mu(\vx)}$,
then
\[
\frac{1}{4}\mq(\vx,\vWeight)\preceq\hess p_{\ve}(\vy)\preceq8\mq(\vx,\vWeight)\enspace\text{ and }\enspace\frac{1}{2}\mh(\vx)\preceq\mh(\vy)\preceq2\mh(\vx)\enspace.
\]
\end{lem}
\begin{proof}
Lemma~\ref{lem:grad_Hessian_vol} shows that
\begin{equation}
\ma_{y}^{T}\left(c_{e}\iMatrix+\mE+\mPsi(\vy)\right)\ma_{y}+\lambda\iMatrix\preceq\hess p_{\ve}(\vy)\preceq\ma_{y}^{T}\left(c_{e}\iMatrix+\mE+3\mPsi(\vy)\right)\ma_{y}+\lambda\iMatrix\enspace.\label{eq:ellip_apprx_hess}
\end{equation}
Since $\mWeight\succeq\mPsi$, we have that $\mq(\vx,\vw)\succeq\mh(\vx)$
and therefore $\norm{\vx-\vy}_{\mh(\vx)}\leq\epsilon\sqrt{c_{e}+\mu(\vx)}$
with $\epsilon=0.1$. Consequently, by Lemma~\ref{lem:ellipsoid:mult_change}
we have $(1-\epsilon)\ms_{x}\preceq\ms_{y}\preceq(1+\epsilon)\ms_{x}$
and therefore
\[
\frac{(1-\epsilon)^{2}}{(1+\epsilon)^{2}}\mPsi(\vx)\preceq\mPsi(\vy)\preceq\frac{(1+\epsilon)^{2}}{(1-\epsilon)^{2}}\mPsi(\vx)
\]
and 
\[
\frac{1}{2}\mh(\vx)\preceq\frac{(1-\epsilon)^{2}}{(1+\epsilon)^{4}}\mh(\vx)\preceq\mh(\vy)\preceq\frac{(1+\epsilon)^{2}}{(1-\epsilon)^{4}}\mh(\vx)\preceq2\mh(\vx)
\]
Furthermore, \eqref{eq:ellip_apprx_hess} shows that
\begin{align*}
\hess p_{\ve}(\vy) & \preceq\ma_{y}^{T}\left(c_{e}\iMatrix+\mE+3\mPsi(\vy)\right)\ma_{y}+\lambda\iMatrix\\
 & \preceq\frac{(1+\epsilon)^{2}}{(1-\epsilon)^{4}}\ma_{x}^{T}\left(c_{e}\iMatrix+\mE+3\mPsi(\vx)\right)\ma_{x}+\lambda\iMatrix\\
 & \preceq\frac{(1+\epsilon)^{2}}{(1-\epsilon)^{4}}\ma_{x}^{T}\left(\frac{3}{2}c_{e}\iMatrix+3\mWeight\right)\ma_{x}+\lambda\iMatrix\\
 & \preceq3\frac{(1+\epsilon)^{2}}{(1-\epsilon)^{4}}\mq(\vx,\vWeight)\preceq8\mq(\vx,\vw)
\end{align*}
and
\begin{align*}
\hess p_{\ve}(\vy) & \succeq\ma_{y}^{T}\left(c_{e}\iMatrix+\mE+\mPsi(\vy)\right)\ma_{y}+\lambda\iMatrix\\
 & \succeq\frac{(1-\epsilon)^{4}}{(1+\epsilon)^{2}}\ma_{x}^{T}\left(c_{e}\iMatrix+\mE+\mPsi(\vx)\right)\ma_{x}+\lambda\iMatrix\\
 & \succeq\frac{(1-\epsilon)^{4}}{(1+\epsilon)^{2}}\ma_{x}^{T}\left(\frac{1}{2}c_{e}\iMatrix+\frac{3}{4}\mWeight\right)\ma_{x}+\lambda\iMatrix\\
 & \succeq\frac{1}{2}\frac{(1-\epsilon)^{4}}{(1+\epsilon)^{2}}\mq(\vx,\vWeight)\succeq\frac{1}{4}\mq(\vx,\vw).
\end{align*}

\end{proof}
To analyze our centering scheme we use standard facts about gradient
descent we prove in Lemma~\ref{lem:ellipsoid:gradient_descent}.
\begin{lem}[Gradient Descent]
\label{lem:ellipsoid:gradient_descent} Let $f:\R^{n}\rightarrow\R$
be twice differentiable and $\mq\in\R^{n\times n}$ be positive definite.
Let $\vx_{0}\in\R^{n}$ and $\vx_{1}\defeq\vx_{0}-\frac{1}{L}\mq^{-1}\grad f(\vx_{0})$.
Furthermore, let $\vx_{\alpha}=\vx_{0}+\alpha(\vx_{1}-\vx)$ and suppose
that $\mu\mq\preceq\hess f(\vx_{\alpha})\preceq L\mq$ for all $\alpha\in[0,1]$.
Then,
\begin{enumerate}
\item $\norm{\grad f(\vx_{1})}_{\mq^{-1}}\leq\left(1-\frac{\mu}{L}\right)\norm{\grad f(\vx_{0})}_{\mq^{-1}}$
\item $f(\vx_{1})\geq f(\vx_{0})-\frac{1}{L}\norm{\grad f(\vx_{0})}_{\mq^{-1}}^{2}$
\end{enumerate}
\end{lem}
\begin{proof}
Integrating we have that
\begin{align*}
\grad f(\vx_{1}) & =\grad f(\vx_{0})+\int_{0}^{1}\hess f(\vx_{\alpha})(\vx_{1}-\vx_{0})d\alpha=\int_{0}^{1}\left(\mq-\frac{1}{L}\hess f(\vx_{\alpha})\right)\mq^{-1}\grad f(\vx_{0})d\alpha
\end{align*}
Consequently, by applying Jensen's inequality we have 
\begin{align*}
\norm{\grad f(\vx_{1})}_{\mq^{-1}} & =\normFull{\int_{0}^{1}\left(\mq-\frac{1}{L}\hess f(\vx_{\alpha})\right)\mq^{-1}\grad f(\vx_{0})d\alpha}_{\mq^{-1}}\\
 & \leq\int_{0}^{1}\normFull{\left(\mq-\frac{1}{L}\hess f(\vx_{\alpha})\right)\mq^{-1}\grad f(\vx_{0})}_{\mq^{-1}}d\alpha\\
 & \leq\norm{\mq^{-1/2}\grad f(\vx_{0})}_{\left[\mq^{-1/2}\left(\mq-\frac{1}{L}\hess f(\vx_{\alpha})\right)\mq^{-1/2}\right]^{2}}
\end{align*}
Now we know that by assumption that 
\[
\mZero\preceq\mq^{-1/2}\left(\mq-\frac{1}{L}\hess f(\vx_{\alpha})\right)\mq^{-1/2}\preceq\left(1-\frac{\mu}{L}\right)\iMatrix
\]
and therefore combining these (1) holds.

Using the convexity of $f$, we have
\begin{eqnarray*}
f(\vx_{1}) & \geq & f(\vx_{0})+\innerproduct{\grad f(\vx_{0})}{\vx_{1}-\vx_{0}}\\
 & \geq & f(\vx_{0})-\norm{\grad f(\vx_{0})}_{\mq^{-1}}\norm{\vx_{1}-\vx_{0}}_{\mq}
\end{eqnarray*}
and since $\norm{\vx_{1}-\vx_{0}}_{\mq}=\frac{1}{L}\norm{\grad f(\vx_{0})}_{\mq^{-1}}$,
(2) holds as well.
\end{proof}
Next we bound the effect of changing $\ve$ on the hybrid barrier
function $p_{\ve}(\vx)$.
\begin{lem}
\label{lem:ellipsoid:change_in_e}For $\vx\in P=\{\vy\in\R^{n}\,:\,\ma\vy\geq\vb\}$,
$\ve,\vf\in\R^{m}$, and $\vw\in\R_{>0}^{m}$ such that $\mWeight\succeq\mPsi_{x}$
\[
\norm{\grad p_{\vf}(\vx)}_{\mq(\vx,\vw)^{-1}}\leq\norm{\grad p_{\ve}(\vx)}_{\mq(\vx,\vw)^{-1}}+\frac{1}{\sqrt{c_{e}+\mu(\vx)}}\norm{\vf-\ve}_{2}
\]
\end{lem}
\begin{proof}
Direct calculation shows the following
\begin{align*}
\norm{\grad p_{\vf}(\vx)}_{\mq(\vx,\vw)^{-1}} & =\norm{-\ma_{x}^{T}(c_{e}\onesVec+\vf+\vpsi_{P}(\vx))+\lambda\vx}_{\mq(\vx,\vw)^{-1}}\tag{Formula for \ensuremath{\grad p_{\vf}(\vx)}}\\
 & \leq\norm{\grad p_{\ve}(\vx)}_{\mq(\vx,\vw)^{-1}}+\norm{\ma_{x}^{T}(\vf-\ve)}_{\mq(\vx,\vw)^{-1}}\tag{Triangle inequality}\\
 & \leq\norm{\grad p_{\ve}(\vx)}_{\mq(\vx,\vw)^{-1}}+\frac{1}{\sqrt{c_{e}+\mu(\vx)}}\norm{\ma_{x}^{T}(\vf-\ve)}_{\left(\ma_{x}^{T}\ma_{x}\right)^{-1}}\tag{Bound on \ensuremath{\mq(\vx,\vw)}}\\
 & \leq\norm{\grad p_{\ve}(\vx)}_{\mq(\vx,\vw)^{-1}}+\frac{1}{\sqrt{c_{e}+\mu(\vx)}}\norm{\vf-\ve}_{2}\tag{Property of projection matrix}
\end{align*}
where in the second to third line we used $\mq(\vx,\vw)\succeq\mh(\vx)\succeq(c_{e}+\mu(\vx))\ma_{x}^{T}\ma_{x}$.
\end{proof}
We now have everything we need to analyze our centering algorithm.

\begin{algorithm2e}

\caption{$\ensuremath{(\vx^{(r)},\vtau^{(r)})=\code{Centering}(\vx^{(0)},\vtau^{(0)},r,c_{\Delta})}$}\label{alg:ellipsoid:centering}

\SetAlgoLined

\textbf{Input: }Initial point $\vx^{(0)}\in P=\{\vy\in\R^{n}\,:\,\ma\vy\geq\vb\}$,
Estimator of leverage scores $\vtau^{(0)}\in\R^{n}$

\textbf{Input: }Number of iterations $r>0$, Accuracy of the estimator
$0\leq c_{\Delta}\leq0.01c_{e}$.

\textbf{Given: }$\norm{\ve^{(0)}}_{\infty}\leq\frac{1}{3}c_{e}\leq\frac{1}{3}$
where $\ve^{(0)}=\ve(\vtau^{(0)},\vx^{(0)}).$

\textbf{Given: $\delta_{\ve^{(0)}}(\vx^{(0)})=\norm{\grad p_{\ve^{(0)}}(\vx^{(0)})}_{\mh(\vx^{(0)})^{-1}}\leq\frac{1}{100}\sqrt{c_{e}+\mu(\vx^{(0)})}$.}

Compute $\vw$ such that $\mPsi(\vx^{(0)})\preceq\mWeight\preceq\frac{4}{3}\mPsi(\vx^{(0)})$
(See Lemma~\ref{lem:computing_leverage_scores})

Let $\mq\defeq\mq(\vx^{(0)},\vWeight)$.

\For{$k=1$ to $r$}{

$\vx^{(k)}:=\vx^{(k-1)}-\frac{1}{8}\mq^{-1}\grad p_{\ve^{(k-1)}}(\vx^{(k-1)}).$

Sample $\vDelta^{(k)}\in\R^{n}$ s.t. 

$\enspace$ $\E[\vDelta^{(k)}]=\vpsi(\vx^{(k)})-\vpsi(\vx^{(k-1)})$
and

$\enspace$with high probability in $n$, $\norm{\vDelta^{(k)}-(\vpsi(\vx^{(k)})-\vpsi(\vx^{(k-1)}))}_{2}\leq c_{\Delta}\norm{\ms_{\vx^{(k-1)}}^{-1}(\vs_{\vx^{(k)}}-\vs_{\vx^{(k-1)}})}_{2}$
(See Section~\ref{sec:lever_change})

$\vtau^{(k)}:=\vtau^{(k-1)}+\vDelta^{(k)}.$

$\ve^{(k)}:=\ve(\vtau^{(k)},\vx^{(k)}).$

}

\textbf{Output}: \textbf{$(\vx^{(r)},\vtau^{(r)})$}

\end{algorithm2e}
\begin{lem}
\label{lem:ellip_centering}Let $\vx^{(0)}\in P=\{\vy\in\R^{n}\,:\,\ma\vy\geq\vb\}$
and let $\vtau^{(0)}\in\R^{m}$ such that $\normInf{\ve(\vtau^{(0)},\vx^{(0)})}\leq\frac{1}{3}c_{e}\leq\frac{1}{3}$.
Assume that $r$ is a positive integer, $0\leq c_{\Delta}\leq0.01c_{e}$
and $\delta_{\ve^{(0)}}(\vx^{(0)})\leq\frac{1}{100}\sqrt{c_{e}+\mu(\vx^{(0)})}.$
With high probability in $n$, the algorithm $\code{Centering}(\vx^{(0)},\vtau^{(0)},r,c_{\Delta})$
outputs $(\vx^{(r)},\vtau^{(r)})$ such that
\begin{enumerate}
\item $\delta_{\ve^{(r)}}(\vx^{(r)})\leq2\left(1-\frac{1}{64}\right)^{r}\delta_{\ve^{(0)}}(\vx^{(0)})$.
\item $\E[p_{\ve^{(k)}}(\vx^{(r)})]\geq p_{\ve^{(0)}}(\vx^{(0)})-8\left(\delta_{\ve^{(0)}}(\vx^{(0)})\right)^{2}.$
\item $\E\ve^{(r)}=\ve^{(0)}$ and $\norm{\ve^{(r)}-\ve^{(0)}}_{2}\leq\frac{1}{10}c_{\Delta}$.
\item $\normFull{\ms_{\vx^{(0)}}^{-1}(\vs(\vx^{(r)})-\vs(\vx^{(0)}))}_{2}\leq\frac{1}{10}$.
\end{enumerate}
where $\ve^{(r)}=\ve(\vtau^{(r)},\vx^{(r)})$. \end{lem}
\begin{proof}
Let $\eta=\norm{\grad p_{\ve^{(0)}}(\vx^{(0)})}_{\mq^{-1}}$. First,
we use induction to prove that $\norm{\vx^{(r)}-\vx^{(0)}}_{\mq}\leq8\eta$,
$\norm{\grad p_{\ve^{(r)}}(\vx^{(r)})}_{\mq^{-1}}\leq\left(1-\frac{1}{64}\right)^{r}\eta$
and $\norm{\ve^{(r)}-\ve^{(0)}}_{2}\leq\frac{1}{10}c_{\Delta}$ for
all $r$.

Clearly the claims hold for $r=0$. We now suppose they hold for all
$r\leq t$ and show that they hold for $r=t+1$. Now, since $\norm{\vx^{(t)}-\vx^{(0)}}_{\mq}\leq8\eta$,
$\vx^{(t+1)}=\vx^{(t)}-\frac{1}{8}\mq^{-1}\grad p_{\ve^{(t)}}(\vx^{(t)})$,
and $\norm{\grad p_{\ve^{(t)}}(\vx^{(t)})}_{\mq^{-1}}\leq\left(1-\frac{1}{64}\right)^{t}\eta\leq\eta$,
we have 
\[
\norm{\vx^{(t+1)}-\vx^{(0)}}_{\mq}\leq\norm{\vx^{(t)}-\vx^{(0)}}_{\mq}+\frac{1}{8}\norm{\grad p_{\ve^{(t)}}(\vx^{(t)})}_{\mq^{-1}}\leq9\eta.
\]
We will improve this estimate later in the proof to finish the induction
on $\norm{\vx^{(t+1)}-\vx^{(0)}}_{\mq}$, but using this, $\eta\leq0.01\sqrt{c_{e}+\mu(\vx^{(0)})}$,
and $\normInf{\ve^{(t)}}\leq\normInf{\ve^{(t)}-\ve^{(0)}}+\norm{\ve^{(0)}}_{\infty}\leq\frac{c_{e}}{2}$,
we can invoke Lemma~\ref{lem:ellipsoid:hess_bound} and Lemma~\ref{lem:ellipsoid:gradient_descent}
and therefore
\[
\norm{\grad p_{\ve^{(t)}}(\vx^{(t+1)})}_{\mq^{-1}}\leq\left(1-\frac{1}{32}\right)\norm{\grad p_{\ve^{(t)}}(\vx^{(t)})}_{\mq^{-1}}\,.
\]
By Lemma~\ref{lem:ellipsoid:change_in_e} we have
\begin{equation}
\norm{\grad p_{\ve^{(t+1)}}(\vx^{(t+1)})}_{\mq^{-1}}\leq\left(1-\frac{1}{32}\right)\norm{\grad p_{\ve^{(t)}}(\vx^{(t)})}_{\mq^{-1}}+\frac{1}{\sqrt{c_{e}+\mu(\vx^{(0)})}}\norm{\ve^{(t+1)}-\ve^{(t)}}_{2}.\label{eq:ellipsoid_grad_p_next}
\end{equation}
To bound $\norm{\ve^{(t+1)}-\ve^{(t)}}_{2}$, we note that Lemma~\ref{lem:ellipsoid:mult_change}
and the induction hypothesis $\norm{\vx^{(t)}-\vx^{(0)}}_{\mh(\vx^{(0)})}\leq\norm{\vx^{(t)}-\vx^{(0)}}_{\mq}\leq8\eta$
shows that $(1-0.1)\ms_{x^{(0)}}\preceq\ms_{x^{(t)}}\preceq(1+0.1)\ms_{x^{(0)}}$
and therefore
\begin{align}
\norm{\ms_{x^{(t)}}^{-1}(\vs_{x^{(t)}}-\vs_{x^{(t+1)}})}_{2} & \leq\frac{1}{1-0.1}\norm{\ms_{x^{(0)}}^{-1}\ma\left(\vx^{(t)}-\vx^{(t+1)}\right)}_{2}\nonumber \\
 & =\frac{1}{1-0.1}\normFull{\frac{1}{8}\mq^{-1}\grad p_{\ve^{(t)}}(\vx^{(t)})}_{\ma^{T}\ms_{x^{(0)}}^{-2}\ma}\nonumber \\
 & \leq\frac{1}{8\left(1-0.1\right)\sqrt{c_{e}+\mu(\vx^{(0)})}}\norm{\grad p_{\ve^{(t)}}(\vx^{(t)})}_{\mq^{-1}}\label{eq:ellipsoid_s_reL_dis}
\end{align}
Now since
\begin{align*}
\ve^{(t+1)}-\ve^{(t)} & =\left(\vtau^{(t+1)}-\vpsi(\vx^{(t+1)})\right)-\left(\vtau^{(t)}-\vpsi(\vx^{(t)})\right)\\
 & =\vDelta^{(t+1)}-\left(\vpsi(\vx^{(t+1)})-\vpsi(\vx^{(t)})\right)
\end{align*}
Consequently, with high probability in $n$,
\begin{align*}
\norm{\ve^{(t+1)}-\ve^{(t)}}_{2} & =\normFull{\vDelta^{(t+1)}-\left(\vpsi(\vx^{(t+1)})-\vpsi(\vx^{(t)})\right)}_{2}\\
 & \leq c_{\Delta}\normFull{\ms_{x^{(t)}}^{-1}(\vs_{x^{(t+1)}}-\vs_{x^{(t)}})}_{2}\\
 & \leq\frac{c_{\Delta}}{8\left(1-0.1\right)\sqrt{c_{e}+\mu(\vx^{(0)})}}\norm{\grad p_{\ve^{(t)}}(\vx^{(t)})}_{\mq^{-1}}\enspace.
\end{align*}
where in the last line we used $\min_{i\in[m]}w_{i}\geq\mu(\vx^{(0)}).$
Since $c_{\Delta}<0.01c_{e}$, by \eqref{eq:ellipsoid_grad_p_next},
we have
\begin{align*}
\norm{\grad p_{\ve^{(t+1)}}(\vx^{(t+1)})}_{\mq^{-1}} & \leq\left(1-\frac{1}{32}\right)\norm{\grad p_{\ve^{(t)}}(\vx^{(t)})}_{\mq^{-1}}+\frac{0.01c_{e}}{8\left(1-0.1\right)(c_{e}+\mu(\vx^{(0)}))}\norm{\grad p_{\ve^{(t)}}(\vx^{(t)})}_{\mq^{-1}}\\
 & \leq\left(1-\frac{1}{64}\right)\norm{\grad p_{\ve^{(t)}}(\vx^{(t)})}_{\mq^{-1}}\enspace.
\end{align*}
Furthermore, this implies that
\[
\norm{\vx^{(t+1)}-\vx^{(0)}}_{\mq}\leq\normFull{\sum_{k=0}^{t}\frac{1}{8}\mq^{-1}\grad p_{\ve^{(k)}}(\vx^{(k)})}_{\mq^{-1}}\leq\frac{1}{8}\sum_{i=0}^{\infty}\left(1-\frac{1}{64}\right)^{k}\eta\leq\frac{64}{8}\eta=8\eta\enspace.
\]
Similarly, we have that
\begin{align*}
\norm{\ve^{(t+1)}-\ve^{(0)}}_{2} & \leq\sum_{k=0}^{t}\frac{c_{\Delta}}{8\left(1-0.1\right)\sqrt{c_{e}+\mu(\vx^{(0)})}}\left(1-\frac{1}{64}\right)^{k}\norm{\grad p_{\ve^{(0)}}(\vx^{(0)})}_{\mq^{-1}}\\
 & \leq\frac{8c_{\Delta}\eta}{(1-0.1)\sqrt{c_{e}+\mu(\vx^{(0)})}}\leq\frac{8c_{\Delta}}{(1-0.1)\sqrt{c_{e}+\mu(\vx^{(0)})}}\delta_{\ve^{(0)}}(\vx^{(0)})\leq\frac{1}{10}c_{\Delta}
\end{align*}
where we used $\eta=\norm{\grad p_{\ve^{(0)}}(\vx^{(0)})}_{\mq^{-1}}\leq\norm{\grad p_{\ve^{(0)}}(\vx^{(0)})}_{\mh^{-1}}=\delta_{\ve^{(0)}}(\vx^{(0)})$
and this finishes the induction on $\norm{\grad p_{\ve^{(t)}}(\vx^{(t)})}_{\mq^{-1}}$,
$\norm{\vx^{(t)}-\vx^{(0)}}_{\mq}$ and $\norm{\ve^{(t)}-\ve^{(0)}}_{2}$. 

Hence, for all $r$, Lemma~\ref{lem:ellipsoid:hess_bound} shows
that
\begin{eqnarray*}
\delta_{\ve^{(r)}}(\vx^{(r)}) & = & \norm{\grad p_{\ve^{(r)}}(\vx^{(r)})}_{\mh(\vx^{(r)})^{-1}}\leq\sqrt{2}\norm{\grad p_{\ve^{(r)}}(\vx^{(r)})}_{\mh(\vx^{(0)})^{-1}}\\
 & \leq & \sqrt{\frac{8}{3}}\norm{\grad p_{\ve^{(r)}}(\vx^{(r)})}_{\mq^{-1}}\leq\sqrt{\frac{8}{3}}\left(1-\frac{1}{64}\right)^{r}\norm{\grad p_{\ve^{(0)}}(\vx^{(0)})}_{\mq^{-1}}\\
 & \leq & 2\left(1-\frac{1}{64}\right)^{r}\delta_{\ve^{(0)}}(\vx^{(0)}).
\end{eqnarray*}
Using that $\E\ve^{(t+1)}=\ve^{(t)}$, we see that the expected change
in function value is only due to the change while taking centering
steps and therefore Lemma \ref{lem:ellipsoid:gradient_descent} shows
that
\[
\E[p_{\ve^{(r)}}(\vx^{(r)})]\geq p_{\ve^{(0)}}(\vx^{(0)})-\frac{1}{8}\sum_{k=0}^{\infty}\left(1-\frac{1}{64}\right)^{2k}\norm{\grad p_{\ve^{(0)}}(\vx^{(0)})}_{\mq^{-1}}^{2}\geq p_{\ve^{(0)}}(\vx^{(0)})-8\left(\delta_{\ve^{(0)}}(\vx^{(0)})\right)^{2}.
\]
Finally, for (4), we note that
\[
\normFull{\frac{s(\vx^{(r)})-s(\vx^{(0)})}{s(\vx^{(0)})}}_{2}=\normFull{\vx^{(r)}-\vx^{(0)}}_{\ma^{T}\ms_{x^{(0)}}^{-2}\ma}\leq\frac{1}{\sqrt{\mu(\vx^{(0)})+c_{e}}}\normFull{\vx^{(r)}-\vx^{(0)}}_{\mq^{-1}}\leq\frac{1}{10}.
\]

\end{proof}

\subsection{Changing Constraints \label{sub:Changing-Constraints}}

Here we bound the effect that adding or a removing a constraint has
on the hybrid barrier function. Much of the analysis in this section
follows from the following lemma which follows easily from the Sherman
Morrison Formula.
\begin{lem}[Sherman Morrison Formula Implications]
\label{lem:sherman_morrison}Let $\mb\in\R^{n\times n}$ be an invertible
symmetric matrix and let $\va\in\R^{n}$ be arbitrary vector satisfying
$\va^{T}\mb^{-1}\va<1$. The following hold:
\begin{enumerate}
\item $\left(\mb\pm\va\va^{T}\right)^{-1}=\mb^{-1}\mp\frac{\mb^{-1}\va\va^{T}\mb^{-1}}{1\pm\va^{T}\mb^{-1}\va}.$
\item $\mZero\preceq\frac{\mb^{-1}\va\va^{T}\mb^{-1}}{1\pm\va^{T}\mb^{-1}\va}\preceq\frac{\va^{T}\mb^{-1}\va}{1\pm\va^{T}\mb^{-1}\va}\mb^{-1}.$
\item $\log\det\left(\mb\pm\va\va^{T}\right)=\ln\det\mb+\ln\left(1\pm\va^{T}\mb^{-1}\va\right).$
\end{enumerate}
\end{lem}
\begin{proof}
(1) follows immediately from Sherman Morrison \cite{sherman1950adjustment}.
(2) follows since $\va\va^{T}$ is PSD, 
\begin{align*}
\frac{\mb^{-1}\va\va^{T}\mb^{-1}}{1\pm\va^{T}\mb^{-1}\va} & =\mb^{-1/2}\left[\frac{\mb^{-1/2}\va\va^{T}\mb^{-1/2}}{1\pm\va^{T}\mb^{-1}\va}\right]\mb^{-1/2}\enspace,
\end{align*}
and $\vy\vy^{T}\preceq\norm{\vy}_{2}^{2}\iMatrix$ for any vector
$\vy$. (3) follows immediately from the Matrix Determinant Lemma.
\end{proof}
We also make use of the following technical helper lemma.
\begin{lem}
\label{lem:ellipsoid:technical} For $\ma\in\R^{n\times m}$ and all
$\va\in\R^{n}$ we have
\[
\sum_{i\in[m]}\frac{1}{\psi_{\ma}[i]}\left(\ma\left(\ma^{T}\ma+\lambda\iMatrix\right)^{-1}\va\right)_{i}^{4}\leq\left(\va^{T}\left(\ma^{T}\ma+\lambda\iMatrix\right)^{-1}\va\right)^{2}
\]
\end{lem}
\begin{proof}
We have by Cauchy Schwarz that
\[
\left(\indicVec i^{T}\ma\left(\ma^{T}\ma+\lambda\iMatrix\right)^{-1}\va\right)^{2}\leq\psi_{\ma}[i]\cdot\va^{T}\left(\ma^{T}\ma+\lambda\iMatrix\right)^{-1}\va
\]
and consequently
\[
\sum_{i\in[m]}\frac{\left(\indicVec i^{T}\ma\left(\ma^{T}\ma+\lambda\iMatrix\right)^{-1}\va\right)^{4}}{\psi_{\ma}[i]}\leq\left(\va^{T}\left(\ma^{T}\ma+\lambda\iMatrix\right)^{-1}\va\right)\sum_{i\in[m]}\left(\indicVec i\ma\left(\ma^{T}\ma+\lambda\iMatrix\right)^{-1}\va\right)^{2}\,.
\]
Since
\begin{align*}
\sum_{i\in[m]}\left(\indicVec i^{T}\ma\left(\ma^{T}\ma+\lambda\iMatrix\right)^{-1}\va\right)^{2} & =\va^{T}\left(\ma^{T}\ma+\lambda\iMatrix\right)^{-1}\ma^{T}\ma\left(\ma^{T}\ma+\lambda\iMatrix\right)^{-1}\va\\
 & \leq\va^{T}\left(\ma^{T}\ma+\lambda\iMatrix\right)^{-1}\va,
\end{align*}
 we have the desired result.
\end{proof}
We now bound the effect of adding a constraint.
\begin{lem}
\label{lem:ellip_add_cont}Let $\ma\in\R^{m\times n}$, $\vb\in\R^{m}$,
$\vtau\in\R^{m}$, and $\vx\in P\defeq\left\{ \vy\in\R^{n}\,:\,\ma\vy\geq\vb\right\} $.
Let $\overline{\ma}\in\R^{(m+1)\times n}$ be $\ma$ with a row $\va_{m+1}$
added, let $\overline{b}\in\R^{m+1}$ be the vector $\vb$ with an
entry $b_{m+1}$ added, and let $\overline{P}\defeq\left\{ \vy\in\R^{n}\,:\,\overline{\ma}\vy\geq\overline{b}\right\} $.
Let $s_{m+1}=\va_{m+1}^{T}\vx-b_{m+1}>0$, $\psi_{a}=\frac{\va_{m+1}^{T}(\ma_{x}^{T}\ma_{x}+\lambda\iMatrix)^{-1}\va_{m+1}}{s_{m+1}^{2}}.$ 

Now, let $\vupsilon\in\R^{m+1}$ be defined so that $\upsilon_{m+1}=\frac{\psi_{a}}{1+\psi_{a}}$
and for all $i\in[m]$ 
\[
\upsilon_{i}=\tau_{i}-\frac{1}{1+\psi_{a}}\left[\ma_{x}\left(\ma_{x}^{T}\ma_{x}+\lambda\iMatrix\right)^{-1}\frac{\va_{m+1}}{s_{m+1}}\right]_{i}^{2}\,.
\]
Then, the following hold\end{lem}
\begin{itemize}
\item {[}Leverage Score Estimation{]} $e_{\overline{P}}(\vec{\upsilon},\vx)_{m+1}=0$
and $e_{\overline{P}}(\vec{\upsilon},\vx)_{i}=e_{P}(\vtau,\vx)_{i}$
for all $i\in[m]$.
\item {[}Function Value Increase{]} $p_{\ve_{\overline{P}}(\vupsilon,\vx)}(\vx)=p_{\ve_{P}(\vtau,\vx)}(\vx)-c_{e}\ln s(\vx)_{m+1}+\ln(1+\psi_{a}).$
\item {[}Centrality Increase{]} $\delta_{\ve_{\overline{P}}(\vupsilon,\vx)}(\vx)\leq\delta_{\ve_{P}(\vupsilon,\vx)}(\vx)+\left(c_{e}+\psi_{a}\right)\sqrt{\frac{\psi_{a}}{\mu(\vx)}}+\psi_{a}.$\end{itemize}
\begin{proof}
By (1) in Lemma \ref{lem:sherman_morrison}, we have that for all
$i\in[m]$
\[
\psi_{\overline{P}}(\vx)_{i}=\psi_{P}(\vx)_{i}-\frac{1}{1+\psi_{a}}\left[\ma_{x}\left(\ma_{x}^{T}\ma_{x}+\lambda\iMatrix\right)^{-1}\frac{\va_{m+1}}{s_{m+1}}\right]_{i}^{2}
\]
and that
\[
\psi_{\overline{P}}(\vx)_{m+1}=\psi_{a}-\frac{\psi_{a}^{2}}{1+\psi_{a}}=\frac{\psi_{a}}{1+\psi_{a}}.
\]
Consequently {[}Leverage Score Estimation{]} holds. Furthermore, by
(3) in Lemma \ref{lem:sherman_morrison} this then implies that {[}Function
Value Change{]} holds.

To bound the change in centrality note that by (2) in Lemma \ref{lem:sherman_morrison}
we have that $\overline{\mh}^{-1}\preceq\mh^{-1}.$ Therefore if let
$\vupsilon'\in\R^{m}$ be defined so that $\vupsilon'_{i}=\vupsilon_{i}$
for all $i\in[m]$ then by triangle inequality we have

\begin{align*}
\delta_{\ve_{\overline{p}}(\vupsilon,\vx)}(\vx) & =\norm{\overline{\ma}_{x}^{T}(c_{e}\onesVec+\vupsilon)}_{\overline{\mh}^{-1}}\leq\norm{\overline{\ma}_{x}^{T}(c_{e}\onesVec+\vupsilon)}_{\mh^{-1}}\\
 & \leq\norm{\ma_{x}^{T}(c_{e}\onesVec+\vtau)}_{\mh^{-1}}+\normFull{\frac{\va_{m+1}}{s_{m+1}}(c_{e}+\upsilon_{m+1})}_{\mh^{-1}}+\norm{\ma_{x}^{T}(\vupsilon'-\vtau)}_{\mh^{-1}}\\
 & =\delta_{\ve_{P}(\vtau,\vx)}(\vx)+\left(c_{e}+\frac{\psi_{a}}{1+\psi_{a}}\right)\normFull{\frac{\va_{m+1}}{s_{m+1}}}_{\mh^{-1}}+\norm{\ma_{x}^{T}(\vupsilon'-\vtau)}_{\mh^{-1}}
\end{align*}
Now, since $\mh^{-1}\preceq\frac{1}{\mu(\vx)}\left(\ma_{x}^{T}\ma_{x}+\lambda\iMatrix\right)^{-1}$,
we have that
\[
\normFull{\frac{\va_{m+1}}{s_{m+1}}}_{\mh^{-1}}\leq\frac{1}{\sqrt{\mu(\vx)}}\normFull{\frac{\va_{m+1}}{s_{m+1}}}_{\left(\ma_{x}^{T}\ma_{x}+\lambda\iMatrix\right)^{-1}}=\sqrt{\frac{\psi_{a}}{\mu(\vx)}}.
\]
Since $\mPsi^{1/2}\ma_{x}\left(\ma_{x}^{T}\mPsi\ma_{x}\right)^{-1}\ma_{x}^{T}\mPsi^{1/2}$
is a projection matrix, we have $\mPsi^{-1}\succeq\ma_{x}\left(\ma_{x}^{T}\mPsi\ma_{x}\right)^{-1}\ma_{x}^{T}\succeq\ma_{x}\mh^{-1}\ma_{x}^{T}$.
By Lemma~\ref{lem:ellipsoid:technical}, we have
\begin{align*}
\norm{\ma_{x}^{T}\left(\vtau'-\vupsilon\right)}_{\mh^{-1}}^{2} & \leq\norm{\vtau'-\vupsilon}_{\mPsi^{-1}}^{2}\\
 & =\sum_{i\in[m]}\frac{1}{\psi(\vx)_{i}}\left(\frac{1}{1+\psi_{a}}\left(\indicVec i\ma_{x}\left(\ma_{x}^{T}\ma_{x}+\lambda\iMatrix\right)^{-1}\frac{\va_{m+1}}{s_{m+1}}\right)^{2}\right)^{2}\\
 & \leq\left(\frac{1}{1+\psi_{a}}\right)^{2}\left(\frac{\va_{m+1}^{T}\left(\ma_{x}^{T}\ma_{x}+\lambda\iMatrix\right)^{-1}\va_{m+1}}{s_{m+1}^{2}}\right)^{2}=\left(\frac{\psi_{a}}{1+\psi_{a}}\right)^{2}
\end{align*}
Combining, we have that 
\begin{eqnarray*}
\delta_{\ve_{\overline{P}}(\vupsilon,\vx)}(\vx) & \leq & \delta_{\ve_{P}(\vtau,\vx)}(\vx)+\left(c_{e}+\frac{\psi_{a}}{1+\psi_{a}}\right)\sqrt{\frac{\psi_{a}}{\mu(\vx)}}+\frac{\psi_{a}}{1+\psi_{a}}\\
 & \leq & \delta_{\ve_{P}(\vtau,\vx)}(\vx)+\left(c_{e}+\psi_{a}\right)\sqrt{\frac{\psi_{a}}{\mu(\vx)}}+\psi_{a}.
\end{eqnarray*}

\end{proof}
We now bound the effect of removing a constraint.
\begin{lem}[Removing a Constraint]
\label{lem:ellip_remove_cont} Let $\ma\in\R^{m\times n}$, $\vb\in\R^{m}$,
$\vtau\in\R^{m}$, and $\vx\in P\defeq\{\vy\in\R^{n}\,:\,\ma\vy\geq\vb\}$.
Let $\overline{\ma}\in\R^{(m-1)\times n}$ be $\ma$ with row $m$
removed, let $\overline{b}\in\R^{m-1}$ denote the first $m-1$ coordinates
of $\vb$, and let $\overline{P}\defeq\{\vy\in\R^{n}\,:\,\overline{\ma}\vy\geq\overline{b}\}$.
Let $\psi_{d}=\psi_{P}(\vx)_{m}$.

Now, let $\vupsilon\in\R^{m-1}$ be defined so that for all $i\in[m-1]$
\[
\upsilon_{i}=\tau_{i}+\frac{1}{1-\psi_{d}}\left(\ma_{x}\left(\ma_{x}^{T}\ma_{x}+\lambda\iMatrix\right)^{-1}\ma_{x}^{T}\indicVec m\right)_{i}^{2}\,.
\]
Assume $\psi_{d}\leq1.1\mu(\vx)\leq\frac{1}{10}$ and $\norm{\ve_{P}(\vtau,\vx)}_{\infty}\leq c_{e}\leq\frac{1}{2}$,
we have the following:\end{lem}
\begin{itemize}
\item {[}Leverage Score Estimation{]} $e_{\overline{P}}(\vec{\upsilon},\vx)_{i}=e_{P}(\vtau,\vx)_{i}$
for all $i\in[m-1]$.
\item {[}Function Value Decrease{]} $p_{\ve_{\overline{p}}(\vupsilon,\vx)}(\vx)=p_{\ve_{P}(\vtau,\vx)}(\vx)+[c_{e}+e_{P}(\vtau,\vx)_{m}]\ln s(\vx)_{m}+\ln(1-\psi_{d})$
\item {[}Centrality Increase{]} $\delta_{\ve_{\overline{p}}(\vupsilon,\vx)}(\vx)\leq\frac{1}{\sqrt{1-2\mu(\vx)}}\delta_{\ve_{P}(\vtau,\vx)}(\vx)+3(c_{e}+\mu(\vx)).$\end{itemize}
\begin{proof}
By (1) in Lemma \ref{lem:sherman_morrison}, we have that for all
$i\in[m-1]$
\[
\psi_{\overline{P}}(\vx)_{i}=\psi_{P}(\vx)_{i}+\frac{1}{1-\psi_{d}}\left(\indicVec i^{T}\ma_{x}\left(\ma_{x}^{T}\ma_{x}+\lambda\iMatrix\right)^{-1}\ma_{x}^{T}\indicVec m\right)^{2}.
\]
Consequently, {[}Leverage Score Estimation{]} holds. Furthermore,
by (3) in Lemma~\ref{lem:sherman_morrison}, this then implies that
{[}Function Value Change{]} holds.

To bound the change in centrality we first note that by (1) and (2)
in Lemma~\ref{lem:sherman_morrison}, we have that the approximate
Hessian for $\overline{P}$, denoted $\overline{\mh}(\vx)$, is bounded
by 
\begin{align*}
\overline{\mh}(\vx)^{-1} & =\left(\mh(\vx)-\ma_{x}^{T}\left(c_{e}\iMatrix+\mPsi_{x}\right)^{1/2}\indicVec m\indicVec m^{T}\left(c_{e}\iMatrix+\mPsi_{x}\right)^{1/2}\ma_{x}\right)^{-1}\\
 & \preceq\left(1+\frac{\alpha}{1-\alpha}\right)\mh(\vx)^{-1}=\left(\frac{1}{1-\alpha}\right)\mh(\vx)^{-1}
\end{align*}
where $\alpha\defeq\indicVec m^{T}\left(c_{e}\iMatrix+\mPsi_{x}\right)^{1/2}\ma_{x}\mh(\vx)^{-1}\ma_{x}^{T}\left(c_{e}\iMatrix+\mPsi_{x}\right)^{1/2}\indicVec m$.
Using $c_{e}+\mu(\vx)\leq\frac{1}{2}+\frac{1}{10}\leq1$, we have
\begin{equation}
\mh(\vx)^{-1}\preceq\left(\ma_{x}^{T}(c_{e}+\mu(\vx))\ma_{x}+\lambda\iMatrix\right)^{-1}\preceq(c_{e}+\mu(\vx))^{-1}\left(\ma_{x}^{T}\ma_{x}+\lambda\iMatrix\right)^{-1}.\label{eq:H_bounded_A_lambda}
\end{equation}
Using this, we have
\begin{align}
\alpha & \leq\left(\frac{c_{e}+\psi_{d}}{c_{e}+\mu(\vx)}\right)\indicVec m^{T}\ma_{x}\left(\ma_{x}^{T}\ma_{x}+\lambda\iMatrix\right)^{-1}\ma_{x}^{T}\indicVec m=\left(\frac{c_{e}+\psi_{d}}{c_{e}+\mu(\vx)}\right)\psi_{d}\,.\label{eq:ellip_alpha_est}
\end{align}
Now let $\vtau'\in\R^{m-1}$ be defined so that $\tau'_{i}=\tau_{i}$
for all $i\in[m-1]$. We have by above that
\[
\delta_{\ve_{\overline{p}}(\vupsilon,\vx)}(\vx)=\norm{\overline{\ma}_{x}^{T}(c_{e}\onesVec+\vupsilon)}_{\overline{\mh}^{-1}}\leq\frac{1}{\sqrt{1-\alpha}}\norm{\overline{\ma}_{x}^{T}(c_{e}\onesVec+\vupsilon)}_{\mh^{-1}}
\]
and therefore, by triangle inequality
\begin{align*}
\norm{\overline{\ma}_{x}^{T}(c_{e}\onesVec+\vupsilon)}_{\mh^{-1}} & \leq\norm{\ma_{x}^{T}(c_{e}\onesVec+\vtau)}_{\mh^{-1}}+\norm{\ma_{x}^{T}\indicVec m(c_{e}+\tau_{m})}_{\mh^{-1}}+\norm{\ma_{x}^{T}(\vtau'-\vupsilon)}_{\mh^{-1}}\\
 & =\delta_{\ve_{P}(\vtau,\vx)}(\vx)+(c_{e}+\tau_{m})\norm{\ma_{x}^{T}\indicVec m}_{\mh^{-1}}+\norm{\ma_{x}^{T}(\vtau'-\vupsilon)}_{\mh^{-1}}\,.
\end{align*}
Now, \eqref{eq:H_bounded_A_lambda} shows that
\[
\norm{\ma_{x}^{T}\indicVec m}_{\mh^{-1}}\leq\frac{1}{\sqrt{c_{e}+\mu(\vx)}}\norm{\ma_{x}^{T}\indicVec m}_{\left(\ma_{x}^{T}\ma_{x}+\lambda\iMatrix\right)^{-1}}\leq\sqrt{\frac{\psi_{d}}{c_{e}+\mu(\vx)}}
\]
Furthermore, since $\mPsi^{-1}\succeq\ma_{x}\left(\ma_{x}^{T}\mPsi\ma_{x}\right)^{-1}\ma_{x}^{T}\succeq\ma_{x}\mh^{-1}\ma_{x}^{T}$,
by Lemma~\ref{lem:ellipsoid:technical} we have
\begin{align*}
\norm{\ma_{x}^{T}\left(\vtau'-\vupsilon\right)}_{\mh^{-1}}^{2} & \leq\norm{\vtau'-\vupsilon}_{\mPsi^{-1}}^{2}\\
 & =\sum_{i\in[m]}\frac{1}{\psi(\vx)_{i}}\left(\frac{1}{1-\psi_{d}}\left(\indicVec i^{T}\ma_{x}\left(\ma_{x}^{T}\ma_{x}+\lambda\iMatrix\right)^{-1}\ma_{x}^{T}\indicVec m\right)^{2}\right)^{2}\\
 & \leq\left(\frac{1}{1-\psi_{d}}\right)^{2}\left(\indicVec m^{T}\ma_{x}\left(\ma_{x}^{T}\ma_{x}+\lambda\iMatrix\right)^{-1}\ma_{x}^{T}\indicVec m\right)^{2}=\left(\frac{\psi_{d}}{1-\psi_{d}}\right)^{2}
\end{align*}
Combining, we have that 
\[
\delta_{\ve_{\overline{p}}(\vupsilon,\vx)}(\vx)\leq\frac{1}{\sqrt{1-\alpha}}\left[\delta_{\ve_{P}(\vtau,\vx)}(\vx)+(c_{e}+\tau_{m})\sqrt{\frac{\psi_{d}}{c_{e}+\mu(\vx)}}+\frac{\psi_{d}}{1-\psi_{d}}\right].
\]
Using the assumption $\psi_{d}\leq1.1\mu(\vx)\leq\frac{1}{10}$, $\norm{\ve_{P}(\vtau,\vx)}_{\infty}\leq c_{e}$
and \eqref{eq:ellip_alpha_est}, we have $\alpha\leq1.1\psi_{d}\leq1.21\mu(\vx)$
and $\tau_{m}\leq\psi_{d}+c_{e}$, and
\begin{eqnarray*}
\delta_{\ve_{\overline{p}}(\vupsilon,\vx)}(\vx) & \leq & \frac{1}{\sqrt{1-1.3\mu(\vx)}}\left[\delta_{\ve_{P}(\vtau,\vx)}(\vx)+(c_{e}+\tau_{m})\sqrt{1.1}+1.2\psi_{d}\right]\\
 & \leq & \frac{1}{\sqrt{1-2\mu(\vx)}}\delta_{\ve_{P}(\vtau,\vx)}(\vx)+\frac{1}{\sqrt{1-\frac{1.3}{10}}}\left(\sqrt{1.1}\cdot2c_{e}+(\sqrt{1.1}+1.2\cdot1.1)\mu(\vx)\right)
\end{eqnarray*}

\end{proof}

\subsection{Hybrid Center Properties\label{sec:ellipsoid:dikin_prop} }

Here we prove properties of points near the hybrid center. First we
bound the distance between points in the $\mh(\vx)$ norm in terms
of the $\ellTwo$ norm of the points.
\begin{lem}
\label{lem:ellipsoid:ellipse_properties} For $\ma\in\R^{m\times n}$
and $\vb\in\R^{m}$ suppose that $\vx\in P=\{\vy:\ma\vy\geq\vb\}$
and $\ve\in\R^{m}$ such that $\normInf{\ve}\leq\frac{1}{2}c_{e}<\frac{1}{20}$
and $\delta_{\ve}\leq0.1\sqrt{c_{e}+\mu(\vx)}$. Then for all $\vy\in P$
we have
\begin{equation}
\norm{\vx-\vy}_{\mh(\vx)}\leq\frac{12mc_{e}+6n+2\lambda\norm{\vy}_{2}^{2}}{\sqrt{c_{e}+\mu(\vx)}}\label{eq:ellipse:ellipse_prop_1}
\end{equation}
and
\[
\norm{\vx}_{2}^{2}\leq4\lambda^{-1}(mc_{e}+n)+2\norm{\vy}_{2}^{2}\enspace.
\]
\end{lem}
\begin{proof}
For notational simplicity let $\vt\defeq c_{e}\onesVec+\ve+\vpsi_{x}$,
$\mt\defeq\mDiag(\vt)$, and $\mq\defeq\ma_{x}^{T}(c_{e}\iMatrix+\mPsi_{x})\ma_{x}$.
We have

\begin{equation}
\norm{\vx-\vy}_{\ma_{x}^{T}\mt\ma_{x}}^{2}=\sum_{i\in[m]}t_{i}\frac{[\vs_{x}-\vs_{y}]_{i}^{2}}{[\vs_{x}]_{i}^{2}}=\sum_{i\in[m]}t_{i}\left(1-2\frac{[\vs_{y}]_{i}}{[\vs_{x}]_{i}}+\frac{[\vs_{y}]_{i}^{2}}{[\vs_{x}]_{i}^{2}}\right)\label{eq:ellipse:hess_bound_1}
\end{equation}
and 
\begin{align}
\sum_{i\in[m]}t_{i}\frac{[\vs_{y}]_{i}^{2}}{[\vs_{x}]_{i}^{2}} & \leq\left(\sum_{i\in[m]}t_{i}\frac{[\vs_{y}]_{i}}{[\vs_{x}]_{i}}\right)\max_{i\in[m]}\frac{[\vs_{y}]_{i}}{[\vs_{x}]_{i}}\leq\left(\sum_{i\in[m]}t_{i}\frac{[\vs_{y}]_{i}}{[\vs_{x}]_{i}}\right)\left(1+\normInf{\ms_{x}^{-1}(\vs_{y}-\vs_{x})}\right)\label{eq:ellipse:hess_bound_2}
\end{align}
and
\begin{align}
\normInf{\ms_{x}^{-1}\left(\vs_{x}-\vs_{y}\right)} & \leq\max_{i\in[m]}\left|\indicVec i\ms_{x}^{-1}\ma\left(\vx-\vy\right)\right|\leq\norm{\vx-\vy}_{\mh(\vx)}\sqrt{\max_{i\in[m]}\left[\ms_{x}^{-1}\ma\mh(\vx)^{-1}\ma^{T}\ms_{x}^{-1}\right]_{ii}}\nonumber \\
 & \leq\left(c_{e}+\mu(\vx)\right)^{-1/2}\norm{\vx-\vy}_{\mh(\vx)}\enspace.\label{eq:ellipse:hess_bound_3}
\end{align}
Now, clearly $\sum_{i\in[m]}t_{i}[\vs_{y}]_{i}/[\vs_{x}]_{i}$ is
positive and since $\norm{\ve}_{\infty}\leq\frac{1}{2}c_{e}$ we know
that $\frac{1}{2}\mq\preceq\ma_{x}^{T}\mt\ma_{x}$. Therefore, by
combining, \eqref{eq:ellipse:hess_bound_1}, \eqref{eq:ellipse:hess_bound_2},
and \eqref{eq:ellipse:hess_bound_3} we have 
\begin{eqnarray}
\frac{1}{2}\norm{\vx-\vy}_{\mq}^{2} & \leq & \norm{\vt}_{1}-\sum_{i\in[m]}t_{i}\frac{[\vs_{y}]_{i}}{[\vs_{x}]_{i}}+\left(\sum_{i\in[m]}t_{i}\frac{[\vs_{y}]_{i}}{[\vs_{x}]_{i}}\right)\frac{\norm{\vx-\vy}_{\mh(\vx)}}{\sqrt{c_{e}+\mu(\vx)}}\nonumber \\
 & \leq & \norm{\vt}_{1}+\left(\sum_{i\in[m]}t_{i}\frac{[\vs_{y}]_{i}}{[\vs_{x}]_{i}}\right)\frac{\norm{\vx-\vy}_{\mh(\vx)}}{\sqrt{c_{e}+\mu(\vx)}}\label{eq:ellipse:hess_bound_4}
\end{eqnarray}

Now since $\grad p_{\ve}(\vx)=-\ma^{T}\ms_{x}^{-1}\mt\onesVec+\lambda\vx$
we have 
\[
\innerproduct{\vx-\vy}{\grad p_{\ve}(\vx)}=-\sum_{i\in[m]}t_{i}\frac{[\vs_{x}-\vs_{y}]_{i}}{[\vs_{x}]_{i}}+\lambda\vx^{T}(\vx-\vy)
\]
and therefore by Cauchy Schwarz and $\vx^{T}\vy\leq\norm{\vx}_{2}^{2}+\frac{1}{4}\norm{\vy}_{2}^{2}$,
\begin{align}
\sum_{i\in[m]}t_{i}\frac{[\vs_{y}]_{i}}{[\vs_{x}]_{i}} & =\norm{\vt}_{1}-\lambda\norm{\vx}_{2}^{2}+\lambda\vx^{T}\vy+\innerproduct{\vx-\vy}{\grad p_{\ve}(\vx)}\label{eq:grad_eq_before}\\
 & \leq\norm t_{1}+\frac{\lambda}{4}\norm{\vy}_{2}^{2}+\norm{\vx-\vy}_{\mh(\vx)}\delta_{\ve}(\vx)\enspace.\label{eq:grad_eq}
\end{align}
Now, using \eqref{eq:ellipse:hess_bound_4}, \eqref{eq:grad_eq} and
the definition of $\mh(\vx)$, we have
\begin{eqnarray*}
\frac{1}{2}\norm{\vx-\vy}_{\mh(\vx)}^{2} & = & \frac{1}{2}\norm{\vx-\vy}_{\mq}^{2}+\frac{\lambda}{2}\norm{\vx-\vy}_{2}^{2}\\
 & \leq & \norm{\vt}_{1}+\left(\sum_{i\in[m]}t_{i}\frac{[\vs_{y}]_{i}}{[\vs_{x}]_{i}}\right)\frac{\norm{\vx-\vy}_{\mh(\vx)}}{\sqrt{c_{e}+\mu(\vx)}}+\frac{\lambda}{2}\norm{\vx-\vy}_{2}^{2}\\
 & \leq & \norm{\vt}_{1}+\frac{\norm t_{1}+\frac{\lambda}{4}\norm{\vy}_{2}^{2}}{\sqrt{c_{e}+\mu(\vx)}}\norm{\vx-\vy}_{\mh(\vx)}+\delta_{e}(\vx)\frac{\norm{\vx-\vy}_{\mh(\vx)}^{2}}{\sqrt{c_{e}+\mu(\vx)}}+\frac{\lambda}{2}\norm{\vx-\vy}_{2}^{2}.
\end{eqnarray*}
Using the fact that $\delta_{e}(\vx)\leq0.1\sqrt{c_{e}+\mu(\vx)}$,
we have
\begin{equation}
\frac{1}{4}\norm{\vx-\vy}_{\mh(\vx)}^{2}\leq\norm{\vt}_{1}+\frac{\lambda}{2}\norm{\vx-\vy}_{2}^{2}+\frac{\norm t_{1}+\frac{\lambda}{4}\norm{\vy}_{2}^{2}}{\sqrt{c_{e}+\mu(\vx)}}\norm{\vx-\vy}_{\mh(\vx)}.\label{eq:ellip_H_est}
\end{equation}
Furthermore, since $\sum_{i\in[m]}t_{i}[\vs_{y}]_{i}/[\vs_{x}]_{i}$
is positive, \eqref{eq:grad_eq_before} shows that
\[
\lambda\vx^{T}(\vx-\vy)=\lambda\norm{\vx}_{2}^{2}-\lambda\vx^{T}\vy\leq\norm{\vt}_{1}+\innerproduct{\vx-\vy}{\grad p_{\ve}(\vx)}\leq\norm{\vt}_{1}+\norm{\vx-\vy}_{\mh(\vx)}\delta_{e}(\vx)
\]
and hence
\begin{align}
\frac{\lambda}{2}\norm{\vx-\vy}_{2}^{2} & \leq\frac{\lambda}{2}\norm{\vx-\vy}_{2}^{2}+\frac{\lambda}{2}\norm{\vx}_{2}^{2}=\lambda\vx^{T}(\vx-\vy)+\frac{\lambda}{2}\norm{\vy}_{2}^{2}\nonumber \\
 & \leq\norm{\vt}_{1}+\frac{\lambda}{2}\norm{\vy}_{2}^{2}+\norm{\vx-\vy}_{\mh(\vx)}\delta_{e}(\vx)\enspace.\label{eq:ellip_x_y_dist}
\end{align}
Putting \eqref{eq:ellip_x_y_dist} into \eqref{eq:ellip_H_est} and
using the fact that $\delta_{e}(\vx)\leq0.1\sqrt{c_{e}+\mu(\vx)}$,
we have
\[
\frac{1}{4}\norm{\vx-\vy}_{\mh(\vx)}^{2}\leq2\norm{\vt}_{1}+\frac{\lambda}{2}\norm{\vy}_{2}^{2}+\left(0.1+\frac{\norm{\vt}_{1}+\frac{\lambda}{4}\norm{\vy}_{2}^{2}}{\sqrt{c_{e}+\mu(\vx)}}\right)\norm{\vx-\vy}_{\mh(\vx)}.
\]
Now, using $\norm{\vt}_{1}\leq2mc_{e}+n$, we have 
\[
\frac{1}{4}\norm{\vx-\vy}_{\mh(\vx)}^{2}\leq2\alpha+\left(0.1+\alpha\right)\norm{\vx-\vy}_{\mh(\vx)}\enspace\text{ for }\enspace\alpha=\frac{2mc_{e}+n+\frac{\lambda}{4}\norm{\vy}_{2}^{2}}{\sqrt{c_{e}+\mu(\vx)}}\enspace.
\]
Since $\sqrt{c_{e}+\mu(\vx)}\leq1.05$, we have $\alpha\geq0.9$ and
hence 
\[
\norm{\vx-\vy}_{\mh(\vx)}\leq\frac{0.1+\alpha+\sqrt{(\alpha+0.1)^{2}+2\alpha}}{2\cdot\frac{1}{4}}\leq6\alpha
\]
yielding \eqref{eq:ellipse:ellipse_prop_1}. 

We also have by \eqref{eq:grad_eq_before} and the fact that $\delta_{e}(\vx)\leq0.1\sqrt{c_{e}+\mu(\vx)}$,
\begin{align*}
\lambda\norm{\vx}_{2}^{2} & =\norm t_{1}+\lambda\vx^{T}\vy+\innerproduct{\vx-\vy}{\grad p_{\ve}(\vx)}-\sum_{i\in[m]}t_{i}\frac{[\vs_{y}]_{i}}{[\vs_{x}]_{i}}\\
 & \leq\norm t_{1}+\frac{\lambda}{2}\norm{\vx}_{2}^{2}+\frac{\lambda}{2}\norm{\vy}_{2}^{2}+\norm{\vx-\vy}_{\mh(\vx)}\delta_{e}(\vx)\\
 & \leq\norm t_{1}+\frac{\lambda}{2}\norm{\vx}_{2}^{2}+\frac{\lambda}{2}\norm{\vy}_{2}^{2}+0.1\sqrt{c_{e}+\mu(\vx)}\norm{\vx-\vy}_{\mh(\vx)}
\end{align*}
Hence, using $\norm{\vt}_{1}\leq2mc_{e}+n$ and $\norm{\vx-\vy}_{\mh(\vx)}\leq6\alpha$,
we have 
\begin{align*}
\frac{\lambda}{2}\norm{\vx}_{2}^{2} & \leq\norm t_{1}+\frac{\lambda}{2}\norm{\vy}_{2}^{2}+0.6\left(2mc_{e}+n+\frac{\lambda}{4}\norm{\vy}_{2}^{2}\right)\\
 & \leq\lambda\norm{\vy}_{2}^{2}+2(mc_{e}+n).
\end{align*}

\end{proof}
In the following lemma we show how we can write one hyperplane in
terms of the others provided that we are nearly centered and show
there is a constraint that the central point is close to. 
\begin{lem}
\label{lem:ellipsoid:apposing_a} Let $\ma\in\R^{m\times n}$ and
$\vb\in\R^{m}$ such that $\norm{a_{i}}_{2}=1$ for all $i$. Suppose
that $\vx\in P=\{\vy:\ma\vy\geq\vb\}$ and $\ve\in\R^{m}$ such that
$\normInf{\ve}\leq\frac{1}{2}c_{e}\leq\frac{1}{2}$. Furthermore,
let $\epsilon=\min_{j\in[m]}s_{j}(\vx)$ and suppose that $i=\argmin_{j\in[m]}s_{j}(\vx)$
then
\[
\normFull{\va_{i}+\sum_{j\neq i}\left(\frac{s(x)_{i}}{s(x)_{j}}\right)\left(\frac{c_{e}+e_{j}+\psi_{j}(\vx)}{c_{e}+e_{i}+\psi_{i}(\vx)}\right)\va_{j}}_{2}\leq\frac{2\epsilon}{(c_{e}+\mu(\vx))}\left[\lambda\norm{\vx}_{2}+\delta_{e}(\vx)\sqrt{\frac{mc_{e}+n}{\varepsilon^{2}}+\lambda}\right].
\]
\end{lem}
\begin{proof}
We know that
\begin{align*}
\grad p_{e}(\vx) & =-\ma^{T}\ms_{x}^{-1}(c_{e}\onesVec+\ve+\vpsi_{x})+\lambda\vx\\
 & =\lambda\vx-\sum_{i\in[m]}\frac{(c_{e}+e_{i}+\psi_{i})}{s(\vx)_{i}}\va_{i}
\end{align*}
Consequently, by $\norm{\ve}_{\infty}\leq\frac{1}{2}c_{e}$, and $\psi_{i}(\vx)\geq\mu(\vx)$
\begin{align*}
\normFull{\va_{i}+\sum_{j\neq i}\left(\frac{s(x)_{i}}{s(x)_{j}}\right)\left(\frac{c_{e}+e_{j}+\psi_{j}(\vx)}{c_{e}+e_{i}+\psi_{i}(\vx)}\right)\va_{j}}_{2} & =\frac{s_{i}(\vx)}{(c_{e}+e_{i}+\psi_{i}(\vx))}\normFull{\ma^{T}\ms_{x}^{-1}(c_{e}\onesVec+\ve+\vpsi_{x})}_{2}\\
 & \leq\frac{2\epsilon}{(c_{e}+\mu(\vx))}\left[\lambda\norm{\vx}_{2}+\norm{\grad p_{e}(\vx)}_{2}\right].
\end{align*}
Using $\norm{\va_{i}}=1$, $\sum_{i}\psi_{i}\leq n$, and $s_{i}(\vx)\geq\epsilon$,
we have
\begin{eqnarray}
\tr(\ma_{x}^{T}(c_{e}\iMatrix+\mPsi_{x})\ma_{x}) & = & \tr(\ma_{x}\ma_{x}^{T}(c_{e}\iMatrix+\mPsi_{x}))\nonumber \\
 & = & \sum_{i}\left(c_{e}+\psi_{i}\right)\frac{\norm{a_{i}}_{2}^{2}}{s_{i}^{2}(\vx)}\leq\frac{mc_{e}+n}{\epsilon^{2}}.\label{eq:tr_hessian}
\end{eqnarray}
Hence, we have $\mh(\vx)\preceq\left(\frac{mc_{e}+n}{\epsilon^{2}}+\lambda\right)\iMatrix$
and $\norm{\grad p_{e}(\vx)}_{2}\leq\delta_{e}(\vx)\sqrt{\frac{mc_{e}+n}{\varepsilon^{2}}+\lambda}$
yielding the result.\end{proof}

\subsection{The Algorithm \label{sub:ellip_algo}}

Here, we put all the results in the previous sections to get our ellipsoid
algorithm. Below is a sketch of the pseudocode; we use $c_{a},c_{d},c_{e},c_{\Delta}$
to denote parameters we decide later. 

\begin{algorithm2e}[h]
\caption{Our Cutting Plane Method}

\label{alg:ellipsoid}

\SetAlgoLined

\textbf{Input:} $\ma^{(0)}\in\R^{m\times n}$, $\vb^{(0)}\in\R^{m},$
$\epsilon>0$, and radius $R>0$.

\textbf{Input: }A separation oracle for a non-empty set $K\subset B_{\infty}(R)$.

\textbf{Check: }Throughout the algorithm, if $s_{i}(\vx^{(k)})<\varepsilon$\textbf{
}output $P^{(k)}$.

\textbf{Check: }Throughout the algorithm, if $\vx^{(k)}\in K$, output
$\vx^{(k)}$.

Set $P^{(0)}=B_{\infty}(R)$.

Set $\vx^{(0)}:=\vzero$ and compute $\tau_{i}^{(0)}=\psi_{P^{(0)}}(\vx^{(0)})_{i}$
for all $i\in[m]$ exactly.

\For{ $k=0$ to $\infty$ }{

Let $m^{(k)}$ be the number of constraints in $P^{(k)}$.

Compute $\vw^{(k)}$ such that $\mPsi_{P^{(k)}}(\vx^{(k)})\preceq\mWeight^{(k)}\preceq(1+c_{\Delta})\mPsi_{P^{(k)}}(\vx^{(k)}).$ 

Let $i^{(k)}\in\argmax_{i\in[m^{(k)}]}\left|w_{i}^{(k)}-\tau_{i}^{(k)}\right|$.

Set $\tau_{i^{(k)}}^{(k+\frac{1}{3})}=\psi_{P^{(k)}}(\vx^{(k)})_{i^{(k)}}$
and $\tau_{j}^{(k+\frac{1}{3})}=\tau_{j}^{(k)}$ for all $j\neq i^{(k)}$.

\uIf{ $\min_{i\in[m^{(k)}]}w_{i}^{(k)}\leq c_{d}$ }{

Remove constraint with minimum $w_{i}^{(k)}$ yielding polytope $P^{(k+1)}.$

Update $\vtau$ according to Lemma \ref{lem:ellip_remove_cont} to
get $\tau_{j}^{(k+\frac{2}{3})}$.

}\uElse{

Use separation oracle at $\vx^{(k)}$ to get a constraint $\{\vx\,:\,\va^{T}\vx\geq\va^{T}\vx^{(k)}\}$
with $\norm{\va}_{2}=1$.

Add constraint $\{\vx\,:\,\va^{T}\vx\geq\va^{T}\vx^{(k)}-c_{a}^{-1/2}\sqrt{\va^{T}(\ma^{T}\ms_{\vx^{(k)}}^{-2}\ma+\lambda\iMatrix)^{-1}\va}\}$
yielding polytope $P^{(k+1)}.$

Update $\vtau$ according to Lemma \ref{lem:ellip_add_cont} to get
$\tau_{j}^{(k+\frac{2}{3})}$.

}

$(\vx^{(k+1)},\vtau^{(k+1)})=\code{Centering}(\vx^{(k)},\vtau^{(k+\frac{2}{3})},200,c_{\Delta}).$

}

\end{algorithm2e}

In the algorithm, there are two main invariants we maintain. First,
we maintain that the centrality $\delta_{P,\ve}(\vx)$, which indicates
how close $\vx$ is to the minimum point of $p_{\ve}$, is small.
Second, we maintain that $\normInf{\ve(\vtau,\vx)}$, which indicates
how accurate the leverage score estimate $\vtau$ is, is small. In
the following lemma we show that we maintain both invariants throughout
the algorithm.
\begin{lem}
\label{lem:ellipsoid_invariant}Assume that $c_{e}\leq c_{d}\leq\frac{1}{10^{6}}$,
$c_{a}\sqrt{c_{a}}\leq\frac{c_{d}}{10^{3}}$, $c_{d}\leq c_{a}$,
and $c_{\Delta}\leq Cc_{e}/\log(n)$ for some small enough universal
constant $C$. During our cutting plane method, for all $k$, with
high probability in $n$, we have
\begin{enumerate}
\item $\normInf{\ve(\vtau^{(k+\frac{1}{3})},\vx^{(k)})}\leq\frac{1}{1000}c_{e}$,
$\normInf{\ve(\vtau^{(k+\frac{2}{3})},\vx^{(k)})}\leq\frac{1}{1000}c_{e}$,
$\normInf{\ve(\vtau^{(k+1)},\vx^{(k+1)})}\leq\frac{1}{400}c_{e}$.
\item $\delta_{P^{(k)},\ve(\vtau^{(k+\frac{2}{3})},\vx^{(k)})}(\vx^{(k)})\leq\frac{1}{100}\sqrt{c_{e}+\min\left(\mu(\vx^{(k)}),c_{d}\right)}$.
\item $\delta_{P^{(k+1)},\ve(\vtau^{(k+1)},\vx^{(k+1)})}(\vx^{(k+1)})\leq\frac{1}{400}\sqrt{c_{e}+\min\left(\mu(\vx^{(k+1)}),c_{d}\right)}$.
\end{enumerate}
\end{lem}
\begin{proof}
Some statements of the proof hold only with high probability in $n$;
we omit mentioning this for simplicity. 

We prove by induction on $k$. Note that the claims are written in
order consistent with the algorithm and proving the statement for
$k$ involves bounding centrality at the point $\vx^{(k+1)}$. Trivially
we define, $\vtau^{(-1)}=\vtau^{(-\frac{2}{3})}=\vtau^{(-\frac{1}{3})}=\vtau^{(0)}$
and note that the claims then hold for $k=-1$ as we compute the initial
leverage scores, $\vtau^{(0)}$, exactly and since the polytope is
symmetric we have $\delta_{\ve(\vtau^{(0)},\vx^{(0)})}(\vzero)=0$.
We now suppose they hold for all $r<t$ and show that they hold for
$r=t$. 

We first bound $\delta$. For notational simplicity, let $\eta_{t}\defeq\sqrt{c_{e}+\min\{\mu(\vx^{(t)}),c_{d}\}}$.
By the induction hypothesis we know that $\delta_{P^{(t)},\ve(\vtau^{(t)},\vx^{(t)})}(\vx^{(t)})\leq\frac{1}{400}\eta_{t}$.
Now, when we update $\tau^{(t)}$ to $\tau^{(t+\frac{1}{3})}$, we
set $\ve_{i^{(t)}}$ to $0$. Consequently, Lemma~\ref{lem:ellipsoid:change_in_e}
and the induction hypothesis $\normInf{\ve(\vtau^{(t)},\vx^{(t)})}\leq\frac{1}{400}c_{e}$
show that
\begin{align}
\delta_{P^{(t)},\ve(\vtau^{(t+\frac{1}{3})},\vx^{(t)})}(\vx^{(t)}) & \leq\delta_{P^{(t)},\ve(\vtau^{(t)},\vx^{(t)})}(\vx^{(t)})+\frac{1}{\sqrt{c_{e}+\mu(\vx^{(t)})}}e_{i^{(t)}}(\vtau^{(t)},\vx^{(t)})\nonumber \\
 & \leq\frac{1}{400}\eta_{t}+\frac{\sqrt{c_{e}}}{400}\leq\frac{\eta_{t}}{200}\label{eq:ellip_delta_bound2_new}
\end{align}

Next, we estimate the $\delta$ changes when we remove or add a constraint. 

For the case of removal, we note that it happens only if $\mu(\vx^{(t)})\leq\min_{i}w_{i}\leq c_{d}\leq\frac{1}{10^{6}}$.
Also, the row we remove has leverage score at most $1.1\mu(\vx^{(t)})$
because we pick the row with minimum $w$. Hence, Lemma~\ref{lem:ellip_remove_cont}
and $c_{e}\leq\frac{1}{10^{6}}$ show that
\begin{align*}
\delta_{P^{(t+1)},\ve(\vtau^{(t+\frac{2}{3})},\vx^{(t)})}(\vx^{(t)}) & \leq\frac{1}{\sqrt{1-2\mu(\vx^{(t)})}}\delta_{P^{(t)},\ve(\vtau^{(t+\frac{1}{3})},\vx^{(t)})}(\vx^{(t)})+2.7(c_{e}+\mu(\vx^{(t)}))\\
 & \leq\frac{1}{\sqrt{1-2\cdot10^{-6}}}\left(\frac{\eta_{t}}{200}\right)+3(c_{e}+\mu(\vx^{(t)}))\leq\frac{\eta_{t}}{100}
\end{align*}
where we used the fact $\mu(\vx^{(t)})\leq c_{d}$ and hence $c_{e}+\mu(\vx^{(t)})\leq\sqrt{c_{e}+c_{d}}\eta_{t}\leq\frac{\sqrt{2}}{1000}\eta_{t}$.

For the case of addition, we note that it happens only if $2\mu(\vx^{(t)})\geq\min_{i}w_{i}\geq c_{d}$.
Furthermore, in this case the hyperplane we add is chosen precisely
so that $\psi_{a}=c_{a}$. Furthermore, since $c_{e}\leq c_{d}\leq c_{a}$
by Lemma~\ref{lem:ellip_add_cont} we have that 
\begin{align*}
\delta_{P^{(t+1)},\ve(\vtau^{(t+\frac{2}{3})},\vx^{(t)})}(\vx^{(t)}) & \leq\delta_{P^{(t)},\ve(\vtau^{(t+\frac{1}{3})},\vx^{(t)})}+\left(c_{e}+\psi_{a}\right)\sqrt{\frac{\psi_{a}}{\mu(\vx^{(t)})}}+\psi_{a}\leq\frac{\eta_{t}}{200}+4c_{a}\sqrt{\frac{c_{a}}{c_{d}}}\enspace.
\end{align*}
 Furthermore, since $c_{a}\sqrt{c_{a}}\leq\frac{c_{d}}{1000}$, $\mu(\vx^{(t)})\geq c_{d}/2$,
and $c_{d}\leq10^{-6}$ we know that $4c_{a}\sqrt{c_{a}/c_{d}}\leq\frac{1}{200}\eta_{t}$
and consequently in both cases we have $\delta_{P^{(t+1)},\ve(\vtau^{(t+\frac{2}{3})},\vx^{(t)})}(\vx^{(t)})\leq\frac{1}{100}\eta_{t}$.

Now, note that Lemmas \ref{lem:ellip_add_cont} and \ref{lem:ellip_remove_cont}
show that $\ve$ does not change during the addition or removal of
an constraint. Hence, we have $\normInf{\ve(\vtau^{(t+\frac{2}{3})},\vx^{(t)})}\leq\normInf{\ve(\vtau^{(t+\frac{1}{3})},\vx^{(t)})}$.
Furthermore, we know the step ``$\vtau_{i^{(k)}}^{(k+\frac{1}{3})}=\psi_{P^{(k)}}(\vx^{(k)})_{i^{(k)}}$''
only decreases $\norm{\ve}_{\infty}$ and hence we have $\normInf{\ve(\vtau^{(t+\frac{2}{3})},\vx^{(t)})}\leq\normInf{\ve(\vtau^{(t)},\vx^{(t)})}\leq\frac{c_{e}}{400}$.
Thus, we have all the conditions needed for Lemma~\ref{lem:ellip_centering}
and consequently
\[
\delta_{P^{(t+1)},\ve(\vtau^{(t+1)},\vx^{(t+1)})}(\vx^{(t+1)})\leq2\left(1-\frac{1}{64}\right)^{200}\delta_{P^{(t+1)},\ve(\vtau^{(t+\frac{2}{3})},\vx^{(t)})}(\vx^{(t)})\leq\frac{1}{1000}\eta_{t}\enspace.
\]
Lemma~\ref{lem:ellip_centering} also shows that that $\normFull{\frac{s(\vx^{(t+1)})-s(\vx^{(t)})}{s(\vx^{(t)})}}_{2}\leq\frac{1}{10}$
and hence $\psi_{i}(\vx^{(t)})\leq2\psi_{i}(\vx^{(t+1)})$ for all
$i$. Therefore, $\eta_{t}\leq2\eta_{t+1}$ and thus
\[
\delta_{P^{(t+1)},\ve(\vtau^{(t+1)},\vx^{(t+1)})}(\vx^{(t+1)})\leq\frac{\sqrt{c_{e}+\min\left(c_{d},\mu(\vx^{(t+1)})\right)}}{400}.
\]
completing the induction case for $\delta$.

Now, we bound $\norm{\ve}_{\infty}$. Lemma \ref{lem:ellip_add_cont}
and \ref{lem:ellip_remove_cont} show that $\ve$ does not change
during the addition or removal of an constraint. Hence, $\ve$ is
affected by only the update step ``$\tau_{i^{(k)}}^{(k+\frac{1}{2})}=\psi_{P^{(k)}}(\vx^{(k)})_{i^{(k)}}$''
and the centering step. Using the induction hypothesis $\delta_{P^{(r)},\ve(\vtau^{(r+\frac{2}{3})},\vx^{(r)})}(\vx^{(r)})\leq\frac{1}{100}\eta_{r}$
and Lemma~\ref{lem:ellip_centering} shows that $\E\ve(\vtau^{(r+1)},\vx^{(r+1)})=\ve(\vtau^{(r+\frac{2}{3})},\vx^{(r)})$
and $\norm{\ve(\vtau^{(r+1)},\vx^{(r+1)})-\ve(\vtau^{(r+\frac{2}{3})},\vx^{(r)})}_{2}\leq\frac{1}{10}c_{\Delta}$
for all $r\leq t$. The goal for the update step is to decrease $\ve$
by updating $\vtau$. In Section~\ref{sec:chasing_0}, we give a
self-contained analysis of the effect of this step as a game. In each
round, the vector $\ve$ is corrupted by some mean $0$ and bounded
variance noise and the problem is how to update $\ve$ such that $\norm{\ve}_{\infty}$
is bounded. Theorem~\ref{thm:chasing_zero} shows that we can do
this by setting the $\ve_{i}=0$ for the almost maximum coordinate
in each iteration. This is exactly what the update step is doing.
Hence, Theorem~\ref{thm:chasing_zero} shows that this strategy guarantees
that after the update step, we have 
\[
\normFull{\ve(\tau^{(r+\frac{1}{3})},\vx^{(r)})}_{\infty}=O\left(c_{\Delta}\log\left(n\right)\right)
\]
for all $r\leq t$. Now, by our choice of $c_{\Delta}$, we have $\normInf{\ve(\vtau^{(t+\frac{1}{3})},\vx^{(t)})}\leq\frac{1}{1000}c_{e}$.
Lemma~\ref{lem:ellip_add_cont} and \ref{lem:ellip_remove_cont}
show that $\ve$ does not change during the addition or removal of
an constraint. Hence, we have $\normInf{\ve(\vtau^{(t+\frac{2}{3})},\vx^{(t)})}\leq\frac{1}{1000}c_{e}$.
Now, we note that again Lemma~\ref{lem:ellip_centering} shows $\norm{\ve(\vtau^{(t+1)},\vx^{(t+1)})-\ve(\vtau^{(t+\frac{2}{3})},\vx^{(t)})}_{2}\leq\frac{1}{10}c_{\Delta}\leq\frac{1}{1000}c_{e}$,
and we have $\normInf{\ve(\vtau^{(t+1)},\vx^{(t+1)})}\leq\frac{c_{e}}{400}$.
This finishes the induction case for $\norm{\ve}_{\infty}$ and proves
this lemma.
\end{proof}
Next, we show the number of constraints is always linear to $n$.
\begin{lem}
\label{lem:small_constraint} Throughout our cutting plane method,
there are at most $1+\frac{2n}{c_{d}}$ constraints.\end{lem}
\begin{proof}
We only add a constraint if $\min_{i}w_{i}\geq c_{d}$. Since $2\psi_{i}\geq w_{i}$,
we have $\psi_{i}\geq\frac{c_{d}}{2}$ for all $i$. Letting $m$
denote the number of constraints after we add that row, we have $n\geq\sum_{i}\psi_{i}\geq(m-1)(c_{d}/2)$.
\end{proof}
Using $K\neq\emptyset$ and $K\subset B_{\infty}(R)$, here we show
that the points are bounded.
\begin{lem}
\label{lem:boundedness} During our Cutting Plane Method, for all
$k$, we have $\norm{\vx^{(k)}}_{2}\leq6\sqrt{n/\lambda}+2\sqrt{n}R$. \end{lem}
\begin{proof}
By Lemma~\ref{lem:ellipsoid_invariant} and Lemma~\ref{lem:ellipsoid:ellipse_properties}
we know that $\norm{\vx^{(k)}}_{2}^{2}\leq4\lambda^{-1}(mc_{e}+n)+2\norm{\vy}_{2}^{2}$
for any $\vy\in P^{(k)}$. Since our method never cuts out any point
in $K$ and since $K$ is nonempty, there is some $\vy\in K\subset P^{(k)}$.
Since $K\subset B_{\infty}(R)$, we have $\norm{\vy}_{2}^{2}\leq nR$.
Furthermore, by Lemma~\ref{lem:small_constraint} we have that $mc_{e}\leq c_{e}+2n\leq3n$
yielding the result.\end{proof}
\begin{lem}
\label{lem:slack_boundedness} $s_{i}\left(\vx^{(k)}\right)\leq12\sqrt{n/\lambda}+4\sqrt{n}R+\sqrt{\frac{1}{c_{a}\lambda}}$
for all $i$ and $k$ in the our cutting plane method.\end{lem}
\begin{proof}
Let $\vx^{(j)}$ be the current point at the time that the constraint
corresponding to $s_{i}$, denoted $\{\vx:\va_{i}^{T}\vx\geq a_{i}^{T}\vx^{(j)}-s_{i}(\vx^{(j)})\}$,
was added. Clearly 
\[
s_{i}(\vx^{(k)})=\va_{i}^{T}\vx^{(k)}-a_{i}^{T}\vx^{(j)}+s_{i}(\vx^{(j)})\leq\norm{\va_{i}}\cdot\norm{\vx^{(k)}}+\left|\va_{i}^{T}\vx^{(j)}-s_{i}(\vx^{(j)})\right|\enspace.
\]
On the one hand, if the constraint for $s_{i}$ comes from the initial
symmetric polytope $P^{(0)}=B_{\infty}(R)$, we know $\left|\va^{T}\vx^{(j)}-\vs_{i}(\vx^{(j)})\right|\leq R$
. On the other hand, if the constraint was added later then we know
that 
\[
s_{i}(\vx^{(j)})=c_{a}^{-1/2}\sqrt{\va^{T}(\ma^{T}\ms_{\vx^{(j)}}^{-2}\ma+\lambda\iMatrix)^{-1}\va}\leq(c_{a}\lambda)^{-1/2}
\]
 and $\left|\va^{T}\vx^{(j)}-s_{i}(\vx^{(j)})\right|\leq\norm{\va_{i}}\cdot\norm{\vx^{(j)}}+\left|s_{i}(\vx^{(j)})\right|$.
Since $\norm{\va_{i}}_{2}=1$ by design and $\norm{\vx^{(j)}}_{2}$
and $\norm{\vx^{(k)}}_{2}$ are upper bounded by $6\sqrt{n/\lambda}+2\sqrt{n}R$
by Lemma \ref{lem:boundedness}, in either case the result follows. 
\end{proof}
Now, we have everything we need to prove that the potential function
is increasing in expectation. 
\begin{lem}
\label{lem:submartinagle} Under the assumptions of Lemma~\ref{lem:ellipsoid_invariant}
if $\lambda=\frac{1}{c_{a}R^{2}}$, $c_{e}=\frac{c_{d}}{6\ln(17nR/\varepsilon)}$,
and $24c_{d}\leq c_{a}\leq\frac{1}{3}$ then for all $k$ we have
\begin{eqnarray*}
\mathbb{E}p_{\ve(\vtau^{(k+1)},\vx^{(k+1)})}(\vx^{(k+1)}) & \geq & p_{\ve(\vtau^{(k)},\vx^{(k)})}(\vx^{(k)})-c_{d}+\ln(1+\beta)
\end{eqnarray*}
where $\beta=c_{a}$ for the case of adding a constraint $\beta=-c_{d}$
for the case of removal. \end{lem}
\begin{proof}
Note that there are three places which affect the function value,
namely the update step for $\tau^{(k+\frac{1}{3})}$, the addition/removal
of constraints, and the centering step. We bound the effect of each
separately.

First, for the update step, we have 
\[
p_{\ve(\vtau^{(k+\frac{1}{3})},\vx^{(k)})}(\vx^{(k)})=-e_{i^{(k)}}\log(s_{i^{(k)}}(\vx^{(k)}))+p_{\ve(\vtau^{(k)},\vx^{(k)})}(\vx^{(k)}).
\]
Lemma~\ref{lem:slack_boundedness}, the termination condition and
$\lambda=\frac{1}{c_{a}R^{2}}$ ensure that
\begin{equation}
\varepsilon\leq s_{i^{(k)}}(\vx^{(k)})\leq12\sqrt{n/\lambda}+4\sqrt{n}R+\sqrt{\frac{1}{c_{a}\lambda}}\leq17\sqrt{n}R\label{eq:bound_on_s}
\end{equation}
and Lemma \ref{lem:ellipsoid_invariant} shows that $\left|e_{i^{(k)}}\right|\leq c_{e}$.
Hence, we have
\[
p_{\ve(\vtau^{(k+\frac{1}{3})},\vx^{(k)})}(\vx^{(k)})\geq p_{\ve(\vtau^{(k)},\vx^{(k)})}(\vx^{(k)})-c_{e}\log(17nR/\varepsilon).
\]

For the addition step, Lemma \ref{lem:ellip_add_cont} shows that
\begin{eqnarray*}
p_{\ve(\vtau^{(k+\frac{2}{3})},\vx^{(k)})}(\vx^{(k)}) & = & p_{\ve(\vtau^{(k+\frac{1}{3})},\vx^{(k)})}(\vx^{(k)})-c_{e}\ln s(\vx)_{m+1}+\ln(1+c_{a})\\
 & \geq & p_{\ve(\vtau^{(k+\frac{1}{3})},\vx^{(k)})}(\vx^{(k)})-c_{e}\log(17nR/\varepsilon)+\ln(1+c_{a})
\end{eqnarray*}
and for the removal step, Lemma~\ref{lem:ellip_remove_cont} and
$\left|e_{i}\right|\leq c_{e}$ shows that
\begin{eqnarray*}
p_{\ve(\vtau^{(k+\frac{2}{3})},\vx^{(k)})}(\vx^{(k)}) & \geq & p_{\ve(\vtau^{(k+\frac{1}{3})},\vx^{(k)})}(\vx^{(k)})-[c_{e}+e_{P}(\vtau,\vx)_{m}]\ln s(\vx)_{m}+\ln(1-c_{d})\\
 & \geq & p_{\ve(\vtau^{(k+\frac{1}{3})},\vx^{(k)})}(\vx^{(k)})-2c_{e}\log(17nR/\varepsilon)+\ln(1-c_{d})
\end{eqnarray*}
 After the addition or removal of a constraint, Lemma~\ref{lem:ellipsoid_invariant}
shows that 
\[
\delta_{P^{(k)},\ve(\vtau^{(k+\frac{2}{3})},\vx^{(k)})}(\vx^{(k)})\leq\frac{1}{100}\sqrt{c_{e}+\min\left(\mu(\vx^{(k)}),c_{d}\right)}
\]
and therefore Lemma~\ref{lem:ellip_centering} and $c_{e}\leq c_{d}$
show that
\begin{eqnarray*}
\mathbb{E}p_{\ve(\vtau^{(k+1)},\vx^{(k+1)})}(\vx^{(k+1)}) & \geq & p_{\ve(\vtau^{(k+\frac{2}{3})},\vx^{(k)})}(\vx^{(k)})-8\left(\frac{\sqrt{c_{e}+\min\left(\mu(\vx^{(k)}),c_{d}\right)}}{100}\right)^{2}\\
 & \geq & p_{\ve(\vtau^{(k+\frac{2}{3})},\vx^{(k)})}(\vx^{(k)})-\frac{c_{d}}{625}.
\end{eqnarray*}

Combining them with $c_{e}=\frac{c_{d}}{6\ln(17nR/\varepsilon)}$,
we have
\begin{eqnarray*}
\mathbb{E}p_{\ve(\vtau^{(k+1)},\vx^{(k+1)})}(\vx^{(k+1)}) & \geq & p_{\ve(\vtau^{(k)},\vx^{(k)})}(\vx^{(k)})-3c_{e}\log(17nR/\varepsilon)-\frac{c_{d}}{625}+\ln(1+\beta)\\
 & \geq & p_{\ve(\vtau^{(k)},\vx^{(k)})}(\vx^{(k)})-c_{d}+\ln(1+\beta)
\end{eqnarray*}
where $\beta=c_{a}$ for the case of addition and $\beta=-c_{d}$
for the case of removal. \end{proof}
\begin{thm}
\label{thm:parameter_choosing} For $c_{a}=\frac{1}{10^{10}}$, $c_{d}=\frac{1}{10^{12}}$,
$c_{e}=\frac{c_{d}}{6\ln(17nR/\varepsilon)}$, $c_{\Delta}=\frac{Cc_{e}}{\log(n)}$
and $\lambda=\frac{1}{c_{a}R^{2}}$ for some small enough universal
constant $C$, then we have
\begin{eqnarray*}
\mathbb{E}p_{\ve(\vtau^{(k+1)},\vx^{(k+1)})}(\vx^{(k+1)}) & \geq & p_{\ve(\vtau^{(k)},\vx^{(k)})}(\vx^{(k)})-\frac{1}{10^{11}}+\frac{9\beta}{10^{11}}
\end{eqnarray*}
where $\beta=1$ for the case of addition and $\beta=0$ for the case
of removal. \end{thm}
\begin{proof}
It is easy to see that these parameters satisfy the requirements of
Lemma~\ref{lem:submartinagle}. 
\end{proof}

\subsection{Guarantees of the Algorithm \label{sub:ellip_guarantee}}

In this section we put everything together to prove Theorem~\ref{thm:main_result},
the main result of this section, providing the guarantees of our cutting
plane method. 

For the remainder of this section we assume that $c_{a}=\frac{1}{10^{10}}$,
$c_{d}=\frac{1}{10^{12}}$, $c_{e}=\frac{c_{d}}{6\ln(17nR/\varepsilon)}$,
$c_{\Delta}=\frac{Cc_{e}}{\log(n)}$ and $\lambda=\frac{1}{c_{a}R^{2}}$.
Consequently, throughout the algorithm we have
\begin{equation}
\norm{\vx}_{2}\leq6\sqrt{n/\lambda}+2\sqrt{n}R=6\sqrt{c_{a}nR^{2}}+2\sqrt{n}R\leq3\sqrt{n}R.\label{eq:x_bounded}
\end{equation}

\begin{lem}
\label{lem:P_too_thin}If $s_{i}(\vx^{(k)})<\epsilon$ for some $i$
and $k$ during our Cutting Plane Method then 
\[
\max_{\vy\in P^{(k)}\cap B_{\infty}(R)}\left\langle \va_{i},\vy\right\rangle -\min_{\vy\in P^{(k)}\cap B_{\infty}(R)}\left\langle \va_{i},\vy\right\rangle \leq\frac{8n\varepsilon}{c_{a}c_{e}}.
\]
\end{lem}
\begin{proof}
Let $\vy\in P^{(k)}\cap B_{\infty}(R)$ be arbitrary. Since $\vy\in B_{\infty}(R)$
clearly $\norm{\vy}_{2}^{2}\leq nR^{2}$. Furthermore, by Lemma~\ref{lem:small_constraint}
and the choice of parameters $mc_{e}+n\leq3n$. Consequently, by Lemma~\ref{lem:ellipsoid:ellipse_properties}
and the fact that $\lambda=\frac{1}{c_{a}R^{2}}$ and $c_{a}<1$ we
have 
\[
\norm{\vx-\vy}_{\mh(\vx)}\leq\frac{12mc_{e}+6n+2\lambda\norm{\vy}_{2}^{2}}{\sqrt{c_{e}+\mu(\vx)}}\leq\frac{30n+2\frac{n}{c_{a}}}{\sqrt{c_{e}+\mu(\vx)}}\leq\frac{4n}{c_{a}\sqrt{c_{e}}}
\]
and therefore
\[
\normFull{\ms_{x^{(k)}}^{-1}(s(\vx^{(k)})-s(\vy))}_{\infty}\leq\frac{1}{\sqrt{c_{e}}}\normFull{\ms_{x^{(k)}}^{-1}(s(\vx^{(k)})-s(\vy))}_{c_{e}\iMatrix+\mPsi}\leq\frac{4n}{c_{a}c_{e}}\enspace.
\]
Consequently, we have $(1-\frac{4n}{c_{a}c_{e}})s_{i}(\vx^{(k)})\leq s_{i}(\vy)\leq(1+\frac{4n}{c_{a}c_{e}})s_{i}(\vx^{(k)})$
for all $\vy\in P^{(k)}\cap B_{\infty}(R)$.
\end{proof}
Now let us show how to compute a proof (or certificate) that the feasible
region has small width on the direction $\va_{i}$.
\begin{lem}
\label{lem:ellipsoid:pos_writing} Suppose that during some iteration
$k$ for $i=\argmin_{j}s_{j}(\vx^{(k)})$ we have $s_{i}(\vx^{(k)})\leq\epsilon$.
Let $(\vx_{*},\vtau_{*})=\code{Centering}(\vx^{(k)},\vtau^{(k)},64\log(2R/\epsilon),c_{\Delta})$
where $\vtau^{(k)}$ is the $\tau$ at that point in the algorithm
and let 
\[
\va^{*}=\sum_{j\neq i}t_{j}\va_{j}\text{ where }t_{j}=\left(\frac{s(\vx_{*})_{i}}{s(\vx_{*})_{j}}\right)\left(\frac{c_{e}+e_{j}(\vx_{*},\vtau_{*})+\psi_{j}(\vx_{*})}{c_{e}+e_{i}(\vx_{*},\vtau_{*})+\psi_{i}(\vx_{*})}\right).
\]
Then, we have that $\norm{\va_{i}+\va^{*}}_{2}\leq\frac{8\sqrt{n}\epsilon}{c_{a}c_{e}R}$
and $t_{j}\geq0$ for all $j$. Furthermore, we have
\[
\left(\sum_{j\neq i}^{O(n)}t_{j}a_{j}\right)^{T}\vx_{*}-\sum_{j\neq i}^{O(n)}t_{j}b_{j}\leq\frac{3n}{c_{e}}s(\vx_{*})_{i}\enspace.
\]
\end{lem}
\begin{proof}
By Lemma~\ref{lem:ellip_centering} and Lemma~\ref{lem:ellipsoid_invariant}
we know that $\ve(\vx_{*},\vtau_{*})\leq\frac{1}{2}c_{e}$ and $\delta_{\ve(\vx_{*},\vtau_{*})}\leq\frac{\epsilon}{R}\sqrt{c_{e}+\mu(\vx_{*})}$.
Since $\ve(\vx_{*},\vtau_{*})\leq\frac{1}{2}c_{e}$, we have $t_{j}\geq0$
for all $j$. Furthermore, by Lemma~\ref{lem:ellipsoid:apposing_a}
and \eqref{eq:x_bounded}, we then have that with high probability
in $n$
\begin{align*}
\norm{\va_{i}+\va^{*}}_{2} & \leq\frac{2\epsilon}{(c_{e}+\mu(\vx_{*}))}\left[\lambda\norm{\vx_{*}}_{2}+\delta_{e}(\vx_{*})\sqrt{\frac{mc_{e}+n}{\varepsilon^{2}}+\lambda}\right]\\
 & \leq\frac{2\epsilon}{c_{e}}\left[\frac{1}{c_{a}R^{2}}(3\sqrt{n}R)+\frac{\epsilon}{R}\sqrt{\frac{3n}{\varepsilon^{2}}+\frac{n}{c_{a}R^{2}}}\right]\\
 & \leq\frac{2\epsilon}{c_{e}}\left[\frac{3\sqrt{n}}{c_{a}R}+\frac{\sqrt{3n}}{4}+\frac{2\sqrt{n}}{\sqrt{c_{a}}R}\right].
\end{align*}
Hence, we have
\[
\norm{\va_{i}+\va^{*}}_{2}\leq\frac{8\sqrt{n}\epsilon}{c_{a}c_{e}R}.
\]
By Lemma~\ref{lem:ellipsoid_invariant} we know that $\ve(\vx_{*},\vtau_{*})\leq\frac{1}{2}c_{e}$
and hence
\begin{eqnarray*}
\left(\sum_{j\neq i}^{O(n)}t_{j}a_{j}\right)^{T}\vx_{*}-\sum_{j\neq i}^{O(n)}t_{j}b_{j} & = & \sum_{j\neq i}^{O(n)}t_{j}s(\vx_{*})_{j}=s_{i}(\vx_{*})\sum_{j\neq i}^{O(n)}\left(\frac{c_{e}+e_{j}(\vx_{*},\vtau_{*})+\psi_{j}(\vx_{*})}{c_{e}+e_{i}(\vx_{*},\vtau_{*})+\psi_{i}(\vx_{*})}\right)\\
 & \leq & s_{i}(\vx_{*})\sum_{j\neq i}^{O(n)}\left(\frac{\frac{3}{2}c_{e}+\psi_{j}(\vx_{*})}{\frac{1}{2}c_{e}+\psi_{i}(\vx_{*})}\right)\leq s_{i}(\vx_{*})\sum_{j\neq i}^{O(n)}\left(\frac{3mc_{e}+2n}{c_{e}}\right)\leq\frac{3n}{c_{e}}s_{i}(\vx_{*})
\end{eqnarray*}
\end{proof}
\begin{lem}
\label{lem:p_too_large}During  our Cutting Plane Method, if $p_{\ve}(\vx^{(k)})\geq n\log(\frac{n}{c_{a}\varepsilon})+\frac{6n}{c_{a}}$,
then we have $s_{i}(\vx^{(k)})\leq\varepsilon$ for some $i$.\end{lem}
\begin{proof}
Recall that 
\[
p_{\ve}(\vx^{(k)})=-\sum_{i\in[m]}\left(c_{e}+e_{i}\right)\log s_{i}(\vx^{(k)})_{i}+\frac{1}{2}\log\det\left(\ma^{T}\ms_{x^{(k)}}^{-2}\ma+\lambda\iMatrix\right)+\frac{\lambda}{2}\norm{\vx^{(k)}}_{2}^{2}.
\]
Using $\norm{\vx^{(k)}}\leq3\sqrt{n}R$ \eqref{eq:x_bounded} and
$\lambda=\frac{1}{c_{a}R^{2}}$, we have
\[
p_{\ve}(\vx^{(k)})\leq-\sum_{i\in[m]}\left(c_{e}+e_{i}\right)\log(s(\vx^{(k)})_{i})+\frac{1}{2}\log\det\left(\ma^{T}\ms_{x^{(k)}}^{-2}\ma+\lambda\iMatrix\right)+\frac{5n}{c_{a}}.
\]
Next, we note that $\norm{e_{i}}_{\infty}\leq c_{e}\leq\frac{1}{12\ln(17nR/\varepsilon)}$
and $s_{i}\left(\vx^{(k)}\right)\leq12\sqrt{n/\lambda}+4\sqrt{n}R+\sqrt{\frac{1}{c_{a}\lambda}}\leq6\sqrt{n}R$
(Lemma \ref{lem:slack_boundedness}). Hence, we have
\[
p_{\ve}(\vx^{(k)})\leq\frac{1}{2}\log\det\left(\ma^{T}\ms_{x^{(k)}}^{-2}\ma+\lambda\iMatrix\right)+\frac{6n}{c_{a}}.
\]
Since $p_{\ve}(\vx^{(k)})\geq n\log(\frac{n}{c_{a}\varepsilon})+\frac{6n}{c_{a}}$,
we have $\frac{1}{2}\log\det\left(\ma^{T}\ms_{x^{(k)}}^{-2}\ma+\lambda\iMatrix\right)\geq n\log(\frac{n}{c_{a}\varepsilon})$.
Using $\varepsilon<R$, we have that $\frac{n^{2}}{c_{a}^{2}\varepsilon^{2}}\geq\frac{n^{2}}{\varepsilon^{2}}+\lambda$
and hence
\[
\sum_{i}\log\lambda_{i}\left(\ma^{T}\ms_{x}^{-2}\ma+\lambda\iMatrix\right)\geq n\log\left(\frac{n}{\varepsilon^{2}}+\lambda\right).
\]
Therefore, we have $\log\lambda_{\max}\left(\ma^{T}\ms_{x}^{-2}\ma+\lambda\iMatrix\right)\geq\log\left(\frac{n}{\varepsilon^{2}}+\lambda\right)$.
Hence, we have some unit vector $\vv$ such that $\vv\ma^{T}\ms_{x}^{-2}\ma\vv+\lambda\vv^{T}\vv\geq\frac{n}{\varepsilon^{2}}+\lambda.$
Thus, 
\[
\sum_{i}\frac{\left(\ma\vv\right)_{i}^{2}}{s(\vx^{(k)})_{i}^{2}}\geq\frac{n}{\varepsilon^{2}}.
\]
Therefore there is some $i$ such that $\frac{\left(\ma\vv\right)_{i}^{2}}{s(\vx^{(k)})_{i}^{2}}\geq\frac{1}{\varepsilon^{2}}$.
Since $\va_{i}$ and $\vv$ are unit vectors, we have $1\geq\left\langle \va_{i},\vv\right\rangle ^{2}\geq s(\vx^{(k)})_{i}^{2}/\varepsilon^{2}$
and hence $s(\vx^{(k)})_{i}\leq\varepsilon$.\end{proof}
\begin{lem}
With constant probability, the algorithm ends in $10^{24}n\log(\frac{nR}{\varepsilon})$
iterations. %
\footnote{We have made no effort on improving this constant and we believe it
can be improved to less than $300$ using techniques in \cite{anstreicher1997vaidya,anstreicher1998towards}.%
} \end{lem}
\begin{proof}
Theorem~\ref{thm:parameter_choosing} shows that for all $k$
\begin{eqnarray}
\mathbb{E}p_{\ve(\vtau^{(k+1)},\vx^{(k+1)})}(\vx^{(k+1)}) & \geq & p_{\ve(\vtau^{(k)},\vx^{(k)})}(\vx^{(k)})-\frac{1}{10^{11}}+\frac{9\beta}{10^{11}}\label{eq:parameter_choicing_again}
\end{eqnarray}
where $\beta=1$ for the case of adding a constraint and $\beta=0$
for the case of removing a constraint. Now, for all $t$ consider
the random variable
\[
\mx_{t}=p_{\ve(\vtau^{(t)},\vx^{(t)})}(\vx^{(t)})-\frac{4.5m^{(t)}}{10^{11}}-\frac{3.5t}{10^{11}}
\]
where $m^{(t)}$ is the number of constraints in iteration $t$ of
the algorithm. Then, since $m^{(t+1)}=m^{(t)}-1+2\beta$, \eqref{eq:parameter_choicing_again}
shows that
\begin{align*}
\E\mx_{t+1} & \geq p_{\ve(\vtau^{(t)},\vx^{(t)})}(\vx^{(t)})-\frac{1}{10^{11}}+\frac{9\beta}{10^{11}}-\frac{4.5m^{(t+1)}}{10^{11}}-\frac{3.5(t+1)}{10^{11}}\\
 & =\mx_{t}-\frac{1}{10^{11}}+\frac{9\beta}{10^{11}}-\frac{4.5(-1+2\beta)}{10^{11}}-\frac{3.5}{10^{11}}=\mx_{t}.
\end{align*}
Hence, it is a sub-martingale. Let $\tau$ be the iteration the algorithm
throws out error or outputs $P^{(k)}$. Optional stopping theorem
shows that
\begin{equation}
\mathbb{E}\mx_{\min(\tau,t)}\geq\mathbb{E}\mx_{0}.\label{eq:EXlesX0}
\end{equation}

Using the diameter of $P^{(0)}$ is $\sqrt{n}R$, we have
\begin{eqnarray*}
p_{\vzero}(\vzero) & = & -\sum_{i\in[m^{(0)}]}c_{e}\log s_{i}(\vzero)+\frac{1}{2}\log\det\left(\ma^{T}\ms_{0}^{-2}\ma+\lambda\iMatrix\right)+\frac{\lambda}{2}\norm{\vzero}_{2}^{2}\\
 & \geq & -c_{e}m^{(0)}\log(\sqrt{n}R)+\frac{n}{2}\log\left(\frac{1}{c_{a}R^{2}}\right)\\
 & \geq & -\left(n+c_{e}m^{(0)}\right)\log(\sqrt{n}R).
\end{eqnarray*}
Using $c_{e}=\frac{c_{d}}{6\ln(17nR/\varepsilon)}$, $c_{d}=\frac{1}{10^{12}}$
and $m^{(0)}=2n$, we have
\begin{eqnarray*}
\mx_{0} & \geq & -\left(n+c_{e}m^{(0)}\right)\log(\sqrt{n}R)-\frac{4.5m^{(0)}}{10^{11}}\\
 & \geq & -n\log(\sqrt{n}R)-100n.
\end{eqnarray*}
Therefore, \eqref{eq:EXlesX0} shows that for all $t$ we have 
\begin{eqnarray}
-n(\log(nR)+100) & \leq & \mathbb{E}\mx_{\min(\tau,t)}\nonumber \\
 & = & p\mathbb{E}\left[\mx_{\min(\tau,t)}|\tau<t\right]+(1-p)\mathbb{E}\left[\mx_{\min(\tau,t)}|\tau\geq t\right]\label{eq:EX_estim}
\end{eqnarray}
where $p\defeq\mathbb{P}(\tau<t)$. 

Note that 
\begin{eqnarray*}
\mathbb{E}\left[\mx_{\min(\tau,t)}|\tau\geq t\right] & \leq & \mathbb{E}\left[p_{\ve(\vtau^{(t)},\vx^{(t)})}(\vx^{(t)})|\tau\geq t\right]-\frac{4.5m^{(t)}}{10^{11}}-\frac{3.5t}{10^{11}}.\\
 & \leq & \mathbb{E}\left[p_{\ve(\vtau^{(t)},\vx^{(t)})}(\vx^{(t)})|\tau\geq t\right]-\frac{3.5t}{10^{11}}.
\end{eqnarray*}
Furthermore, by Lemma~\ref{lem:p_too_large} we know that when $p_{\ve(\vtau^{(t)},\vx^{(t)})}(\vx^{(t)})\geq n\log(\frac{n}{c_{a}\varepsilon})+\frac{6n}{c_{a}}$,
there is a slack that is too small and the algorithm terminates. Hence,
we have
\[
\E\left[\mx_{\min(\tau,t)}|\tau\geq t\right]\leq n\log(\frac{n}{c_{a}\varepsilon})+\frac{6n}{c_{a}}-\frac{3.5t}{10^{11}}.
\]
The proof of Lemma~\ref{lem:ellipsoid_invariant} shows that the
function value does not change by more than $1$ in one iteration
by changing $\vx$ and can change by at most $mc_{e}\log(\frac{3nR}{\epsilon})$
by changing $\tau$. Since by Lemma~\ref{lem:small_constraint} we
know that $m\leq1+\frac{2n}{c_{a}}$ and $c_{e}=\frac{c_{d}}{6\ln(17nR/\varepsilon)}$,
we have that $p_{\ve}(\vx)\leq n\log(\frac{n}{c_{a}\varepsilon})+\frac{7n}{c_{a}}$
throughout the execution of the algorithm. Therefore, we have 
\[
\E\left[\mx_{\min(\tau,t)}|\tau\leq t\right]\leq\mathbb{E}_{\tau<t}p_{\ve(\vtau^{(t)},\vx^{(t)})}(\vx^{(t)})\leq n\log(\frac{n}{c_{a}\varepsilon})+\frac{7n}{c_{a}}.
\]
Therefore, \eqref{eq:EX_estim} shows that
\[
-n(\log(nR)+100)\leq n\log\left(\frac{n}{c_{a}\varepsilon}\right)+\frac{7n}{c_{a}}-(1-p)\frac{3.5t}{10^{11}}.
\]
Hence, we have
\begin{eqnarray*}
(1-p)\frac{3.5t}{10^{11}} & \leq & n\log\left(\frac{n}{c_{a}\varepsilon}\right)+\frac{7n}{c_{a}}+n(\log(nR)+100)\\
 & \leq & n\log\left(\frac{Rn^{2}}{c_{a}\varepsilon}\right)+\frac{7n}{c_{a}}+100n\\
 & = & n\log\left(\frac{Rn^{2}}{c_{a}\varepsilon}\right)+8\cdot10^{10}n.
\end{eqnarray*}
Thus, we have
\[
\mathbb{P}(\tau<t)=p\geq1-\frac{1}{t}\left(10^{11}n\log\left(\frac{n^{2}R}{\varepsilon}\right)+10^{22}n\right).
\]

\end{proof}
Now, we gather all the result as follows:
\begin{thm}[Our Cutting Plane Method]
\label{thm:main_result} Let $K\subseteq\R^{n}$ be a non-empty set
contained in a box of radius $R$, i.e. $K\subseteq B_{\infty}(R)$.
For any $\epsilon\in(0,R)$ in expected time $O(n\SO_{\Omega(\varepsilon/\sqrt{n})}(K)\log(nR/\epsilon)+n^{3}\log^{O(1)}(nR/\epsilon))$
our cutting plane method either outputs $\vx\in K$ or finds a polytope
$P=\{\vx\,:\,\ma\vx\geq\vb\}\supseteq K$ such that
\begin{enumerate}
\item $P$ has $O(n)$ many constraints (i.e. $\ma\in\R^{O(n)\times n}$
and $\vb\in\R^{O(n)}$).
\item Each constraint of $P$ is either an initial constraint from $B_{\infty}(R)$
or of the form $\innerproduct{\va}{\vx}\geq b-\delta$ where $\innerproduct{\va}{\vx}\geq b$
is a normalized hyperplane (i.e. $\norm{\va}_{2}=1$) returned by
the separation oracle and $\delta=\Omega\left(\frac{\varepsilon}{\sqrt{n}}\right)$.
\item The polytope $P$ has small width with respect to some direction $\va_{1}$
given by one of the constraints, i.e. 
\[
\max_{\vy\in P\cap B_{\infty}(R)}\innerproduct{\va_{1}}{\vy}-\min_{\vy\in P\cap B_{\infty}(R)}\left\langle \va_{1},\vy\right\rangle \leq O\left(n\varepsilon\ln(R/\varepsilon)\right)\enspace
\]

\item Furthermore, the algorithm produces a proof of the fact above involving
convex combination of the constraints, namely, non-negatives $t_{2},...,t_{O(n)}$
and $\vx\in P$ such that 

\begin{enumerate}
\item $\norm{\vx}_{2}\leq3\sqrt{n}R$, 
\item $\left\Vert \va_{1}+\sum_{i=2}^{O(n)}t_{i}\va_{i}\right\Vert _{2}=O\left(\frac{\epsilon}{R}\sqrt{n}\log(R/\epsilon)\right)$,
\item $\va_{1}^{T}\vx-\vb_{1}\leq\epsilon$,
\item $\left(\sum_{i=2}^{O(n)}t_{i}a_{i}\right)^{T}\vx-\sum_{i=2}^{O(n)}t_{i}b_{i}\leq O(n\epsilon\log(R/\epsilon))\enspace.$
\end{enumerate}
\end{enumerate}
\end{thm}
\begin{proof}
Our algorithm either finds $\vx\in K$ or we have $s_{i}(\vx^{(k)})<\varepsilon$.
When $s_{i}(\vx^{(k)})<\varepsilon$, we apply Lemma \ref{lem:ellipsoid:pos_writing}
to construct the polytope $P$ and the linear combination $\sum_{i=2}^{O(n)}t_{i}\va_{i}$.

Notice that each iteration of our algorithm needs to solve constant
number of linear systems and implements the sampling step to find
$\vDelta^{(k)}\in\R^{n}$ s.t. $\E[\vDelta^{(k)}]=\vpsi(\vx^{(k)})-\vpsi(\vx^{(k-1)})$.
Theorem~\ref{thm:est_leverage_score} shows how to do the sampling
in $\tilde{O}(1)$ many linear systems. Hence, in total, each iterations
needs to solve $\tilde{O}(1)$ many linear systems plus nearly linear
work. To output the proof for (4), we use Lemma \ref{lem:ellipsoid:pos_writing}.

Note that the linear systems the whole algorithm need to solve is
of the form
\[
(\ma^{T}\ms_{x}^{-2}\ma+\lambda\iMatrix)^{-1}\vx=\vy.
\]
where the matrix $\ma^{T}\ms_{x}^{-2}\ma+\lambda\iMatrix$ can be
written as $\overline{\ma}^{T}\md\overline{\ma}$ for the matrix $\overline{\ma}=[\ma\ \iMatrix]$
and diagonal matrix 
\[
\md=\left[\begin{array}{cc}
\ms^{-2} & \mZero\\
\mZero & \lambda\iMatrix
\end{array}\right].
\]
Note that Lemma~\ref{lem:ellip_centering} shows that $\norm{\left(\ms^{(k)}\right)^{-1}(\vs^{(k+1)}-\vs^{(k)})}_{2}\leq\frac{1}{10}$
for the $k^{th}$ and $(k+1)^{th}$ linear systems we solved in the
algorithm. Hence, we have $\norm{\left(\md^{(k)}\right)^{-1}(\vd^{(k+1)}-\vd^{(k)})}_{2}\leq\frac{1}{10}$.
In \cite{lee2015efficient}, they showed how to solve such sequence
of systems in $\tilde{O}(n^{2})$ amortized cost. Moreover, since
our algorithm always changes the constraints by $\delta$ amount where
$\delta=\Omega(\frac{\varepsilon}{\sqrt{n}})$ an inexact separation
oracle $\SO_{\Omega(\varepsilon/\sqrt{n})}$ suffices. (see Def \ref{def:weak_sep}).
Consequently, the total work $O(n\SO_{\Omega(\varepsilon/\sqrt{n})}(K)\log(nR/\varepsilon)+n^{3}\log^{O(1)}(nR/\varepsilon))$.
Note that as the running time holds with only constant probability,
we can restart the algorithm whenever the running time is too large.

To prove (2), we note that from the algorithm description, we know
the constraints are either from $B_{\infty}(R)$ or of the form $\va^{T}\vx\geq\va^{T}\vx^{(k)}-\delta$
where
\[
\delta=\sqrt{\frac{\va^{T}(\ma^{T}\ms_{\vx^{(k)}}^{-2}\ma+\lambda\iMatrix)^{-1}\va}{c_{a}}}.
\]
From the proof of Lemma \ref{lem:p_too_large}, we know that if $\lambda_{\max}(\ma^{T}\ms_{x}^{-2}\ma+\lambda\iMatrix)\geq\frac{n}{\varepsilon^{2}}$,
then there is $s_{i}<\varepsilon$. Hence, we have $\lambda_{\min}((\ma^{T}\ms_{x}^{-2}\ma+\lambda\iMatrix)^{-1})\leq\frac{\varepsilon^{2}}{n}$.
Since $\va$ is a unit vector, we have
\[
\sqrt{\frac{\va^{T}(\ma^{T}\ms_{\vx^{(k)}}^{-2}\ma+\lambda\iMatrix)^{-1}\va}{c_{a}}}\geq\sqrt{\frac{\varepsilon^{2}}{nc_{a}}}.
\]
\end{proof}

\section{Technical Tools\label{sec:ellipsoid:tools}}

In this section we provide stand-alone technical tools we use in our
cutting plane method in Section~\ref{sec:ellipsoid:method}. In Section~\ref{sec:lever_change}
we show how to efficiently compute accurate estimates of changes in
leverage scores using access to a linear system solver. In Section~\ref{sec:chasing_0}
we study what we call the ``Stochastic Chasing $\vzero$ Game''
and show how to maintain that a vector is small in $\ellInf$ norm
by making small coordinate updates while the vector changes randomly
in $\ellTwo$.

\subsection{Estimating Changes in Leverage Scores\label{sec:lever_change} }

\label{sec:est_leveragescore}In previous sections, we needed to compute
leverage scores accurately and efficiently for use in our cutting
plane method. Note that the leverage score definition we used was
\[
\psi(\vWeight)_{i}=\onesVec_{i}^{T}\sqrt{\mWeight}\ma\left(\ma^{T}\mWeight\ma+\lambda\iMatrix\right)^{-1}\ma^{T}\sqrt{\mWeight}\onesVec_{i}
\]
for some $\lambda>0$ which is different from the standard definition
\[
\sigma(\vWeight)_{i}=\onesVec_{i}^{T}\sqrt{\mWeight}\ma\left(\ma^{T}\mWeight\ma\right)^{-1}\ma^{T}\sqrt{\mWeight}\onesVec_{i}.
\]
However, note that the matrix $\ma^{T}\mWeight\ma+\lambda\iMatrix$
can be written as $\overline{\ma}^{T}\md\overline{\ma}$ for the matrix
$\overline{\ma}=[\ma\ \iMatrix]$ and diagonal matrix 
\[
\md=\left[\begin{array}{cc}
\mWeight & \mZero\\
\mZero & \lambda\iMatrix
\end{array}\right].
\]
and therefore computing $\psi$ is essentially strictly easier than
computing typical leverage scores. Consequently, we use the standard
definition $\sigma$ to simplify notation.

In \cite{spielmanS08sparsRes}, Spielman and Srivastava observed that
leverage scores can be written as the norm of certain vectors
\[
\sigma(\vWeight)_{i}=\normFull{\sqrt{\mWeight}\ma\left(\ma^{T}\mWeight\ma\right)^{-1}\ma^{T}\sqrt{\mWeight}\onesVec_{i}}_{2}^{2}
\]
and therefore leverage scores can be approximated efficiently using
dimension reduction. Unfortunately, the error incurred by this approximation
is too large to use inside the cutting point method. In this section,
we show how to efficiently approximate the \textit{change} of leverage
score more accurately. 

In particular, we show how to approximate $\sigma(\vWeight)-\sigma(\vv)$
for any given $\vWeight,\vv$ with $\norm{\log(\vWeight)-\log(\vv)}_{2}\ll1$.
Our algorithm breaks $\sigma(\vWeight)_{i}-\sigma(\vv)_{i}$ into
the sum of the norm of small vectors and then uses the Johnson-Lindenstrauss
dimension reduction to approximate the norm of each vector separately.
Our algorithm makes use of the following version of Johnson-Lindenstrauss.
\begin{lem}
[{\cite{achlioptas2003database}}]\label{lem:JL}Let $0\leq\epsilon\leq\frac{1}{2}$
and let $\vx_{1},...,\vx_{m}\in\R^{n}$ be arbitrary $m$ points.
For $k=O(\epsilon^{-2}\log(m))$ let $\mq$ be a $k\times n$ random
matrix with each entry sampled from $\{-\frac{1}{\sqrt{k}},\frac{1}{\sqrt{k}}\}$
uniformly and independently. Then, $\mathbb{E}\normFull{\mq\vx_{i}}^{2}=\norm{\vx_{i}}^{2}$
for all $i\in[m]$ and with high probability in $m$ we have that
for all $i\in[m]$
\[
(1-\varepsilon)\norm{\vx_{i}}^{2}\leq\normFull{\mq\vx_{i}}^{2}\leq(1+\varepsilon)\norm{\vx_{i}}^{2}\enspace.
\]

\end{lem}
\begin{algorithm2e}[H]

\caption{$\ensuremath{\widehat{h}=\code{LeverageChange}(\ma,\vv,\vw,\epsilon,\alpha)}$}\label{alg:leverage_score_estimation}

\SetAlgoLined

\textbf{Input: }$\ma\in\R^{m\times n}$, $\vv,\vWeight\in\Rm_{>0}$,
$\epsilon\in(0,0.5)$.

\textbf{Given: }$\norm{\mv^{-1}(\vv-\vw)}_{2}\leq\frac{1}{10}$ and
$\ma^{T}\mv\ma$ and $\ma^{T}\mWeight\ma$ are invertible.

Sample $\mq_{d}\in\R^{O(\epsilon^{-2}\log(m))\times n}$ as in Lemma~\ref{lem:JL}.

Let $\hat{d}_{i}=\norm{\mq_{d}\sqrt{\mWeight}\ma\left(\ma^{T}\mWeight\ma\right)^{-1}\ma^{T}\indicVec i}_{2}^{2}$
for all $i\in[n].$

Let $t=O\left(\log(\epsilon^{-1})\right)$.

Sample $\mq_{f}\in\R^{O(\epsilon^{-2}\log(mt))\times n}$ as in Lemma~\ref{lem:JL}.

Pick positive integer $u$ randomly such that $\Pr[u=i]=(\frac{1}{2})^{i}$.

\For{ $j\in\{1,2,\cdots,t\}\cup\{t+u\}$}{

\uIf{ $j$ is even }{

Let $\hat{f}_{i}^{(j)}=\norm{\mq_{f}\sqrt{\mv}\ma\left(\ma^{T}\mv\ma\right)^{-1}\left(\ma^{T}\left(\mv-\mWeight\right)\ma\left(\ma^{T}\mv\ma\right)^{-1}\right)^{\frac{j}{2}}\ma^{T}\onesVec_{i}}_{2}^{2}.$ 

} \Else{

Let $\Delta^{+}\defeq\left(\mv-\mWeight\right)^{+}$, i.e. the matrix
$\mv-\mWeight$ with negative entries set to 0.

Let $\Delta^{-}\defeq\left(\mWeight-\mv\right)^{+}$, i.e. the matrix
$\mWeight-\mv$ with negative entries set to 0.

Let $\hat{\alpha}_{i}^{(j)}=\norm{\mq_{f}\sqrt{\Delta^{+}}\ma(\ma^{T}\mv\ma)^{-1}(\ma^{T}\left(\mv-\mWeight\right)\ma\left(\ma^{T}\mv\ma\right)^{-1})^{\frac{j-1}{2}}\ma^{T}\onesVec_{i}}_{2}^{2}$.

Let $\hat{\beta}_{i}^{(j)}=\norm{\mq_{f}\sqrt{\Delta^{-}}\ma(\ma^{T}\mv\ma)^{-1}(\ma^{T}(\mv-\mWeight)\ma(\ma^{T}\mv\ma)^{-1})^{\frac{j-1}{2}}\ma^{T}\onesVec_{i}}_{2}^{2}.$

Let $\hat{f}_{i}^{(j)}=\hat{\alpha}_{i}^{(j)}-\hat{\beta}_{i}^{(j)}.$}}

Let $\hat{f}_{i}=2^{u}\hat{f}_{i}^{(t+u)}+\sum_{j=1}^{t}\hat{f}_{i}^{(j)}$.

\textbf{Output:} $\hat{h}_{i}=(w_{i}-v_{i})\hat{d}_{i}+v_{i}\hat{f}_{i}$.
for all $i\in[m]$

\end{algorithm2e}
\begin{thm}
\label{thm:est_leverage_score}Let $\ma\in\R^{m\times n}$ and $\vv,\vWeight\in\Rm_{>0}$
be such that $\alpha\defeq\norm{\mv^{-1}(\vv-\vw)}_{2}\leq\frac{1}{10}$
and both $\ma^{T}\mv\ma$ and $\ma^{T}\mWeight\ma$ are invertible.
For any $\epsilon\in(0,0.5)$, Algorithm~\ref{alg:leverage_score_estimation}
generates a random variable $\widehat{h}$ such that $\E\hat{h}=\sigma(\vWeight)-\sigma(\vv)$
and with high probability in $m$, we have $\|\widehat{h}-\left(\sigma(\vWeight)-\sigma(\vv)\right)\|_{2}\leq O\left(\alpha\epsilon\right)$.
Furthermore, the expected running time is $\otilde((\nnz(\ma)+\LO)/\epsilon^{2})$
where $\LO$ is the amount of time needed to apply $\left(\ma^{T}\mv\ma\right)^{-1}$
and $\left(\ma^{T}\mWeight\ma\right)^{-1}$ to a vector.\end{thm}
\begin{proof}
First we bound the running time. To compute $\hat{d}_{i},\hat{f}_{i}^{(j)},\hat{\alpha}_{i}^{(j)},\hat{\beta}_{i}^{(j)}$,
we simply perform matrix multiplications from the left and then consider
the dot products with each of the rows of $\ma$. Naively this would
take time $\otilde((t+u)^{2}\log(mt)(\nnz(\ma)+\LO))$. However, we
can reuse the computation in computing high powers of $j$ to only
take time $\otilde((t+u)\log(mt)(\nnz(\ma)+\LO))$. Now since $\E[u]$
is constant we see that the total running time is as desired. It only
remains to prove the desired properties of $\hat{h}$.

First we note that we can re-write leverage score differences using
\[
\sigma(\vWeight)_{i}-\sigma(\vv)_{i}=\left(w_{i}-v_{i}\right)\left[\ma\left(\ma^{T}\mWeight\ma\right)^{-1}\ma^{T}\right]_{ii}+v_{i}\left[\ma\left(\left(\ma^{T}\mWeight\ma\right)^{-1}-\left(\ma^{T}\mv\ma\right)^{-1}\right)\ma^{T}\right]_{ii}\enspace.
\]
Consequently, for all $i\in[m]$, if we let
\begin{eqnarray*}
d_{i} & \defeq & \onesVec_{i}^{T}\ma\left(\ma^{T}\mWeight\ma\right)^{-1}\ma^{T}\onesVec_{i},\\
f_{i} & \defeq & \onesVec_{i}^{T}\ma\left[\left(\ma^{T}\mWeight\ma\right)^{-1}-\left(\ma^{T}\mv\ma\right)^{-1}\right]\ma^{T}\onesVec_{i}.
\end{eqnarray*}
then
\begin{equation}
\sigma(\vw)_{i}-\sigma(\vv)_{i}=(w_{i}-v_{i})d_{i}+(v_{i})f_{i}\enspace.\label{eq:diff_by_d_and_f}
\end{equation}
We show that $\hat{d}_{i}$ approximates $d$ and $\hat{f}_{i}$ approximate
$\hat{f}$ well enough to satisfy the statements in the Theorem.

First we bound the quality of $\hat{d}_{i}$. Note that $d_{i}=\norm{\sqrt{\mWeight}\ma\left(\ma^{T}\mWeight\ma\right)^{-1}\ma^{T}\onesVec_{i}}_{2}^{2}$.
Consequently, Lemma~\ref{lem:JL} shows that $\E[\hat{d}_{i}]=d_{i}$
and that with high probability in $m$ we have $(1-\epsilon)d_{i}\leq\hat{d}_{i}\leq(1+\epsilon)d_{i}$
for all $i\in[m]$. Therefore, with high probability in $m$, we have
\begin{align*}
\norm{\left(\vWeight-\vv\right)\widehat{d}-\left(\vWeight-\vv\right)\vd}_{2}^{2} & =\sum_{i\in[m]}(w_{i}-v_{i})^{2}\left(\widehat{d}_{i}-d_{i}\right)^{2}\leq\epsilon^{2}\sum_{i\in[m]}(w_{i}-v_{i})^{2}d_{i}^{2}\\
 & =\varepsilon^{2}\sum_{i\in[m]}(w_{i}-v_{i})^{2}\left(\frac{\sigma(\vWeight)_{i}}{\vWeight_{i}}\right)^{2}\leq2\varepsilon^{2}\sum_{i\in[m]}\left(\frac{w_{i}-v_{i}}{v_{i}}\right)^{2}\enspace.
\end{align*}

Next we show how to estimate $f$. Let $\mx\defeq\left(\ma^{T}\mv\ma\right)^{-1/2}\ma^{T}\left(\mv-\mWeight\right)\ma\left(\ma^{T}\mv\ma\right)^{-1/2}$.
By the assumption on $\alpha$ we know $-\frac{1}{2}\mv\prec\mv-\mWeight\prec\frac{1}{2}\mv$
and therefore $-\frac{1}{2}\iMatrix\prec\mx\prec\frac{1}{2}\iMatrix$.
Consequently we have that 
\begin{eqnarray*}
\left(\ma^{T}\mWeight\ma\right)^{-1} & = & \left(\ma^{T}\mv\ma\right)^{-1/2}\left(\iMatrix-\mx\right)^{-1}\left(\ma^{T}\mv\ma\right)^{-1/2}\\
 & = & \sum_{j=0}^{\infty}\left(\ma^{T}\mv\ma\right)^{-1/2}\mx^{j}\left(\ma^{T}\mv\ma\right)^{-1/2}.
\end{eqnarray*}
and therefore
\begin{align*}
f_{i} & =\indicVec i^{T}\ma\left(\sum_{j=0}^{\infty}\left(\ma^{T}\mv\ma\right)^{-1/2}\mx^{j}\left(\ma^{T}\mv\ma\right)^{-1/2}-\left(\ma^{T}\mv\ma\right)^{-1}\right)\ma^{T}\indicVec i\\
 & =\sum_{j=1}^{\infty}f_{i}^{(j)}\enspace\text{ where }\enspace f_{i}^{(j)}\defeq\indicVec i^{T}\ma\left(\ma^{T}\mv\ma\right)^{-1/2}\mx^{j}\left(\ma^{T}\mv\ma\right)^{-1/2}\ma^{T}\indicVec i\enspace.
\end{align*}
Furthermore, using the definition of $\mx$ we have that for even
j
\begin{eqnarray*}
f_{i}^{(j)} & = & \normFull{\mx^{\frac{j}{2}}\left(\ma^{T}\mv\ma\right)^{-1/2}\ma^{T}\onesVec_{i}}_{2}^{2}\\
 & = & \normFull{\left(\ma^{T}\mv\ma\right)^{-1/2}\left(\ma^{T}\left(\mv-\mWeight\right)\ma\left(\ma^{T}\mv\ma\right)^{-1}\right)^{\frac{j}{2}}\ma^{T}\onesVec_{i}}_{2}^{2}\\
 & = & \normFull{\sqrt{\mv}\ma\left(\ma^{T}\mv\ma\right)^{-1}\left(\ma^{T}\left(\mv-\mWeight\right)\ma\left(\ma^{T}\mv\ma\right)^{-1}\right)^{\frac{j}{2}}\ma^{T}\onesVec_{i}}_{2}^{2}
\end{eqnarray*}
For odd $j$, using our definition of $\mDelta^{+}$ and $\mDelta^{-}$
we have that 
\begin{eqnarray*}
f_{i}^{(j)} & = & \onesVec_{i}^{T}\ma\left(\ma^{T}\mv\ma\right)^{-1/2}\mx^{j}\left(\ma^{T}\mv\ma\right)^{-1/2}\ma^{T}\onesVec_{i}\\
 & = & \onesVec_{i}^{T}\ma\left(\left(\ma^{T}\mv\ma\right)^{-1}\ma^{T}\left(\mv-\mWeight\right)\ma\right)^{\frac{j-1}{2}}\left(\ma^{T}\mv\ma\right)^{-1}\ma^{T}\left(\mv-\mWeight\right)\\
 &  & \times\ma\left(\ma^{T}\mv\ma\right)^{-1}\left(\ma^{T}\left(\mv-\mWeight\right)\ma\left(\ma^{T}\mv\ma\right)^{-1}\right)^{\frac{j-1}{2}}\ma^{T}\onesVec_{i}\\
 & = & \alpha_{i}^{(j)}-\beta_{i}^{(j)}
\end{eqnarray*}
where
\begin{eqnarray*}
\alpha_{i}^{(j)} & \defeq & \normFull{\sqrt{\Delta^{+}}\ma\left(\ma^{T}\mv\ma\right)^{-1}\left(\ma^{T}\left(\mWeight-\mv\right)\ma\left(\ma^{T}\mv\ma\right)^{-1}\right)^{\frac{j-1}{2}}\ma^{T}\onesVec_{i}}_{2}^{2}\ ,\\
\beta_{i}^{(j)} & \defeq & \normFull{\sqrt{\Delta^{-}}\ma\left(\ma^{T}\mv\ma\right)^{-1}\left(\ma^{T}\left(\mWeight-\mv\right)\ma\left(\ma^{T}\mv\ma\right)^{-1}\right)^{\frac{j-1}{2}}\ma^{T}\onesVec_{i}}_{2}^{2}\enspace.
\end{eqnarray*}
Consequently, by Lemma~\ref{lem:JL} and the construction, we see
that
\[
\E\hat{f}_{i}=\E\left[\sum_{j=1}^{t}\hat{f}_{i}^{(j)}+\sum_{u=1}^{\infty}\frac{2^{u}}{2^{u}}\hat{f}_{i}^{(t+u)}\right]=\sum_{j=1}^{\infty}f_{i}^{(j)}=f_{i}
\]
and therefore $\E\hat{h}=\sigma(\vw)-\sigma(\vv)$ as desired. All
that remains is to bound the variance of $\hat{f}_{i}$.

To bound the variance of $\hat{f}$, let $\left|\mx\right|=\left(\ma^{T}\mv\ma\right)^{-1/2}\ma^{T}\left|\mWeight-\mv\right|\ma\left(\ma^{T}\mv\ma\right)^{-1/2}$.
Note that $-\frac{1}{4}\iMatrix\preceq-\left|\mx\right|\preceq\mx\preceq\left|\mx\right|\preceq\frac{1}{4}\iMatrix$
and consequently for all $j$

\begin{align*}
g_{i}^{(j)} & \defeq\onesVec_{i}^{T}\ma\left(\ma^{T}\mv\ma\right)^{-1/2}|\mx|^{j}\left(\ma^{T}\mv\ma\right)^{-1/2}\ma^{T}\onesVec_{i}\\
 & \leq\frac{1}{4^{j-1}}\onesVec_{i}^{T}\ma\left(\ma^{T}\mv\ma\right)^{-1/2}|\mx|\left(\ma^{T}\mv\ma\right)^{-1/2}\ma^{T}\onesVec_{i}\\
 & \defeq\frac{1}{v_{i}4^{j-1}}\onesVec_{i}^{T}\mpr_{v}\mDelta\mpr_{v}\onesVec_{i}
\end{align*}
where $\mpr_{v}=\sqrt{\mv}\ma\left(\ma^{T}\mv\ma\right)^{-1/2}\ma^{T}\sqrt{\mv}$
and $\mDelta$ is a diagonal matrix with $\mDelta_{ii}=\left|\frac{w_{i}-v_{i}}{v_{i}}\right|$.
Using that $\mZero\preceq\mpr_{v}\preceq\iMatrix$, we have that for
all $j$

\begin{eqnarray*}
(4^{j-1})^{2}\sum_{i=1}^{m}\left(v_{i}g_{i}^{(j)}\right)^{2} & = & \sum_{i=1}^{m}\left(\onesVec_{i}^{T}\mpr_{v}\mDelta\mpr_{v}\onesVec_{i}\right)^{2}=\tr\left(\mpr_{v}\mDelta\mpr_{v}\mpr_{v}\mDelta\mpr_{v}\right)\\
 & \leq & \tr\left(\mpr_{v}\mDelta\mDelta\mpr_{v}\right)=\tr\left(\mDelta\mpr_{v}\mpr_{v}\mDelta\right)\\
 & \leq & \tr\left(\mDelta^{2}\right)=\sum_{i=1}^{m}\left(\frac{w_{i}-v_{i}}{v_{i}}\right)^{2}\leq\alpha^{2}
\end{eqnarray*}
and thus $\norm{\mv\vg^{(j)}}_{2}\leq\frac{4\alpha}{4^{j}}$. Furthermore,
since $\mDelta^{+}\preceq\left|\mWeight-\mv\right|$ and $\mDelta^{-}\preceq\left|\mWeight-\mv\right|$
we have that $\abs{\alpha_{i}^{(j)}}\leq g_{i}^{(j)}$ and $\abs{\beta_{i}^{(j)}}\leq g_{i}^{(j)}$.
Consequently, by Lemma~\ref{lem:JL} again, we have
\begin{eqnarray*}
\norm{\mv\hat{f}^{(j)}-\mv\vf^{(j)}}_{2}^{2} & = & \sum_{i}v_{i}^{2}\left(\hat{f}_{i}^{(j)}-f_{i}^{(j)}\right)^{2}\\
 & \leq & 2\sum_{i}v_{i}^{2}\left(\hat{\alpha}_{i}^{(j)}-\alpha_{i}^{(j)}\right)^{2}+2\sum_{i}v_{i}^{2}\left(\hat{\beta}_{i}^{(j)}-\beta_{i}^{(j)}\right)^{2}\\
 & \leq & 2\varepsilon^{2}\sum_{i}v_{i}^{2}\left(\left(\alpha_{i}^{(j)}\right)^{2}+\left(\beta_{i}^{(j)}\right)^{2}\right)\\
 & \leq & 2\varepsilon^{2}\sum_{i}\left(v_{i}g_{i}^{(j)}\right)^{2}\leq\frac{2\alpha^{2}\epsilon^{2}}{\left(4^{j-1}\right)^{2}}\ .
\end{eqnarray*}
Putting this all together we have that 
\begin{align*}
\norm{\mv\hat{f}-\mv\vf}_{2} & \leq\norm{2^{u}\mv\hat{f}^{(t+u)}+\sum_{j=1}^{t}\mv\hat{f}^{(j)}-\sum_{j=1}^{\infty}\mv\vf^{(j)}}_{2}\\
 & \leq2^{u}\norm{\mv\hat{f}^{(t+u)}}_{2}+\sum_{j=1}^{t}\norm{\mv\hat{f}^{(j)}-\mv\vf^{(j)}}_{2}+\sum_{j=t+1}^{\infty}\norm{\mv\vf^{(j)}}_{2}\\
 & \leq2^{u}\frac{4\alpha}{4^{t+u}}+\sum_{j=1}^{t}\frac{\sqrt{2}\alpha\epsilon}{4^{j-1}}+\sum_{j=t+1}^{\infty}\frac{4\alpha}{4^{j}}\\
 & =O\left(\alpha\epsilon+\frac{\alpha}{4^{t}}\right).
\end{align*}
Consequently, since $t=O(\log(\epsilon{}^{-1}))$ we have the desired
result.\end{proof}

\subsection{The Stochastic Chasing $\protect\vzero$ Game\label{sec:chasing_0}}

To avoid computing leverage scores exactly, in Section~\ref{sec:lever_change}
we showed how to estimate the difference of leverage scores and use
these to update the leverage scores. However, if we only applied this
technique the error of leverage scores would accumulate in the algorithm
and we need to fix it. Naturally, one may wish to use dimension reduction
to compute a multiplicative approximation to the leverage scores and
update our computed value if the error is too large. However, this
strategy would fail if there are too many rows with inaccurate leverage
scores in the same iteration. In this case, we would change the central
point too much that we are not able to recover. In this section, we
present this update problem in a general form that we call \emph{Stochastic
Chasing 0 game} and provide an effective strategy for playing this
game.

The \emph{Stochastic chasing 0 game} is as follows. There is a player,
a stochastic adversary, and a point $\vx\in\Rm$. The goal of the
player is to keep the point close to $\vzero\in\Rm$ in $\ellInf$
norm and the goal of the stochastic adversary is to move $\vx$ away
from $\vzero.$ The game proceeds for an infinite number of iterations
where in each iteration the stochastic adversary moves the current
point $\vx^{(k)}\in\Rm$ to some new point $\vx^{(k)}+\vDelta^{(k)}\in\Rm$
and the player needs to respond. The stochastic adversary cannot move
the $\vDelta^{(k)}$ arbitrarily, instead he is only allowed to choose
a probability distribution $\mathcal{D}^{(k)}$ and sample $\vDelta^{(k)}$
from it. Furthermore, it is required that $\mathbb{E}_{\mathcal{D}^{(k)}}\vDelta=\vzero$
and $\norm{\vDelta}_{2}^{2}\leq c$ for some fixed $c$ and all $\vDelta\in\mathcal{D}^{(k)}$.
The player does not know $\vx^{(k)}$ or the distribution $\mathcal{D}^{(k)}$
or the move $\vDelta^{(k)}$ of the stochastic adversary. All the
player knows is some $\vy^{(k)}\in\Rn$ that is close to $\vx^{(k)}$
in $\ellInf$ norm. With this information, the player is allowed to
choose one coordinate $i$ and set $x_{i}^{(k+1)}$ to be zero and
for other $j$, we have $x_{j}^{(k+1)}=x_{j}^{(k)}+\Delta_{j}^{(k)}$.

The question we would like to address is, what strategy the player
should choose to keep $\vx^{(k)}$ close to $\vzero$ in $\ellInf$
norm? We show that there is a trivial strategy that performs well:
simply pick the largest coordinate and set it to $0$.

\begin{algorithm2e}[H]
\caption{Stochastic chasing $\vvar{0}$ game}

\label{alg:chasing_zero}

\SetAlgoLined

\textbf{Constant: }$c>0,R>0$.

Let $\vx^{(1)}=\vzero\in\Rm$.

\For{ $k=1$ to $\infty$ }{

\textbf{Stochastic Adversary}: Pick $\mathcal{D}^{(k)}$ such that
$\mathbb{E}_{\mathcal{D}^{(k)}}\vDelta=\vzero$ and $\norm{\vDelta}_{2}\leq c$
all $\vDelta\in\mathcal{D}^{(k)}$.

\textbf{Stochastic Adversary}: Pick $\vy^{(k)}\in\Rm$ such that $\norm{\vy^{(k)}-\vx^{(k)}}_{\infty}\leq R$.

\textbf{Player}: Pick a coordinate $i^{(k)}$ using only $\vy^{(k)}$.

Sample $\vDelta^{(k)}$ from $\mathcal{D}^{(k)}$. 

Set $x_{i^{(k)}}^{(k+1)}=0$ and $x_{j}^{(k+1)}=x_{j}^{(k)}+\Delta_{j}^{(k)}$
for all $j\neq i$.

}

\end{algorithm2e}
\begin{thm}
\label{thm:chasing_zero}Using the strategy $i^{(k)}=\argmax_{i}\left|y_{i}^{(k)}\right|$,
with probability at least $1-p$, we have
\[
\norm{\vx^{(k)}}_{\infty}\leq2(c+R)\log\left(4mk^{2}/p\right)
\]
for all $k$ in the Stochastic Chasing $\vzero$ Game.\end{thm}
\begin{proof}
Consider the potential function $\Phi(\vx)=\sum_{i}e^{\alpha x_{i}}+\sum_{i}e^{-\alpha x_{i}}$
where $\alpha$ is to be determined. Now for all $x$ we know that
$e^{x}\leq1+x+\frac{x^{2}}{2}e^{\left|x\right|}$ and therefore for
all $\left|\delta\right|\leq c$, $x$ and $\alpha$, we have
\[
e^{\alpha x+\alpha\delta}\leq e^{\alpha x}+\alpha\delta e^{\alpha x}+\frac{1}{2}\alpha^{2}\delta^{2}e^{\alpha x+\left|\alpha\right|c}\enspace.
\]
Consequently,
\begin{align*}
\mathbb{E}_{\vDelta\in\mathcal{D}^{(k)}}\Phi(\vx^{(k)}+\vDelta) & \leq\Phi(\vx^{(k)})+\alpha\mathbb{E}_{\vDelta\in\mathcal{D}^{(k)}}\left(\sum_{i\in[m]}e^{\alpha x_{i}^{(k)}}\Delta_{i}-\sum_{i\in[m]}e^{-\alpha x_{i}^{(k)}}\Delta_{i}\right)\\
 & \enspace\enspace+\frac{\alpha^{2}}{2}e^{\alpha\norm{\vDelta}_{\infty}}\mathbb{E}_{\vDelta\in\mathcal{D}^{(k)}}\left(\sum_{i\in[m]}e^{\alpha x_{i}^{(k)}}\Delta_{i}^{2}+\sum_{i\in[m]}e^{-\alpha x_{i}^{(k)}}\Delta_{i}^{2}\right)\enspace.
\end{align*}
 Since $\mathbb{E}_{\mathcal{D}^{(k)}}\vDelta=\vzero$ and $\norm{\vDelta}_{2}\leq c$,
we have $\mathbb{E}_{\vDelta\in\mathcal{D}^{(k)}}\left(\sum_{i}e^{\alpha x_{i}^{(k)}}\Delta_{i}-\sum_{i}e^{-\alpha x_{i}^{(k)}}\Delta_{i}\right)=0$
and
\begin{align*}
\mathbb{E}_{\vDelta\in\mathcal{D}^{(k)}}\left(\sum_{i}e^{\alpha x_{i}^{(k)}}\Delta_{i}^{2}+\sum_{i}e^{-\alpha x_{i}^{(k)}}\Delta_{i}^{2}\right) & \leq\mathbb{E}_{\vDelta\in\mathcal{D}^{(k)}}\left(\sum_{i}\Delta_{i}^{2}\right)\left(\max_{i}e^{\alpha x_{i}^{(k)}}+\max_{i}e^{-\alpha x_{i}^{(k)}}\right)\\
 & \leq c^{2}\left(\max_{i}e^{\alpha x_{i}^{(k)}}+\max_{i}e^{-\alpha x_{i}^{(k)}}\right)\enspace.
\end{align*}
Letting $\eta^{(k)}=\max\left(\max_{i}e^{\alpha x_{i}^{(k)}},\max_{i}e^{-\alpha x_{i}^{(k)}}\right)$,
we then have
\[
\mathbb{E}_{\vDelta\in\mathcal{D}^{(k)}}\Phi(\vx^{(k)}+\vDelta)\leq\Phi(\vx^{(k)})+\alpha^{2}e^{\alpha c}c^{2}\eta^{(k)}.
\]
Since $i^{(k)}=\argmax_{i}\left|y_{i}^{(k)}\right|$ and $\norm{\vy^{(k)}-\vx^{(k)}}_{\infty}\leq R$,
the player setting $x_{i^{(k)}}^{(k+1)}=0$ decreases $\Phi$ by at
least $e^{-\alpha(R+c)}\eta^{(k)}.$ Hence, we have 
\begin{eqnarray*}
\mathbb{E}_{\vDelta\in\mathcal{D}^{(k)}}\Phi(\vx^{(k+1)}) & \leq & \Phi(\vx^{(k)})+\alpha^{2}e^{\alpha c}c^{2}\eta^{(k)}-e^{-\alpha R}\eta^{(k)}.
\end{eqnarray*}
Picking $\alpha=\frac{1}{2(c+R)}$, we have $e^{2\alpha(c+R)}(\alpha(c+R))^{2}\leq1$
and hence $\alpha^{2}e^{\alpha c}c^{2}\leq e^{-\alpha(R+c)}$. Therefore,
we have that
\[
\mathbb{E}_{\vDelta\in\mathcal{D}^{(k)}}\Phi(\vx^{(k+1)})\leq\E\Phi(\vx^{(k)})\leq...\leq\Phi(\vx^{(1)})=2m\,.
\]
Consequently, by Markov's inequality we have that $\Pr[\Phi(\vx^{(k)})\geq\lambda_{k}]\leq\frac{2m}{\lambda_{k}}$
for any $\lambda_{k}$. Furthermore, since clearly $\Phi(\vx)\geq e^{\alpha\|\vx\|_{\infty}}$
we have that $\Pr[\|\vx^{(k)}\|_{\infty}\geq\log(\lambda_{k})/\alpha]\leq\frac{2m}{\lambda_{k}}$
for all $k$. Choosing $\lambda_{k}=\frac{4mk^{2}}{p}$ and taking
a union bound over all $k$, we have that 
\[
\norm{\vx^{(k)}}_{\infty}\leq2(c+R)\log\left(4mk^{2}/p\right)
\]
for all $k$ with probability at least
\[
1-\sum_{i=1}^{\infty}\frac{2m}{\lambda_{k}}=1-\sum_{k=1}^{\infty}\frac{p}{2k^{2}}\geq1-p\enspace.
\]
\end{proof}

\pagebreak{}
\part{\label{part:app} A User's Guide to Cutting Plane Methods }

\section{Introduction}

Cutting plane methods have long been employed to obtain polynomial
time algorithms for solving optimization problems. However, for many
problems cutting plane methods are often regarded as inefficient both
in theory and in practice. Here, in Part~\ref{part:app} we provide
several techniques for applying cutting plane methods efficiently.
Moreover, we illustrate the efficacy and versatility of these techniques
by applying them to achieve improved running times for solving multiple
problems including semidefinite programming, matroid intersection,
and submodular flow. 

We hope these results revive interest in ellipsoid and cutting plane
methods. We believe these results demonstrate how cutting plan methods
are often useful not just for showing that a problem is solvable in
polynomial time, but in many yield substantial running time improvements.
We stress that while some results in Part~\ref{part:app} are problem-specific,
the techniques introduced here are quite general and are applicable
to a wide range of problems.

In the remainder of this introduction we survey the key techniques
we use to apply our cutting plane method (Section~\ref{sec:app:intro:techniques})
and the key results we obtain on improving the running time for solving
various optimization problems (Section~\ref{sec:app:results_and_prev}).
We conclude in Section~\ref{sub:app:overview} by providing an overview
of where to find additional technical result in Part~\ref{part:app}.

\subsection{Techniques \label{sec:app:intro:techniques}}

Although cutting plane methods are typically introduced as algorithms
for finding a point in a convex set (as we did with the feasibility
problem in Part~\ref{part:Ellipsoid}), this is often not the easiest
way to apply the methods. Moreover, improperly applying results on
the feasibility problem to solve convex optimization problems can
lead to vastly sub-optimal running times. Our central goal, here,
in Part~\ref{part:app} is to provide tools that allow cutting plane
methods to be efficiently applied to solve complex optimization problems.
Some of these tools are new and some are extensions of previously
known techniques. Here we briefly survey the techniques we cover in
Section~\ref{sub:conv_opt} and Section~\ref{sec:Intersection}.

\subsection*{Technique 0: From Feasibility to Optimization}

In Section~\ref{sub:app:convex-opt}, we explain how to use our cutting
plane method to solve convex optimization problems using an approximate
subgradient oracle. Our result is based on a result of Nemirovski~\cite{Nemirovski1994}
in which he showed how to use a cutting plane method to solve convex
optimization problems without smoothness assumptions on the function
and with minimal assumptions on the size of the function's domain.
We generalize his proof to accommodate for an approximate separation
oracle, an extension which is essential for our applications. We use
this result as the starting point for two new techniques we discuss
below.

\subsection*{\noindent Technique 1: Dimension Reduction through Duality}

In Section~\ref{sub:sdp}, we discuss how cutting plane methods can
be applied to obtain both primal and dual solutions to convex optimization
problems. Moreover, we show how this can be achieved while only applying
the cutting plane method in the space, primal or dual, which has a
fewer number of variables. Thus we show how to use duality to improve
the convergence of cutting plane methods while still solving the original
problem. 

To illustrate this idea consider the following very simple linear
program (LP)
\[
\min_{x_{i}\geq0,\sum x_{i}=1}\sum_{i=1}^{n}w_{i}x_{i}
\]
where $\vx\in\Rn$ and $\vw\in\Rn$. Although this LP has $n$ variables,
it should to be easy to solve purely on the grounds that it only has
one equality constraint and thus dual linear program is simply 
\[
\max_{y\leq w_{i}\forall i}y\,,
\]
i.e. a LP with only one variable. Consequently, we can apply our cutting
plane method to solve it efficiently. 

However, while this simple example demonstrates how we can use duality
to decrease dimensions, it is not always obvious how to recover the
optimal primal solution $x$ variable given the optimal dual solution
$y$. Indeed, for many problems their dual is significantly simpler
than itself (primal), so some work is required to show that working
in the space suffices to require a primal solution.

One such recent example of this approach proving successful is a recent
linear programming result \cite{leeS14}. In this result, the authors
show how to take advantage of this observation and get a faster LP
solver and maximum flow algorithm. It is interesting to study how
far this technique can extend, that is, in what settings can one recover
the solution to a more difficult dual problem from the solution to
its easier primal problem?

There is in fact another precedent for such an approach. Grötschel,
Lovász and Schrijver\cite{grotschel1988ellipsoid} showed how to obtain
the primal solution for linear program by using a cutting plane method
to solve the linear program exactly. This is based on the observation
that cutting plane methods are able to find the active constraints
of the optimal solution and hence one can take dual of the linear
program to get the dual solution. This idea was further extended in
\cite{krishnan2006unifying} which also observed that cutting plane
methods are incrementally building up a LP relaxation of the optimization
problem. Hence, one can find a dual solution by taking the dual of
that relaxation.

In Section~\ref{sub:sdp}, we provide a fairly general technique
to recover a dual optimal solution from an approximately optimal primal
solution. Unfortunately, the performance of this technique seems quite
problem-dependent. We therefore only analyze this technique for semidefinite
programming (SDP), a classic and popular convex optimization problem.
As a result, we obtain a faster SDP solver in both the primal and
dual formulations of the problem.

\subsection*{Technique 2: Using Optimization Oracles Directly}

In the seminal works of Grötschel, Lovász, Schrijver and independently
Karp and Papadimitriou \cite{grotschel1981ellipsoid,karp1982linear},
they showed the equivalence between optimization oracles and separation
oracles, and gave a general method to construct a separation oracle
for a convex set given an optimization oracle for that set, that is
an oracle for minimizing linear functionals over the set. This seminal
result led to the first weakly polynomial time algorithm for many
algorithms such as submodular function minimization. Since then, this
idea has been used extensively in various settings \cite{jansen2003approximate,cai2012optimal,cai2013reducing,daskalakis2015bayesian}. 

Unfortunately, while this equivalence of separation and optimization
is a beautiful and powerful tool for polynomial time solvability of
problems, in many case it may lead to inefficient algorithms. In order
to use this reduction to get a separation oracle, the optimization
oracle may need to be called multiple times -- essentially the number
of times needed to run a cutting plane method and hence may be detrimental
to obtaining small asymptotic running times. Therefore, it is an interesting
question of whether there is a way of using an optimization oracle
more directly.

In Section~\ref{sec:Intersection} we provide a partial answer to
this question for the case of a broad class of problems, that we call
the \emph{intersection problem}. For these problems we demonstrate
how to achieve running time improvements by using optimization oracles
directly. The problem we consider is as follows. We wish to solve
the problem for some cost vector $\vc\in\R^{n}$ and convex set $K$.
We assume that the convex set $K$ can be decomposed as $K=K_{1}\cap K_{2}$
such that $\max_{\vx\in K_{1}}\left\langle \vc,\vx\right\rangle $
and $\max_{\vx\in K_{2}}\left\langle \vc,\vx\right\rangle $ can each
be solved efficiently. Our goal is to obtain a running time for this
problem comparable to that of minimizing $K$ given only a separation
oracle for it.

We show that by considering a carefully regularized variant, we obtain
a problem such that optimization oracles for $K_{1}$ and $K_{2}$
immediately yield a separation oracle for this regularized problem.
By analyzing the regularizer and bounding the domains of the problem
we are able to show that this allows us to efficiently compute highly
accurate solutions to the intersection problem by applying our cutting
plane method once. In other words, we do not need to use a complicated
iterative scheme or directly invoke the equivalence between separation
and optimization and thereby save $O(\poly(n))$ factors in our running
times.

We note that this intersection problem can be viewed as a generalization
of the matroid intersection problem and in Section~\ref{sub:Mat_int},
we show our reduction gives a faster algorithm in certain parameter
regimes. As another example, in Section~\ref{sub:Sub_flow} we show
our reduction gives a substantial polynomial improvement for the submodular
flow problem. Furthermore, in Section~\ref{sec:Intersection2} we
show how our techniques allow us to minimize a linear function over
the intersection of a convex set and an affine subspace in a number
of iterations that depends only on the co-dimension of the affine
space.

\subsection{Applications \label{sec:app:results_and_prev}}

Our main goal in Part~\ref{part:app} is to provide general techniques
for efficiently using cutting plane methods for various problems.
Hence, in Part~\ref{part:app} we use minimally problem-specific
techniques to achieve the best possible running time. However, we
also demonstrate the efficacy of our approach by showing how techniques
improve upon the previous best known running times for solve several
classic problems in combinatorial and continuous optimization. Here
we provide a brief overview of these applications, previous work on
these problems, and our results.

In order to avoid deviating from our main discussion, our coverage
of previous methods and techniques is brief. Given the large body
of prior works on SDP, matroid intersection and submodular flow, it
would be impossible to have an in-depth discussion on all of them.
Therefore, this section focuses on running time comparisons and explanations
of relevant preivous techniques.

\subsection*{Semidefinite Programming}

In Section~\ref{sub:sdp} we consider the classic semidefinite programming
(SDP) problem:
\[
\max_{\mx\succeq\mZero}\mc\bullet\mx\text{ s.t. }\ma_{i}\bullet\mx=b_{i}\text{ (primal)}\quad\quad\min_{\vy}\vb^{T}\vy\text{ s.t. }\sum_{i=1}^{n}y_{i}\ma_{i}\succeq\mc\text{ (dual)}
\]
where $\mx$, $\mc$, $\ma_{i}$ are $m\times m$ symmetric matrices,
$\vb,\vy\in\Rn$, and $\ma\bullet\mb\defeq\tr(\mathbf{A}^{T}\mathbf{B})$.
For many problems, $n\ll m^{2}$ and hence the dual problem has fewer
variables than the primal. There are many results and applications
of SDP; see \cite{vandenberghe1996semidefinite,todd2001semidefinite,monteiro2003first}
for a survey on this topic. Since our focus is on polynomial time
algorithms, we do not discuss pseudo-polynomial algorithms such as
the spectral bundle method \cite{helmberg2000spectral}, multiplicative
weight update methods \cite{arora2005fast,arora2007combinatorial,jain2011parallel,allen2015using},
etc.

Currently, there are two competing approaches for solving SDP problems,
namely interior point methods (IPM) and cutting plane methods. Typically,
IPMs require fewer iterations than the cutting plane methods, however
each iteration of these methods is more complicated and possibly more
computationally expensive. For SDP problems, interior point methods
require the computations of the Hessian of the function $-\log\det\left(\mc-\sum_{i=1}^{n}y_{i}\ma_{i}\right)$
whereas cutting plane methods usually only need to compute minimum
eigenvectors of the slack matrix $\mc-\sum_{i=1}^{n}y_{i}\ma_{i}$. 

In \cite{anstreicher2000volumetric}, Anstreicher provided the current
fastest IPM for solving the dual SDP problem using a method based
on the volumetric barrier function. This method takes $O((mn)^{1/4})$
iterations and each iteration is as cheap as usual IPMs. For general
matrices $\mc,\mx,\ma_{i}$, each iteration takes $O(nm^{\omega}+n^{\omega-1}m^{2})$
time where $\omega$ is the fast matrix multiplication exponent. If
the constraint matrices $\ma_{i}$ are rank one matrices, the iteration
cost can be improved to $O(m^{\omega}+nm^{2}+n^{2}m)$ \cite{krishnan2005interior}.
If the matrices are sparse, then \cite{fukuda2001exploiting,nakata2003exploiting}
show how to use matrix completion inside the IPM. However, the running
time depends on the extended sparsity patterns which can be much larger
than the total number of non-zeros.

In \cite{krishnan2003properties}, Krishnan and Mitchell observed
that the separation oracle for dual SDP takes only $O(m^{\omega}+S)$
time, where $S=\sum_{i=1}^{n}\nnz(\ma_{i})$ be the total number of
non-zeros in the constant matrix. Hence, the cutting plane method
by \cite{vaidya1996new} gives a faster algorithm for SDP for many
regimes. For $\omega=2.38$, the cutting plane method is faster when
$\ma_{i}$ is not rank 1 and the problem is not too dense, i.e. $\sum_{i=1}^{n}\nnz(\ma_{i})<n^{0.63}m^{2.25}.$
While there are previous methods for using cutting plane methods to
obtain primal solutions\cite{krishnan2006unifying} , to the best
of our knowledge, there are no worst case running time analysis for
these techniques. 

In Section~\ref{sub:sdp}, show how to alleviate this issue. We provide
an improved algorithm for finding the dual solution and prove carefully
how to obtain a comparable primal solution as well. See Figure 9.1
for a summary of the algorithms for SDP and their running times.

\begin{table}
\begin{centering}
\begin{tabular}{|c|c|c|}
\hline 
Authors & Years & Running times\tabularnewline
\hline 
\hline 
Nesterov, Nemirovsky\cite{nesterov1992conic} & 1992 & $\tilde{O}(\sqrt{n}(nm^{\omega}+n^{\omega-1}m^{2}))$\tabularnewline
\hline 
Anstreicher \cite{anstreicher2000volumetric} & 2000 & $\tilde{O}((mn)^{1/4}(nm^{\omega}+n^{\omega-1}m^{2}))$ \tabularnewline
\hline 
Krishnan, Mitchell \cite{krishnan2003properties} & 2003 & $\tilde{O}(m(n^{\omega}+m^{\omega}+S))$ (dual SDP)\tabularnewline
\hline 
\textbf{This paper} & 2015 & $\tilde{O}(m(n^{2}+m^{\omega}+S))$\tabularnewline
\hline 
\end{tabular}
\par\end{centering}

\protect\caption{Previous algorithms for solving a $n\times n$ SDP with $m$ constraints
and $S$ non-zeros entries}
\end{table}

\subsection*{Matroid Intersection}

In Section~\ref{sub:Mat_int} we show how our optimization oracle
technique can be used to improve upon the previous best known running
times for matroid intersection. Matroid intersection is one of the
most fundamental problems in combinatorial optimization. The first
algorithm for matroid intersection is due to the seminal paper by
Edmonds \cite{edmonds1968matroid}. In Figures 9.2 and 9.3 we provide
a summary of the previous algorithms for unweighted and weighted matroid
intersection as well as the new running times we obtain in this paper.
While there is no total ordering on the running times of these algorithms
due to the different dependence on various parameters, we would like
to point out that our algorithms outperform the previous ones in regimes
where $r$ is close to $n$ and/or the oracle query costs are relatively
expensive. In particular, in terms of oracle query complexity our
algorithms are the first to achieve the quadratic bounds of $\tilde{O}(n^{2})$
and $\tilde{O}(nr)$ for independence and rank oracles. We hope our
work will revive the interest in the problem of which progress has
been mostly stagnated for the past 20-30 years.

\begin{table}
\begin{centering}
\begin{tabular}{|c|c|c|}
\hline 
Authors & Years & Running times\tabularnewline
\hline 
\hline 
Edmonds \cite{edmonds1968matroid} & 1968 & not stated\tabularnewline
\hline 
Aigner, Dowling \cite{aigner1971matching} & 1971 & $O(nr^{2}\mathcal{T_{\text{ind}}})$\tabularnewline
\hline 
Tomizawa, Iri \cite{tomizawa1974algorithm} & 1974 & not stated\tabularnewline
\hline 
Lawler \cite{lawler1975matroid} & 1975 & $O(nr^{2}\mathcal{T_{\text{ind}}})$\tabularnewline
\hline 
Edmonds \cite{edmonds1979matroid} & 1979 & not stated\tabularnewline
\hline 
Cunningham \cite{cunningham1986improved} & 1986 & $O(nr^{1.5}\mathcal{T_{\text{ind}}})$\tabularnewline
\hline 
\textbf{This paper} & 2015 & $\begin{array}{c}
O(n^{2}\log n\mathcal{T_{\text{ind}}}+n^{3}\log^{O(1)}n)\\
O(nr\log^{2}n\mathcal{T_{\text{rank}}}+n^{3}\log^{O(1)}n)
\end{array}$\tabularnewline
\hline 
\end{tabular}
\par\end{centering}

\protect\caption{Previous algorithms for (unweighted) matroid intersection. Here $n$
is the size of the ground set, $r=\max\{r_{1},r_{2}\}$ is the maximum
rank of the two matroids, $\mathcal{T_{\text{ind}}}$ is the time
needed to check if a set is independent (independence oracle), and
$\mathcal{T_{\text{rank}}}$ is the time needed to compute the rank
of a given set (rank oracle).}
\end{table}

\begin{table}
\begin{centering}
\begin{tabular}{|c|c|c|}
\hline 
Authors & Years & Running times\tabularnewline
\hline 
\hline 
Edmonds \cite{edmonds1968matroid} & 1968 & not stated\tabularnewline
\hline 
Tomizawa, Iri \cite{tomizawa1974algorithm} & 1974 & not stated\tabularnewline
\hline 
Lawler \cite{lawler1975matroid} & 1975 & $O(nr^{2}\mathcal{T_{\text{ind}}}+nr^{3})$\tabularnewline
\hline 
Edmonds \cite{edmonds1979matroid} & 1979 & not stated\tabularnewline
\hline 
Frank \cite{frank1981weighted} & 1981 & $O(n^{2}r(\mathcal{T_{\text{circuit}}}+n))$\tabularnewline
\hline 
Orlin, Ahuja \cite{orlin1983primal} & 1983 & not stated\tabularnewline
\hline 
Brezovec, Cornuéjols, Glover\cite{brezovec1986two} & 1986 & $O(nr(\mathcal{T_{\text{circuit}}}+r+\log n))$\tabularnewline
\hline 
Fujishige, Zhang \cite{fujishige1995efficient} & 1995 & $O(n^{2}r^{0.5}\log rM\cdot\mathcal{T_{\text{ind}}})$\tabularnewline
\hline 
Shigeno, Iwata \cite{shigeno1995dual} & 1995 & $O((n+\mathcal{T_{\text{circuit}}})nr^{0.5}\log rM)$\tabularnewline
\hline 
\textbf{This paper} & 2015 & $\begin{array}{c}
O((n^{2}\log n\mathcal{T_{\text{ind}}}+n^{3}\log^{O(1)}n)\log nM)\\
O((nr\log^{2}n\mathcal{T_{\text{rank}}}+n^{3}\log^{O(1)}n)\log nM)
\end{array}$\tabularnewline
\hline 
\end{tabular}
\par\end{centering}

\protect\caption{Previous algorithms for weighted matroid intersection. In additions
to the notations used in the unweighted table, $\mathcal{T_{\text{circuit}}}$
is the time needed to find a fundamental circuit and $M$ is the bit
complexity of the weights.}
\end{table}

\subsection*{Minimum-Cost Submodular Flow}

In Section~\ref{sub:Sub_flow} we show how our optimization oracle
technique can be used to improve upon the previous best known running
times for (Minimum-cost) Submodular Flow. Submodular flow is a very
general problem in combinatorial optimization which generalizes many
problems such as minimum cost flow, the graph orientation, polymatroid
intersection, directed cut covering \cite{fujishige2000algorithms}.
In Figure 9.4 we provide an overview of the previous algorithms for
submodular flow as well as the new running times we obtain in this
paper.

Many of the running times are in terms of a parameter $h$, which
is the time required for computing an ``exchange capacity''. To
the best of our knowledge, the most efficient way of computing an
exchange capacity is to solve an instance of submodular minimization
which previously took time $\tilde{O}(n^{4}\EO+n^{5})$ (and now takes
$\tilde{O}(n^{2}\EO+n^{3})$ time using our result in Part~\ref{part:Submodular-Function-Minimization}).
Readers may wish to substitute $h=\tilde{O}(n^{2}\EO+n^{3})$ when
reading the table.

The previous fastest weakly polynomial algorithms for submodular flow
are by \cite{iwata2000fast,fleischer2000improved,fleischer2002faster},
which take time $\tilde{O}(n^{6}\EO+n^{7})$ and $O(mn^{5}\log nU\cdot\EO)$,
assuming $h=\tilde{O}(n^{2}\EO+n^{3})$. Our algorithm for submodular
flow has a running time of $\tilde{O}(n^{2}\EO+n^{3})$, which is
significantly faster by roughly a factor of $\tilde{O}(n^{4})$.

For strongly polynomial algorithms, our results do not yield a speedup
but we remark that our faster strongly polynomial algorithm for submodular
minimization in Part \ref{part:Submodular-Function-Minimization}
improves the previous algorithms by a factor of $\tilde{O}(n^{2})$
as a corollary (because $h$ requires solving an instance of submodular
minimization).

\begin{figure}
\begin{centering}
\begin{tabular}{|c|c|c|}
\hline 
Authors & Years & Running times\tabularnewline
\hline 
\hline 
Fujishige \cite{1978Fujishige} & 1978 & not stated\tabularnewline
\hline 
Grötschel, Lovász, Schrijver{\small{}\cite{grotschel1981ellipsoid}} & 1981 & weakly polynomial\tabularnewline
\hline 
Zimmermann \cite{zimmermann1982minimization} & 1982 & not stated\tabularnewline
\hline 
Barahona, Cunningham \cite{barahona1984submodular} & 1984 & not stated\tabularnewline
\hline 
Cunningham, Frank \cite{cunningham1985primal} & 1985 & $\rightarrow O(n^{4}h\log C)$\tabularnewline
\hline 
Fujishige \cite{fujishige1987out} & 1987 & not stated\tabularnewline
\hline 
Frank, Tardos \cite{frank1987application} & 1987 & strongly polynomial\tabularnewline
\hline 
Cui, Fujishige \cite{1988Fujishige} & 1988 & not stated\tabularnewline
\hline 
Fujishige, Röck, Zimmermann\cite{fujishige1989strongly} & 1989 & $\rightarrow O(n^{6}h\log n)$\tabularnewline
\hline 
Chung, Tcha \cite{chung1991dual} & 1991 & not stated\tabularnewline
\hline 
Zimmermann \cite{zimmermann1992negative} & 1992 & not stated\tabularnewline
\hline 
McCormick, Ervolina \cite{mccormick1993canceling} & 1993 & $O(n^{7}h^{*}\log nCU)$\tabularnewline
\hline 
Wallacher, Zimmermann \cite{wallacher1999polynomial} & 1994 & $O(n^{8}h\log nCU)$\tabularnewline
\hline 
Iwata \cite{iwata1997capacity} & 1997 & $O(n^{7}h\log U)$\tabularnewline
\hline 
Iwata, McCormick, Shigeno \cite{iwata1998faster} & 1998 & $O\left(n^{4}h\min\left\{ \log nC,n^{2}\log n\right\} \right)$\tabularnewline
\hline 
Iwata, McCormick, Shigeno \cite{iwata1999strongly} & 1999 & $O\left(n^{6}h\min\left\{ \log nU,n^{2}\log n\right\} \right)$\tabularnewline
\hline 
Fleischer, Iwata, McCormick\cite{fleischer2002faster} & 1999 & $O\left(n^{4}h\min\left\{ \log U,n^{2}\log n\right\} \right)$\tabularnewline
\hline 
Iwata, McCormick, Shigeno \cite{iwata2000fast} & 1999 & $O\left(n^{4}h\min\left\{ \log C,n^{2}\log n\right\} \right)$\tabularnewline
\hline 
Fleischer, Iwata \cite{fleischer2000improved} & 2000 & $O(mn^{5}\log nU\cdot\EO)$\tabularnewline
\hline 
\textbf{This paper} & 2015 & $O(n^{2}\log nCU\cdot\EO+n^{3}\log^{O(1)}nCU)$\tabularnewline
\hline 
\end{tabular}
\par\end{centering}

\protect\caption{Previous algorithms for Submodular Flow with $n$ vertices, maximum
cost $C$ and maximum capacity $U$. The factor $h$ is the time for
an exchange capacity oracle, $h^{*}$ is the time for a ``more complicated
exchange capacity oracle'' and $\protect\EO$ is the time for evaluation
oracle of the submodular function. The arrow,$\rightarrow$, indicates
that it used currently best maximum submodular flow algorithm as subroutine
which was non-existent at the time of the publication.}
\end{figure}

\subsection{Overview \label{sub:app:overview}}

After providing covering some preliminaries on convex analysis in
Section~\ref{sec:app:preliminaries} we split the remainder of Part~\ref{part:app}
into Section~\ref{sub:conv_opt} and Section~\ref{sec:Intersection}.
In Section~\ref{sub:conv_opt} we cover our algorithm for convex
optimization using an approximate subgradient oracle (Section~\ref{sub:app:convex-opt})
as well as our technique on using duality to decrease dimensions and
improve the running time of semidefinite programming (Section~\ref{sub:sdp}).
In Section~\ref{sec:Intersection} we provide our technique for using
minimization oracles to minimize functions over the intersection of
convex sets and provide several applications including matroid intersection
(Section~\ref{sub:Mat_int}), submodular flow (Section~\ref{sub:Sub_flow}),
and minimizing a linear function over the intersection of an affine
subspace and a convex set (Section~\ref{sec:Intersection2}).

\section{Preliminaries \label{sec:app:preliminaries}}

In this section we review basic facts about convex functions that
we use throughout Part~\ref{part:app}. We also introduce two oracles
that we use throughout Part~\ref{part:app}, i.e. subgradient and
optimization oracles, and provide some basic reductions between them.
Note that we have slightly extended some definitions and facts to
accommodate for the noisy separation oracles used in this paper. 

First we recall the definition of strong convexity
\begin{defn}[Strong Convexity ]
 \label{def:strong-convexity} A real valued function $f$ on a convex
set $\Omega$ is $\alpha$-strongly convex if for any $\vx,\vy\in\Omega$
and $t\in[0,1]$, we have
\[
f(t\vx+(1-t)\vy)+\frac{1}{2}\alpha t(1-t)\norm{\vx-\vy}^{2}\leq tf(\vx)+(1-t)f(\vy).
\]

\end{defn}
Next we define an approximate subgradient.
\begin{defn}[Subgradient]
\label{def:wek_subgrad}For any convex function $f$ on a convex
set $\Omega$, the $\delta$-subgradients of $f$ at $x$ are defined
to be
\[
\partial_{\delta}f(\vx)\defeq\{\vg\in\Omega\text{ : }f(\vy)+\delta\geq f(\vx)+\left\langle \vg,\vy-\vx\right\rangle \text{ for all }\vy\in\Omega\}.
\]

\end{defn}
Here we provide some basic facts regarding convexity and subgradients.
These statements are natural extensions of well known facts regarding
convex functions and their proof can be found in any standard textbook
on convex optimization.
\begin{fact}
\label{fact:conv_facts}For any convex set $\Omega$ and $\vx$ be
a point in the interior of $\Omega$, we have the following:
\begin{enumerate}
\item If $f$ is convex on $\Omega$, then $\partial_{0}f(\vx)\neq\emptyset$
and $\partial_{s}f(\vx)\subseteq\partial_{t}f(\vx)$ for all $0\leq s\leq t$.Otherwise,
we have $\norm{\vg}_{2}>\frac{1}{2}\sqrt{\frac{\delta}{D}}$. For
any $f(\vy)\leq f(\vx)$, we have $\delta\geq\left\langle \vg,\vy-\vx\right\rangle $
and hence
\item If $f$ is a differential convex function on $\Omega$, then $\nabla f(\vx)\in\partial_{0}f(\vx)$. 
\item If $f_{1}$ and $f_{2}$ are convex function on $\Omega$, $\vg_{1}\in\partial_{\delta_{1}}f_{1}(\vx)$
and $\vg_{2}\in\partial_{\delta_{2}}f_{1}(\vx)$, then $\alpha\vg_{1}+\beta\vg_{2}\in\partial_{\alpha\delta_{1}+\beta\delta_{2}}(\vg_{1}+\vg_{2})(\vx)$.
\item If $f$ is $\alpha$-strongly convex on $\Omega$ with minimizer $x^{*}$,
then for any $\vy$ with $f(\vy)\leq f(\vx^{*})+\varepsilon$, we
have $\frac{1}{2}\alpha\norm{\vx^{*}-\vy}^{2}\leq\varepsilon$.
\end{enumerate}
\end{fact}
Next we provide a reduction from subgradients to separation oracles.
We will use this reduction several times in Part~\ref{part:app}
to simplify our construction of separation oracles.
\begin{lem}
\label{lem:opt_implies_sep} Let $f$ be a convex function. Suppose
we have $\vx$ and $\vg\in\partial_{\delta}f(\vx)$ with $\norm{\vx}_{2}\leq1\leq D$
and $\delta\leq1$. If $\norm{\vg}_{2}\leq\frac{1}{2}\sqrt{\frac{\delta}{D}}$,
then $f(\vx)\leq\min_{\|\vy_{2}\|_{2}\leq D}f(\vy)+2\sqrt{\delta D}$
and if $\norm{\vg}_{2}\leq\frac{1}{2}\sqrt{\frac{\delta}{D}}$ then
\[
\{\norm{\vy}_{2}\leq D\,:\, f(\vy)\leq f(\vx)\}\subset\{\vy:\vd^{T}\vy\leq\vd^{T}\vx+2\sqrt{\delta D}\}
\]
with $\vd=\vg/\norm{\vg}_{2}$. Hence, this gives a $(2\sqrt{\delta D},2\sqrt{\delta D})$-separation
oracle on the set $\{\norm{\vx}_{2}\leq D\}$.\end{lem}
\begin{proof}
Let $\vy$ such that $\norm{\vy}_{2}\leq D$. By the definition of
$\delta$-subgradient, we have
\[
f(\vy)+\delta\geq f(\vx)+\left\langle \vg,\vy-\vx\right\rangle .
\]
If $\norm{\vg}\leq\frac{1}{2}\sqrt{\frac{\delta}{D}}$, then, we have
$\left|\left\langle \vg,\vy-\vx\right\rangle \right|\leq\sqrt{\delta D}$
because $\norm{\vx}\leq D$ and $\norm{\vy}_{2}\leq D$. Therefore,
\[
\min_{\norm{\vy}_{2}\leq D}f(\vy)+2\sqrt{\delta D}\geq f(\vx).
\]
Otherwise, we have $\norm{\vg}_{2}>\frac{1}{2}\sqrt{\frac{\delta}{D}}$.
For any $f(\vy)\leq f(\vx)$, we have $\delta\geq\left\langle \vg,\vy-\vx\right\rangle $
and hence 
\[
2\sqrt{\delta D}\geq\left\langle \frac{\vg}{\norm{\vg}},\vy-\vx\right\rangle .
\]

\end{proof}
At several times in Part~\ref{part:app} we will wish to construct
subgradient oracles or separation oracles given only the ability to
approximately maximize a linear function over a convex set. In the
remainder of this section we formally define such a \emph{optimization
oracle} and prove this equivalence.
\begin{defn}[Optimization Oracle]
 \label{def:weak_opt} Given a convex set $K$ and $\delta>0$ a
$\delta$-optimization oracle for $K$ is a function on $\Rn$ such
that for any input $\vc\in\Rn$, it outputs $\vy$ such that
\[
\max_{\vx\in K}\left\langle \vc,\vx\right\rangle \leq\left\langle \vc,\vy\right\rangle +\delta.
\]
We denote by $\OO_{\delta}(K)$ the time complexity of this oracle.\end{defn}
\begin{lem}
\label{lem:weak_opt_weak_subgrad} Given a convex set $K$, any $\varepsilon$-optimization
oracle for $K$ is a $\varepsilon$-subgradient oracle for $f(\vc)\defeq\max_{\vx\in K}\left\langle \vc,\vx\right\rangle .$\end{lem}
\begin{proof}
Let $\vx_{c}$ be the output of $\varepsilon$-optimization oracle
on the cost vector $\vc$. We have
\[
\max_{\vx\in K}\left\langle \vc,\vx\right\rangle \leq\left\langle \vc,\vx_{c}\right\rangle +\varepsilon.
\]
Hence, for all $\vd$, we have and therefore 
\[
\left\langle \vx_{c},\vd-\vc\right\rangle +f(\vc)\leq f(\vd)+\varepsilon.
\]
Hence, $\vx_{c}\in\partial_{\delta}f(\vc)$.
\end{proof}
Combining these lemmas shows that having an $\epsilon$-optimization
oracle for a convex set $K$ contained in a ball of radius $D$ yields
a $O(\sqrt{D\epsilon},\sqrt{D\epsilon})$ separation oracle for $\max_{x\in K}\innerproduct{\vc}{\vx}$.
We use these ideas to construction separation oracles throughout Part~\ref{part:app}.

\section{Convex Optimization \label{sub:conv_opt} }

In this section we show how to apply our cutting plane method to efficiently
solve problems in convex optimization. First, in Section~\ref{sub:app:convex-opt}
we show how to use our result to minimize a convex function given
an approximate subgradient oracle. Then, in Section~\ref{sub:sdp}
we illustrate how this result can be used to obtain both primal and
dual solutions for a standard convex optimization problems. In particular,
we show how our result can be used to obtain improved running times
for semidefinite programming across a range of parameters.

\subsection{From Feasibility to Optimization}

\label{sub:app:convex-opt}

In this section we consider the following standard optimization problem.
We are given a convex function $f\,:\R^{n}\rightarrow\R\cup\{+\infty\}$
and we want to find a point $\vx$ that approximately solves the minimization
problem 
\[
\min_{\vx\in\Rn}f(\vx)
\]
given only a subgradient oracle for $f$. 

Here we show how to apply the cutting plane method from Part~\ref{part:Ellipsoid}
turning the small width guarantee of the output of that algorithm
into a tool to find an approximate minimizer of $f$. Our result is
applicable to any convex optimization problem armed with a separation
or subgradient oracle. This result will serve as the foundation for
many of our applications in Part~\ref{part:app}.

Our approach is an adaptation of Nemiroski's method \cite{Nemirovski1994}
which applies the cutting plane method to solve convex optimiziation
problems, with only minimal assumption on the cutting plane method.
The proof here is a generalization that accommodates for the noisy
separation oracle used in this paper. In the remainder of this subsection
we provide a key definition we will use in our algorithm (Defintion~\ref{def:min-width}),
provide our main result (Theorem~\ref{thm:conv_opt}), and conclude
with a brief discussion of this result.
\begin{defn}
\label{def:min-width} For any compact set $K$, we define the \emph{minimum
width} by $\width(K)\defeq\min_{\|\va\|_{2}=1}\max_{\vx,\vy\in K}\left\langle \va,\vx-\vy\right\rangle .$\end{defn}
\begin{thm}
\label{thm:conv_opt} Let $f$ be a convex function on $\Rn$ and
$\Omega$ be a convex set that contains a minimizer of $f$. Suppose
we have a $(\eta,\delta)$-separation oracle for $f$ and $\Omega$
is contained inside $B_{\infty}(R)$. Using $B_{\infty}(R)$ as the
initial polytope for our Cutting Plane Method, for any $0<\alpha<1$,
we can compute $\vx\in\Rn$ such that 
\begin{equation}
f(\vx)-\min_{\vy\in\Omega}f(\vy)\leq\eta+\alpha\left(\max_{\vy\in\Omega}f(\vy)-\min_{\vy\in\Omega}f(\vy)\right)\enspace.\label{eq:cutting_plane_function}
\end{equation}
with an expected running time of 
\[
O\left(n\SO_{\eta,\delta}(f)\log\left(\frac{n\kappa}{\alpha}\right)+n^{3}\log^{O(1)}\left(\frac{n\kappa}{\alpha}\right)\right),
\]
where $\delta=\Theta\left(\frac{\alpha\width(\Omega)}{n^{3/2}\ln\left(\kappa\right)}\right)$
and $\kappa=\frac{R}{\width(\Omega)}$. Furthermore, we only need
the oracle defined on the set $B_{\infty}(R)$. \end{thm}
\begin{proof}
Let $\vx^{*}\in\arg\min_{\vx\in\Omega}f(\vx)$. Since $B_{\infty}(R)\supset\Omega$
contains a minimizer of $f$, by the definition of $(\eta,\delta)$-separation
oracles, our Cutting Plane Method (Theorem~\ref{thm:main_result})
either returns a point $\vx$ that is almost optimal or returns a
polytope $P$ of small width. In the former case we have a point $\vx$
such that $f(\vx)\leq\min_{\vy}f(\vy)+\eta$. Hence, the error is
clearly at most $\eta+\alpha\left(\max_{\vz\in\Omega}f(\vz)-\min_{\vx\in\Omega}f(\vx)\right)$
as desired. Consequently, we assume the latter case.

Theorem~\ref{thm:main_result} shows $\width(P)<Cn\varepsilon\ln(R/\varepsilon)$
for some universal constant $C$. Picking 
\begin{equation}
\varepsilon=C'\frac{\alpha\width(\Omega)}{n\ln\left(\frac{n\kappa}{\alpha}\right)}\label{eq:ellip_eps}
\end{equation}
for small enough constant $C'$, we have $\width(P^{(i)})<\alpha\width(\Omega)$.
Let $\Omega^{\alpha}=\vx^{*}+\alpha(\Omega-\vx^{*})$, namely, $\Omega^{\alpha}=\{\vx^{*}+\alpha(\vz-\vx^{*}):\vz\in\Omega\}$.
Then, we have
\[
\width(\Omega^{\alpha})=\alpha\width(\Omega)>\width(P).
\]
Therefore, $\Omega^{\alpha}$ is not a subset of $P^{(i)}$ and hence
there is some point $\vy\in\Omega^{\alpha}\backslash P$. Since $\Omega^{\alpha}\subseteq\Omega\subseteq B_{\infty}(R)$,
we know that $\vy$ does not violate any of the constraints of $P^{(0)}$
and therefore must violate one of the constraints added by querying
the separation oracle. Therefore, for some $j\leq i$, we have
\[
\left\langle \vc^{(j-1)},\vy\right\rangle >\left\langle \vc^{(j-1)},\vx^{(j-1)}\right\rangle +c_{s}\varepsilon/\sqrt{n}\enspace.
\]
By the definition of $(\eta,c_{s}\varepsilon/\sqrt{n})$-separation
oracle (Definition \ref{def:weak_sep2}), we have $f(\vy)>f(\vx^{(j-1)})$.
Since $\vy\in\Omega^{\alpha}$, we have $\vy=(1-\alpha)\vx^{*}+\alpha\vz$
for some $\vz\in\Omega$. Thus, the convexity of $f$ implies that
\[
f(\vy)\leq(1-\alpha)f(\vx^{*})+\alpha f(\vz).
\]
Therefore, we have
\[
\min_{1\leq k\leq i}f(\vx^{(k)})-\min_{\vx\in\Omega}f(\vx)<f(\vy)-f(\vx^{*})\leq\alpha\left(\max_{\vz\in\Omega}f(\vz)-\min_{\vx\in\Omega}f(\vx)\right).
\]
Hence, we can simply output the best $\vx$ among all $\vx^{(j)}$
and in either case $\vx$ satisfies \eqref{eq:cutting_plane_function}.

Note that we need to call $(\eta,\delta)$-separation oracle with
$\delta=\Omega(\varepsilon/\sqrt{n})$ to ensure we do not cut out
$\vx^{*}$. Theorem~\ref{thm:main_result} shows that the algorithm
takes $O(n\SO_{\eta,\delta}(f)\log(nR/\varepsilon)+n^{3}\log^{O(1)}(nR/\varepsilon))$
expected time, as promised. Furthermore, the oracle needs only be
defined on $B_{\infty}(R)$ as our cutting plane method guarantees
$\vx^{(k)}\in B_{\infty}(R)$ for all $k$ (although if needed, an
obvious separating hyperplane can be returned for a query point outside
$B_{\infty}(R)$ ).
\end{proof}
Observe that this algorithm requires no information about $\Omega$
(other than that $\Omega\subseteq B_{\infty}(R)$) and does not guarantee
that the output is in $\Omega$. Hence, even though $\Omega$ can
be complicated to describe, the algorithm still gives a guarantee
related to the gap $\max_{\vx\in\Omega}f(\vx)-\min_{\vx\in\Omega}f(\vx)$.
For specific applications, it is therefore advantageous to pick a
$\Omega$ as large as possible while the bound on function value is
as small as possible. 

Before indulging into specific applications, we remark on the dependence
on $\kappa$. Using John's ellipsoid, it can be shown that any convex
set $\Omega$ can be transformed linearly such that (1) $B_{\infty}(1)$
contains $\Omega$ and, (2) $\width(\Omega)=\Omega(n^{-3/2})$. In
other words, $\kappa$ can be effectively chosen as $O(n^{3/2})$.
Therefore if we are able to find such a linear transformation, the
running time is simply $O\left(n\SO(f)\log\left(n/\alpha\right)+n^{3}\log^{O(1)}\left(n/\alpha\right)\right)$.
Often this can be done easily using the structure of the particular
problem and the running time does not depend on the size of domain
at all.

\subsection{Duality and Semidefinite Programming\label{sub:sdp}}

In this section we illustrate how our result in Section~\ref{sub:app:convex-opt}
can be used to obtain both primal and dual solutions for standard
problems in convex optimization. In particular we show how to obtain
improved running times for semidefinite programming.

To explain our approach, consider the following minimax problem
\begin{equation}
\min_{\vy\in Y}\max_{\vx\in X}\left\langle \ma\vx,\vy\right\rangle +\left\langle \vc,\vx\right\rangle +\left\langle \vd,\vy\right\rangle \label{eq:minimax_problem}
\end{equation}
where $\vx\in\Rm$ and $\vy\in\Rn$. When $m\gg n$, solving this
problem by directly using Part~\ref{part:Ellipsoid} could lead to
an inefficient algorithm with running time at least $m^{3}$. In many
situations, for any fixed $\vy$, the problem $\max_{\vx\in X}\left\langle \ma\vx,\vy\right\rangle $
is very easy and hence one can use it as a separation oracle and apply
Part~\ref{part:Ellipsoid} and this would gives a running time almost
independent of $m$. However, this would only give us the $\vy$ variable
and it is not clear how to recover $\vx$ variable from it.

In this section we show how to alleviate this issue and give semidefinite
programming (SDP) as a concrete example of how to apply this general
technique. We do not write down the general version as the running
time of the technique seems to be problem specific and faster SDP
is already an interesting application.

For the remainder of this section we focus on the semidefinite programming
(SDP) problem:
\begin{equation}
\max_{\mx\succeq\mZero}\mc\bullet\mx\text{ s.t. }\ma_{i}\bullet\mx=b_{i}\label{eq:SDP}
\end{equation}
and its dual
\begin{equation}
\min_{\vy}\vb^{T}\vy\text{ s.t. }\sum_{i=1}^{n}y_{i}\ma_{i}\succeq\mc\label{eq:dual_SDP}
\end{equation}
where $\mx$, $\mc$, $\ma_{i}$ are $m\times m$ symmetric matrices
and $\vb,\vy\in\Rn$. Our approach is partially inspired by one of
the key ideas of \cite{helmberg2000spectral,krishnan2003properties}.
These results write down the dual SDP in the form
\begin{equation}
\min_{y}\vb^{T}\vy-K\min(\lambda_{\min}(\sum_{i=1}^{n}y_{i}\ma_{i}-\mc),0)\label{eq:dual_SDP_v2}
\end{equation}
for some large number $K$ and use non-smooth optimization techniques
to solve the dual SDP problem. Here, we follow the same approach but
instead write it as a max-min problem $\min_{\vy}f_{K}(\vy)$ where
\begin{equation}
f_{K}(\vy)=\max_{\tr\mx\leq K,\mx\succeq\mZero}\left(\vb^{T}\vy+\left\langle \mx,\mc-\sum_{i=1}^{n}y_{i}\ma_{i}\right\rangle \right).\label{eq:sdp_f}
\end{equation}
Thus the SDP problem in fact assumes the form \eqref{eq:minimax_problem}
and many ideas in this section can be generalized to the minimax problem
\eqref{eq:minimax_problem}.

To get a dual solution, we notice that the cutting plane method maintains
a subset of the primal feasible solution $\conv(\mx_{i})$ such that
\[
\min_{\vy}\vb^{T}\vy+\max_{\tr\mx\leq K,\mx\succeq\mZero}\left\langle \mx,\mc-\sum_{i=1}^{n}y_{i}\ma_{i}\right\rangle \sim\min_{\vy}\vb^{T}\vy+\max_{\mx\in\conv(\mx_{i})}\left\langle \mx,\mc-\sum_{i=1}^{n}y_{i}\ma_{i}\right\rangle .
\]
Applying minimax theorem, this shows that there exists an approximation
solution $\mx$ in $\conv(\mx_{i})$ for the primal problem. Hence,
we can restrict the primal SDP on the polytope $\conv(\mx_{i})$,
this reduces the primal SDP into a linear program which can be solved
very efficiently. This idea of getting primal/dual solution from the
cutting plane method is quite general and is the main purpose of this
example. As a by-product, we have a faster SDP solver in both primal
and dual! We remark that this idea has been used as a heuristic to
obtain \cite{krishnan2006unifying} for getting the primal SDP solution
and our contribution here is mainly the asymptotic time analysis.

We first show how to construct the separation oracle for SDP. For
that we need to compute smallest eigenvector of a matrix. Below, for
completeness we provide a folklore result showing we can do this using
fast matrix multiplication. 
\begin{lem}
\label{lem:computing_eigenvalue} Given a $n\times n$ symmetric matrix
$\my$ such that $-R\iMatrix\preceq\my\preceq R\iMatrix$, for any
$\varepsilon>0$, with high probability in $n$ in time $O(n^{\omega+o(1)}\log^{O(1)}(R/\varepsilon))$
we can find a unit vector $\vu$ such that $\vu^{T}\my\vu\geq\lambda_{\max}(\my)-\varepsilon$.\end{lem}
\begin{proof}
Let $\mb\defeq\frac{1}{R}\my+\iMatrix$. Note that $\mb\succeq\mZero$.
Now, we consider the repeated squaring $\mb_{0}=\mb$ and $\mb_{k+1}=\frac{\mb_{k}^{2}}{\tr\mb_{k}^{2}}.$
Let $0\leq\lambda_{1}\leq\lambda_{2}\leq\cdots\leq\lambda_{n}$ be
the eigenvalues of $\mb$ and $\vv_{i}$ be the corresponding eigenvectors.
Then, it is easy to see the the eigenvalues of $\mb_{k}$ are $\frac{\lambda_{i}^{2^{k}}}{\sum_{i=1}^{n}\lambda_{i}^{2^{k}}}.$ 

Let $\vq$ be a random unit vector and $\vr\defeq\mb_{k}\vq$. Now
$\vq=\sum\alpha_{i}\vv_{i}$ for some $\alpha_{i}$ such that $\sum\alpha_{i}^{2}=1$.
Letting 
\[
\vp=\frac{\sum_{\lambda_{i}>(1-\delta)\lambda_{n}}\alpha_{i}\lambda_{i}^{2^{k}}\vv_{i}}{\sum_{i=1}^{n}\lambda_{i}^{2^{k}}}
\]
we have
\[
\normFull{\vr-\vp}_{2}=\normFull{\frac{\sum_{\lambda_{i}\leq(1-\delta)\lambda_{n}}\alpha_{i}\lambda_{i}^{2^{k}}\vv_{i}}{\sum_{i=1}^{n}\lambda_{i}^{2^{k}}}}_{2}\leq\frac{\sum_{\lambda_{i}\leq(1-\delta)\lambda_{n}}\lambda_{i}^{2^{k}}}{\sum_{i=1}^{n}\lambda_{i}^{2^{k}}}\leq(1-\delta)^{2^{k}}n.
\]
Letting $k=\log_{2}\left(\frac{\log(n^{3/2}/\delta)}{\delta}\right)$,
we have $\normFull{\vr-\vp}_{2}\leq\delta/\sqrt{n}.$ Since $\mZero\preceq\mb\preceq2\iMatrix$,
we have 
\begin{eqnarray*}
\sqrt{\vr^{T}\mb\vr} & \geq & \sqrt{\vp^{T}\mb\vp}-\sqrt{(\vr-\vp)^{T}\mb(\vr-\vp)}\\
 & \geq & \sqrt{\vp^{T}\mb\vp}-2\delta/\sqrt{n}.
\end{eqnarray*}
Note that $\vp$ involves only eigenvectors between $(1-\delta)\lambda_{n}$
to $\lambda_{n}$. Hence, we have 
\[
\sqrt{\vr^{T}\mb\vr}\geq\sqrt{(1-\delta)\lambda_{n}}\norm{\vp}_{2}-2\delta/\sqrt{n}.
\]
With constant probability, we have $\alpha_{n}=\Omega(1/\sqrt{n})$.
Hence, we have $\norm{\vp}_{2}=\Omega(1/\sqrt{n})$. Using $\mb\preceq2\iMatrix$
and $\norm{\vp}_{2}\geq\norm{\vr}_{2}-\delta/\sqrt{n}$ we have that
so long as $\delta$ is a small enough universal constant 
\begin{eqnarray*}
\frac{\sqrt{\vr^{T}\mb\vr}}{\norm{\vr}_{2}} & \geq & \frac{\sqrt{(1-\delta)\lambda_{n}}\norm{\vp}_{2}-2\delta/\sqrt{n}}{\norm{\vp}_{2}+\delta/\sqrt{n}}\\
 & = & (1-O(\delta))\sqrt{\lambda_{n}}-O(\delta)\\
 & = & \sqrt{\lambda_{n}}-O(\delta\sqrt{R}).
\end{eqnarray*}
Therefore, we have $\frac{\vr^{T}\my\vr}{\norm{\vr}^{2}}\geq\lambda_{\max}(\my)-O(R\delta)$.
Hence, we can find vector $\vr$ by computing $k$ matrix multiplications.
\cite{demmel2007fast} showed that fast matrix multiplication is stable
under Frobenius norm, i.e., for any $\eta>0$, using $O(\log(n/b))$
bits, we can find $\mc$ such that $\norm{\mc-\ma\mb}_{F}\leq\frac{1}{b}\norm{\ma}\norm{\mb}$
in time $O(n^{\omega+\eta})$ where $\omega$ is the matrix multiplicative
constant. Hence, this algorithm takes only $O(n^{\omega+o(1)}\log^{O(1)}(\delta^{-1}))$
time. The result follows from renormalizing the vector $\vr$, repeating
the algorithm $O(\log n)$ times to boost the probability and taking
$\delta=\Omega(\varepsilon/R)$.
\end{proof}
The following lemma shows how to compute a separation for $f_{K}$
defined in \eqref{eq:sdp_f}.
\begin{lem}
\label{lem:subgrad_sdp} Suppose that $\norm{\ma_{i}}_{F}\leq M$
and $\norm{\mc}_{F}\leq M$. For any $0<\varepsilon<1$ and $\vy$
with $\norm{\vy}_{2}=O(L)$, with high probability in $m$, we can
compute a $(\varepsilon,\varepsilon)$-separation of $f_{K}$ on $\{\norm{\vx}_{2}\leq L\}$
at $\vy$ in time $O(S+m^{\omega+o(1)}\log^{O(1)}(nKML/\varepsilon))$
where where $S$ is the sparsity of the problem defined as $\nnz(\mc)+\sum_{i=1}^{n}\nnz(\ma_{i})$.\end{lem}
\begin{proof}
Note that $-O(nML)\iMatrix\preceq\mc-\sum_{i=1}^{n}y_{i}\ma_{i}\preceq O(nML)\iMatrix.$
Using Lemma~\ref{lem:computing_eigenvalue}, we can find a vector
$\vv$ with $\norm{\vv}_{2}=K$ in time $O(m^{\omega+o(1)}\log^{O(1)}(nKML/\delta))$
such that
\begin{equation}
\vv^{T}\left(\mc-\sum_{i=1}^{n}y_{i}\ma_{i}\right)\vv\geq\max_{\tr\mx\leq K,\mx\succeq\mZero}\left\langle \mx,\mc-\sum_{i=1}^{n}y_{i}\ma_{i}\right\rangle -\delta.\label{eq:sdp_vv_apr_eig}
\end{equation}
In other words, we have a $\delta$-optimization oracle for the function
$f_{K}$. Lemma~\ref{lem:weak_opt_weak_subgrad} shows this yields
a $\delta$-subgradient oracle and Lemma~\ref{lem:opt_implies_sep}
then shows this yields a $\left(O(\sqrt{\delta L}),O(\sqrt{\delta L})\right)$-separation
oracle on the set $\{\norm{\vx}_{2}\leq L\}$. By picking $\delta=\varepsilon^{2}/L$,
we have the promised oracle. 
\end{proof}
With the separation oracle in hand, we are ready to give the algorithm
for SDP:
\begin{thm}
Given a primal-dual semidefinite programming problem in the form \eqref{eq:SDP}
and \eqref{eq:dual_SDP}, suppose that for some $M\geq1$ we have
\begin{enumerate}
\item $\norm b_{2}\leq M$, $\norm{\mc}_{F}\leq M$ and $\norm{\ma_{i}}_{F}\leq M$
for all $i$.
\item The primal feasible set lies inside the region $\tr\mx\leq M$.
\item The dual feasible set lies inside the region $\norm{\vy}_{\infty}\leq M$.
\end{enumerate}
Let $\OPT$ be the optimum solution of \eqref{eq:SDP} and \eqref{eq:dual_SDP}.
Then, with high probability, we can find $\mx$ and $\vy$ such that
\begin{enumerate}
\item $\mx\succeq\mZero$, $\tr\mx=O(M)$, $\sum_{i}\left|b_{i}-\left\langle \mx,\ma_{i}\right\rangle \right|\leq\varepsilon$
for all $i$ and $\mc\bullet\mx\geq\OPT-\varepsilon.$
\item $\norm{\vy}_{\infty}=O(M)$, $\sum_{i=1}^{n}y_{i}\ma_{i}\succeq\mc-\varepsilon\iMatrix$
and $\vb^{T}\vy\leq\OPT+\varepsilon.$
\end{enumerate}
in expected time $O\left(\left(nS+n^{3}+nm^{\omega+o(1)}\right)\log^{O(1)}\left(\frac{nM}{\varepsilon}\right)\right)$
where $S$ is the sparsity of the problem defined as $\nnz(\mc)+\sum_{i=1}^{n}\nnz(\ma_{i})$
and $\omega$ is the fast matrix multiplication constant.\end{thm}
\begin{proof}
Let $K\geq M$ be some parameter to be determined. Since the primal
feasible set is lies inside the region $\tr\mx\leq M\leq K$, we have
\begin{eqnarray*}
\min_{\sum_{i=1}^{n}y_{i}\ma_{i}\succeq\mc}\vb^{T}\vy & = & \max_{\mx\succeq\mZero,\tr\mx\leq K,\ma_{i}\bullet\mx=b_{i}}\mc\bullet\mx\\
 & = & \max_{\mx\succeq\mZero,\tr\mx\leq K}\min_{\vy}\mc\bullet\mx-\sum_{i}y_{i}\left(\ma_{i}\bullet\mx-b_{i}\right)\\
 & = & \min_{\vy}\max_{\mx\succeq\mZero,\tr\mx\leq K}\left(\vb^{T}\vy+(\mc-\sum_{i}y_{i}\ma_{i})\bullet\mx\right)\\
 & = & \min_{\vy}f_{K}(\vy).
\end{eqnarray*}

Lemma~\ref{lem:subgrad_sdp} shows that it takes $\SO_{\delta,\delta}(f_{K})=O(S+m^{\omega+o(1)}\log(nKML/\delta))$
time to compute a $(\delta,\delta)$-separation oracle of $f_{K}$
for any point $\vy$ with $\|\vy\|_{\infty}=O(L)$ where $L$ is some
parameter with $L\geq M$. Taking the radius $R=L$, Theorem \ref{thm:conv_opt}
shows that it takes $O\left(n\SO_{\delta,\delta}(f_{K})\log\left(\frac{n}{\alpha}\right)+n^{3}\log^{O(1)}\left(\frac{n}{\alpha}\right)\right)$
expected time with $\delta=\Theta\left(\alpha n^{-3/2}L\right)$ to
find $\vy$ such that
\[
f_{K}(\vy)-\min_{\norm{\vy}_{\infty}\leq L}f_{K}(\vy)\leq\delta+\alpha\left(\max_{\norm{\vy}_{\infty}\leq L}f_{K}(\vy)-\min_{\norm{\vy}_{\infty}\leq L}f_{K}(\vy)\right)\leq\delta+2\alpha\left(nML+2nKML\right).
\]
Picking $\alpha=\frac{\varepsilon}{7nMKL}$, we have $f_{K}(\vy)\leq\min_{\vy}f_{K}(\vy)+\varepsilon.$
Therefore,
\[
\vb^{T}\vy+K\max(\lambda_{\max}(\mc-\sum_{i=1}^{n}y_{i}\ma_{i}),0)\leq\OPT+\varepsilon.
\]
Let $\beta=\max(\lambda_{\max}(\mc-\sum_{i=1}^{n}y_{i}\ma_{i}),0)$.
Then, we have that $\sum_{i=1}^{n}y_{i}\ma_{i}\succeq\mc-\beta\iMatrix$
and
\begin{eqnarray*}
\vb^{T}\vy & \geq & \min_{\sum_{i=1}^{n}y_{i}\ma_{i}\succeq\mc-\beta\iMatrix}\vb^{T}\vy\\
 & = & \max_{\mx\succeq\mZero\ma_{i}\bullet\mx=b_{i}}\left(\mc-\beta\iMatrix\right)\bullet\mx\\
 & \geq & \OPT-\beta M
\end{eqnarray*}
because $\tr\mx\leq M$. Hence, we have 
\[
\OPT-\beta M+\beta K\leq\vb^{T}\vy+K\max(\lambda_{\max}(\mc-\sum_{i=1}^{n}y_{i}\ma_{i}),0)\leq\OPT+\varepsilon
\]
Putting $K=M+1$, we have $\beta\leq\varepsilon$. Thus,
\[
\sum_{i=1}^{n}y_{i}\ma_{i}\succeq\mc-\varepsilon\iMatrix.
\]
This gives the result for the dual with the running time $O\left(\left(nS+n^{3}+nm^{\omega+o(1)}\right)\log^{O(1)}\left(\frac{nML}{\varepsilon}\right)\right)$.

Our Cutting Plane Method accesses the sub-problem 
\[
\max_{\mx\succeq\mZero,\tr\mx\leq K}(\mc-\sum_{i}y_{i}\ma_{i})\bullet\mx
\]
only through the separation oracle. Let $\vz$ be the output of our
Cutting Plane Method and $\{\vv_{i}\vv_{i}^{T}\}_{i=1}^{O(n)}$ be
the matrices used to construct the separation for the $O(n)$ hyperplanes
the algorithm maintains at the end. Let $\vu$ be the maximum eigenvector
of $\mc-\sum_{i=1}^{n}z_{i}\ma_{i}$. Now, we consider a realization
of $f_{K}$
\[
\tilde{f}_{K}(\vy)=\vb^{T}\vy+\max_{\mx\in\conv(K\vu\vu^{T},\vv_{i}\vv_{i}^{T})}\left\langle \mx,\mc-\sum_{i=1}^{n}y_{i}\ma_{i}\right\rangle .
\]
Since applying our Cutting Plane Method to either $f_{K}$ or $\tilde{f}_{K}$
gives the same result, the correctness of the our Cutting Plane Method
shows that
\[
\tilde{f}_{K}(\vz)\leq\min_{\norm{\vy}_{\infty}\leq L}\tilde{f}_{K}(\vy)+\varepsilon.
\]
Note that the function $\tilde{f}_{K}$ is defined such that $\tilde{f}_{K}(\vz)=f_{K}(\vz)$.
Hence, we have
\[
\min_{\norm{\vy}_{\infty}\leq L}f_{K}(\vy)\leq f_{K}(\vz)\leq\tilde{f}_{K}(\vz)\leq\min_{\norm{\vy}_{\infty}\leq L}\tilde{f}_{K}(\vy)+\varepsilon.
\]
Also, note that $\tilde{f}_{K}(\vx)\leq f_{K}(\vx)$ for all $\vx$.
Hence, we have 
\[
\min_{\norm{\vy}_{\infty}\leq L}f_{K}(\vy)-\varepsilon\leq\min_{\norm{\vy}_{\infty}\leq L}\tilde{f}(\vy)\leq\min_{\norm{\vy}_{\infty}\leq L}f_{K}(\vy).
\]
Now, we consider the primal version of $\tilde{f}$, namely
\[
g(\mx)\defeq\min_{\norm{\vy}_{\infty}\leq L}\vb^{T}\vy+\left\langle \mx,\mc-\sum_{i=1}^{n}y_{i}\ma_{i}\right\rangle .
\]
Sion's minimax theorem \cite{sion1958general} shows that 
\[
\OPT\geq\max_{\mx\in\conv(K\vu\vu^{T},\vv_{i}\vv_{i}^{T})}g(\mx)=\min_{\norm{\vy}_{\infty}\leq L}\tilde{f}(\vy)\geq\OPT-\varepsilon.
\]
Therefore, to get the primal solution, we only need to find $\vu$
by Lemma \ref{lem:computing_eigenvalue} and solve the maximization
problem on $g$. Note that
\begin{eqnarray*}
g(\mx) & = & \min_{\norm{\vy}_{\infty}\leq L}\sum_{i=1}^{n}y_{i}\left(b_{i}-\left\langle \mx,\ma_{i}\right\rangle \right)+\left\langle \mx,\mc\right\rangle \\
 & = & -L\sum_{i}\left|b_{i}-\left\langle \mx,\ma_{i}\right\rangle \right|+\left\langle \mx,\mc\right\rangle .
\end{eqnarray*}
For notation simplicity, we write $K\vu\vu^{T}=\vv_{0}\vv_{0}^{T}$.
Then, $\mx=\sum_{j=0}^{O(n)}\alpha_{j}\vv_{j}\vv_{j}^{T}$ for some
$\sum\alpha_{j}=1$ and $\alpha_{j}\geq0$. Substituting this into
the function $g$, we have
\[
g(\valpha)=-L\sum_{j}\left|b_{i}-\sum_{j}\alpha_{j}\vv_{j}^{T}\ma_{i}\vv_{j}\right|+\sum_{j}\alpha_{j}\vv_{j}^{T}\mc\vv_{j}.
\]
Hence, this can be easily written as a linear program with $O(n)$
variables and $O(n)$ constraints in time $O(nS)$. Now, we can apply
interior point method to find $\valpha$ such that
\[
g(\valpha)\geq\max_{\mx\in\conv(K\vu\vu^{T},\vv_{i}\vv_{i}^{T})}g(\mx)-\varepsilon\geq\OPT-2\varepsilon.
\]
Let the corresponding approximate solution be $\widetilde{\mx}=\sum\alpha_{j}\vv_{j}\vv_{j}^{T}$.
Then, we have 
\[
\left\langle \widetilde{\mx},\mc\right\rangle -L\sum_{i}\left|b_{i}-\left\langle \mx,\ma_{i}\right\rangle \right|\geq\OPT-2\varepsilon.
\]
Now, we let $\tilde{b}_{i}=\left\langle \widetilde{\mx},\ma_{i}\right\rangle $.
Then, we note that
\begin{eqnarray*}
\left\langle \widetilde{\mx},\mc\right\rangle  & \leq & \max_{\mx\succeq\mZero\ma_{i}\bullet\mx=\tilde{b}_{i}}\mc\bullet\mx\\
 & = & \min_{\sum_{i=1}^{n}y_{i}\ma_{i}\succeq\mc}\tilde{b}_{i}^{T}\vy\\
 & \leq & \OPT+M\sum_{i}\left|b_{i}-\left\langle \widetilde{\mx},\ma_{i}\right\rangle \right|
\end{eqnarray*}
because $\norm{\vy}_{\infty}\leq M$. Hence, we have
\[
\OPT+\left(M-L\right)\sum_{i}\left|b_{i}-\left\langle \widetilde{\mx},\ma_{i}\right\rangle \right|\geq\left\langle \widetilde{\mx},\mc\right\rangle -L\sum_{i}\left|b_{i}-\left\langle \widetilde{\mx},\ma_{i}\right\rangle \right|\geq\OPT-2\varepsilon.
\]
Now, we put $L=M+2$, we have 
\[
\sum_{i}\left|b_{i}-\left\langle \widetilde{\mx},\ma_{i}\right\rangle \right|\leq\varepsilon.
\]
This gives the result for the primal. Note that it only takes $O(n^{5/2}\log^{O(1)}(nM/\varepsilon))$
to solve a linear program with $O(n)$ variables and $O(n)$ constraints
because we have an explicit interior point deep inside the feasible
set, i.e. $\alpha_{i}=\frac{1}{m}$ for some parameter $m$ \cite{lee2015efficient}.%
\footnote{Without this, the running time of interior point method depends on
the bit complexity of the linear programs.%
} Hence, the running time is dominated by the cost of cutting plane
method which is $O\left(\left(nS+n^{3}+nm^{\omega+o(1)}\right)\log^{O(1)}\left(\frac{nM}{\varepsilon}\right)\right)$
by putting $L=M+2$.
\end{proof}
We leave it as an open problem if it is possible to improve this result
by reusing the computation in the separation oracle and achieve a
running time of $O\left(\left(nS+n^{3}+nm^{2}\right)\log^{O(1)}\left(\frac{nM}{\varepsilon}\right)\right)$. 

\section{Intersection of Convex Sets\label{sec:Intersection}}

In this section we introduce a general technique to optimize a linear
function over the intersection of two convex sets, whenever the linear
optimization problem on each of them can be done efficiently. At the
very high level, this is accomplished by applying cutting plane to
a suitably regularized version of the problem. In Section~\ref{sub:app:intersect:problem}
we present the technique and in the remaining sections we provide
several applications including, matroid intersection (Section~\ref{sub:Mat_int}),
submodular flow (Section~\ref{sub:Sub_flow}), and minimizing over
the intersection of an affine subspace and a convex set (Section~\ref{sec:Intersection2}).

\subsection{The Technique \label{sub:app:intersect:problem}}

Throughout this section we consider variants of the following general
optimization problem
\begin{equation}
\max_{\vx\in K_{1}\cap K_{2}}\left\langle \vc,\vx\right\rangle \label{eq:in_conv}
\end{equation}
where $\vx,\vc\in\Rn$, $K_{1}$ and $K_{2}$ are convex subsets of
$\Rn$. We assume that
\begin{equation}
\max_{\vx\in K_{1}}\normFull{\vx}_{2}<M,\ \max_{\vx\in K_{2}}\normFull{\vx}_{2}<M,\ \normFull{\vc}_{2}\leq M\label{eq:int_conv_ass}
\end{equation}
for some constant $M\geq1$ and we assume that 
\begin{equation}
K_{1}\cap K_{2}\neq\emptyset.\label{eq:int_conv_ass2}
\end{equation}

Instead of a separation oracle, we assume that $K_{1}$ and $K_{2}$
each have optimization oracles (see Section~\ref{sec:app:preliminaries}).

To solve this problem we first introduce a relaxation for the problem
\eqref{eq:in_conv} that we can optimize efficiently. Because we have
only the optimization oracles for $K_{1}$ and $K_{2}$, we simply
have variables $\vx$ and $\vy$ for each of them in the objective.
Since the output should (approximately) be in the intersection of
$K_{1}$ and $K_{2}$, a regularization term $-\frac{\lambda}{2}\norm{\vx-\vy}_{2}^{2}$
is added to force $\vx\approx\vy$ where $\lambda$ is a large number
to be determined later. Furthermore, we add terms to make the problem
strong concave.
\begin{lem}
\label{lem:int_conv_rel}Assume \eqref{eq:int_conv_ass} and \eqref{eq:int_conv_ass2}.
For $\lambda\geq1$, let
\begin{equation}
f_{\lambda}(\vx,\vy)\defeq\frac{1}{2}\left\langle \vc,\vx\right\rangle +\frac{1}{2}\left\langle \vc,\vy\right\rangle -\frac{\lambda}{2}\norm{\vx-\vy}_{2}^{2}-\frac{1}{2\lambda}\norm{\vx}_{2}^{2}-\frac{1}{2\lambda}\norm{\vy}_{2}^{2}\enspace.\label{eq:int_conv_rel}
\end{equation}
There is an unique maximizer $(\vx_{\lambda},\vy_{\lambda})$ for
the problem $\max_{\vx\in K_{1},\vy\in K_{2}}f_{\lambda}(\vx,\vy)$.
The maximizer $(\vx_{\lambda},\vy_{\lambda})$ is a good approximation
of the solution of \eqref{eq:in_conv}, i.e. $\norm{\vx_{\lambda}-\vy_{\lambda}}_{2}^{2}\leq\frac{6M^{2}}{\lambda}$
and 
\begin{equation}
\max_{\vx\in K_{1}\cap K_{2}}\left\langle \vc,\vx\right\rangle \leq f_{\lambda}(\vx_{\lambda},\vy_{\lambda})+\frac{M^{2}}{\lambda}.\label{eq:int_conv_rel_err}
\end{equation}
\end{lem}
\begin{proof}
Let $\vx^{*}$ be a maximizer of $\max_{\vx\in K_{1}\cap K_{2}}\left\langle \vc,\vx\right\rangle $.
By assumption \eqref{eq:int_conv_ass}, $\norm{\vx^{*}}_{2}\leq M$,
and therefore
\begin{equation}
f_{\lambda}(\vx^{*},\vx^{*})=\left\langle \vc,\vx^{*}\right\rangle -\frac{\norm{\vx^{*}}_{2}^{2}}{\lambda}\geq\max_{\vx\in K_{1}\cap K_{2}}\left\langle \vc,\vx\right\rangle -\frac{M^{2}}{\lambda}\enspace.\label{eq:app_f_lower}
\end{equation}
This shows \eqref{eq:int_conv_rel_err}. Since $f_{\lambda}$ is strongly
concave in $\vx$ and $\vy$, there is a unique maximizer $(\vx_{\lambda},\vy_{\lambda})$.
Let $\OPT_{\lambda}=f_{\lambda}(\vx_{\lambda},\vy_{\lambda})$. Then,
we have
\begin{eqnarray*}
\OPT_{\lambda} & \leq & \frac{1}{2}\norm{\vc}_{2}\norm{\vx_{\lambda}}_{2}+\frac{1}{2}\norm{\vc}_{2}\norm{\vy_{\lambda}}_{2}-\frac{\lambda}{2}\norm{\vx_{\lambda}-\vy_{\lambda}}_{2}^{2}\\
 & \leq & \frac{M^{2}}{2}+\frac{M^{2}}{2}-\frac{\lambda}{2}\norm{\vx_{\lambda}-\vy_{\lambda}}_{2}^{2}.
\end{eqnarray*}
On the other hand, using $\lambda\geq1$, \eqref{eq:app_f_lower}
shows that
\[
\OPT_{\lambda}\geq f_{\lambda}(\vx^{*},\vx^{*})\geq\max_{\vx\in K_{1}\cap K_{2}}\left\langle \vc,\vx\right\rangle -\frac{M^{2}}{\lambda}\geq-2M^{2}.
\]
Hence, we have 
\begin{eqnarray}
\norm{\vx_{\lambda}-\vy_{\lambda}}_{2}^{2} & \leq & \frac{2\left(M^{2}-\OPT_{\lambda}\right)}{\lambda}\leq\frac{6M^{2}}{\lambda}.\label{eq:conv_int_xy_bound}
\end{eqnarray}

\end{proof}
Now we write $\max f_{\lambda}(\vx,\vy)$ as a max-min problem. The
reason for doing this is that the dual approximate solution is much
easier to obtain and there is a way to read off a primal approximate
solution from a dual approximate solution. This is analogous to the
idea in \cite{lee2013new} which showed how to convert a cut solution
to a flow solution by adding regularization terms into the problem.
\begin{lem}
\label{lem:int_conv_dual}Assume \eqref{eq:int_conv_ass} and \eqref{eq:int_conv_ass2}.
Let $\lambda\geq2$. For any $\vx\in K_{1}$ and $\vy\in K_{2}$,
the function $f_{\lambda}$ can be represented as
\begin{equation}
f_{\lambda}(\vx,\vy)=\min_{(\vtheta_{1},\vtheta_{2},\vtheta_{3})\in\Omega}g_{\lambda}(\vx,\vy,\vtheta_{1},\vtheta_{2},\vtheta_{3})\label{eq:app_f_g_rel}
\end{equation}
where $\Omega=\{(\vtheta_{1},\vtheta_{2},\vtheta_{3}):\norm{\vtheta_{1}}_{2}\leq2M,\ \norm{\vtheta_{2}}_{2}\leq M,\ \norm{\vtheta_{3}}_{2}\leq M\}$
and
\begin{equation}
g_{\lambda}(\vx,\vy,\vtheta_{1},\vtheta_{2},\vtheta_{3})=\left\langle \frac{\vc}{2}+\lambda\vtheta_{1}+\frac{\vtheta_{2}}{\lambda},\vx\right\rangle +\left\langle \frac{\vc}{2}-\lambda\vtheta_{1}+\frac{\vtheta_{3}}{\lambda},\vy\right\rangle +\frac{\lambda}{2}\norm{\vtheta_{1}}_{2}^{2}+\frac{1}{2\lambda}\norm{\vtheta_{2}}_{2}^{2}+\frac{1}{2\lambda}\norm{\vtheta_{3}}_{2}^{2}.\label{eq:int_conv_rel_dual}
\end{equation}

Let $h_{\lambda}(\vtheta_{1},\vtheta_{2},\vtheta_{3})=\max_{\vx\in K_{1},\vy\in K_{2}}g_{\lambda}(\vx,\vy,\vtheta_{1},\vtheta_{2},\vtheta_{3})$.
For any $(\vtheta_{1}',\vtheta_{2}',\vtheta_{3}')$ such that $h_{\lambda}(\vtheta_{1}',\vtheta_{2}',\vtheta_{3}')\leq\min_{(\vtheta_{1},\vtheta_{2},\vtheta_{3})\in\Omega}h_{\lambda}(\vtheta_{1},\vtheta_{2},\vtheta_{3})+\varepsilon,$
we know $\vz=\frac{1}{2}(\vtheta_{2}'+\vtheta_{3}')$ satisfies 
\[
\max_{\vx\in K_{1}\cap K_{2}}\left\langle \vc,\vx\right\rangle \leq\left\langle \vc,\vz\right\rangle +\frac{20M^{2}}{\lambda}+20\lambda^{3}\varepsilon.
\]
and $\norm{\vz-\vx_{\lambda}}_{2}+\norm{\vz-\vy_{\lambda}}_{2}\leq4\sqrt{2\lambda\varepsilon}+\sqrt{\frac{6M^{2}}{\lambda}}$
where $(\vx_{\lambda},\vy_{\lambda})$ is the unique maximizer for
the problem $\max_{\vx\in K_{1},\vy\in K_{2}}f_{\lambda}(\vx,\vy)$.\end{lem}
\begin{proof}
Note that for any $\norm{\vxi}_{2}\leq\alpha$, we have
\[
-\frac{1}{2}\norm{\vxi}_{2}^{2}=\min_{\norm{\vtheta}_{2}\leq\alpha}\left\langle \vtheta,\vxi\right\rangle +\frac{1}{2}\norm{\vtheta}_{2}^{2}
\]
Using this and \eqref{eq:int_conv_ass}, we have \eqref{eq:app_f_g_rel}
for all $\vx\in K_{1}$ and $\vy\in K_{2}$ as desired. Since $\Omega$
is closed and bounded set and the function $g_{\lambda}$ is concave
in $(\vx,\vy)$ and convex in $(\vtheta_{1},\vtheta_{2},\vtheta_{3})$,
Sion's minimax theorem \cite{sion1958general} shows that
\begin{equation}
\max_{\vx\in K_{1},\vy\in K_{2}}f_{\lambda}(\vx,\vy)=\min_{(\vtheta_{1},\vtheta_{2},\vtheta_{3})\in\Omega}\ h_{\lambda}(\vtheta_{1},\vtheta_{2},\vtheta_{3})\label{eq:mat_int_min_max_max_min}
\end{equation}
Since $f_{\lambda}$ is strongly concave, there is an unique maximizer
$(\vx_{\lambda},\vy_{\lambda})$ of $f_{\lambda}$. Since $h_{\lambda}$
is strongly convex, there is a unique minimizer $(\vtheta_{1}^{*},\vtheta_{2}^{*},\vtheta_{3}^{*})$.
By the definition of $f_{\lambda}$ and $h_{\lambda}$, we have
\[
h_{\lambda}(\vtheta_{1}^{*},\vtheta_{2}^{*},\vtheta_{3}^{*})\geq g_{\lambda}(\vx_{\lambda},\vy_{\lambda},\vtheta_{1}^{*},\vtheta_{2}^{*},\vtheta_{3}^{*})\geq f_{\lambda}(\vx_{\lambda},\vy_{\lambda})\enspace.
\]
Using \eqref{eq:mat_int_min_max_max_min}, the equality above holds
and hence $(\vtheta_{1}^{*},\vtheta_{2}^{*},\vtheta_{3}^{*})$ is
the minimizer of $g_{\lambda}(\vx_{\lambda},\vy_{\lambda},\vtheta_{1},\vtheta_{2},\vtheta_{3})$
over $(\vtheta_{1},\vtheta_{2},\vtheta_{3})$. Since the domain $\Omega$
is large enough that $(\vtheta_{1}^{*},\vtheta_{2}^{*},\vtheta_{3}^{*})$
is an interior point in $\Omega$, the optimality condition of $g_{\lambda}$
shows that we have $\vtheta_{2}^{*}=\vx_{\lambda}$ and $\vtheta_{3}^{*}=\vy_{\lambda}$.

Since $h_{\lambda}$ is $\frac{1}{\lambda}$ strongly convex, we have
$\norm{\vtheta_{1}'-\vtheta_{1}^{*}}_{2}^{2}+\norm{\vtheta_{2}'-\vtheta_{2}^{*}}_{2}^{2}+\norm{\vtheta_{3}'-\vtheta_{3}^{*}}_{2}^{2}\leq2\lambda\varepsilon$
(Fact \ref{fact:conv_facts}). Since $\vtheta_{2}^{*}=\vx_{\lambda}$
and $\vtheta_{3}^{*}=\vy_{\lambda}$, we have 
\begin{equation}
\norm{\vtheta_{2}'-\vx_{\lambda}}_{2}^{2}+\norm{\vtheta_{3}'-\vy_{\lambda}}_{2}^{2}\leq2\lambda\varepsilon.\label{eq:mat_int_theta_x_y_bound}
\end{equation}
Therefore, we have $\norm{\vx_{\lambda}-\vy_{\lambda}}_{2}\geq\norm{\vtheta_{2}'-\vtheta_{3}'}_{2}-2\sqrt{2\lambda\varepsilon}$,
$\norm{\vx_{\lambda}}_{2}\geq\norm{\vtheta_{2}'}_{2}-\sqrt{2\lambda\varepsilon}$
and $\norm{\vy_{\lambda}}_{2}\geq\norm{\vtheta_{3}'}_{2}-\sqrt{2\lambda\varepsilon}$.
Using these, $\norm{\vx_{\lambda}}_{2}\leq M$ and $\norm{\vy_{\lambda}}_{2}\leq M$,
we have
\begin{eqnarray*}
f_{\lambda}(\vtheta_{2}',\vtheta_{3}') & = & \frac{1}{2}\left\langle \vc,\vtheta_{2}'\right\rangle +\frac{1}{2}\left\langle \vc,\vtheta_{3}'\right\rangle -\frac{\lambda}{2}\norm{\vtheta_{2}'-\vtheta_{3}'}_{2}^{2}-\frac{1}{2\lambda}\norm{\vtheta_{2}'}_{2}^{2}-\frac{1}{2\lambda}\norm{\vtheta_{3}'}_{2}^{2}\\
 & \geq & \frac{1}{2}\left\langle \vc,\vx_{\lambda}\right\rangle +\frac{1}{2}\left\langle \vc,\vy_{\lambda}\right\rangle -M\sqrt{2\lambda\varepsilon}\\
 &  & -\frac{\lambda}{2}\left(\norm{\vx_{\lambda}-\vy_{\lambda}}_{2}+2\sqrt{2\lambda\varepsilon}\right)^{2}\\
 &  & -\frac{1}{2\lambda}\left(\norm{\vx_{\lambda}}_{2}+\sqrt{2\lambda\varepsilon}\right)^{2}-\frac{1}{2\lambda}\left(\norm{\vy_{\lambda}}_{2}+\sqrt{2\lambda\varepsilon}\right)^{2}\\
 & = & \frac{1}{2}\left\langle \vc,\vx_{\lambda}\right\rangle +\frac{1}{2}\left\langle \vc,\vy_{\lambda}\right\rangle -\frac{\lambda}{2}\norm{\vx_{\lambda}-\vy_{\lambda}}_{2}^{2}-\frac{1}{2\lambda}\norm{\vx_{\lambda}}_{2}^{2}-\frac{1}{2\lambda}\norm{\vy_{\lambda}}_{2}^{2}\\
 &  & -M\sqrt{2\lambda\varepsilon}-2\lambda\sqrt{2\lambda\varepsilon}\norm{\vx_{\lambda}-\vy_{\lambda}}_{2}-4\lambda^{2}\varepsilon\\
 &  & -\frac{1}{\lambda}\norm{\vx_{\lambda}}_{2}\sqrt{2\lambda\varepsilon}-\varepsilon-\frac{1}{\lambda}\norm{\vy_{\lambda}}_{2}\sqrt{2\lambda\varepsilon}-\varepsilon.
\end{eqnarray*}
Using $\norm{\vx_{\lambda}-\vy_{\lambda}}_{2}\leq\sqrt{\frac{6M^{2}}{\lambda}}$
(Lemma \ref{lem:int_conv_rel}), $\norm{\vx_{\lambda}}_{2}<M$ and
$\norm{\vy_{\lambda}}_{2}<M$, we have
\begin{eqnarray*}
f_{\lambda}(\vtheta_{2}',\vtheta_{3}') & \geq & f_{\lambda}(\vx_{\lambda},\vy_{\lambda})\\
 &  & -M\sqrt{2\lambda\varepsilon}-2\lambda\sqrt{2\lambda\varepsilon}\norm{\vx_{\lambda}-\vy_{\lambda}}_{2}-4\lambda^{2}\varepsilon\\
 &  & -\frac{1}{\lambda}\norm{\vx_{\lambda}}_{2}\sqrt{2\lambda\varepsilon}-\varepsilon-\frac{1}{\lambda}\norm{\vy_{\lambda}}_{2}\sqrt{2\lambda\varepsilon}-\varepsilon.\\
 & \geq & f_{\lambda}(\vx_{\lambda},\vy_{\lambda})\\
 &  & -M\sqrt{2\lambda\varepsilon}-2\lambda\sqrt{12\varepsilon}M-4\lambda^{2}\varepsilon\\
 &  & -2M\sqrt{2\frac{\varepsilon}{\lambda}}-2\varepsilon.
\end{eqnarray*}
Since $\lambda\geq2$, we have
\begin{eqnarray*}
f_{\lambda}(\vtheta_{2}',\vtheta_{3}') & \geq & f_{\lambda}(\vx_{\lambda},\vy_{\lambda})-20M\lambda\sqrt{\varepsilon}-10\lambda^{2}\varepsilon.
\end{eqnarray*}

Let $\vz=\frac{\vtheta_{2}'+\vtheta_{3}'}{2}$. Lemma \ref{lem:int_conv_rel}
shows that
\begin{eqnarray*}
\max_{\vx\in K_{1}\cap K_{2}}\left\langle \vc,\vx\right\rangle  & \leq & \max_{\vx\in K_{1},\vy\in K_{2}}f_{\lambda}(\vx,\vy)+\frac{M^{2}}{\lambda}\\
 & \leq & f_{\lambda}(\vtheta_{2}',\vtheta_{3}')+\frac{M^{2}}{\lambda}+20M\lambda\sqrt{\varepsilon}+10\lambda^{2}\varepsilon\\
 & \leq & \left\langle \vc,\vz\right\rangle +\frac{20M^{2}}{\lambda}+20\lambda^{3}\varepsilon
\end{eqnarray*}
because $20M\lambda\sqrt{\varepsilon}\leq10\frac{M^{2}}{\lambda}+10\lambda^{3}\varepsilon$.
Furthermore, we have 
\begin{eqnarray*}
\norm{\vz-\vx_{\lambda}}_{2}+\norm{\vz-\vy_{\lambda}}_{2} & \leq & \norm{\vtheta_{2}'-\vx_{\lambda}}_{2}+\norm{\vtheta_{3}'-\vy_{\lambda}}_{2}+\norm{\vtheta_{2}'-\vtheta_{3}'}_{2}\\
 & \leq & 4\sqrt{2\lambda\varepsilon}+\sqrt{\frac{6M^{2}}{\lambda}}.
\end{eqnarray*}

\end{proof}
We now apply our cutting plane method to solve the optimization problem
\eqref{eq:in_conv}. First we show how to transform the optimization
oracles for $K_{1}$ and $K_{2}$ to get a separation oracle for $h_{\lambda}$,
with the appropriate parameters.
\begin{lem}
\label{lem:weak_separation}Suppose we have a $\varepsilon$-optimization
oracle for $K_{1}$ and $K_{2}$ for some $0<\varepsilon<1$. Then
on the set $\{\norm{\vtheta}_{2}\leq D\}$, we have a $(O(\sqrt{\varepsilon\lambda D}),O(\sqrt{\varepsilon\lambda D}))$-separation
oracle for $h_{\lambda}$ with time complexity $\OO_{\varepsilon}(K_{1})+\OO_{\varepsilon}(K_{2})$.\end{lem}
\begin{proof}
Recall that the function $h_{\lambda}$ is defined by 
\begin{eqnarray*}
 &  & h_{\lambda}(\vtheta_{1},\vtheta_{2},\vtheta_{3})\\
 & = & \max_{\vx\in K_{1},\vy\in K_{2}}\left(\left\langle \frac{\vc}{2}+\lambda\vtheta_{1}+\frac{\vtheta_{2}}{\lambda},\vx\right\rangle +\left\langle \frac{\vc}{2}-\lambda\vtheta_{1}+\frac{\vtheta_{3}}{\lambda},\vy\right\rangle +\frac{\lambda}{2}\norm{\vtheta_{1}}_{2}^{2}+\frac{1}{2\lambda}\norm{\vtheta_{2}}_{2}^{2}+\frac{1}{2\lambda}\norm{\vtheta_{3}}_{2}^{2}\right)\\
 & = & \max_{\vx\in K_{1}}\left\langle \frac{\vc}{2}+\lambda\vtheta_{1}+\frac{\vtheta_{2}}{\lambda},\vx\right\rangle +\max_{\vy\in K_{2}}\left\langle \frac{\vc}{2}-\lambda\vtheta_{1}+\frac{\vtheta_{3}}{\lambda},\vy\right\rangle +\frac{\lambda}{2}\norm{\vtheta_{1}}_{2}^{2}+\frac{1}{2\lambda}\norm{\vtheta_{2}}_{2}^{2}+\frac{1}{2\lambda}\norm{\vtheta_{3}}_{2}^{2}.
\end{eqnarray*}

Lemma~\ref{lem:weak_opt_weak_subgrad} shows how to compute the subgradient
of functions of the form $f(\vc)=\max_{\vx\in K}\left\langle \vc,\vx\right\rangle $
using the optimization oracle for $K$. The rest of the term are differentiable
so its subgradient is just the gradient. Hence, by addition rule for
subgradients (Fact \ref{fact:conv_facts}), we have a $O(\varepsilon\lambda)$-subgradient
oracle for $f_{\lambda}$ using a $O(\varepsilon)$-optimization oracle
for $K_{1}$ and $K_{2}$. The result then follows from Lemma~\ref{lem:opt_implies_sep}.\end{proof}
\begin{thm}
\label{thm:intersection}Assume \eqref{eq:int_conv_ass} and \eqref{eq:int_conv_ass2}.
Suppose that we have $\varepsilon$-optimization oracle for every
$\varepsilon>0$. For $0<\delta<1$, we can find $\vz\in\Rn$ such
that 
\[
\max_{\vx\in K_{1}\cap K_{2}}\left\langle \vc,\vx\right\rangle \leq\delta+\left\langle \vc,\vz\right\rangle 
\]
and $\norm{\vz-\vx}_{2}+\norm{\vz-\vy}_{2}\leq\delta$ for some $\vx\in K_{1}$
and $\vy\in K_{2}$ in time 
\[
O\left(n\left(\OO_{\eta}(K_{1})+\OO_{\eta}(K_{2})\right)\log\left(\frac{nM}{\delta}\right)+n^{3}\log^{O(1)}\left(\frac{nM}{\delta}\right)\right)
\]
where $\eta=\Omega\left(\left(\frac{\delta}{nM}\right)^{O(1)}\right)$.
\end{thm}

\begin{proof}
Setting $\lambda=\frac{40M^{2}}{\delta^{2}}$ and $\varepsilon=\frac{\delta^{7}}{10^{7}M^{6}}$
in Lemma~\ref{lem:int_conv_dual} we see that so long as we obtain
any approximate solution $(\vtheta_{1}',\vtheta_{2}',\vtheta_{3}')$
such that
\[
h_{\lambda}(\vtheta_{1}',\vtheta_{2}',\vtheta_{3}')\leq\min_{(\vtheta_{1},\vtheta_{2},\vtheta_{3})\in\Omega}h_{\lambda}(\vtheta_{1},\vtheta_{2},\vtheta_{3})+\varepsilon,
\]
then we obtain the point we want. To apply Theorem~\ref{thm:conv_opt},
we use 
\[
\tilde{h}(\vtheta_{1},\vtheta_{2},\vtheta_{3})=\begin{cases}
h_{\lambda}(\vtheta_{1},\vtheta_{2},\vtheta_{3}) & \text{if }(\vtheta_{1},\vtheta_{2},\vtheta_{3})\in\Omega\\
+\infty & \text{else}
\end{cases}.
\]
Lemma~\ref{lem:weak_separation} shows that for any $\gamma>0$ we
can obtain a $(\gamma,\gamma)$-separation oracle of $h_{\lambda}(\vtheta)$
by using sufficiently accurate optimization oracles. Since $\Omega$
is just a product of $\ell^{2}$ balls, we can produce a separating
hyperplane easily when $(\vtheta_{1},\vtheta_{2},\vtheta_{3})\notin\Omega$.
Hence, we can obtain a $(\gamma,\gamma)$-separation oracle of $\tilde{h}(\vtheta)$.
For simplicity, we use $\vtheta$ to represent $(\vtheta_{1},\vtheta_{2},\vtheta_{3})$.
Note that $B_{\infty}(2M)\supseteq\Omega$ and therefore we can apply
Theorem~\ref{thm:conv_opt} with $R=2M$ to compute $\vtheta'$ such
\[
\tilde{h}(\vtheta')-\min_{\vtheta\in\Omega}\tilde{h}(\vtheta)\leq\gamma+\alpha\left(\max_{\vtheta\in\Omega}\tilde{h}(\vtheta)-\min_{\vtheta\in\Omega}\tilde{h}(\vtheta)\right)
\]
in time $O\left(n\SO_{\gamma,\gamma}\log\left(\frac{n\kappa}{\alpha}\right)+n^{3}\log^{O(1)}\left(\frac{n\kappa}{\alpha}\right)\right)$
where $\gamma=\Omega\left(\alpha\width(\Omega)/n^{O(1)}\right)=\Omega\left(\alpha M/n^{O(1)}\right)$
and $\kappa=\frac{2M}{\width(\Omega)}=O(1)$. Using $\lambda\geq1$
and $M\geq1$, we have
\[
\max_{\vtheta\in\Omega}\tilde{h}(\vtheta)-\min_{\vtheta\in\Omega}\tilde{h}(\vtheta)\leq O\left(\lambda M^{2}\right)\leq O\left(\frac{M^{4}}{\delta^{2}}\right)\,.
\]
Setting $\alpha=\Theta\left(\frac{\delta^{9}}{M^{10}}\right)$ with
some small enough constant, we have that we can find $\vtheta'$ such
that
\begin{eqnarray*}
h_{\lambda}(\vtheta') & \leq & \min_{\vtheta\in P}h_{\lambda}(\vtheta)+\gamma+\alpha O\left(\frac{M^{4}}{\delta^{2}}\right)\\
 & = & \min_{\vtheta\in P}h_{\lambda}(\vtheta)+O\left(\frac{\delta^{7}}{M^{6}}\right)\\
 & = & \min_{\vtheta\in P}h_{\lambda}(\vtheta)+\varepsilon
\end{eqnarray*}
in time $O\left(n\SO_{\gamma,\gamma}\log\left(\frac{nM}{\delta}\right)+n^{3}\log^{O(1)}\left(\frac{nM}{\delta}\right)\right)$
where $\gamma=\Omega\left(\left(\frac{\delta}{nM}\right)^{O(1)}\right)$.
Lemma~\ref{lem:weak_separation} shows that the cost of $(\gamma,\gamma)$-separation
oracle is just $O(\OO_{\eta}(K_{1})+\OO_{\eta}(K_{2}))$ where $\eta=\Omega\left(\left(\frac{\delta}{nM}\right)^{O(1)}\right)$.\end{proof}
\begin{rem}
Note that the algorithm does not promise that we obtain a point close
to $K_{1}\cap K_{2}$. It only promises to give a point that is close
to both some point in $K_{1}$ and some point in $K_{2}$. It appears
to the authors that a further assumption is needed to get a point
close to $K_{1}\cap K_{2}$. For example, if $K_{1}$ and $K_{2}$
are two almost parallel lines, it would be difficult to get an algorithm
that does not depend on the angle. However, as far as we know, most
algorithms tackling this problem are pseudo-polynomial and have polynomial
dependence on the angle. Our algorithm depends on the logarithmic
of the angle which is useful for combinatorial problems.
\end{rem}
This reduction is very useful for problems in many areas including
linear programming, semi-definite programming and algorithmic game
theory. In the remainder of this section we demonstrate its power
by applying it to classical combinatorial problems.

There is however one issue with applying our cutting plane algorithm
to these problems. As with other convex optimization methods, only
an approximately optimal solution is found. On the other hand, typically
an exact solution is insisted in combinatorial optimization. To overcome
this gap, we introduce the following lemma which (1) transforms the
objective function so that there is only one optimal solution and
(2) shows that an approximate solution is close to the optimal solution
whenever it is unique. As we shall see in the next two subsections,
this allows us to round an approximate solution to an optimal one.
\begin{lem}
\label{lem:isolation_lemma} Given a linear program $\min_{\ma\vx\geq\vb}\vc^{T}\vx$
where $\vx,\vc\in\mathbb{Z}^{n}$, $\vb\in\mathbb{Z}^{m}$ and $\ma\in\mathbb{Z}^{m\times n}$.
Suppose $\{\ma\vx\geq\vb\}$ is an integral polytope (i.e. all extreme
points are integral) contained in the set $\{\norm{\vx}_{\infty}\leq M\}$.
Then we can find a random cost vector $\vz\in\mathbb{Z}^{n}$ with
$\norm{\vz}_{\infty}\leq O(n^{2}M^{2}\norm{\vc}_{\infty})$ such that
with constant probability, $\min_{\ma\vx\geq\vb}\vz^{T}\vx$ has an
unique minimizer $\vx^{*}$ and this minimizer is one of the minimizer(s)
of $\min_{\ma\vx\geq\vb}\vc^{T}\vx$. Furthermore, if there is an
interior point $\vy$ such that $\vz^{T}\vy<\min_{\ma\vx\geq\vb}\vz^{T}\vx+\delta$,
then $\norm{\vy-\vx^{*}}_{\infty}\leq2nM\delta.$\end{lem}
\begin{proof}
The first part of the lemma follows by randomly perturbing the cost
vector $\vc$. We consider a new cost vector $\vz=100n^{2}M^{2}\vc+\vr$
where each coordinate of $\vr$ is sampled randomly from $\{0,1,\cdots,10nM\}$.
\cite[Lem 4]{klivans2001randomness} shows that the linear program
$\min_{\ma\vx\geq\vb}\vz^{T}\vx$ has a unique minimizer with constant
probability. Furthermore, it is clear that the minimizer of $\min_{\ma\vx\geq\vb}\vz^{T}\vx$
is a minimizer of $\min_{\ma\vx\geq\vb}\vc^{T}\vx$ (as $\vr_{i}\ll100n^{2}M^{2}|\vc_{i}|$). 

Now we show the second part of the lemma. Given an interior point
$\vy$ of the polytope $\{\ma\vx\geq\vb\}$, we can write $\vy$ as
a convex combination of the vertices of $\{\ma\vx\geq\vb\}$, i.e.
$\vy=\sum t_{i}\vv_{i}$. Note that $\vz^{T}\vy=\sum t_{i}\vz^{T}\vv_{i}$.
If all $\vv_{i}$ are not the minimizer, then $\vz^{T}\vv_{i}\geq\OPT+1$
and hence $\vz^{T}\vy\geq\OPT+1$ which is impossible. Hence, we can
assume that $\vv_{1}$ is the minimizer. Hence, $\vz^{T}\vv_{i}=\OPT$
if $i=1$ and $\vz^{T}\vv_{i}\geq\OPT+1$ otherwise. We then have
$\vz^{T}\vy\geq\OPT+(1-t_{1})$ which gives $1-t_{1}<\delta$. Finally,
the claim follows from $\norm{\vy-\vv_{1}}_{\infty}\leq\sum_{i\neq1}t_{i}\norm{\vv_{i}-\vv_{1}}_{\infty}\leq2nM\delta$. 
\end{proof}

\subsection{Matroid Intersection\label{sub:Mat_int} }

Let $M_{1}=(E,\mathcal{I}_{1})$ and $M_{2}=(E,\mathcal{I}_{2})$
be two matroids sharing the same ground set. In this section we consider
the weighted matroid intersection problem
\[
\min_{S\in\mathcal{I}_{1}\cap\mathcal{I}_{2}}\vw(S).
\]
where $\vw\in\R^{E}$ and $w(S)\defeq\sum_{e\in S}w_{e}$. 

For any matroid $M=(E,\mathcal{I})$, it is well known that the polytope
of all independent sets has the following description \cite{edmonds1979matroid}:
\begin{equation}
\text{conv}(\mathcal{I}_{1})=\{\vx\in\R^{E}\text{ s.t. }0\leq x(S)\leq r(S)\text{ for all }S\subseteq E\}\label{eq:matroid_formula}
\end{equation}
where $r$ is the rank function for $M$, i.e. $r(S)$ is the size
of the largest independent set that is a subset of $S$. Furthermore,
the polytope of the matroid intersection satisfies $\text{conv}(\mathcal{I}_{1}\cap\mathcal{I}_{2})=\text{conv}(\mathcal{I}_{1})\cap\text{conv}(\mathcal{I}_{2})$.

It is well known that the optimization problem 
\[
\min_{S\in\mathcal{I}_{1}}w(S)\text{ and }\min_{S\in\mathcal{I}_{2}}w(S)
\]
can be solved efficiently by the greedy method. Given a matroid (polytope),
the greedy method finds a maximum weight independent subset by maintaining
a candidate independent subset $S$ and iteratively attempts to add
new element to $S$ in descending weight. A element $i$ is added
to $S$ if $S\cup\{i\}$ is still independent. A proof of this algorithm
is well-known and can be found in any standard textbook on combinatorial
optimization.

Clearly, the greedy method can be implemented by $O(n)$ calls to
the independence oracle (also called membership oracle). For rank
oracle, it requires $O(r\log n)$ calls by finding the next element
to add via binary search. Therefore, we can apply Theorem \ref{thm:intersection}
to get the following result (note that this algorithm is the fastest
if $r$ is close to $n$ for the independence oracle).
\begin{thm}
\label{thm:matroid_inter}Suppose that the weights $\vw$ are integer
with $\norm w_{\infty}\leq M$. Then, we can find
\[
S\in\argmin_{S\in\mathcal{I}_{1}\cap\mathcal{I}_{2}}w(S)
\]
in time $O\left(n\GO\log\left(nM\right)+n^{3}\log^{O(1)}\left(nM\right)\right)$
where $\GO$ is the cost of greedy method for $\mathcal{I}_{1}$ and
$\mathcal{I}_{2}$.\end{thm}
\begin{proof}
Applying Lemma \ref{lem:isolation_lemma}, we can find a new cost
$\vz$ such that 
\[
\min_{\vx\in\conv(\mathcal{I}_{1})\cap\conv(\mathcal{I}_{2})}\vz^{T}\vx
\]
has an unique solution. Note that for any $\vx\in\conv(\mathcal{I}_{1})$,
we have $\norm{\vx}_{\infty}\leq1$. Hence, applying theorem \ref{thm:intersection},
we can find $\vq$ such that $\vq^{T}\vz\leq\OPT+\varepsilon$ and
$\norm{\vq-\vx}_{2}+\norm{\vq-\vy}_{2}\leq\varepsilon$ for some $\vx\in\conv(\mathcal{I}_{1})$
and $\vy\in\conv(\mathcal{I}_{2})$. Using \eqref{eq:matroid_formula},
we have the coordinate wise minimum of $\vx,\vy$, i.e. $\min\{\vx,\vy\}$,
is in $\conv(\mathcal{I}_{1})\cap\conv(\mathcal{I}_{2})$. Since $\norm{\vq-\min\{\vx,\vy\}}_{2}\leq\norm{\vq-\vx}_{2}+\norm{\vq-\vy}_{2}\leq\varepsilon$,
we have 
\[
\left(\min\{\vx,\vy\}\right)^{T}\vz\leq\OPT+nM\varepsilon.
\]
Hence, we have a feasible point $\min\{\vx,\vy\}$ which has value
close to optimal and Lemma~\ref{lem:isolation_lemma} shows that
$\norm{\min(\vx,\vy)-\vs}_{\infty}\leq2n^{2}M^{2}\varepsilon$ where
$\vs$ is the optimal solution. Hence, we have $\norm{\vq-\vs}_{\infty}\leq2n^{2}M^{2}\varepsilon+\varepsilon$.
Picking $\varepsilon=\frac{1}{6n^{2}M^{2}}$, we have $\norm{\vq-\vs}_{\infty}<\frac{1}{2}$
and hence, we can get the optimal solution by rounding to the nearest
integer. 

Since optimization over $\mathcal{I}_{1}$ and $\mathcal{I}_{2}$
involves applying greedy method on certain vectors, it takes only
$O(\GO)$ time. Theorem~\ref{thm:intersection} shows it only takes
$O\left(n\GO\log\left(nM\right)+n^{3}\log^{O(1)}\left(nM\right)\right)$
in finding such $\vq$.
\end{proof}
This gives the following corollary.
\begin{cor}
We have $O(n^{2}\mathcal{T_{\text{ind}}}\log(nM)+n^{3}\log^{O(1)}nM)$
and $O(nr\mathcal{T_{\text{rank}}}\log n\log(nM)+n^{3}\log^{O(1)}nM)$
time algorithms for weighted matroid intersection. Here $\mathcal{T_{\text{ind}}}$
is the time needed to check if a subset is independent, and $\mathcal{T_{\text{rank}}}$
is the time needed to compute the rank of a given subset.\end{cor}
\begin{proof}
By Theorem \ref{thm:matroid_inter}, it suffices to show that the
optimization oracle for the matroid polytope can be implemented in
$O(n\mathcal{T_{\text{ind}}})$ and $O(r\mathcal{T_{\text{rank}}}\log n)$
time. This is simply attained by the well-known greedy algorithm which
iterates through all the positively-weighted elements in decreasing
order, and adds an element to our candidate independent set whenever
possible.

For the independence oracle, this involves one oracle call for each
element. On the other hand, for the rank oracle, we can find the next
element to add by binary search which takes time $O(\mathcal{T_{\text{rank}}}\log n)$.
Since there are at most $r$ elements to add, we have the desired
running time.
\end{proof}

\subsection{Submodular Flow\label{sub:Sub_flow}}

Let $G=(V,E)$ be a directed graphwith $\left|E\right|=m$, let $f$
be a submodular function on $\R^{V}$ with $\left|V\right|=n$, $f(\emptyset)=0$
and $f(V)=0$, and let $A$ be the incidence matrix of $G$. In this
section we consider the submodular flow problem
\begin{eqnarray}
 & \text{Minimize} & \left\langle c,\varphi\right\rangle \label{eq:submodular_flow}\\
 & \text{subject to } & l(e)\leq\varphi(e)\leq u(e)\quad\forall e\in E\nonumber \\
 &  & x(v)=(A\varphi)(v)\quad\forall v\in V\nonumber \\
 &  & \sum_{v\in S}x(v)\leq f(S)\quad\forall S\subseteq V\nonumber 
\end{eqnarray}
where $c\in\Z^{E}$, $l\in\mathbb{Z}^{E}$, $u\in\mathbb{Z}^{E}$
where $C=\norm{\vc}_{\infty}$ and $U=\max\left(\norm u_{\infty},\norm l_{\infty},\max_{S\subset V}\left|f(S)\right|\right)$.
Here $c$ is the cost on edges, $\varphi$ is the flow on edges, $l$
and $u$ are lower and upper bounds on the amount of flow on the edges,
and $x(v)$ is the net flow out of vertex $v$. The submodular function
$f$ upper bounds the total net flow out of any subset $S$ of vertices
by $f(S)$.
\begin{thm}
Suppose that the cost vector $\vc$ is integer weight with $\norm{\vc}_{\infty}\leq C$
and the capacity vector and the submodular function satisfy $U=\max\left(\norm u_{\infty},\norm l_{\infty},\max_{S\subset V}\left|f(S)\right|\right)$.
Then, we can solve the submodular flow problem \eqref{eq:submodular_flow}
in time $O\left(n^{2}\EO\log(mCU)+n^{3}\log^{O(1)}(mCU)\right)$ where
$\EO$ is the cost of function evaluation. \end{thm}
\begin{proof}
First, we can assume $l(e)\leq u(e)$ for every edge $e$, otherwise,
the problem is infeasible. Now, we apply a similar transformation
in \cite{grotschel1981ellipsoid} to modify the graph. We create a
new vertex $v_{0}$. For every vertex $v$ in $V$, we create a edge
from $v_{0}$ to $v$ with capacity lower bounded by $0$, upper bounded
by $4nU$, and with cost $2mCU$. Edmonds and Giles showed that the
submodular flow polytope is integral \cite{edmonds1977min}. Hence,
there is an integral optimal flow on this new graph. If the optimal
flow passes through the newly created edge, then it has cost at least
$2mCU-mCU$ because the cost of all other edges in total has at least
$-mCU$. That means the optimal flow has the cost larger than $mCU$
which is impossible. So the optimal flow does not use the newly created
edges and vertex and hence the optimal flow in the new problem gives
the optimal solution of the original problem. Next, we note that for
any $\varphi$ on the original graph such that $l(e)\leq\varphi(e)\leq u(e)$,
we can send suitable amount of flow from $v_{0}$ to $v$ to make
$\varphi$ feasible. Hence, this modification makes the feasibility
problem trivial.

Lemma~\ref{lem:isolation_lemma} shows that we can assume the new
problem has an unique solution and it only blows up $C$ by a $(mU)^{O(1)}$
factors. 

Note that the optimal value is an integer and its absolute value at
most $mCU$. By binary search, we can assume we know the optimal value
$\OPT$. Now, we reduce the problem to finding a feasible $\varphi$
with $\left\{ \left\langle d,\varphi\right\rangle \leq\OPT+\varepsilon\right\} $
with $\varepsilon$ determined later. Let $P_{\varepsilon}$ be the
set of such $\varphi$. Note that $P_{\varepsilon}=K_{1,\varepsilon}\cap K_{2,\varepsilon}$
where 
\begin{eqnarray*}
K_{1,\varepsilon} & = & \left\{ x\in\R^{V}\text{ such that }\begin{array}{c}
l(e)\leq\varphi(e)\leq u(e)\quad\forall e\in E\\
x(v)=(A\varphi)(v)\quad\forall v\in V\\
\left\langle d,\varphi\right\rangle \leq\OPT+\varepsilon
\end{array}\text{ for some }\varphi\right\} ,\\
K_{2,\varepsilon} & = & \left\{ y\in\R^{V}\text{ such that }\sum_{v\in S}y(v)\leq f(S)\quad\forall S\subseteq V,\sum_{v\in V}y(v)=f(V)\right\} .
\end{eqnarray*}
Note that the extra condition $\sum_{v}y(v)=f(V)$ is valid because
$\sum_{v}y(v)=\sum_{v}(A\varphi)(v)=0$ and $f(V)=0$, and $K_{1,\varepsilon}$
has radius bounded by $O((mCU)^{O(1)})$ and $K_{2,\varepsilon}$
has radius bounded by $O(nU)$. Furthermore, for any vector $\vc\in\R^{V}$,
we note that
\begin{eqnarray*}
\max_{x\in K_{1,\varepsilon}}\left\langle c,x\right\rangle  & = & \max_{l\leq\varphi\leq u,\left\langle d,\varphi\right\rangle \leq\OPT+\varepsilon,x=A\varphi}\left\langle c,x\right\rangle \\
 & = & \max_{l\leq\varphi\leq u,\left\langle d,\varphi\right\rangle \leq\OPT+\varepsilon}\left\langle c,A\varphi\right\rangle \\
 & = & \max_{l\leq\varphi\leq u,\left\langle d,\varphi\right\rangle \leq\OPT+\varepsilon}\left\langle A^{T}c,\varphi\right\rangle .
\end{eqnarray*}
To solve this problem, again we can do a binary search on $\left\langle d,\varphi\right\rangle $
and reduce the problem to
\[
\max_{l\leq\varphi\leq u,\left\langle d,\varphi\right\rangle =K}\left\langle A^{T}c,\varphi\right\rangle 
\]
for some value of $K$. Since $A^{T}c$ is fixed, this is a linear
program with only the box constraints and an extra equality constraint.
Hence, it can be solved in nearly linear time \cite[Thm 17, ArXiv v1]{lee2015efficient}.
As the optimization oracle for $K_{1,\varepsilon}$ involves only
computing $A^{T}c$ and solving this simple linear program, it takes
only $O(n^{2}\log^{O(1)}(mCU/\varepsilon))$ time. On the other hand,
since $K_{2,\varepsilon}$ is just a base polyhedron, the optimization
oracle for $K_{2,\varepsilon}$ can be done by greedy method and only
takes $O(n\EO)$ time.

Applying Theorem \ref{thm:intersection}, we can find $q$ such that
$\norm{q-x}_{2}+\norm{q-y}_{2}\leq\delta$ for some $x\in K_{1,\varepsilon}$,
$y\in K_{2,\varepsilon}$ and $\delta$ to be chosen later. According
to the definition of $K_{1,\varepsilon}$, there is $\varphi$ such
that $l(e)\leq\varphi(e)\leq u(e)$ and $x(v)=(A\varphi)(v)$ for
all $v$ and $\left\langle d,\varphi\right\rangle \leq\OPT+\varepsilon$.
Since $\norm{y-x}_{2}\leq2\delta$, that means $\left|y(v)-(A\varphi)(v)\right|\leq2\delta$
for all $v$. 
\begin{itemize}
\item Case 1) If $y(v)\geq(A\varphi)(v)$, then we can replace $y(v)$ by
$(A\varphi)(v)$, note that $y$ is still in $K_{2,\varepsilon}$
because of the submodular constraints.
\item Case 2) If $y(v)\leq(A\varphi)(v)$, then we can send a suitable amount
of flow from $v_{0}$ to $v$ to make $\varphi$ feasible $y(v)\leq(A\varphi)(v)$.
\end{itemize}
Note that under this modification, we increased the objective value
by $(\delta n)(2mCU)$ because the new edge cost $2mCU$ per unit
of flow. Hence, we find a flow $\varphi$ which is feasible in new
graph with objective value $\varepsilon+(\delta n)(2mCU)$ far from
optimum value. By picking $\delta=\frac{1}{2mnCU}$, we have the value
$2\varepsilon$ far from $\OPT$. Now, we use Lemma \ref{lem:isolation_lemma}
to shows that when $\varepsilon$ is small enough, i.e, $\frac{1}{(mCU)^{c}}$
for some constant $c$, then we can guarantee that $\norm{y-x^{*}}_{\infty}\leq\frac{1}{4}$
where $x^{*}$ is the optimal demand. Now, we note that $\norm{q-y}_{2}\leq\delta$
and we note that we only modify $y$ by a small amount, we in fact
have $\norm{q-x^{*}}_{\infty}<\frac{1}{2}$. Hence, we can read off
the solution $x^{*}$ by rounding $q$ to the nearest integer. Note
that we only need to solve the problem $K_{1,\varepsilon}\cap K_{2,\varepsilon}$
to $\frac{1}{(mCU)^{\Theta(1)}}$ accuracy and the optimization oracle
for $K_{1,\varepsilon}$ and $K_{2,\varepsilon}$ takes time $O(n^{2}\log^{O(1)}(mCU))$
and $O(n\EO)$ respectively. Hence, Theorem \ref{thm:intersection}
shows that it takes $O\left(n^{2}\EO\log(mCU)+n^{3}\log^{O(1)}(mCU)\right)$
time to find $x^{*}$ exactly. 

After getting $x^{*}$, one can find $\varphi^{*}$ by solving a min
cost flow problem using interior point method \cite{lsMaxflow}, which
takes $O(m\sqrt{n}\log^{O(1)}(mCU))$ time. \end{proof}

\subsection{Affine Subspace of Convex Set\label{sec:Intersection2}}

In this section, we give another example about using optimization
oracle directly via regularization. We consider the following optimization
problem
\begin{equation}
\max_{\vx\in K\text{ and }\ma\vx=\vb}\left\langle \vc,\vx\right\rangle \label{eq:in_conv_2}
\end{equation}
where $\vx,\vc\in\Rn$, $K$ is a convex subset of $\Rn$, $\ma\in\mathbb{R}^{r\times n}$
and $\vb\in\Rm$. We suppose that $r\ll n$ and thus, the goal of
this subsection is to show how to obtain an algorithm takes only $\tilde{O}(r)$
many iterations. To do this, we assume a slightly stronger optimization
oracle for $K$:
\begin{defn}
\label{def:weak_opt2} Given a convex set $K$ and $\delta>0$. A
$\delta$-2nd-order-optimization oracle for $K$ is a function on
$\Rn$ such that for any input $\vc\in\Rn$ and $\lambda>0$, it outputs
$\vy$ such that
\[
\max_{\vx\in K}\left(\left\langle \vc,\vx\right\rangle -\lambda\norm{\vx}^{2}\right)\leq\delta+\left\langle \vc,\vy\right\rangle -\lambda\norm{\vy}^{2}.
\]
We denote by $\OO_{\delta,\lambda}^{(2)}(K)$ the time complexity
of this oracle.
\end{defn}
The strategy for solving this problem is very similar to the intersection
problem and hence some details are omitted.
\begin{thm}
Assume that $\max_{\vx\in K}\normFull{\vx}_{2}<M$, $\norm{\vb}_{2}<M$,
$\norm{\vc}_{2}<M$, $\norm{\ma}_{2}<M$ and $\lambda_{\min}(\ma)>1/M$.
Assume that $K\cap\{\ma\vx=\vb\}\neq\emptyset$ and we have $\varepsilon$-2nd-order-optimization
oracle for every $\varepsilon>0$. For $0<\delta<1$, we can find
$\vz\in K$ such that 
\[
\max_{\vx\in K\text{ and }\ma\vx=\vb}\left\langle \vc,\vx\right\rangle \leq\delta+\left\langle \vc,\vz\right\rangle 
\]
and $\norm{\ma\vz-\vb}_{2}\leq\delta$. This algorithm takes time
\[
O\left(r\OO_{\eta,\lambda}^{(2)}(K)\log\left(\frac{nM}{\delta}\right)+r^{3}\log^{O(1)}\left(\frac{nM}{\delta}\right)\right)
\]
where $r$ is the number of rows in $\ma$, $\eta=\left(\frac{\delta}{nM}\right)^{\Theta(1)}$
and $\lambda=\left(\frac{\delta}{nM}\right)^{\Theta(1)}$.\end{thm}
\begin{proof}
The proof is based on the minimax problem
\[
\text{OPT}_{\lambda}\defeq\min_{\norm{\veta}_{2}\leq\lambda}\max_{\vx\in K}\left\langle \vc,\vx\right\rangle +\left\langle \veta,\ma\vx-\vb\right\rangle -\frac{1}{\lambda}\norm{\vx}_{2}^{2}
\]
where $\lambda=\left(\frac{\delta}{nM}\right)^{c}$ for some large
constant $c$. We note that
\begin{eqnarray*}
\text{OPT}_{\lambda} & = & \max_{\vx\in K}\min_{\norm{\veta}_{2}\leq\lambda}\left\langle \vc,\vx\right\rangle +\left\langle \veta,\ma\vx-\vb\right\rangle -\frac{1}{\lambda}\norm{\vx}_{2}^{2}\\
 & = & \max_{\vx\in K}\left\langle \vc,\vx\right\rangle -\lambda\norm{\ma\vx-\vb}_{2}-\frac{1}{\lambda}\norm{\vx}_{2}^{2}.
\end{eqnarray*}
Since $\lambda_{\min}(\ma)>1/M$ and the set $K$ is bounded by $M$,
one can show that the saddle point $(\vx^{*},\veta^{*})$ of the minimax
problem gives a good enough solution $\vx$ for the original problem
for large enough constant $c$. 

For any $\veta$, we define
\[
\vx_{\veta}=\argmax_{\vx\in K}\left\langle \vc,\vx\right\rangle +\left\langle \veta,\ma\vx-\vb\right\rangle -\frac{1}{\lambda}\norm{\vx}_{2}^{2}.
\]
Since the problem is strongly concave in $\vx$, one can prove that
\[
\norm{\vx_{\veta}-\vx^{*}}_{2}\leq\left(\frac{nM}{\delta}\right)^{O(c)}\norm{\veta-\veta^{*}}_{2}.
\]
Hence, we can first find an approximate minimizer of the function
$f(\veta)=\max_{\vx\in K}\left\langle \vc,\vx\right\rangle +\left\langle \veta,\ma\vx-\vb\right\rangle -\frac{1}{\lambda}\norm{\vx}_{2}^{2}$
and use the oracle to find $\vx_{\veta}$.

To find an approximate minimizer of $f$, we note that the subgradient
of $f$ can be found using the optimization oracle similar to Theorem
\ref{thm:intersection}. Hence, the result follows from our cutting
plane method and the fact that $\veta\in\R^{r}$.\end{proof}
\begin{rem}
In \cite{lsMaxflow}, they considered the special case $K=\{\vx:0\leq x_{i}\leq1\}$
and showed that it can be solved in $\tilde{O}(\sqrt{r})$ iterations
using interior point methods. This gives the current fastest algorithm
for the maximum flow problem on directed weighted graphs. Our result
generalizes their result to any convex set $K$ but with $\tilde{O}(r)$
iterations. This suggests the following open problem: under what condition
on $K$ can one optimize linear functions over affine subspaces of
$K$ with $r$ constraints in $\tilde{O}(\sqrt{r})$ iterations?\end{rem}

\pagebreak{}\newpage{}
\begin{description}
\item [{\global\long\def\gap{\mathtt{gap}}
}]~
\end{description}

\part{Submodular Function Minimization\label{part:Submodular-Function-Minimization}}
\begin{description}
\item [{\global\long\def\upi{\mathtt{upper}(i)}
\global\long\def\lowi{\mathtt{lower}(i)}
\global\long\def\upp{\mathtt{upper}(p)}
\global\long\def\lowp{\mathtt{lower}(p)}
\global\long\def\basepolytope{\mathcal{B}(f)}
}]~
\end{description}

\section{Introduction\label{sec:submodular_intro}}

Submodular functions and submodular function minimization (SFM) are
fundamental to the field of combinatorial optimization. Examples of
submodular functions include graph cut functions, set coverage function,
and utility functions from economics. Since the seminal work by Edmonds
in 1970 \cite{edmonds1970submodular}, submodular functions and the
problem of minimizing such functions (i.e. submodular function minimization)
have served as a popular modeling and optimization tool in various
fields such as theoretical computer science, operations research,
game theory, and most recently, machine learning. Given its prevalence,
fast algorithms for SFM are of immense interest both in theory and
in practice.

Throughout Part~\ref{part:Submodular-Function-Minimization}, we
consider the standard formulation of SFM: we are given a submodular
function $f$ defined over the subsets of a $n$-element ground set.
The values of $f$ are integers, have absolute value at most $M$,
and are evaluated by querying an oracle that takes time $\text{EO}$.
Our goal is to produce an algorithm that solves this SFM problem,
i.e. finds a minimizer of $f$, while minimizing both the number of
oracle calls made and the total running time. 

We provide new $O(n^{2}\log nM\cdot\text{EO}+n^{3}\log^{O(1)}nM)$
and $O(n^{3}\log^{2}n\cdot\text{EO}+n^{4}\log^{O(1)}n)$ time algorithms
for SFM. These algorithms improve upon the previous fastest weakly
and strongly polynomial time algorithms for SFM which had a a running
time of $O((n^{4}\cdot\text{EO}+n^{5})\log M)$ \cite{iwata2003faster}
and $O(n^{5}\cdot\text{EO}+n^{6})$ \cite{orlin2009faster} respectively.
Consequently, we improve the running times in both regimes by roughly
a factor of $O(n^{2})$.

Both of our algorithms bear resemblance to the classic approach of
Grötschel, Lovász and Schrijver \cite{grotschel1981ellipsoid,grotschel1988ellipsoid}
using the Lovász extension. In fact our weakly polynomial time algorithm
directly uses the Lovász extension as well as the results of Part~\ref{part:app}
to achieve these results. Our strongly polynomial time algorithm also
uses the Lovász extension, along with more modern tools from the past
15 years.

At a high level, our strongly polynomial algorithms apply our cutting
plane method in conjunction with techniques originally developed by
Iwata, Fleischer, and Fujishige (IFF) \cite{iwata2001combinatorial}.
Our cutting plane method is performed for enough iterations to sandwich
the feasible region in a narrow strip from which useful structural
information about the minimizers can be deduced. Our ability to derive
the new information hinges on a significant extension of IFF techniques.

Over the past few decades, SFM has drawn considerable attention from
various research communities, most recently in machine learning \cite{bach2011learning,suborg}.
Given this abundant interest in SFM, we hope that our ideas will be
of value in various practical applications. Indeed, one of the critiques
against existing theoretical algorithms is that their running time
is too slow to be practical. Our contribution, on the contrary, shows
that this school of algorithms can actually be made fast theoretically
and we hope it may potentially be competitive against heuristics which
are more commonly used.

\subsection{Previous Work}

Here we provide a brief survey of the history of algorithms for SFM.
For a more comprehensive account of the rich history of SFM, we refer
the readers to recent surveys \cite{mccormicksurvey,iwata2008submodular}.

The first weakly and strongly polynomial time algorithms for SFM were
based on the ellipsoid method \cite{khachiyan1980polynomial} and
were established in the foundational work of Grötschel, Lovász and
Schrijver in 1980's \cite{grotschel1981ellipsoid,grotschel1988ellipsoid}.
Their work was complemented by a landmark paper by Cunningham in 1985
which provided a pseudopolynomial algorithm that followed a flow-style
algorithmic framework \cite{cunningham1985submodular}. His tools
foreshadowed much of the development in SFM that would take place
15 years later. Indeed, modern algorithms synthesize his framework
with inspirations from various max flow algorithms. 

The first such ``flow style'' strongly polynomial algorithms for
SFM were discovered independently in the breakthrough papers by Schrijver
\cite{schrijver2000combinatorial} and Iwata, Fleischer, and Fujishige
(IFF) \cite{iwata2001combinatorial}. Schrijver's algorithm has a
running of $O(n^{8}\cdot\text{EO}+n^{9})$ and borrows ideas from
the push-relabel algorithms \cite{goldberg1988new,dinic1970algorithm}
for the maximum flow problem. On the other hand, IFF's algorithm runs
in time $O(n^{7}\log n\cdot\text{EO})$ and $O(n^{5}\cdot\text{EO}\log M)$,
and applies a flow-scaling scheme with the aid of certain proximity-type
lemmas as in the work of Tardos \cite{tardos1985strongly}. Their
method has roots in flow algorithms such as \cite{iwata1997capacity,goldberg1990finding}.

Subsequent work on SFM provided algorithms with considerably faster
running time by extending the ideas in these two ``genesis'' papers
\cite{schrijver2000combinatorial,iwata2001combinatorial} in various
novel directions \cite{vygen2003note,fleischer2003push,iwata2003faster,orlin2009faster,iwata2009simple}.
Currently, the fastest weakly and strongly polynomial time algorithms
for SFM have a running time of $O((n^{4}\cdot\text{EO}+n^{5})\log M)$
\cite{iwata2003faster} and $O(n^{5}\cdot\text{EO}+n^{6})$ \cite{orlin2009faster}
respectively. Despite this impressive track record, the running time
has not been improved in the last eight years.

We remark that all of the previous algorithms for SFM proceed by maintaining
a convex combination of $O(n)$ BFS's of the base polyhedron, and
incrementally improving it in a relatively local manner. As we shall
discuss in Section~\ref{sub:Our-results-submo}, our algorithms do
not explicitly maintain a convex combination. This may be one of the
fundamental reasons why our algorithms achieve a faster running time.

Finally, beyond the distinction between weakly and strongly polynomial
time algorithms for SFM, there has been interest in another type of
SFM algorithm, known as fully combinatorial algorithms in which only
additions and subtractions are permitted. Previous such algorithms
include \cite{iwata2009simple,iwata2003faster,iwata2002fully}. We
do not consider such algorithms in the remainder of the paper and
leave it as an open question if it is possible to turn our algorithms
into fully combinatorial ones.

\subsection{Our Results and Techniques\label{sub:Our-results-submo}}

In Part~\ref{part:Submodular-Function-Minimization} we show how
to improve upon the previous best known running times for SFM by a
factor of $O(n^{2})$ in both the strongly and weakly polynomial regimes.
In Table~\ref{fig:SFM_table_1-1} summarizes the running time of
the previous algorithms as well as the running times of the fastest
algorithms presented in this paper. 
\begin{table}
\begin{centering}
\begin{tabular}{|c|c|c|c|}
\hline 
Authors & Years & Running times & Remarks\tabularnewline
\hline 
\hline 
Grötschel, Lovász,  & \multirow{2}{*}{1981,1988} & \multirow{2}{*}{$\widetilde{O}(n^{5}\cdot\text{EO}+n^{7})$\cite{mccormicksurvey}} & first weakly\tabularnewline
Schrijver \cite{grotschel1981ellipsoid,grotschel1988ellipsoid} &  &  &  and strongly\tabularnewline
\hline 
Cunningham \cite{cunningham1985submodular} & 1985 & $O(Mn^{6}\log nM\cdot\text{EO})$ & first pseudopoly\tabularnewline
\hline 
Schrijver \cite{schrijver2000combinatorial} & 2000 & $O(n^{8}\cdot\text{EO}+n^{9})$ & first combin. strongly\tabularnewline
\hline 
Iwata, Fleischer,  & \multirow{2}{*}{2000} & \multirow{2}{*}{$\begin{array}{c}
O(n^{5}\cdot\text{EO}\log M)\\
O(n^{7}\log n\cdot\text{EO})
\end{array}$} & \multirow{2}{*}{first combin. strongly}\tabularnewline
Fujishige\cite{iwata2001combinatorial} &  &  & \tabularnewline
\hline 
Iwata, Fleischer \cite{fleischer2003push} & 2000 & $O(n^{7}\cdot\text{EO}+n^{8})$ & \tabularnewline
\hline 
Iwata \cite{iwata2003faster} & 2003 & $\begin{array}{c}
O((n^{4}\cdot\text{EO}+n^{5})\log M)\\
O((n^{6}\cdot\text{EO}+n^{7})\log n)
\end{array}$ & current best weakly\tabularnewline
\hline 
Vygen \cite{vygen2003note} & 2003 & $O(n^{7}\cdot\text{EO}+n^{8})$ & \tabularnewline
\hline 
Orlin \cite{orlin2009faster} & 2007 & $O(n^{5}\cdot\text{EO}+n^{6})$ & current best strongly\tabularnewline
\hline 
Iwata, Orlin \cite{iwata2009simple} & 2009 & $\begin{array}{c}
O((n^{4}\cdot\text{EO}+n^{5})\log nM)\\
O((n^{5}\cdot\text{EO}+n^{6})\log n)
\end{array}$ & \tabularnewline
\hline 
\textbf{Our algorithms} & 2015 & $\begin{array}{c}
O(n^{2}\log nM\cdot\text{EO}+n^{3}\log^{O(1)}nM)\\
O(n^{3}\log^{2}n\cdot\text{EO}+n^{4}\log^{O(1)}n)
\end{array}$ & \tabularnewline
\hline 
\end{tabular}\,
\par\end{centering}

\protect\caption{\label{fig:SFM_table_1-1} Algorithms for submodular function minimization.
Note that some of these algorithms were published in both conferences
and journals, in which case the year we provided is the earlier one.}
\end{table}

Both our weakly and strongly polynomial algorithms for SFM utilize
a convex relaxation of the submodular function, called the \emph{Lovász
extension}. Our algorithms apply our cutting plane method from Part~\ref{part:Ellipsoid}
using a separation oracle given by the subgradient of the Lov\emph{á}sz
extension. To the best of the author's knowledge, Grötschel, Lovász
and Schrijver were the first to formulate this convex optimization
framework for SFM \cite{grotschel1981ellipsoid,grotschel1988ellipsoid}.

For weakly polynomial algorithms, our contribution is two-fold. First,
we show that cutting plane methods such as Vaidya's \cite{vaidya1996new}
can be applied to SFM to yield faster algorithms. Second, as our cutting
plane method, Theorem~\ref{thm:conv_opt}, improves upon previous
cutting plane algorithms and consequently the running time for SFM
as well. This gives a running time of $O(n^{2}\log nM\cdot\text{EO}+n^{3}\log^{O(1)}nM)$,
an improvement over the previous best algorithm by Iwata {\small{}\cite{iwata2003faster}
}by a factor of almost $O(n^{2})$.

Our strongly polynomial algorithms, on the other hand, require substantially
more innovation. We first begin with a very simple geometric argument
that SFM can be solved in $O(n^{3}\log n\cdot\text{EO})$ oracle calls
(but in exponential time). This proof only uses Grunbaum's Theorem
from convex geometry and is completely independent from the rest of
the paper. It was the starting point of our method and suggests that
a running time of $\widetilde{O}(n^{3}\cdot\text{EO}+n^{O(1)})$ for
submodular minimization is in principle achievable.

To make this existence result algorithmic, we first run cutting plane,
Theorem~\ref{thm:main_result}, for enough iterations such that we
compute either a minimizer or a set $P$ containing the minimizers
that fits within in a narrow strip. This narrow strip consists of
the intersection of two approximately parallel hyperplanes. If our
narrow strip separates $P$ from one of the faces $x_{i}=0$, $x_{i}=1$,
we can effectively eliminate the element $i$ from our consideration
and reduce the dimension of our problem by 1. Otherwise a pair of
elements $p,q$ can be identified for which $q$ is guaranteed to
be in any minimizer containing $p$ (but $p$ may not be contained
in a minimizer). Our first algorithm deduces only one such pair at
a time. This technique immediately suffices to achieve a $\widetilde{O}(n^{4}\cdot\text{EO}+n^{5})$
time algorithm for SFM (See Section~\ref{sub:submodular-n4}). We
then improve the running time to $\widetilde{O}(n^{3}\cdot\text{EO}+n^{4})$
by showing how to deduce many such pairs simultaneously. Similar to
past algorithms, this structural information is deduced from a point
in the so-called base polyhedron (See Section~\ref{sec:Background-on-Submodular}).

Readers well-versed in SFM literature may recognize that our strongly
polynomial algorithms are reminiscent of the scaling-based approach
first used by IFF {\small{}\cite{iwata2001combinatorial}} and later
in {\small{}\cite{iwata2003faster,iwata2009simple}}. While both approaches
share the same skeleton, there are differences as to how structural
information about minimizers is deduced. A comparison of our algorithms
and previous ones are presented in Section~\ref{sec:submodular_discussion}.

Finally, there is one more crucial difference between these algorithms
which we believe is responsible for much of our speedup. One common
feature shared by all the previous algorithms is that they maintain
a convex combination of $O(n)$ BFS's of the base polyhedron, and
incrementally improve on it by introducing new BFS's by making local
changes to existing ones. Our algorithms, on the other hand, choose
new BFS's by the cutting plane method. Because of this, our algorithm
considers the geometry of the existing BFS's where each of them has
influences over the choice of the next BFS. In some sense, our next
BFS is chosen in a more ``global'' manner.

\subsection{Organization}

The rest of Part~\ref{part:Submodular-Function-Minimization} is
organized as follows. We first begin with a gentle introduction to
submodular functions in Section~\ref{sec:Background-on-Submodular}.
In Section~\ref{sec:Sub-Weakly-Polynomial}, we apply our cutting
plane method to SFM to obtain a faster weakly polynomial algorithms.
In Section~\ref{sec:Sub-Strongly-Polynomial} we then present our
results for achieving better strongly polynomial algorithms, where
a warm-up $\widetilde{O}(n^{4}\cdot\text{EO}+n^{5})$ algorithm is
given before the full-fledged $\widetilde{O}(n^{3}\cdot\text{EO}+n^{4})$
algorithm. Finally, we end the part with a discussion and comparison
between our algorithms and previous ones in Section~\ref{sec:submodular_discussion}.

We note that there are a few results in Part~\ref{part:Submodular-Function-Minimization}
that can be read fairly independently of the rest of the paper. In
Theorem~\textit{\ref{thm:santosh_CG} }\textit{\emph{we show how
Vaidya's algorithm can be applied to SFM to obtain a faster weakly
polynomial running time. Also in Theorem~\ref{thm:n3logn} we present
a simple geometric argument that SFM can be solved with $O(n^{3}\log n\cdot\text{EO})$
oracle calls but with exponential time. These results can be read
with only }}a working knowledge of the Lovász extension of submodular
functions.

\section{Preliminaries\label{sec:Background-on-Submodular}}

Here we introduce background on submodular function minimization (SFM)
and notation that we use throughout Part~\ref{part:Submodular-Function-Minimization}.
Our exposition is kept to a minimal amount sufficient for our purposes.
We refer interested readers to the extensive survey by McCormick \cite{mccormicksurvey}
for further intuition.

\subsection{Submodular Function Minimization}

Throughout the rest of the paper, let $V=\{1,...,n\}=[n]$ denote
a ground set and let $f:2^{V}\longrightarrow\mathbb{Z}$ denote a
\emph{submodular} \emph{function} defined on subsets of this ground
set. We use $V$ and $[n]$ interchangeably and let $[0]\defeq\emptyset$.
We abuse notation by letting $S+i\defeq S\cup\{i\}$ and $S-i\defeq S\backslash\{i\}$
for an element $i\in V$ and a set $S\subseteq2^{V}$. Formally, we
call a function submodular if it obeys the following property of \textit{diminishing
marginal difference}s: 
\begin{defn}[Submodularity]
 A function $f:2^{V}\longrightarrow\mathbb{Z}$ is \emph{submodular}
if $f(T+i)-f(T)\leq f(S+i)-f(S)$ for any $S\subseteq T$ and $i\in V\backslash T$.
\end{defn}
For convenience we assume without loss of generality that $f(\emptyset)=0$
by replacing $f(S)$ by $f(S)-f(\emptyset)$ for all $S$. We also
let $M\defeq\max_{S\in2^{V}}|f(S)|$. 

The central goal of Part~\ref{part:Submodular-Function-Minimization}
is to design algorithms for SFM, i.e. computing the minimizer of $f$.
We call such an algorithm \emph{strongly polynomial} if its running
time depends only polynomially on $n$ and $\EO$, the time needed
to compute $f(S)$ for a set $S$, and we call such an algorithm \emph{weakly
polynomial }if it also depends polylogarithmically on $M$.

\subsection{Lovász Extension }

Our new algorithms for SFM all consider a convex relaxation of a submodular
function, known as the Lovász extension, and then carefully apply
our cutting plane methods to it. Here we formally introduce the Lovász
extension and present basic facts that we use throughout Part~\ref{part:Submodular-Function-Minimization}. 

The Lovász extension of $\hat{f}:[0,1]^{n}\longrightarrow\mathbb{R}$
of our submodular function $f$ is defined for all $\vx$ by 
\[
\hat{f}(\vx)\defeq\mathbb{E}_{t\sim[0,1]}[f(\{i:x_{i}\geq t\})],
\]
where $t\sim[0,1]$ is drawn uniformly at random from $[0,1]$. The
Lovász extension allows us to reduce SFM to minimizing a convex function
defined over the interior of the hypercube. Below we state that the
Lovász extension is a convex relaxation of $f$ and that it can be
evaluated efficiently.
\begin{thm}
\label{thm:The-Lovasz-extension} The Lovász extension $\hat{f}$
satisfies the following properties:
\begin{enumerate}
\item $\hat{f}$ is convex and $\min_{\vx\in[0,1]^{n}}\hat{f}(\vx)=\min_{S\subset[n]}f(S)$;
\item $f(S)=\hat{f}(I_{S})$, where $I_{S}$ is the characteristic vector
for $S$, i.e. $I_{S}(i)=\begin{cases}
1 & \text{if }i\in S\\
0 & \text{if }i\notin S
\end{cases}$;
\item If $S$ is a minimizer of $f$, then $I_{S}$ is a minimizer of $\hat{f}$;
\item Suppose $x_{1}\geq\cdots\geq x_{n}\geq x_{n+1}\defeq0$, then 
\[
\hat{f}(\vx)=\sum_{i=1}^{n}f([i])(x_{i}-x_{i+1})=\sum_{i=1}^{n}(f([i])-f([i-1]))x_{i}\,.
\]

\end{enumerate}
\end{thm}
\begin{proof}
See \cite{grotschel1988ellipsoid} or any standard textbook on combinatorial
optimization, e.g. \cite{schrijver2003combinatorial}.
\end{proof}
Next we show that we can efficiently compute a subgradient of the
Lovász or alternatively, a separating hyperplane for the set of minimizers
of our submdoular function $f$. First we remind the reader of the
definition of a separation oracle, and then we prove the necessary
properties of the hyperplane, Theorem~\ref{thm:sep}.
\begin{defn}[separation oracle, Defintion~\ref{def:weak_sep} restated for Lovász
extension]
\label{def:submod:separation-reminder} Given a point $\bar{x}$
and a convex function $\hat{f}$ over a convex set $P$, $\va^{T}\vx\leq b$
is a separating hyperplane if $\va^{T}\bar{x}\ge b$ and any minimizer
$x^{*}$ of $\hat{f}$ over $P$ satisfies $\va^{T}x^{*}\leq b$.\end{defn}
\begin{thm}
\label{thm:sep} Given a point $\bar{x}\in[0,1]^{n}$ assume without
loss of generality (by re-indexing the coordinates) that $\bar{x}_{1}\geq\cdots\geq\bar{x}_{n}$.
Then the following inequality is a valid separating hyperplane for
$\vx$ and $f$:
\[
\sum_{i=1}^{n}(f([i])-f([i-1]))x_{i}\leq\hat{f}(\bar{x})
\]
i.e., it satisfies the following:
\begin{enumerate}
\item (separating) $\bar{x}$ lies on $\sum_{i=1}^{n}(f([i])-f([i-1]))x_{i}\leq\hat{f}(\bar{x})$.
\item (valid) For any $\vx$, we have $\sum_{i=1}^{n}(f([i])-f([i-1]))x_{i}\leq\hat{f}(\vx)$.
In particular, $\sum_{i=1}^{n}(f([i])-f([i-1]))x_{i}^{*}\leq\hat{f}(\bar{x})$
for any minimizer $\vx^{*}$, i.e. the separating hyperplane does
not cut out any minimizer. 
\end{enumerate}
Moreover, such a hyperplane can be computed with $n$ oracle calls
to $f$ and in time $O(n\cdot\text{EO}+n^{2})$.\end{thm}
\begin{proof}
Note that by Theorem~\ref{thm:The-Lovasz-extension} we have that
$\sum_{i\in[n]}(f([i])-f([i-1]))x_{i}=\hat{f}(\bar{x})$ and thus
the hyperplane satisfies the separating condition. Moreover, clearly
computing it only takes time $O(n\cdot\text{EO}+n^{2})$ as we simply
need to sort the coordinates and evaluate $f$ at $n$ points, i.e.
each of the $[i]$. All that remains is to show that the hyperplane
satisfies the valid condition.

Let $L^{(t)}\defeq\{i\,:\, x_{i}\geq t\}$. Recall that $\hat{f}(\vx)=\mathbb{E}_{t\sim[0,1]}[f(L_{t})]$.
Thus $\hat{f}(\vx)$ can be written as a convex combination $\hat{f}(\vx)=\sum_{t}\alpha_{t}f(L^{(t)})$,
where $\alpha_{t}\geq0$ and $\sum_{t}\alpha_{t}=1$. However, by
diminishing marginal differences we see that for all $t$ 
\begin{eqnarray*}
\sum_{i\in[n]}(f([i])-f([i-1]))\left(I_{L^{(t)}}\right)_{i} & = & \sum_{i\in L^{(t)}}\left(f([i])-f([i-1])\right)\\
 & \leq & \sum_{i\in L^{(t)}}\left(f([i]\cap L^{(t)})-f([i-1]\cap L^{(t)})\right)\\
 & = & f(L^{(t)})-f(\emptyset)=f(L^{(t)})
\end{eqnarray*}
and therefore since $\sum_{t}\alpha_{t}I_{L^{(t)}}=\vx$ we have 
\[
\sum_{i\in[n]}(f[i]-f([i-1])x_{i}=\sum_{t}\alpha_{t}\sum_{i=1}^{n}(f([i])-f([i-1]))\left(I_{L^{(t)}}\right)_{i}\leq\sum_{t}\alpha_{t}f(L^{(t)})=\hat{f}(\vx).
\]

\end{proof}

\subsection{Polyhedral Aspects of SFM}

Here we provide a natural primal dual view of SFM that we use throughout
the analysis. We provide a dual convex optimization program to minimizing
the Lovász extension and provide several properties of these programs.
We believe the material in this section helps crystallize some of
the intuition behind our algorithm and we make heavy use of the notation
presented in this section. However, we will not need to appeal to
the strong duality of these programs in our proofs.

Consider the following primal and dual programs, where we use the
shorthands $y(S)=\sum_{i\in S}y_{i}$ and $y_{i}^{-}=\min\{0,y_{i}\}$.
Here the primal constraints are often called the \emph{base polyhedron
$\mathcal{B}(f)\defeq\{\vy\in\mathbb{R}^{n}\,:\, y(S)\leq f(S)\forall S\subsetneq V,y(V)=f(V)\}$
and the dual program directly corresponds to minimizing the }\textit{Lovász}\emph{
extension and thus $f$.}

\begin{center}
\begin{tabular}{|c|c|}
\hline 
Primal & Dual\tabularnewline
\hline 
\hline 
$\begin{aligned}\max & y^{-}(V)\\
 & y(S)\leq f(S)\forall S\subsetneq V\\
 & y(V)=f(V)
\end{aligned}
$ & $\begin{aligned}\min & \hat{f}(\vx)\\
 & 0\leq\vx\leq1
\end{aligned}
$\tabularnewline
\hline 
\end{tabular}
\par\end{center}
\begin{thm}
\label{thm:bfs}$\vh$ is a basic feasible solution (BFS) of the base
polyhedron $\mathcal{B}(f)$ if and only if 
\[
h_{i}=f(\{v_{1},...,v_{i}\})-f(\{v_{1},...,v_{i-1}\})
\]
for some permutation $v_{1},...,v_{n}$ of the ground set $V$. We
call $v_{1},...,v_{n}$ the defining permutation of $\vh$. We call
$v_{i}$ precedes $v_{j}$ for $i<j$.
\end{thm}
This theorem gives a nice characterization of the BFS's of $\mathcal{B}(f)$.
It also gives the key observation underpinning our approach: \textbf{the
coefficients of each separating hyperplane in Theorem~\ref{thm:sep}
precisely corresponds to a primal BFS} \textbf{(Theorem~\ref{thm:bfs})}.
Our analysis relies heavily on this connection. We re-state Theorem
\ref{thm:sep} in the language of BFS.
\begin{lem}
\label{lem:bfsxx}We have $\vh^{T}\vx\leq\hat{f}(\vx)$ for any $\vx\in[0,1]^{n}$
and BFS $\vh$.\end{lem}
\begin{proof}
Any BFS is given by some permutation. Thus this is just Theorem \ref{thm:sep}
in disguise.
\end{proof}
We also note that since the objective function of the primal program
is non-linear, we cannot say that the optimal solution to the primal
program is a BFS. Instead we only know that it is a convex combination
of the BFS's that satisfy the following property. A proof can be found
in any standard textbook on combinatorial optimization.
\begin{thm}
\label{thm:duality} The above primal and dual programs have no duality
gap. Moreover, there always exists a primal optimal solution $\vy=\sum_{k}\lambda^{(k)}\vh^{(k)}$
with $\sum_{k}\lambda^{(k)}=1$ (a convex combination of BFS $\vh^{(k)}$)
s.t. any $i$ with $y_{i}<0$ precedes any $j$ with $y_{j}>0$ in
the defining permutation for each BFS $\vh^{(k)}$.
\end{thm}
Our algorithms will maintain collections of BFS and use properties
of $\vh\in\basepolytope$, i.e. convex combination of BFS. To simplify
our analysis at several points we will want to assume that such a
vector $\vh\in\basepolytope$ is \emph{non-degenerate}, meaning it
has both positive and negative entries. Below, we prove that such
degenerate points in the base polytope immediately allow us to trivially
solve the SFM problem. 
\begin{lem}[Degenerate Hyperplanes]
\label{lem:degenerate_hyperplane} If $\vh\in\basepolytope$ is non-negative
then $\emptyset$ is a minimizer of $f$ and if $\vh$ is non-positive
then $V$ is a minimizer of $f$. \end{lem}
\begin{proof}
While this follows immediately from Theorem~\ref{thm:duality}, for
completeness we prove this directly. Let $S\in2^{V}$ be arbitrary.
If $\vh\in\basepolytope$ is non-negative then by the we have 
\[
f(S)\geq\vh(S)=\sum_{i\in S}h_{i}\geq0=f(\emptyset)\,.
\]
On the other hand if $\vh$ is non-positive then by definition we
have
\[
f(S)\geq\vh(S)=\sum_{i\in S}h_{i}\geq\sum_{i\in V}h_{i}=h(V)=f(V)\,.
\]

\end{proof}

\section{Improved Weakly Polynomial Algorithms for SFM\label{sec:Sub-Weakly-Polynomial}}

In this section we show how our cutting plane method can be used to
obtain a $O(n^{2}\log nM\cdot\text{EO}+n^{3}\log^{O(1)}nM)$ time
algorithm for SFM. Our main result in this section is the following
theorem, which shows how directly applying our results from earlier
parts to minimize the Lovász extension yields the desired running
time.
\begin{thm}
We have an $O(n^{2}\log nM\cdot\text{EO}+n^{3}\log^{O(1)}nM)$ time
algorithm for submodular function minimization.\end{thm}
\begin{proof}
We apply Theorem~\ref{thm:conv_opt} to the Lovász extension $\hat{f}:[0,1]^{n}\longrightarrow\mathbb{R}$
with the separation oracle given by Theorem \ref{thm:sep}. $\hat{f}$
fulfills the requirement on the domain as its domain $\Omega=[0,1]^{n}$
is symmetric about the point $(1/2,\ldots,1/2)$ and has exactly $2n$
constraints.

In the language of Theorem \ref{thm:conv_opt}, our separation oracle
is a $(0,0)$-separation oracle with $\eta=0$ and $\delta=0$.

We first show that $\delta=0$. Firstly, our separating hyperplane
can be written as
\[
\sum_{i=1}^{n}(f([i])-f([i-1]))x_{i}\leq\hat{f}(\bar{x})=\sum_{i=1}^{n}(f([i])-f([i-1]))\bar{x}_{i},
\]
where the equality follows from Theorem \ref{thm:The-Lovasz-extension}.
Secondly, for any $\vx$ with $\hat{f}(\vx)\leq\hat{f}(\bar{x})$
we have by Theorem \ref{thm:sep} that
\[
\sum_{i=1}^{n}(f([i])-f([i-1]))x_{i}\leq\hat{f}(\vx)\leq\hat{f}(\bar{x})
\]
which implies that $\vx$ is not cut away by the hyperplane.

Next we show that $\eta=0$. Our separating hyperplane induces a valid
halfspace whenever it is not nonzero, i.e. $f([i])\neq f([i-1])$
for some $i$. In the case that it is zero $f([i])=f([i-1])\forall i$,
by the same argument above, we have $\hat{f}(\bar{x})=\sum_{i=1}^{n}(f([i])-f([i-1]))\bar{x}_{i}=0$
and 
\[
\hat{f}(\vx)\geq\sum_{i=1}^{n}(f([i])-f([i-1]))x_{i}=0=\hat{f}(\bar{x}).
\]

In other words, $\bar{x}$ is an exact minimizer, i.e. $\eta=0$.

Note that $\left|\hat{f}(\vx)\right|=\left|\mathbb{E}_{t\sim[0,1]}[f(\{i:x_{i}\geq t\})]\right|\leq M$
as $M=\max_{S}|f(S)|$. Now plugging in $\alpha=\frac{1}{4M}$ in
the guarantee of Theorem \ref{thm:main_result}, we can find a point
$x^{*}$ such that
\begin{eqnarray*}
\hat{f}(x^{*})-\min_{\vx\in[0,1]^{n}}\hat{f}(\vx) & \leq & \frac{1}{4M}\left(\max_{\vx\in[0,1]^{n}}\hat{f}(\vx)-\min_{\vx\in[0,1]^{n}}\hat{f}(\vx)\right)\\
 & \leq & \frac{1}{4M}(2M)\\
 & < & 1
\end{eqnarray*}

We claim that $\min_{t\in[0,1]}f(\{i:x_{i}^{*}\geq t\})$ is minimum.
To see this, recall from \ref{thm:The-Lovasz-extension} that $\hat{f}$
has an integer minimizer and hence $\min_{\vx\in[0,1]^{n}}\hat{f}(\vx)=\min_{S}f(S)$.
Moreover, $\hat{f}(x^{*})$ is a convex combination of $f(\{i:x_{i}^{*}\geq t\})$
which gives
\[
1>\hat{f}(x^{*})-\min_{\vx\in[0,1]^{n}}\hat{f}(\vx)=\hat{f}(x^{*})-\min_{S}f(S)\geq\min_{t\in[0,1]}f(\{i:x_{i}^{*}\geq t\})-\min_{S}f(S).
\]

Since $f$ is integer-valued, we must then have $\min_{t\in[0,1]}f(\{i:x_{i}^{*}\geq t\})=\min_{S}f(S)$
as desired. Since our separation oracle can be computed by $n$ oracle
calls and runs in time $O(n\cdot\text{EO}+n^{2})$, by Theorem \ref{thm:conv_opt}
the overall running time is then $O(n^{2}\log nM\cdot\text{EO}+n^{3}\log^{O(1)}nM)$
as claimed.
\end{proof}
Needless to say the proof above completely depends on Theorem \ref{thm:conv_opt}.
We remark that one can use the Vaidya's cutting plane instead of ours
to get a time complexity $O(n^{2}\log nM\cdot\text{EO}+n^{\omega+1}\log^{O(1)}n\cdot\log M)$.
There is actually an alternate argument that gives a time complexity
of $O(n^{2}\log M\cdot\text{EO}+n^{O(1)}\cdot\log M)$. Thus it requires
slightly fewer oracle calls at the expense of slower running time.
A proof is offered in this section, which can be skipped without any
risk of discontinuation. This proof relies the following cutting plane
method.
\begin{thm}[\cite{bertsimas2004solving} ]
 \label{thm:santosh_CG}Given any convex set $K\subset[0,1]^{n}$
with a separation oracle of cost $SO$, in time $O(kSO+kn^{O(1)})$
one can find either find a point $\vx\in K$ or find a polytope $P$
such that $K\subset P$ and the volume of $K$ is at most $\left(\frac{2}{3}\right)^{k}$.
\end{thm}
The Theorem allows us to decrease the volume of the feasible region
by a factor of $\left(\frac{2}{3}\right)^{k}$ after $k$ iterations.
Similar to above, we apply cutting plane to minimize $\hat{f}$ over
the hypercube $[0,1]^{n}$ for $O(n\log M)$ iterations, and outputs
\textit{any} integral point in the remaining feasible region $P$.
\begin{lem}
Let $x^{*}$ achieve the minimum function value $\hat{f}(x^{*})$
among the points used to query the separation oracle. Then
\begin{enumerate}
\item $x^{*}\in P^{(k)}$, the current feasible region.
\item Any $\vx$ with $\hat{f}(\vx)\leq\hat{f}(x^{*})$ belongs to $P^{(k)}$.
\item suppose $x_{i_{1}}^{*}\geq\cdots\geq x_{i_{n}}^{*}$ and let $S_{j}=\{i_{1},\ldots,i_{j}\}$.
Then $S_{l}\in\arg\min_{S_{j}}f(S_{j})$ also belongs to $P^{(k)}$.
\end{enumerate}
\end{lem}
\begin{proof}
For any separating hyperplane $\vh^{T}x\leq\hat{f}(\bar{x})$ given
by $\bar{x}$, we have by Lemma \ref{lem:bfsxx} that $\vh^{T}x^{*}\leq\hat{f}(x^{*})$.
Since $\hat{f}(x^{*})$ is the minimum among all $\hat{f}(\bar{x})$,
$\vh^{T}x^{*}\leq\hat{f}(\bar{x})$ and hence $x^{*}$ is not removed
by any new separating hyperplane. In other words, $x^{*}\in P^{(k)}$
. The argument for (2) is analogous.

For (3), recall that by the definition of Lovász extension $\hat{f}(x^{*})$
is a convex combination of $f(S_{j})$ and thus the indicator variable
$I_{S_{l}}$ for $S_{l}$ satisfies $f(I_{S_{l}})\leq\hat{f}(x^{*})$.
By Lemma \ref{lem:bfsxx} again, this implies $\vh^{T}I_{S_{l}}\leq f(I_{S_{l}})\leq\hat{f}(x^{*})\leq\hat{f}(\bar{x})$
for any separating hyperplane $\vh^{T}x\leq\hat{f}(\bar{x})$.\end{proof}
\begin{thm}
Suppose that we run Cutting Plane in Theorem \ref{thm:santosh_CG}
for $O(n\log M)$ iterations. Then $S_{l}$ from the last lemma also
minimizes $f$.\end{thm}
\begin{proof}
We use the notations from the last lemma. After $k=Kn\log_{2/3}M$
iterations, the volume of the feasible region $P^{(k)}$ is at most
$1/M{}^{Kn}$. By the last lemma, $I_{S_{l}}\in P^{(k)}$\@.

Suppose for the sake of contradiction that $S$ minimizes $f$ but
$f(S)<f(S_{l})$. Since $f$ is integer-valued, $f(S)+1\leq f(S_{l})$.
Let $r\defeq1/6M$. Consider the set $B\defeq\{\vx:0\leq x_{i}\leq r\:\forall\mbox{i}\notin S,\,1-r\leq x_{i}\leq1\:\forall i\in S\}$.
We claim that for $\vx\in B$,
\[
\hat{f}(\vx)\leq f(S)+1.
\]

To show this, note that $f(\{i:x_{i}\geq t\})=f(S)$ for $r<t\leq1-r$
as $x_{i}\leq r$ for $i\notin S$ and $x_{i}\geq1-r$ for $i\in S$.
Now using conditional probability and $|f(T)|\le M$ for any $T$,
\begin{eqnarray*}
\hat{f}(\vx) & = & \mathbb{E}_{t\sim[0,1]}[f(\{i:x_{i}\geq t\})]\\
 & = & \left(1-2r\right)\mathbb{E}[f(\{i:x_{i}\geq t\})|r<t\leq1-r]+\\
 &  & r\left(\mathbb{E}[f(\{i:x_{i}\geq t\})|0\leq t\leq r]+\mathbb{E}[f(\{i:x_{i}\geq t\})|1-r\leq t\leq1]]\right)\\
 & = & \left(1-r\right)f(S)+r\left(\mathbb{E}[f(\{i:x_{i}\geq t\})|0\leq t\leq r+\mathbb{E}[f(\{i:x_{i}\geq t\})|1-r\leq t\leq1]]\right)\\
 & \leq & \left(1-2r\right)f(S)+2rM\\
 & \leq & f(S)+4rM\\
 & \leq & f(S)+1
\end{eqnarray*}
But now $B\subseteq P^{(k)}$ as $\hat{f}(\vx)\leq f(S)+1\leq f(S_{l})$
and by (2) of the last lemma. This would lead to a contradiction since
\[
\text{vol}(B)=\frac{1}{(6M)^{n}}>\frac{1}{M{}^{Kn}}\geq\text{vol}(P^{(k)})
\]
for sufficiently large $K$.\end{proof}
\begin{cor}
There is an $O(n^{2}\log M\cdot\text{EO}+n^{O(1)}\log M)$ time algorithm
for submodular function minimization.\end{cor}
\begin{proof}
This simply follows from the last lemma, Theorem \ref{thm:santosh_CG},
and the fact that our separation oracle runs in time $O(n\cdot\text{EO}+n^{2})$.
\end{proof}
Curiously, we obtained $O(\log M)$ rather than $O(\log nM)$ as in
our algorithm. We leave it as an open problem whether one can slightly
improve our running time to $O(n^{2}\log M\cdot\text{EO}+n^{3}\log^{O(1)}n\cdot\log M)$.
The rest of this paper is devoted to obtaining better strongly polynomial
running time.

\section{Improved Strongly Polynomial Algorithms for SFM\label{sec:Sub-Strongly-Polynomial}}

In this section we show how our cutting plane method can be used to
obtain a $\widetilde{O}(n^{3}\cdot\text{EO}+n^{4})$ time algorithm
for SFM, which improves over the currently fastest $O(n^{5}\cdot\text{EO}+n^{6})$
time algorithm by Orlin.

\subsection{Improved Oracle Complexity}

\label{sub:submod:strongly:oracle-complexity}

We first present a simple geometric argument that $f$ can be minimized
with just $O(n^{3}\log n\cdot\text{EO})$ oracle calls. While this
is our desired query complexity (and it improves upon the previous
best known bounds by a factor of $O(n^{2})$ unfortunately the algorithm
runs in exponential time. Nevertheless, it does provide some insight
into how our more efficient algorithms should proceed and it alone,
does suggests that information theoretically, $O(n^{3}\log n\cdot\text{EO})$
calls suffice to solve SFM. In the rest of the paper, we combine this
insight with some of the existing SFM tools developed over the last
decade to get improved polynomial time algorithms.
\begin{thm}
\label{thm:n3logn} Submodular functions can be minimized with $O(n^{3}\log n\cdot\text{EO})$
oracle calls.\end{thm}
\begin{proof}
We use the cutting plane method in Theorem \ref{thm:santosh_CG} with
the separation oracle given by Theorem~\ref{thm:sep}. This method
reduce the volume of the feasible region by a factor of $(\frac{2}{3})^{k}$
after $k$ iterations if the optimal has not found yet.

Now, we argue that after $O(n\log n)$ iterations of this procedure
we have either found a minimizer of $f$ or we have enough information
to reduce the dimension of the problem by 1. To see this, first note
that if the separation oracle ever returns a degenerate hyperplane,
then by Lemma~\ref{lem:degenerate_hyperplane} then either $\emptyset$
or $V$ is the minimizer, which we can determine in time $O(\EO+n).$
Otherwise, after $100n\log n$ iterations, our feasible region $P$
must have a volume of at most $1/n^{10n}$ . In this case, we claim
that the remaining integer points in $P$ all lie on a hyperplane.
This holds, as if this was not the case, then there is a simplex $\triangle$,
with integral vertices $v_{0},v_{1},\ldots,v_{n}$, contained in $P$.
But then
\[
\text{vol}(P)\geq\text{vol}(\triangle)=\frac{1}{n!}\left|\det\left(v_{1}-v_{0}\: v_{2}-v_{0}\:\ldots\: v_{n}-v_{0}\right)\right|\geq\frac{1}{n!}
\]
where the last inequality holds since the determinant of an integral
matrix is integral, yielding a contradiction.

In other words after $O(n\log n)$ iterations, we have reduced the
dimension of all viable solutions by at least 1. Thus, we can recurse
by applying the cutting plane method to the lower dimensional feasible
region, i.e. $P$ is (replaced by) the convex combination of all the
remaining integer points. There is a minor technical issue we need
to address as our space is now lower dimensional and the starting
region is not necessarily the hypercube anymore and the starting volume
is not necessarily equal to 1.

We argue that the starting volume is bounded by $n^{O(n)}$. If this
is indeed the case, then our previous argument still works as the
volume goes down by a factor of $1/n^{O(n)}$ in $O(n\log n)$ iterations.

Let $v\in P$ be an integer point. Now the $\dim(P)$-dimensional
ball of radius $\sqrt{n}$ centered at $v$ must contain all the other
integer points in $P$ as any two points of $\{0,1\}^{n}$ are at
most $\sqrt{n}$ apart. Thus the volume of $P$ is bounded by the
volume of the ball which is $n^{O(n)}$. Now to get the volume down
to $1/n^{10n}$, the number of iterations is still $O(n\log n)$.

In summary, we have reduced our dimension by 1 using $O(n\log n)$
iterations which requires $O(n^{2}\log n\cdot\text{EO})$ oracle calls
(as each separating hyperplane is computed with $n\cdot\text{EO}$
oracle calls). This can happen at most $n$ times. The overall query
complexity is then $O(n^{3}\log n\cdot\text{EO})$.

Note that the minimizer $\vx$ obtained may not be integral. This
is not a problem as the definition of Lovász extension implies that
if $\hat{f}(\vx)$ is minimal, then $f(\{i:x_{i}\geq t\})$ is minimal
for any $t\in[0,1]$.

We remark that this algorithm does not have a polynomial runtime.
Even though all the integral vertices of $P$ lie on a hyperplane,
the best way we know of that identifies it takes exponential time
by checking for all the integer points $\{0,1\}^{n}$.\end{proof}
\begin{rem}
Note that this algorithm works for minimizing any convex function
over the hypercube that obtains its optimal value at a vertex of the
hypercube. Formally, our proof of Theorem~\ref{thm:n3logn} holds
whenever a function $f:2^{V}\longrightarrow\mathbb{R}^{n}$ admits
a convex relaxation $\hat{f}$ with the following properties:
\begin{enumerate}
\item For every $S\subseteq V$, $\hat{f}(I_{S})=f(S)$.
\item Every $\hat{f}(\vx)$ can be written as a convex combination $\sum_{S\in\mathcal{S}}\alpha_{S}f(S)$,
where $\sum\alpha_{S}=1$, $|\mathcal{S}|=O(n)$, and $\mathcal{S}$
can be computed without any oracle call.
\item A subgradient $\partial\hat{f}(\vx)$ of $\hat{f}$ at any point $\vx\in[0,1]^{n}$
can be computed with $O(n\cdot\text{EO})$ oracle calls.
\end{enumerate}
In this case, the proof of Theorem~\ref{thm:n3logn}, implies that
$\hat{f}$ and $f$ can be minimized with $O(n^{3}\log n\cdot\text{EO})$
oracle calls by using the separating hyperplane $\partial\hat{f}(\bar{x})^{T}(\vx-\bar{x})\leq0$.
\end{rem}

\subsection{Technical Tools}

To improve upon the running time of the algorithm in the previous
section, we use more structure of our submodular function $f$. Rather
than merely showing that we can decrease the dimension of our SFM
problem by 1 we show how we can reduce the degrees of freedom of our
problem in a more principled way. In Section~\ref{sub:ring_family}
we formally define the abstraction we use for this and discuss how
to change our separation oracle to accommodate this abstraction, and
in Section~\ref{sub:valid_arcs} we show how we can deduce these
constraints. These tools serve as the foundation for the faster strongly
polynomial time SFM algorithms we present in Section~\ref{sub:submodular-n4}
and Section~\ref{sub:submod:strongly_n3}.

\subsubsection{SFM over Ring Family}

\label{sub:ring_family}

For the remainder of the paper we consider a more general problem
than SFM in which we wish to compute a minimizer of our submodular
function $f$ over a ring family of the ground set $V=[n]$. A ring
family $\mathcal{F}$ is a collection of subsets of $V$ such that
for any $S_{1},S_{2}\in\mathcal{F}$, we have $S_{1}\cup S_{2},S_{1}\cap S_{2}\in\mathcal{F}$.
Thus SFM corresponds to the special case where $\mathcal{F}$ consists
of every subset of $V$. This generalization has been considered before
in the literature and was essential to the IFF algorithm.

It is well known that any ring family $\mathcal{F}$ over $V$ can
be represented by a directed graph $D=(V,A)$ where $S\in\mathcal{F}$
iff $S$ contains all of the descendants of any $i\in S$. An equivalent
definition is that for any arc $(i,j)\in A$, $i\in S$ implies $j\in S$.
It is customary to assume that $A$ is acyclic as any (directed) cycle
of $A$ can be contracted (see section \ref{sub:Consolidating}).

We denote by $R(i)$ the set of descendants of $i$ (including $i$
itself) and $Q(i)$ the set of ancestors of $i$ (including $i$ itself).
Polyhedrally, an arc $(i,j)\in A$ can be encoded as the constraint
$x_{i}\leq x_{j}$ as shown by the next lemma.
\begin{lem}
Let $\mathcal{F}$ be a ring family over $V$ and $D=(V,A)$ be its
directed acyclic graph representation. Suppose $f:V\longrightarrow\mathbb{R}$
is submodular with Lovász extension $\hat{f}$. Then the characteristic
vector $I_{S}$ of any minimizer $S=\arg\min_{S\in\mathcal{F}}f(S)$
over $\mathcal{F}$ is also the solution to
\begin{equation}
\begin{aligned}\min & \hat{f}(\vx)\\
 & x_{i}\leq x_{j}\,\forall(i,j)\in A\\
 & 0\leq\vx\leq1
\end{aligned}
\label{eq:ring}
\end{equation}
\end{lem}
\begin{proof}
Let $x^{*}$ be a minimizer, and $L^{(t)}=\{i:x_{i}^{*}\geq t\}$.
It is easy to check that the indicator variable $I_{L^{(t)}}$ satisfies
\eqref{eq:ring} since $x^{*}$ does. Moreover, recall that $\hat{f}(x^{*})=\mathbb{E}_{t\sim[0,1]}[f(L_{t})]$.
Thus $\hat{f}(x^{*})$ can be written as a convex combination $\hat{f}(x^{*})=\sum_{t}\alpha_{t}f(L^{(t)})=\sum_{t}\alpha_{t}\hat{f}(I_{L^{(t)}})$,
where $\alpha_{t}>0$ and $\sum_{t}\alpha_{t}=1$. Thus all such $\hat{f}(I_{L^{(t)}})$
are minimal, i.e. \eqref{eq:ring} has no ``integrality gap''.
\end{proof}
We also modify our separation oracle to accommodate for this generalization
as follows. Before doing so we need a definition which relates our
BFS to the ring family formalism.
\begin{defn}
A permutation $(v_{1},\ldots,v_{n})$ of $V$ is said to be \textit{consistent}
with an arc $(i,j)$ if $j$ precedes $i$ in $(v_{1},\ldots,v_{n})$.
Similarly, a BFS of the base polyhedron is consistent with $(i,j)$
if $j$ precedes $i$ in its defining permutation. $(v_{1},\ldots,v_{n})$
(or a BFS) is consistent with $A$ if it is consistent with every
$(i,j)\in A$.
\end{defn}
Readers may find it helpful to keep in mind the following picture
which depicts the relative positions between $R(i),i,Q(i)$ in the
defining permutation of $\vh$ that is consistent with $A$:

\begin{center}
\textit{\Large{}$\cdots\cdots\; R(i)\backslash\{i\}\;\cdots\cdots\; i\;\cdots\cdots\; Q(i)\backslash\{i\}\;\cdots\cdots$}
\par\end{center}{\Large \par}

In Theorem~\ref{thm:sep}, given $\bar{x}\in[0,1]^{n}$ our separating
hyperplane is constructed by sorting the entries of $\bar{x}$. This
hyperplane is associated with some BFS $\vh$ of the base polyhedron.
As we shall see towards the end of the section, we would like $\vh$
to be consistent with every arc $(i,j)\in A$.

This task is easy initially as $\bar{x}$ satisfies $x_{i}\leq x_{j}$
for $(i,j)\in A$ for the starting polytope of \eqref{eq:ring}. If
$x_{i}<x_{j}$, nothing special has to be done as $j$ must precede
$i$ in the ordering. On the other hand, whenever $x_{i}=x_{j}$,
we can always break ties by ranking $j$ ahead of $i$.

However, a technical issue arises due to the fact that our cutting
plane algorithm may drop constraints from the current feasible region
$P$. In other words, $\bar{x}$ may violate $x_{i}\geq0$, $x_{j}\le1$
or $x_{i}\leq x_{j}$ if it is ever dropped. Fortunately this can
be fixed by reintroducing the constraint. We summarize the modification
needed in the pseudocode below and formally show that it fulfills
our requirement.

\begin{algorithm2e}
\caption{Modified Separation Oracle}

\SetAlgoLined

\textbf{Input:} $\bar{x}\in\mathbb{R}^{n}$ and the set of arcs $A$

\uIf{ $\bar{x}_{i}<0$ for some $i$ }{

\textbf{Output}: $x_{i}\ge0$

}\uElseIf{ $\bar{x}_{j}>1$ for some $j$ }{

\textbf{Output}: $x_{j}\le1$

}\uElseIf{ $\bar{x}_{i}>\bar{x}_{j}$ for some $(i,j)\in A$ }{

\textbf{Output}: $x_{i}\leq x_{j}$

}\uElse{  

Let $i_{1},\ldots,i_{n}$ be a permutation of $V$ such that $\bar{x}_{i_{1}}\ge\ldots\ge\bar{x}_{i_{n}}$and
for all $(i,j)\in A$, $j$ precedes $i$ in $i_{1},\ldots,i_{n}$.

\textbf{Output}: $\vh^{T}\vx\leq\hat{f}(\bar{x})$, where $\vh$ is
the BFS defined by the permutation $i_{1},\ldots,i_{n}$.

}  

\end{algorithm2e}
\begin{lem}
\label{lem:Our-modified-separation} Our modified separation oracle
returns either some BFS $\vh=0$ or a valid separating hyperplane,
i.e.
\begin{enumerate}
\item $\bar{x}$ either lies on the separating hyperplane or is cut away
by it.
\item Any minimizer of \eqref{eq:ring} is not cut away by the separating
hyperplane.
\end{enumerate}
Such a hyperplane can be computed with $n$ oracle calls to $f$ and
in time $O(n\cdot\text{EO}+n^{2})$.\end{lem}
\begin{proof}
If we get $x_{i}\geq0$, $x_{j}\le1$ or $x_{i}\leq x_{j}$ (if loop
or the first two else loops), then clearly $\bar{x}$ is cut away
by it and any minimizer must of course satisfy $x_{i}\geq0$, $x_{j}\le1$
and $x_{i}\leq x_{j}$ as they are the constraints in \eqref{eq:ring}.
This proves (1) and (2) for the case of getting $x_{i}\geq0$, $x_{j}\le1$
or $x_{i}\leq x_{j}$.

Thus it remains to consider the case $\vh^{T}\vx\leq\hat{f}(\bar{x})$
(last else loop). First of all, $\bar{x}$ lies on it as $\hat{f}(\bar{x})=\vh^{T}\bar{x}$.
This proves (1). For (2), we have from Lemma \ref{lem:bfsxx} that
$\vh^{T}\vx\leq\hat{f}(\vx)$. If $x^{*}$ is a minimizer of \eqref{eq:ring},
we must then have $\vh^{T}x^{*}\leq\hat{f}(x^{*})\le\hat{f}(\bar{x})$
as $\bar{x}$ is also feasible for \eqref{eq:ring}. 

Finally we note that the running time is self-evident.
\end{proof}
We stress again that the main purpose of modifying our separation
oracle is to ensure that any BFS $\vh$ used to define a new separating
hyperplane must be consistent with every $(i,j)\in A$.

\subsubsection{Identifying New Valid Arcs}

\label{sub:valid_arcs}

The reason for considering the ring family generalization of SFM is
that our algorithms (and some previous algorithms too) work by adding
new arcs to our digraph $D$. This operation yields a strongly polynomial
algorithm since there are only $2\cdot{n \choose 2}$ possible arcs
to add. Of course, a new arc $(i,j)$ is valid only if $i\in S_{\min}\implies j\in S_{\min}$
for some minimizer $S_{\min}$. Here we show how to identify such
valid arcs by extracting information from certain nice elements of
the base polyhedron.

This is guaranteed by the next four lemmas, which are stated in a
way different from previous works e.g. our version is extended to
the ring family setting. This is necessary as our algorithms require
a more general formulation. We also give a new polyhedral proof, which
is mildly simpler than the previous combinatorial proof. On the other
hand, Lemma \ref{lem:newarc-new-1} is new and unique to our work.
It is an important ingredient of our $\widetilde{O}(n^{3}\cdot\text{EO}+n^{4})$
time algorithm.

Recall that each BFS of the base polyhedron is defined by some permutation
of the ground set elements.

First, we prove the following two lemmas which show that should we
ever encounter a non-degenerate point in the base polytope with a
coordinate of very large value, then we can immediately conclude that
that coordinate must be or must not be in solution to SFM over the
ring family. 
\begin{lem}
\label{lem:mustnotopt} If $\vy\in\basepolytope$ is non-degenerate
and satisfies $y{}_{i}>-(n-1)\min_{j}y{}_{j}$, then $i$ is not in
any minimizer of $f$ (over the ring family $A$).\end{lem}
\begin{proof}
We proceed by contradiction and suppose that $S$ is a minimizer of
$f$ that contains $i$. Now since $\vy$ is non-degenerate we know
that $\min_{j}y_{j}\leq0$ and by the definition of $\vy$ we have
the following contradiction
\[
0<y_{i}+(n-1)\min_{j}y_{j}\leq\sum_{j\in S}y_{j}=\vy(S)\leq f(S)\leq f(\emptyset)=0\,.
\]
\end{proof}
\begin{lem}
\label{lem:mustbeopt} If $\vy\in\basepolytope$ is non-degenerate
and satisfies $y{}_{i}<-(n-1)\max_{j}y{}_{j}$, then $i$ is in every
minimizer of $f$ (over the ring family $A$).\end{lem}
\begin{proof}
We proceed by contradiction and suppose that $S$ is a minimizer of
$f$ that does not contain $i$. Now since $\vy$ is non-degenerate
we know that $\max_{j}y_{j}\geq0$ and therefore
\[
\sum_{j\in[n]}y_{j}=y_{i}+\sum_{j\in S}y_{j}+\sum_{j\in V-(S+i)}y_{j}<-(n-1)\max_{j}y_{j}+\sum_{j\in S}y_{j}+(|V|-|S|-1)\max_{j}y_{j}\leq\sum_{j\in S}y_{j}\,.
\]
However by the definition of $\vy$ we have
\[
\sum_{j\in S}y_{j}=\vy(S)\leq f(S)\leq f(V)=\sum_{j\in[n]}y_{j}\,.
\]
Thus we have a contradiction and the result follows.
\end{proof}
Now we are ready to present conditions under which a new valid arc
can be added. We begin with a simple observation. Let $\upi\defeq f(R(i))-f(R(i)-i)$
and $\lowi\defeq f(V\backslash Q(i)+i)-f(V\backslash Q(i))$. As the
names suggest, they bound the value of $h_{i}$ for any BFS used.
\begin{lem}
\label{lem:boundonsmall-1}For any BFS $\vh$ used to construct a
separating hyperplane given by our modified separation oracle, we
have $\lowi\leq h_{i}\leq\upi$.\end{lem}
\begin{proof}
Note that by Lemma \ref{lem:Our-modified-separation}, $\vh$ is consistent
with every $(j_{1},j_{2})\in A$ and hence $i$ must precede $Q(i)$
and be preceded by $R(i)$. Let $S$ be the set of elements preceding
$i$ in the defining permutation of $\vh$. Then $h_{i}=f(S+i)-f(S)\leq f(R(i))-f(R(i)-i)$
because of diminishing return and $R(i)-i\subseteq S$. The lower
bound follows from the same argument as $Q(i)-i$ comes after $i$,
and so $Q(i)\subseteq V\backslash S$.
\end{proof}
In the following two lemmas, we show that if $\upi$ is ever sufficiently
positive or $\lowi$ is sufficiently negative, then we find a new
arc. 

While these lemmas may appear somewhat technical but actually has
an intuitive interpretation. Suppose an element $p$ is in a minimizer
$S_{\min}$ of $f$ over the ring family $D$. Then $R(p)$ must also
be part of $S_{\min}$. Now if $f(R(p))$ is very large relative to
$f(R(p)-p)$, there should be some element $q\in S_{\min}\backslash R(p)$
compensating for the discrepancy. The lemma says that such an element
$q$ can in fact be found efficiently.
\begin{lem}[new arc]
\label{lem:newarc}Let $\vy=\sum_{k}\lambda^{(k)}\vy^{(k)}$ be a
non-degenerate convex combination of $O(n)$ base polyhedron BFS's
$\vy^{(k)}$ which are consistent with every arc $(i,j)\in A$. If
some element $p$ satisfies $\upp>n^{4}\max y_{j}$, then we can find,
using $O(n\cdot\text{EO})$ oracle calls and $O(n^{2})$ time, some
$q\notin R(p)$ such that the arc $(p,q)$ is valid, i.e. if $p$
is in a minimizer, then so is $q$.\end{lem}
\begin{proof}
If $\max y_{j}<0$ then we are immediately done by Lemma \ref{lem:degenerate_hyperplane}.
We assume $\max y_{j}\geq0$ in the proof. For all $k$ let $\vy'^{(k)}$
be the BFS obtained by taking the defining permutation of $\vy^{(k)}$
and moving $R(p)$ to the front while preserving the relative ordering
of $R(p)$ within each permutation). Furthermore, let $\vy'\defeq\sum_{k}\lambda^{(k)}\vy'^{(k)}$.
Then since $y'{}_{p}^{(k)}=f(R(p))-f(R(p)-p)=\upp$ we have $\upp=y'_{p}=f(R(p))-f(R(p)-p)$.
Moreover, 
\begin{equation}
y'_{j}\geq y_{j}\quad\forall j\in R(p)\text{ and }y'_{j}\leq y_{j}\quad\forall j\notin R(p)\label{eq:haha}
\end{equation}
by diminishing marginal return.

Now, suppose $p$ is in a minimizer $S_{\min}$. Then $R(p)\subseteq S_{\min}$
by definition. We then define $f'(S)=f(S\cup R(p))$ for $S\subseteq V\backslash R(p)$.
It can be checked readily that $f'$ is submodular and $S_{\min}\backslash R(p)$
is a minimizer of $f'$ (over the corresponding ring family). Note
that now $\vy'_{V\backslash R(p)}$ (the restriction of $\vy'$ to
$V\backslash R(p)$) is a convex combination of the BFS's of the base
polyhedron $\mathcal{B}(f')$ of $f'$. We shall show that $\vy'_{V\backslash R(p)}$
has the desired property in Lemma~\ref{lem:mustbeopt}.

Note that $y'(V\backslash R(p)+p)\leq y(V\backslash R(p)+p)$ since
\[
y'(V\backslash R(p)+p)=y'(V)-y'(R(p)-p)=y(V)-y'(R(p)-p)\leq y(V)-y(R(p)-p)=y(V\backslash R(p)+p).
\]
But now since $\vy$ is non-degenerate $\max_{j}y_{j}\geq0$ and therefore
\begin{eqnarray}
y'(V\backslash R(p)) & \leq & y(V\backslash R(p)+p)-y_{p}'\nonumber \\
 & = & y(V\backslash R(p)+p)-\left(f(R(p))-f(R(p)-p)\right)\label{eq:blah}\\
 & \leq & n\max y_{j}-\left(f(R(p))-f(R(p)-p)\right)\nonumber \\
 & < & (n-n^{4})\max y_{j}\nonumber 
\end{eqnarray}

Therefore by the Pigeonhole Principle some $q\notin R(p)$ must satisfy
\begin{eqnarray*}
y'_{q} & < & \left((n-n^{4})\max y_{j}\right)/(n-1)\\
 & = & -(n^{3}+n^{2}+n)\max y_{j}\\
 & \leq & -(n^{3}+n^{2}+n)\max_{j\notin R(p)}y_{j}\\
 & \leq & -(n^{3}+n^{2}+n)\max_{j\notin R(p)}y_{j}'\quad\text{by }\eqref{eq:haha}
\end{eqnarray*}
 By Lemma \ref{lem:mustbeopt}, this $q$ must be in any minimizer
of $f'$. In other words, whenever $p$ is in a minimizer of $f$,
then so is $q$.

Note however that computing all $\vy'$ would take $O(n^{2})$ oracle
calls in the worst case as there are $O(n)$ $\vy'^{(k)}$'s. We use
the following trick to identify some $q$ with $y'_{q}<-(n-1)\max y_{j}$
using just $O(n)$ calls. The idea is that we actually only want to
have sufficient decreases in $y'(V\backslash R(p))$ which can be
accomplished by having a large corresponding decrease in some $\vy'^{(k)}$.

For each $k$, by the same argument above (see \eqref{eq:blah}) 
\begin{equation}
y'^{(k)}(V\backslash R(p))-y{}^{(k)}(V\backslash R(p))\leq y{}_{p}^{(k)}-\left(f(R(p))-f(R(p)-p)\right)\label{eq:hahaha}
\end{equation}
The ``weighted decrease'' $\lambda^{(k)}\left(y{}_{p}^{(k)}-\left(f(R(p))-f(R(p)-p)\right)\right)$
for $\vy'^{(k)}$ sum up to
\[
\sum\lambda^{(k)}\left(y{}_{p}^{(k)}-\left(f(R(p))-f(R(p)-p)\right)\right)=y_{p}-\left(f(R(p))-f(R(p)-p)\right)<(1-n^{4})\max y_{j}
\]

Thus by the Pigeonhole Principle, some $l$ will have 
\[
\lambda^{(l)}\left(y{}_{p}^{(l)}-\left(f(R(p))-f(R(p)-p)\right)\right)<\left((1-n^{4})\max y_{j}\right)/O(n)<-n^{2}\max y_{j}.
\]
For this $\vy^{(l)}$ we compute $\vy'^{(l)}$. We show that $\vy''=\lambda^{(l)}\vy'^{(l)}+\sum_{k\neq l}\lambda^{(k)}\vy^{(k)}$
has the same property as $\vy'$ above.
\begin{eqnarray*}
y''(V\backslash R(p)) & = & \lambda^{(l)}y'^{(l)}(V\backslash R(p))+\sum_{k\neq l}\lambda^{(k)}y^{(k)}(V\backslash R(p))\\
 & = & y(V\backslash R(p))+\lambda^{(l)}\left(y'^{(l)}(V\backslash R(p))-y{}^{(l)}(V\backslash R(p))\right)\\
 & \leq & y(V\backslash R(p))+\lambda^{(l)}\left(y{}_{p}^{(l)}-\left(f(R(p))-f(R(p)-p)\right)\right)\quad\text{by }\eqref{eq:hahaha}\\
 & < & (n-1)\max y_{j}-n^{2}\max y_{j}\\
 & < & (n-n^{2})\max y_{j}
\end{eqnarray*}

Then some $q\in V\backslash R(p)$ must satisfy 
\[
y''_{q}<\frac{n-n^{2}}{n-1}\max y_{j}=-n\max y_{j}
\]

That is, the arc $(p,q)$ is valid. This takes $O(n)$ oracle calls
as given $\vy=\sum_{k}\lambda^{(k)}\vy^{(k)}$ , computing $\vy''$
requires knowing only $f(R(p))$, $f(R(p)-p)$, and $\vy'^{(l)}$
which can be computed from $\vy{}^{(l)}$ with $n$ oracle calls.
The runtime is $O(n^{2})$ which is needed for computing $\vy''$.\end{proof}
\begin{lem}
\label{lem:newarc-new-1}Let $\vy=\sum_{k}\lambda^{(k)}\vy^{(k)}$be
a non-degenerate convex combination of base polyhedron BFS $\vy^{(k)}$
which is consistent with every arc $(i,j)\in A$. If $\lowp<n^{4}\min y_{j}$,
then we can find, using $O(n\cdot\text{EO})$ oracle calls and $O(n^{2})$
time, some $q\notin Q(p)$ such that the arc $(q,p)$ is valid, i.e.
if $p$ is not in a minimizer, then $q$ is not either.\end{lem}
\begin{proof}
It is possible to follow the same recipe in the proof of Lemma \ref{lem:newarc}
but using Lemma~\ref{lem:mustnotopt} instead of Lemma~\ref{lem:mustbeopt}.
Here we offer a proof which directly invokes Lemma \ref{lem:mustbeopt}
on a different submodular function.

Let $g$ be defined by $g(S)\defeq f(V\backslash S)$ for any $S$,
and $A_{g}$ be the set of arcs obtained by reversing the directions
of the arcs of $A$. Consider the problem of minimizing $g$ over
the ring family $A_{g}$. Using subscripts to avoid confusion with
$f$ and $g$, e.g. $R_{g}(i)$ is the set of descendants of $i$
w.r.t. $A_{g}$, it is not hard to verify the following:
\begin{itemize}
\item $g$ is submodular
\item $R_{g}(i)=Q_{f}(i)$
\item $g(R_{g}(p))-g(R_{g}(p)-p)=-\left(f(V\backslash Q_{f}(p)+p)-f(V\backslash Q_{f}(p))\right)$
\item $-\vy^{(k)}$ is a BFS of $\mathcal{B}(g)$ if and only if $\vy^{(k)}$
is a BFS of $\mathcal{B}(f)$
\item $\max(-y_{j})=-\min y_{j}$
\end{itemize}
By using the above correspondence and applying Lemma~\ref{lem:newarc}
to $g$ and $A_{g}$, we can find, using $O(n)$ oracle calls and
$O(n^{2})$ time, some $q\notin R_{g}(p)=Q(p)$ such that the arc
$(p,q)$ is valid for $g$ and $A_{g}$. In other words, the reverse
$(q,p)$ will be valid for $f$ and $A$.
\end{proof}
These lemmas lay the foundation of our algorithm. They suggests that
if the positive entries of a point in the base polyhedron are small
relative to some $\upp=f(R(p))-f(R(p)-p)$, a new arc $(p,q)$ can
be added to $A$. This can be seen as a robust version of Lemma~\ref{lem:degenerate_hyperplane}.

Finally, we end the section with a technical lemma that will be used
crucially for both of our algorithms. The importance of it would become
obvious when it is invoked in our analyses.
\begin{lem}
\label{lem:submod:sandwhich-stronger} Let $\vh''$ denote a convex
combination of two vectors $\vh$ and $\vh'$ in the base polyhedron,
i.e. $\vh''=\lambda\vh+(1-\lambda)\vh'$ for some $\lambda\in[0,1]$.
Further suppose that 
\[
\norm{\vh''}_{2}\leq\alpha\min\left\{ \lambda\norm{\vh}_{2},(1-\lambda)\norm{\vh'}_{2}\right\} 
\]
for some $\alpha\leq\frac{1}{2\sqrt{n}}$. Then for $p=\argmax_{j}(\max\{\lambda|h_{j}|,(1-\lambda)|h'_{j}|\})$
we have
\[
\lowp\leq-\frac{1}{2\alpha\sqrt{n}}\cdot\normInf{\vh''}\enspace\text{ and }\enspace\upp\geq\frac{1}{2\alpha\sqrt{n}}\cdot\norm{\vh''}_{\infty}\enspace.
\]
\end{lem}
\begin{proof}
Suppose without loss of generality that $\lambda|h_{p}|\geq(1-\lambda)|h'_{p}|$.
Then by assumptions we have
\[
\norm{\vh''}_{\infty}\leq\norm{\vh''}_{2}\leq\alpha\cdot\min\left\{ \lambda\norm{\vh}_{2},(1-\lambda)\norm{\vh'}_{2}\right\} \leq\alpha\sqrt{n}\abs{\lambda h_{p}}\enspace.
\]
However, since $\alpha\leq\frac{1}{2\sqrt{n}}$ we see that
\[
\abs{\lambda h_{p}+(1-\lambda)h'_{p}}\leq\normInf{h''}\leq\alpha\sqrt{n}\abs{\lambda h_{p}}\leq\frac{1}{2}\abs{\lambda h_{p}}\enspace.
\]
Consequently, $\lambda h_{p}$ and $(1-\lambda)h'_{p}$ have opposite
signs and $\abs{(1-\lambda)h'_{p}}\geq\frac{1}{2}\abs{\lambda h'_{p}}$.
We then have, 
\[
\lowp\leq\min\left\{ h_{p},h'_{p}\right\} \leq\min\left\{ \lambda h_{p},(1-\lambda)h'_{p}\right\} \leq-\frac{1}{2}\abs{\lambda h{}_{p}}\leq-\frac{1}{2\alpha\sqrt{n}}\norm{h''}_{\infty}
\]
and
\[
\upp\geq\max\left\{ h_{p},h'_{p}\right\} \geq\max\left\{ \lambda h_{p},(1-\lambda)h'_{p}\right\} \geq\frac{1}{2}\abs{\lambda h{}_{p}}\geq\frac{1}{2\alpha\sqrt{n}}\norm{h''}_{\infty}\,.
\]

\end{proof}

\subsection{$\widetilde{O}(n^{4}\cdot\text{EO}+n^{5})$ Time Algorithm\label{sub:submodular-n4}}

Here we present a $\widetilde{O}(n^{4}\cdot\text{EO}+n^{5})$ time,
i.e. strongly polynomial time algorithm, for SFM. We build upon the
algorithm achieved in the section to achieve a faster running time
in Section~\ref{sub:submod:strongly_n3}.

Our new algorithm combines the existing tools for SFM developed over
the last decade with our cutting plane method. While there are certain
similarities with previous algorithms (especially {\small{}\cite{iwata2003faster,iwata2009simple,iwata2001combinatorial}}),
our approach significantly departs from all the old approaches in
one important aspect.

All of the previous algorithms actively maintain a point in the base
polyhedron and represent it as a \textit{convex combination }of BFS's.
At each step, a new BFS may enter the convex combination and an old
BFS may exit. Our algorithm, on the other hand, maintains only a \textit{collection}
of BFS's (corresponding to our separating hyperplanes), rather than
an explicit convex combination. A ``good'' convex combination is
computed from the \textit{collection} of BFS's only after running
Cutting Plane for enough iterations. We believe that this crucial
difference is the fundamental reason which offers the speedup. This
is achieved by the Cutting Plane method which considers the \textit{geometry}
of the collection of BFS's. On the other hand, considering only a
convex combination of BFS's effectively narrows our sight to only
one \textit{point }in the base polyhedron.
\begin{description}
\item [{Overview}]~
\end{description}
Now we are ready to describe our strongly polynomial time algorithm.
Similar to the weakly polynomial algorithm, we first run our cutting
plane for enough iterations on the initial feasible region $\{\vx\in[0,1]^{n}:x_{i}\leq x_{j}\:\forall(i,j)\in A\}$,
after which a pair of approximately parallel supporting hyperplanes
$F_{1},F_{2}$ of width $1/n^{\Theta(1)}$ can be found. Our strategy
is to write $F_{1}$ and $F_{2}$ as a nonnegative combination of
the facets of remaining feasible region $P$. This combination is
made up of newly added separating hyperplanes as well as the inequalities
$x_{i}\geq0$, $x_{j}\leq1$ and $x_{i}\leq x_{j}$. We then argue
that one of the following updates can be done:
\begin{itemize}
\item Collapsing: $x_{i}=0$, $x_{j}=1$ or $x_{i}=x_{j}$
\item Adding a new arc $(i,j)$: $x_{i}\leq x_{j}$ for some $(i,j)\notin A$
\end{itemize}
The former case is easy to handle by elimination or contraction. If
$x_{i}=0$, we simply eliminate $i$ from the ground set $V$; and
if $x_{i}=1$, we redefine $f$ so that $f(S)=f(S+i)$ for any $S\subseteq V-i$.
$x_{i}=x_{j}$ can be handled in a similar fashion. In the latter
case, we simply add the arc $(i,j)$ to $A$. We then repeat the same
procedure on the new problem.

Roughly speaking, our strongly polynomial time guarantee follows as
eliminations and contractions can happen at most $n$ times and at
most $2\cdot{n \choose 2}$ new arcs can be added. While the whole
picture is simple, numerous technical details come into play in the
execution. We advise readers to keep this overview in mind when reading
the subsequent sections.
\begin{description}
\item [{Algorithm}]~
\end{description}
Our algorithm is summarized below. Again, we remark that our algorithm
simply uses Theorem~\ref{thm:cuttingplanerestated} regarding our
cutting plane and is agnostic as to how the cutting plane works, thus
it could be replaced with other methods, albeit at the expense of
slower runtime.
\begin{enumerate}
\item Run cutting plane on \eqref{eq:ring} (Theorem~\ref{thm:cuttingplanerestated}
with $\tau=\Theta(1)$) using our modified separation oracle (Section~\ref{sub:ring_family}).
\item Identify a pair of ``narrow'' approximately parallel supporting
hyperplanes or get some BFS $\vh=0$ (in which case both $\emptyset$
and $V$ are minimizers).
\item Deduce from the hyperplanes some new constraint of the forms $x_{i}=0,x_{j}=1,x_{i}=x_{j}$
or $x_{i}\leq x_{j}$ (Section \ref{sub:Deducing-new-constraints}).
\item Consolidate $A$ and $f$ (Section~\ref{sub:Consolidating}).
\item Repeat by running our cutting plane method on \eqref{eq:ring} with
updated $A$ and $f$. (Note that Any previously found separating
hyperplanes are discarded.)
\end{enumerate}
We call step (1) a \textit{phase} of cutting plane. The minimizer
can be constructed by unraveling the recursion.

\subsubsection{Consolidating $A$ and $f$\label{sub:Consolidating}}

Here we detail how the set of valid arcs $A$ and submodular function
$f$ should be updated once we deduce new information $x_{i}=0,x_{i}=1,x_{i}=x_{j}$
or $x_{i}\leq x_{j}$. Recall that $R(i)$ and $Q(i)$ are the sets
of descendants and ancestors of $i$ respectively (including $i$
itself). The changes below are somewhat self-evident, and are actually
used in some of the previous algorithms so we only sketch how they
are done without a detailed justification.

Changes to the digraph representation $D$ of our ring family include:
\begin{itemize}
\item $x_{i}=0$: remove $Q(i)$ from the ground set and all the arcs incident
to $Q(i)$
\item $x_{i}=1$: remove $R(i)$ from the ground set and all the arcs incident
to $R(i)$
\item $x_{i}=x_{j}$: contract $i$ and $j$ in $D$ and remove any duplicate
arcs
\item $x_{i}\leq x_{j}$: insert the arc $(i,j)$ to $A$
\item For the last two cases, we also contract the vertices on a directed
cycle of $A$ until there is no more. Remove any duplicate arcs.
\end{itemize}
Here we can contract any cycle $(i_{1},\ldots,i_{k})$ because the
inequalities $x_{i_{1}}\leq x_{i_{2}},\ldots,x_{i_{k-1}}\leq x_{i_{k}},x_{i_{k}}\leq x_{i_{1}}$
imply $x_{i_{1}}=\ldots=x_{i_{k}}$.

Changes to $f$:
\begin{itemize}
\item $x_{i}=0$: replace $f$ by $f':2{}^{V\backslash Q(i)}\longrightarrow\mathbb{R}$,
$f'(S)=f(S)$ for $S\subseteq V\backslash Q(i)$
\item $x_{i}=1$: replace $f$ by $f':2^{V\backslash R(i)}\longrightarrow\mathbb{R}$,
$f'(S)=f(S\cup R(i))$ for $S\subseteq V\backslash R(i)$
\item $x_{i}=x_{j}$: see below
\item $x_{i}\leq x_{j}$: no changes to $f$ needed if it does not create
a cycle in $A$; otherwise see below
\item Contraction of $C=\{i_{1},\ldots,i_{k}\}$: replace $f$ by $f':2^{V\backslash C+l}\longrightarrow\mathbb{R}$,
$f'(S)=f(S)$ for $S\subseteq V\backslash C$ and $f'(S)=f((S-l)\cup C)$
for $S\ni l$
\end{itemize}
Strictly speaking, these changes are in fact \textit{not} needed as
they will automatically be taken care of by our cutting plane method.
Nevertheless, performing them lends a more natural formulation of
the algorithm and simplifies its description.

\subsubsection{Deducing New Constraints $x_{i}=0$, $x_{j}=1$, $x_{i}=x_{j}$ or
$x_{i}\leq x_{j}$\label{sub:Deducing-new-constraints}}

Here we show how to deduce new constraints through the result of our
cutting plane method. This is the most important ingredient of our
algorithm. As mentioned before, similar arguments were used first
by IFF {\small{}\cite{iwata2001combinatorial}} and later in {\small{}\cite{iwata2003faster,iwata2009simple}}.
There are however two important differences for our method:
\begin{itemize}
\item We maintain a collection of BFS's rather a convex combination; a convex
combination is computed and needed only after each phase of cutting
plane.
\item As a result, our results are proved mostly \textit{geometrically}
whereas the previous ones were proved mostly \textit{combinatorially}.
\end{itemize}
Our ability to deduce such information hinges on the power of the
cutting plane method in Part~\ref{part:Ellipsoid}. We re-state our
main result Theorem~\ref{thm:main_result} in the language of SFM.
Note that Theorem~\ref{thm:cuttingplanerestated} is formulated in
a fairly general manner in order to accommodate for the next section.
Readers may wish to think $\tau=\Theta(1)$ for now.
\begin{thm}[Theorem~\ref{thm:main_result} restated for SFM]
\label{thm:cuttingplanerestated} For any $\tau\geq100$, applying
our cutting plane method, Theorem~\ref{thm:cuttingplanerestated},
to \eqref{eq:ring} with our modified separation oracle (or its variant
in Section \ref{sub:submod:strongly_n3}) with high probability in
$n$ either
\begin{enumerate}
\item Finds a degenerate BFS $\vh\geq\vec{0}$ or $\vh\leq\vzero$.
\item Finds a polytope $P$ consisting of $O(n)$ constraints which are
our separating hyperplanes or the constraints in \eqref{eq:ring}.
Moreover, $P$ satisfies the following inequalities
\[
\vc^{T}\vx\le M\,\,\,\,\,\,\text{ and }\,\,\,\,\,\,\vc'^{T}\vx\leq M',
\]
both of which are nonnegative combinations of the constraints of $P$,
where $||\vc+\vc'||_{2}\le\min\{||\vc||_{2},||\vc'||_{2}\}/n^{\Theta(\tau)}$
and $|M+M'|\le\min\{||\vc||_{2},||\vc'||_{2}\}/n^{\Theta(\tau)}$. 
\end{enumerate}
Furthermore, the algorithm runs in expected time $O(n^{2}\tau\log n\cdot\text{EO}+n^{3}\tau^{O(1)}\log^{O(1)}n)$.\end{thm}
\begin{proof}
In applying Theorem~\ref{thm:cuttingplanerestated} we let $K$ be
the set of minimizers of $f$ over the ring family and the box is
the hypercube with $R=1$. We run cutting plane with our modified
separation oracle (Lemma~\ref{lem:Our-modified-separation}). The
initial polytope $P^{(0)}$ can be chosen to be, say, the hypercube.
If some separating hyperplane is degenerate, then we have the desired
result (and know that either $\emptyset$ or $V$ is optimal). Otherwise
let $P$ be the current feasible region. Note that $P\neq\emptyset$,
because our minimizers of $\hat{f}$ are all in $P^{(0)}$ and $P^{(k)}$
as they are never cut away by the separating hyperplanes.

Let $\mathcal{S}$ be the collection of inequalities \eqref{eq:ring}
as well as the separating hyperplanes $\vh^{T}\vx\leq\hat{f}(\bar{x}_{h})=\vh^{T}\bar{x}_{h}$
used. By Theorem~\ref{thm:main_result}, all of our minimizers will
be contained in $P$, consisting of $O(n)$ constraints $\ma\vx\geq\vb$.
Each such constraint $\va_{i}^{T}\vx\geq b_{i}$ is a scaling and
shifting of some inequality $\vp_{i}^{T}\vx\geq q_{i}$ in $\mathcal{S}$,
i.e. $\va_{i}=\vp_{i}/||\vp_{i}||_{2}$ and $b_{i}\leq q_{i}/||\vp_{i}||_{2}$.

By taking $\epsilon=1/n^{\Theta(\tau)}$ with sufficiently large constant
in $\Theta$, our theorem certifies that $P$ has a narrow width by
$\va_{1}$, some nonnegative combination $\sum_{i=2}^{O(n)}t_{i}\va_{i}$
and point $\vx_{o}\in P$ with $||\vx_{o}||_{\infty}\leq3\sqrt{n}R=3\sqrt{n}$
satisying the following:
\[
\left\Vert \va_{1}+\sum_{i=2}^{O(n)}t_{i}\va_{i}\right\Vert _{2}\leq1/n^{\Theta(\tau)}
\]
\[
0\leq\va_{1}^{T}\vx_{o}-\vb_{1}\leq1/n^{\Theta(\tau)}
\]
\[
0\leq\left(\sum_{i=2}^{O(n)}t_{i}a_{i}\right)^{T}\vx_{o}-\sum_{i=2}^{O(n)}t_{i}b_{i}\leq1/n^{\Theta(\tau)}
\]

We convert these inequalities to $\vp$ and $q$. Let $t_{i}'\defeq t_{i}\cdot||\vp_{1}||_{2}/||\vp_{i}||_{2}\geq0$.
\[
\left\Vert \vp_{1}+\sum_{i=2}^{O(n)}t_{i}'\vp_{i}\right\Vert _{2}\leq||\vp_{1}||_{2}/n^{\Theta(\tau)}
\]
\[
0\leq\vp_{1}^{T}\vx_{o}-q_{1}\leq||\vp_{1}||_{2}/n^{\Theta(\tau)}
\]
\[
0\leq\left(\sum_{i=2}^{O(n)}t_{i}'\vp_{i}\right)^{T}\vx_{o}-\sum_{i=2}^{O(n)}t_{i}'q_{i}\leq||\vp_{1}||_{2}/n^{\Theta(\tau)}
\]

We claim that%
\footnote{Minus signs is needed because we express our inequalities as e.g.
$\vh^{T}\vx\leq\vh^{T}\bar{x}_{h}$ whereas in Theorem \ref{thm:main_result},
$\va_{i}^{T}\vx\geq b_{i}$ is used. We apologize for the inconvenience.%
} $\vc=-\vp_{1}$, $M=-q_{1}$, $\vc'=-\sum_{i=2}^{O(n)}t_{i}'\vp_{i}$,
$M'=-\sum_{i=2}^{O(n)}t_{i}'q_{i}$ satisfy our requirement.

We first show that $||\vc+\vc'||_{2}\leq\min\{||\vc||_{2},||\vc'||_{2}\}/n^{\Theta(\tau)}$.
We have $||\vc+\vc'||_{2}\leq||\vc||_{2}/n^{\Theta(\tau)}$ from the
first inequality. If $||\vc||_{2}\leq||\vc'||_{2}$ we are done. Otherwise,
by triangle inequality
\[
||\vc'||_{2}-||\vc||_{2}\leq||\vc+\vc'||_{2}\leq||\vc||_{2}/n^{\Theta(\tau)}\implies2||\vc||_{2}\geq||\vc'||_{2}
\]
and hence $||\vc+\vc'||_{2}\leq||\vc||_{2}/n^{\Theta(\tau)}\leq||\vc'||_{2}/2n^{\Theta(\tau)}=||\vc'||_{2}/n^{\Theta(\tau)}$.

We also need to prove $|M+M'|\leq\min\{||\vc||_{2},||\vc'||_{2}\}/n^{\Theta(\tau)}$.
Summing the second and third inequalities,
\[
-||\vc||_{2}/n^{\Theta(\tau)}\leq(\vc+\vc')^{T}\vx_{o}-(M+M')\leq0
\]

Recall that we have $||\vx_{o}||_{\infty}\leq3\sqrt{n}$. Then
\begin{eqnarray*}
|M+M'| & \leq & |(\vc+\vc')^{T}\vx_{o}-(M+M')|+|(\vc+\vc')^{T}\vx_{o}|\\
 & \leq & ||\vc||_{2}/n^{\Theta(\tau)}+3\sqrt{n}||\vc+\vc'||_{2}\\
 & \leq & ||\vc||_{2}/n^{\Theta(\tau)}+3\sqrt{n}||\vc||_{2}/n^{\Theta(\tau)}\\
 & = & ||\vc||_{2}/n^{\Theta(\tau)}
\end{eqnarray*}
as desired. Our result then follows as we proved $2||\vc'||_{2}\geq||\vc||_{2}$.

Finally, we have the desired runtime as our modified separation oracle
runs in time $O(n\cdot\text{EO}+n^{2}\log^{O(1)}n)$.
\end{proof}
Informally, the theorem above simply states that after $O(n\tau\log n)$
iterations of cutting plane, the remaining feasible region $P$ can
be sandwiched between two approximately parallel supporting hyperplanes
of width $1/n^{O(\tau)}$. A good intuition to keep in mind is that
every $O(n)$ iterations of cutting plane reduces the minimum width
by a constant factor.
\begin{rem}
As shown in the proof of Theorem~\ref{thm:cuttingplanerestated},
one of the two approximately parallel hyperplanes can actually be
chosen to be a constraint of our feasible region $P$. However we
do not exploit this property as it does not seem to help us and would
break the notational symmetry in $\vc$ and $\vc'$.\end{rem}
\begin{description}
\item [{Setup}]~
\end{description}
In each phase, we run cutting plane using Theorem~\ref{thm:cuttingplanerestated}
with $\tau=\Theta(1)$. If some separating hyperplane used is degenerate,
we have found the minimizer by Lemma~\ref{lem:degenerate_hyperplane}.

Now assume none of the separating hyperplanes is degenerate. By Theorem~\ref{thm:cuttingplanerestated},
$P$ is sandwiched by a pair of approximately parallel supporting
hyperplanes $F,F'$ which are of width $1/10n^{10}$ apart. The width
here can actually be $1/n^{c}$ for any constant $c$ by taking a
sufficiently large constant in Theta.

Here, we show how to deduce from $F$ and $F'$ some $x_{i}=0$,$,x_{j}=1$,$x_{i}=x_{j}$,
or $x_{i}\leq x_{j}$ constraint on the minimizers of $f$ over the
ring family. Let 
\[
\vc^{T}\vx=\sum c_{i}x_{i}\leq M\,\,\,\,\,\,\text{ and }\,\,\,\,\,\,\vc'^{T}\vx=\sum c_{i}'x_{i}\leq M'
\]
 be the inequality for $F$ and $F'$ such that 
\[
|M+M'|,\:||\vc+\vc'||_{2}\leq\gap,\quad\text{where }\gap\defeq\frac{1}{10n^{10}}\min\{||\vc||_{2},||\vc'||_{2}\}.
\]

By the same theorem we can write $\vc^{T}\vx\leq M$ as a nonnegative
combination of the constraints for $P$. Recall that the constraints
for $P$ take on four different forms: (1) $-x_{i}\leq0$; (2) $x_{j}\leq1$;
(3) $-(x_{j}-x_{i})\leq0$; (4) $\vh^{T}\vx=\sum h_{i}x_{i}\leq\hat{f}(\bar{x}_{h})$.
Here the first three types are present initially whereas the last
type is the separating hyperplane added. As alleged previously, the
coefficient vector $\vh$ corresponds to a BFS of the base polyhedron
for $f$. Our analysis crucially exploits this property.

Thus suppose $\vc^{T}\vx=\sum_{i}c_{i}x_{i}\leq M$ is a nonnegative
combination of our constraints with weights $\alpha_{i},\beta_{j},\gamma_{ij},\lambda_{h}\geq0$.
The number of (positive) $\alpha_{i},\beta_{j},\gamma_{ij},\lambda_{h}$
is at most $O(n)$. Here we denote separating hyperplanes by $\vh^{T}\vx\leq\hat{f}(\bar{x}_{h})$.
Let $H$ be the set of BFS's used to construct separating hyperplanes.
\begin{equation}
\vc^{T}\vx=-\sum_{i}\alpha_{i}x_{i}+\sum_{j}\beta_{j}x_{j}+\sum_{(i,j)\in A}\gamma_{ij}(x_{i}-x_{j})+\sum_{h\in H}\lambda_{h}\vh^{T}\vx\,\,\,\,\,\,\text{ and }\,\,\,\,\,\, M=\sum_{j}\beta_{j}+\sum_{h\in H}\lambda_{h}\hat{f}(\bar{x}_{h}).\label{eq:combin}
\end{equation}

Similarly, we write the inequality for $F'$ as a nonnegative combination
of the constraints for $P$ and the number of (positive) $\alpha_{i}',\beta_{j}',\gamma_{ij}',\lambda_{h}'$
is $O(n)$:
\begin{equation}
\vc'^{T}\vx=-\sum\alpha_{i}'x_{i}+\sum\beta_{j}'x_{j}+\sum_{(i,j)\in A}\gamma_{ij}'(x_{i}-x_{j})+\sum_{h\in H}\lambda_{h}'\vh^{T}\vx\,\,\,\,\,\,\text{ and }\,\,\,\,\,\, M'=\sum\beta_{j}'+\sum_{h\in H}\lambda_{h}'\hat{f}(\bar{x}_{h}).\label{eq:combin2}
\end{equation}

We also scale $\vc,\vc',\alpha,\alpha',\beta,\beta',\gamma,\gamma',\lambda,\lambda'$
so that
\[
\sum_{h\in H}(\lambda_{h}+\lambda_{h}')=1
\]
as this does not change any of our preceding inequalities regarding
$F$ and $F'$.

Now that $F,F'$ have been written as combinations of our constraints,
we have gathered the necessary ingredients to derive our new arc.
We first give a geometric intuition why we would expect to be able
to derive a new constraint. Consider the nonnegative combination making
up $F$. We think of the coefficient $\beta_{j}$ as the contribution
of $x_{j}\leq1$ to $F$. Now if $\beta_{j}$ is very large, $F$
is ``very parallel'' to $x_{j}\leq1$ and consequently $F'$ would
miss $x_{j}=0$ as the gap between $F$ and $F'$ is small. $P$ would
then miss $x_{j}=0$ too as it is sandwiched between $F$ and $F'$.
Similarly, a large $\alpha_{i}$ and a large $\gamma_{ij}$ would
respectively imply that $x_{i}=1$ and $(x_{i}=0,x_{j}=1)$ would
be missed. The same argument works for $F'$ as well.

But on the other hand, if the contributions from $x_{i}\geq0,x_{j}\leq1,x_{i}\leq x_{j}$
to both $F$ and $F'$ are small, then the supporting hyperplanes
$\vc^{T}\vx\leq...$ and $\vc'^{T}\vx\le...$ would be mostly made
up of separating hyperplanes $\vh^{T}\vx\leq\hat{f}(\bar{x}_{h})$.
By summing up these separating hyperplanes (whose coefficients form
BFS's), we would then get a point in the base polyhedron which is
very close to the origin 0. Moreover, by Lemma~\ref{lem:submod:sandwhich-stronger}
and Lemma \ref{lem:newarc} we should then be able to deduce some
interesting information about the minimizer of $f$ over $D$.

The rest of this section is devoted to realizing the vision sketched
above. We stress that while the algebraic manipulations may be long,
they are simply the execution of this elegant geometric picture.

Now, consider the following weighted sum of $\vh^{T}\vx\leq\hat{f}(\bar{x}_{h})$:
\[
\left(\sum_{h\in H}\lambda_{h}\vh^{T}+\sum_{h\in H}\lambda_{h}'\vh^{T}\right)\vx=\sum_{h\in H}\lambda_{h}\vh^{T}\vx+\sum_{h\in H}\lambda_{h}'\vh^{T}\vx\leq\sum_{h\in H}\lambda_{h}\hat{f}(\bar{x}_{h})+\sum_{h\in H}\lambda_{h}'\hat{f}(\bar{x}_{h}).
\]

Observe that $\sum_{h\in H}\lambda_{h}\vh^{T}+\sum_{h\in H}\lambda_{h}'\vh^{T}$
is in the base polyhedron since it is a convex combination of BFS
$\vh$. Furthermore, using (\ref{eq:combin}) and (\ref{eq:combin2})
this can also be written as
\begin{equation}
\begin{aligned}\left(\sum_{h\in H}\lambda_{h}\vh^{T}+\sum_{h\in H}\lambda_{h}'\vh^{T}\right)\vx & =\left(\vc^{T}\vx+\sum\alpha_{i}x_{i}-\sum\beta_{j}x_{j}+\sum_{(i,j)\in A}\gamma_{ij}(x_{j}-x_{i})\right)\\
 & \enspace\enspace\enspace+\left(\vc'^{T}\vx+\sum\alpha_{i}'x_{i}-\sum\beta_{j}'x_{j}+\sum_{(i,j)\in A}\gamma_{ij}'(x_{j}-x_{i})\right)
\end{aligned}
\label{eq:bfs}
\end{equation}
and 
\[
\begin{aligned}\sum_{h\in H}\lambda_{h}\hat{f}(\bar{x}_{h})+\sum_{h\in H}\lambda_{h}'\hat{f}(\bar{x}_{h}) & =\left(M-\sum\beta_{j}\right)+\left(M'-\sum\beta_{j}'\right)\\
 & =(M+M')-\sum\beta_{j}-\sum\beta_{j}'
\end{aligned}
\]

Furthermore, we can bound $\vc^{T}\vx+\vc'^{T}\vx$ by $\vc^{T}\vx+\vc'^{T}\vx\geq-||\vc+\vc'||_{1}\geq-\sqrt{n}||\vc+\vc'||_{2}\geq-\sqrt{n}\gap$
as $\vx\leq1$. Since $M+M'\leq\gap$, we obtain
\[
LHS\defeq\sum\alpha_{i}x_{i}+\sum\alpha_{i}'x_{i}-\sum\beta_{j}x_{j}-\sum\beta_{j}'x_{j}+\sum_{(i,j)\in A}\gamma_{ij}(x_{j}-x_{i})+\sum_{(i,j)\in A}\gamma_{ij}'(x_{j}-x_{i})
\]
\[
\leq2\sqrt{n}\gap-\sum\beta_{j}-\sum\beta_{j}'
\]

Geometrically, the next lemma states that if the contribution from,
say $x_{i}\geq0$, to $F$ is too large, then $F'$ would be forced
to miss $x_{i}=1$ because they are close to one another.
\begin{lem}
\label{lem:dimcut} Suppose $\vx$ satisfies \eqref{eq:ring} and
$LHS\leq2\sqrt{n}\gap-\sum\beta_{j}-\sum\beta_{j}'$ with $\alpha_{i},\beta_{j},\gamma_{ij},\alpha_{i}',\beta_{j}',\gamma_{ij}'\geq0$.
\begin{enumerate}
\item If $\alpha_{i}>2\sqrt{n}\gap$ or $\alpha_{i}'>2\sqrt{n}\gap$, then
$x_{i}<1$.
\item If $\beta_{j}>2\sqrt{n}\gap$ or $\beta_{j}'>2\sqrt{n}\gap$, then
$x_{j}>0$.
\item If $\gamma_{ij}>2\sqrt{n}\gap$ or $\gamma_{ij}'>2\sqrt{n}\gap$,
then $0\leq x_{j}-x_{i}<1$.
\end{enumerate}
\end{lem}
\begin{proof}
We only prove it for $\alpha_{i},\beta_{j},\gamma_{ij}$ as the other
case follows by symmetry.

Using $0\leq x\le1$ and $x_{i}\leq x_{j}$ for $(i,j)\in A$, we
have $LHS\geq\alpha_{i}x_{i}-\sum\beta_{j}-\sum\beta_{j}'$. Hence
$\alpha_{i}x_{i}\leq2\sqrt{n}\gap$ and we get $x_{i}<1$ if $\alpha_{i}>2\sqrt{n}\gap$.

Similarly, $LHS\geq-\beta_{k}x_{k}-\sum_{j\neq k}\beta_{j}-\sum\beta_{j}'$
which gives $-\beta_{k}x_{k}\leq2\sqrt{n}\gap-\beta_{k}$. Then $x_{k}>0$
if $\beta_{k}>2\sqrt{n}\gap$.

Finally, $LHS\ge\gamma_{ij}(x_{j}-x_{i})-\sum\beta_{j}-\sum\beta_{j}'$
which gives $\gamma_{ij}(x_{j}-x_{i})\le2\sqrt{n}\gap$. Then $x_{j}-x_{i}<1$
if $\gamma_{ij}>2\sqrt{n}\gap$. We have $x_{i}\leq x_{j}$ since
$(i,j)\in A$.
\end{proof}
So if either condition of Lemma~\ref{lem:dimcut} holds, we can set
$x_{i}=0$ or $x_{j}=1$ or $x_{i}=x_{j}$ since our problem \eqref{eq:ring}
has an integral minimizer and any minimizer of $\hat{f}$ is never
cut away by Lemma~\ref{lem:Our-modified-separation}. Consequently,
in this case we can reduce the dimension by at least 1. From now on
we may assume that
\begin{equation}
\max\{\alpha_{i},\alpha_{i}',\beta_{j},\beta_{j}',\gamma_{ij},\gamma_{ij}'\}\leq2\sqrt{n}\gap.\label{eq:small}
\end{equation}

Geometrically, \eqref{eq:small} says that if the supporting hyperplanes
are both mostly made up of the separating hyperplanes, then their
aggregate contributions to $F$ and $F'$ should be small in absolute
value.

The next lemma identifies some $p\in V$ for which $f(R(p))-f(R(p)-p)$
is ``big''. This prepares for the final step of our approach which
invokes Lemma \ref{lem:newarc}.
\begin{lem}
\label{lem:newarcfound}Let $\vy\defeq\sum_{h\in H}\lambda_{h}\vh$
and $\vy'\defeq\sum_{h\in H}\lambda'_{h}\vh$ and let $p\in\argmax_{l}\{\max\{|y_{l}|,|y'_{l}|\}\}$
then
\[
\upp\geq n^{7}\normInf{\vy+\vy'}
\]
assuming \eqref{eq:small}.\end{lem}
\begin{proof}
Recall that $\norm{\vc+\vc'}_{2}\leq\gap$ where $\gap=\frac{1}{10n^{10}}\min\{||\vc||_{2},||\vc'||_{2}\}$,
\[
\vc=\vy-\sum_{i}\alpha_{i}\indicVec i+\sum_{j}\beta_{j}\indicVec j+\sum_{(i,j)}\gamma_{ij}(\indicVec i-\indicVec j)\enspace\text{ and }\vc'=\vy'-\sum_{i}\alpha_{i}'\indicVec i+\sum_{j}\beta_{j}'\indicVec j+\sum_{(i,j)}\gamma_{ij}'(\indicVec i-\indicVec j)\,.
\]
By \eqref{eq:small} we know that $\norm{\vc-\vy}_{2}\leq4n^{2}\gap\leq\frac{4}{10n^{8}}\norm{\vc}_{2}$
and $\norm{\vc'-\vy'}_{2}\leq4n^{2}\gap\leq\frac{4}{10n^{8}}\norm{\vc'}_{2}$.
Consequently, by the triangle inequality we have that
\begin{align*}
\norm{\vy+\vy'}_{2} & \leq\norm{\vc+\vc'}_{2}+\norm{\vc-\vy}_{2}+\norm{\vc'-\vy'}_{2}\leq9n^{2}\gap
\end{align*}
and
\[
\norm{\vc}_{2}\leq\norm{\vc-\vy}_{2}+\norm{\vy}_{2}\leq\frac{4}{10n^{8}}\norm{\vc}_{2}+\norm{\vy}_{2}\enspace\Rightarrow\enspace\norm{\vc}_{2}\leq2\norm{\vy}_{2}
\]
Similarly, we have that $\norm{\vc'}_{2}\leq2\norm{\vy'}_{2}$. Consequently
since $\gap\leq\frac{1}{10n^{10}}\min\{||\vc||_{2},||\vc'||_{2}\}$,
we have that 
\[
\norm{\vy+\vy'}_{2}\leq\frac{2}{n^{8}}\min\left\{ \norm{\vy}_{2},\norm{\vy'}_{2}\right\} 
\]
and thus, invoking Lemma~\ref{lem:submod:sandwhich-stronger} yields
the result.
\end{proof}
We summarize the results in the lemma below.
\begin{cor}
\label{cor:summary} Let $P$ be the feasible region after running
cutting plane on \eqref{eq:ring}. Then one of the following holds:
\begin{enumerate}
\item We found a degenerate BFS and hence either $\emptyset$ or $V$ is
a minimizer.
\item The integral points of $P$ all lie on some hyperplane $x_{i}=0$,
$x_{j}=1$ or $x_{i}=x_{j}$ which we can find.
\item Let $H$ be the collection of BFS's $\vh$ used to construct our separating
hyperplanes for $P$. Then there is a convex combination $\vy$ of
$H$ such that $n^{4}|y_{i}|<\max_{p}\upp$ for all $i$.
\end{enumerate}
\end{cor}
\begin{proof}
As mentioned before, (1) happens if some separating hyperplane is
degenerate. We have (2) if one of the conditions in Lemma \ref{lem:dimcut}
holds. Otherwise, $y=\sum_{h\in H}\lambda_{h}\vh+\sum_{h\in H}\lambda_{h}'\vh$
is a candidate for Case 3 by Lemma \ref{lem:newarcfound}.
\end{proof}
Let us revisit the conditions of Lemma \ref{lem:newarc} and explain
that they are satisfied by Case 3 of the last lemma.
\begin{itemize}
\item $\vy$ is a convex combination of at most $O(n)$ BFS's. This holds
in Case 3 since our current feasible region consists of only $O(n)$
constraints thanks to the Cutting Plane method.
\item Those BFS's must be consistent with every arc of $A$. This holds
because Case 3 uses the BFS's for constructing our separating hyperplane.
Our modified separation oracle guarantees that they are consistent
with $A$.
\end{itemize}
Thus in Case 3 of the last corollary, Lemma \ref{lem:newarc} allows
us to deduce a new constraint $x_{p}\leq x_{q}$ for some $q\notin R(p)$.

\subsubsection{Running Time}

Here we bound the total running time of our algorithm and prove the
following.
\begin{thm}
Our algorithm runs in time $O(n^{4}\log n\cdot\text{EO}+n^{5}\log^{O(1)}n)$.\end{thm}
\begin{proof}
To avoid being repetitive, we appeal to Corollary \ref{cor:summary}.
Each phase of cutting plane takes time $O(n^{2}\log n\cdot\text{EO}+n^{3}\log^{O(1)}n)$
(Theorem~\ref{thm:cuttingplanerestated} with $\tau$ being a big
constant. Given $F$ and $F'$ represented as a nonnegative combination
of facets, we can check for the conditions in Lemma~\ref{lem:dimcut}
in $O(n)$ time as there are only this many facets of $P$. This settles
Case 2 of Corollary~\ref{cor:summary}. Finally, Lemma~\ref{lem:newarc}
tells us that we can find a new arc in $O(n\cdot\text{EO}+n^{2})$
time for Case 3 of Corollary~\ref{cor:summary}. Our conclusion follows
from the fact that we can get $x_{i}=0$, $x_{i}=1$, $x_{i}=x_{j}$
at most $n$ times and $x_{i}\leq x_{j}$ at most $O(n^{2})$ times.
\end{proof}

\subsection{$\widetilde{O}(n^{3}\cdot\text{EO}+n^{4})$ Time Algorithm\label{sub:submod:strongly_n3}}

Here we show how to improve our running time for strongly polynomial
SFM to $\widetilde{O}(n^{3}\cdot\text{EO}+n^{4})$. Our algorithm
can be viewed as an extension of the algorithm we presented in the
previous Section~\ref{sub:submodular-n4}. The main bottleneck of
our previous algorithm was the time needed to identify a new arc,
which cost us $\widetilde{O}(n^{2}\cdot\text{EO}+n^{3})$. Here we
show how to reduce our amortized cost for identifying a valid arc
down to $\widetilde{O}(n\cdot\text{EO}+n^{2})$ and thereby achieve
our result.

The key observation we make to improve this running time is that our
choice of $p$ for adding an arc in the previous lemma can be relaxed.
$p$ actually need not be $\arg\max_{i}\upi$; instead it is enough
to have $\upp>n^{4}\max\{\alpha_{i},\alpha_{i}',\beta_{j},\beta_{j}',\gamma_{ij},\gamma_{ij}'\}$.
For each such $p$ a new constraint $x_{p}\leq x_{q}$ can be identified
via Lemma~\ref{lem:newarc}. So if there are many $p$'s satisfying
this we will be able to obtain many new constraints and hence new
valid arcs $(p,q)$.

On the other hand, the bound in Lemma \ref{lem:newarcfound} says
that our point in the base polyhedron is small in \textit{absolute}
value. This is actually stronger than what we need in Lemma~\ref{lem:newarc}
which requires only its positive entries to be ``small''. However
as we saw in Lemma~\ref{lem:newarc-new-1} we can generate a constraint
of the form $x_{q}\leq x_{p}$ whenever $\lowp$ is sufficiently negative.

Using this idea, we divide $V$ into different buckets according to
$\upp$ and $\lowp$. This will allow us to get a speedup for two
reasons. 

First, bucketing allows us to disregard unimportant elements of $V$
during certain executions of our cutting plane method. If both $\upi$
and $\lowi$ are small in absolute value, then $i$ is essentially
negligible because for a separating hyperplane $\vh^{T}\vx\leq\hat{f}(\bar{x})$,
any $h_{i}\in[\lowi,\upi]$ small in absolute value would not really
make a difference. We can then run our cutting plane algorithm only
on those non-negligible $i$'s, thereby reducing our time complexity.
Of course, whether $h_{i}$ is small is something relative. This suggests
that partitioning the ground set by the relative size of $\upi$ and
$\lowi$ is a good idea. 

Second, bucketing allows us to ensure that we can always add an arc
for many edges simultaneously. Recall that we remarked that all we
want is $n^{O(1)}|y_{i}|\leq\upp$ for some $\vy$ in the base polyhedron.
This would be sufficient to identify a new valid arc $(p,q)$. Now
if the marginal differences $\upp$ and $\code{upper}(p')$ are close
in value, knowing $n^{O(1)}|y_{i}|\leq\upp$ would effectively give
us the same for $p'$ for free. This suggests that elements with similar
marginal differences should be grouped together. 

The remainder of this section simply formalizes these ideas. In Section~\ref{sub:submod:n4:partition}
we discuss how we partition the ground set $V$. In Section~\ref{sub:submod:n4:separating},
we present our cutting plane method on a subset of the coordinates.
Then in Section~\ref{sub:submod:n4:new-constraints} we show how
we find new arcs. Finally, in Section~\ref{sub:submod:n4:runtime}
we put all of this together to achieve our desired running time.

\subsubsection{Partitioning Ground Set into Buckets\label{sub:submod:n4:partition}}

We partition the ground set $V$ into different buckets according
to the values of $\upi$ and $\lowi$. This is reminiscent to Iwata-Orlin's
algorithm \cite{iwata2009simple} which considers elements with big
$\upi$. However they did not need to do bucketing by size or to consider
$\lowi$, whereas these seem necessary for our algorithm.

Let $N=\max_{i}\{\max\{\upi,-\lowi\}\}$ be the largest marginal difference
in absolute value. By Lemma \eqref{lem:boundonsmall-1}, $N\geq0$.
We partition our ground set $V$ as follows:
\[
B_{1}=\{i:\upi\geq N/n^{10}\text{ or }\lowi\leq-N/n^{10}\}
\]
\begin{eqnarray*}
B_{k} & =\{i\notin B_{1}\cup\ldots\cup B_{k-1}: & N/n^{10k}\leq\upi<N/n^{10(k-1)}\\
 &  & \text{or }-N/n^{10(k-1)}<\lowi\le-N/n^{10k}\},\quad k\geq2
\end{eqnarray*}
We call $B_{k}$ \textit{buckets}. Our buckets group elements by the
values of $\upi$ and $\lowi$ at $1/n^{10}$ ``precision''. There
are two cases.
\begin{itemize}
\item Case 1: the number of buckets is at most $\log n$%
\footnote{More precisely, $B_{k}=\emptyset$ for $k>\ceil{\log n}$.%
}, in which case $\upi>N/n^{O(\log n)}$ or $\lowi<-N/n^{O(\log n)}$
for all $i$.
\item Case 2: there is some $k$ for which $|B_{1}\cup\ldots\cup B_{k}|\geq|B_{k+1}|$.
\end{itemize}
This is because if there is no such $k$ in Case 2, then by induction
each bucket $B_{k+1}$ has at least $2^{k}|B_{1}|\geq2^{k}$ elements
and hence $k\leq\log n$.

Case 1 is easier to handle, and is in fact a special case of Case
2. We first informally sketch the treatment for Case 1 which should
shed some light into how we deal with Case 2.

We run Cutting Plane for $O(n\log^{2}n)$ iterations (i.e. $\tau=\Theta(\log n)$).
By Theorem \ref{thm:cuttingplanerestated}, our feasible region $P$
would be sandwiched by a pair of approximately parallel supporting
hyperplanes of width at most $1/n^{\Theta(\log n)}$. Now proceeding
as in the last section, we would be able to find some $\vy$ in the
base polyhedron and some element $p$ such that $n^{\Theta(\log n)}|y_{i}|\leq\upp$.
This gives
\[
n^{\Theta(\log n)}|y_{i}|\leq\frac{\upp}{n^{\Theta(\log n)}}\leq\frac{N}{n^{\Theta(\log n)}}.
\]

Since $\upi>N/n^{\Theta(\log n)}$ or $\lowi<-N/n^{\Theta(\log n)}$
for all $i$ in Case 1, we can then conclude that some valid arc $(i,q)$
or $(q,i)$ can be added for every $i$. Thus we add $n/2$ arcs simultaneously
in one phase of the algorithm at the expense of blowing up the runtime
by $O(\log n)$. This saves a factor of $n/\log n$ from our runtime
in the last section, and the amortized cost for an arc would then
be $\widetilde{O}(n\cdot\text{EO}+n^{2})$.

On the other hand, in Case 2 we have a ``trough'' at $B_{k+1}$.
Roughly speaking, this trough is useful for acting as a soft boundary
between $B_{1}\cup\ldots\cup B_{k}$ and $\bigcup_{l\geq k+2}B_{l}$.
Recall that we are able to ``ignore'' $\bigcup_{l\geq k+2}B_{l}$
because their $h_{i}$ is relatively small in absolute value. In particular,
we know that for any $p\in B_{1}\cup\ldots\cup B_{k}$ and $i\in B_{l}$,
where $l\geq k+2$,
\[
\max\{\upp,-\lowp\}\geq n^{10}\max\{\upi,-\lowi\}.
\]

This is possible because $B_{k+1}$, which is sandwiched in between,
acts like a shield preventing $B_{l}$ to ``mess with'' $B_{1}\cup\ldots\cup B_{k}$.
This property comes at the expense of sacrificing $B_{k+1}$ which
must confront $B_{l}$.

Furthermore, we require that $|B_{1}\cup\ldots\cup B_{k}|\geq|B_{k+1}|$,
and run Cutting Plane on $B=(B_{1}\cup\ldots\cup B_{k})\cup B_{k+1}$.
If $|B_{k+1}|\gg|B_{1}\cup\ldots\cup B_{k}|$, our effort would mostly
be wasted on $B_{k+1}$ which is sacrificed, and the amortized time
complexity for $B_{1}\cup\ldots\cup B_{k}$ would then be large.

Before discussing the algorithm for Case 2, we need some preparatory
work.

\subsubsection{Separating Hyperplane: Project and Lift\label{sub:submod:n4:separating}}

Our speedup is achieved by running our cutting plane method on the
projection of our feasible region onto $B:=(B_{1}\cup\cdots\cup B_{k})\cup B_{k+1}$.
More precisely, we start by running our cutting plane on $P^{B}=\{\vx\in\mathbb{R}^{B}:\exists\vx'\in\mathbb{R}^{\bar{B}}\text{ s.t. }(\vx,\vx')\text{ satisfies }\eqref{eq:ring}\}$,
which has a lower dimension. However, to do this, we need to specify
a separation oracle for $P^{B}$. Here we make one of the most natural
choices.

We begin by making an apparently immaterial change to our set of arcs
$A$. Let us take the \textit{transitive closure }of $A$ by adding
the arc $(i,j)$ whenever there is a path from $i$ to $j$. Clearly
this would not change our ring family as a path from $i$ to $j$
implies $j\in R(i)$. Roughly speaking, we do this to handle pathological
cases such as $(i,k),(k,j)\in A,(i,j)\notin A$ and $i,j\in B,k\notin B$.
Without introducing the arc $(i,j)$, we risk confusing a solution
containing $i$ but not $j$ as feasible since we are restricting
our attention to $B$ and ignoring $k\notin B$.
\begin{defn}
Given a digraph $D=(V,A)$, the transitive closure of $A$ is the
set of arcs $(i,j)$ for which there is a directed path from $i$
to $j$. We say that $A$ is \textit{complete} if it is equal to its
transitive closure.
\end{defn}
Given $\bar{x}\in[0,1]^{B}$, we define the completion of $\bar{x}$
with respect to $A$ as follows.
\begin{defn}
Given $\bar{x}\in[0,1]^{B}$ and a set of arcs $A$, $x^{\mathcal{C}}\in[0,1]^{n}$
is a completion of $\bar{x}$ if $x_{B}^{\mathcal{C}}=\bar{x}$ and
$x_{i}^{\mathcal{C}}\leq x_{j}^{\mathcal{C}}$ for every $(i,j)\in A$.
Here $x_{B}^{\mathcal{C}}$ denotes the restriction of $x^{\mathcal{C}}$
to $B$.\end{defn}
\begin{lem}
\label{lem:completion}Given $\bar{x}\in[0,1]^{B}$ and a complete
set of arcs $A$, there is a completion of $\bar{x}$ if $\bar{x}_{i}\leq\bar{x}_{j}$
\textup{for every $(i,j)\in A\cap(B\times B)$. Moreover, it can be
computed in $O(n^{2})$ time.}\end{lem}
\begin{proof}
We set $x_{B}^{\mathcal{C}}=\bar{x}$. For $i\notin B$, we set 
\[
x_{i}^{\mathcal{C}}=\begin{cases}
1 & \text{if }\nexists j\in B\text{ s.t. }(i,j)\in A\\
\min_{(i,j)\in A,j\in B}x_{j}^{\mathcal{C}} & \text{otherwise}
\end{cases}
\]

One may verify that $x^{\mathcal{C}}$ satisfies our requirement as
$A$ is complete. Computing each $x_{i}^{\mathcal{C}}$ takes $O(n)$
time. Since $|V\backslash B|=|\bar{B}|\leq n$, computing the whole
$x^{\mathcal{C}}$ takes $O(n^{2})$ time.
\end{proof}
This notion of completion is needed since our original separation
oracle requires a full dimensional input $\bar{x}$. Now that $\bar{x}\in\mathbb{R}^{B}$,
we need a way of extending it to $\mathbb{R}^{n}$ while retaining
the crucial property that $\vh$ is consistent with every arc in $A$.

\begin{algorithm2e}
\caption{Projected Separation Oracle}

\SetAlgoLined

\textbf{Input:} $\bar{x}\in\mathbb{R}^{B}$ and a complete set of
arcs $A$

\uIf{ $\bar{x}_{i}<0$ for some $i\in B$ }{

\textbf{Output}: $x_{i}\ge0$

}\uElseIf{ $\bar{x}_{j}>1$ for some $j\in B$ }{

\textbf{Output}: $x_{j}\le1$

}\uElseIf{ $\bar{x}_{i}>\bar{x}_{j}$ for some $(i,j)\in A\cap B^{2}$
}{

\textbf{Output}: $x_{i}\leq x_{j}$

}\uElse{  

Let $x^{\mathcal{C}}\in\mathbb{R}^{n}$ be a completion of $\bar{x}$

Let $i_{1},\ldots,i_{n}$ be a permutation of $V$ such that $x_{i_{1}}^{\mathcal{C}}\ge\ldots\ge x_{i_{n}}^{\mathcal{C}}$
and for all $(i,j)\in A$, $j$ precedes $i$ in $i_{1},\ldots,i_{n}$.

\textbf{Output}: $\vh_{B}^{T}\vx_{B}=\sum_{i\in B}h_{i}x_{i}\leq\sum_{i\in B}h_{i}\bar{x}_{i}$,
where $\vh$ is the BFS defined by the permutation $i_{1},\ldots,i_{n}$.

}  

\end{algorithm2e}

Note that the runtime is still $O(n\cdot\text{EO}+n^{2}\log^{O(1)}n)$
as $x^{\mathcal{C}}$ can be computed in $O(n^{2})$ time by the last
lemma. 

We reckon that the hyperplane $\vh_{B}^{T}\vx_{B}\leq\sum_{i\in B}h_{i}\bar{x}_{i}$
returned by the oracle is \textit{not} a valid separating hyperplane
(i.e. it may cut out the minimizers). Nevertheless, we will show that
it is a decent ``proxy'' to the true separating hyperplane $\vh^{T}\vx\leq\hat{f}(x^{\mathcal{C}})=\sum_{i\in V}h_{i}x_{i}^{\mathcal{C}}$
and is good enough to serve our purpose of sandwiching the remaining
feasible region in a small strip. To get a glimpse, note that the
terms missing $\vh_{B}^{T}\vx_{B}\leq\sum_{i\in B}h_{i}\bar{x}_{i}$
all involve $h_{i}$ for $i\notin B$, which is ``negligible'' compared
to $B_{1}\cup\cdots\cup B_{k}$.

One may try to make $\vh_{B}^{T}\vx_{B}\leq\sum_{i\in B}h_{i}\bar{x}_{i}$
valid, say, by $\vh_{B}^{T}\vx_{B}\leq\sum_{i\in B}h_{i}\bar{x}_{i}+\sum_{i\notin B}|h_{i}|$.
The problem is that such hyperplanes would not be separating for $\bar{x}$
anymore as $\vh_{B}^{T}\bar{x}=\sum_{i\in B}h_{i}\bar{x}_{i}<\sum_{i\in B}h_{i}\bar{x}_{i}+\sum_{i\notin B}|h_{i}|$.
Consequently, we lose the width (or volume) guarantee of our cutting
plane algorithm. Although this seems problematic, it is actually still
possible to show a guarantee sufficient for our purpose as $\sum_{i\notin B}|h_{i}|$
is relatively small. We leave it as a nontrivial exercise to interested
readers.

In conclusion, it seems that one cannot have the best of both worlds:
the hyperplane returned by the oracle cannot be simultaneously valid
and separating.
\begin{description}
\item [{Algorithm}]~
\end{description}
We take $k$ to be the first for which $|B_{1}\cup\ldots\cup B_{k}|\geq|B_{k+1}|$,
i.e. $|B_{1}\cup\ldots\cup B_{l}|<|B_{l+1}|$ for $l\leq k-1$. Thus
$k\leq\log n$. Let $b=|B|$, and so $|B_{1}\cup\cdots\cup B_{k}|\ge b/2$.
Case 1 is a special case by taking $B=V$.

Our algorithm is summarized below. Here $A$ is always complete as
$A$ is replaced its transitive closure whenever a new valid arc is
added.
\begin{enumerate}
\item Run Cutting Plane on $P^{B}=\{x\in\mathbb{R}^{B}:\exists x'\in\mathbb{R}^{\bar{B}}\text{ s.t. }(x,x')\text{ satisfies }\eqref{eq:ring}\}$
with the new projected separation oracle.
\item Identify a pair of ``narrow'' approximately parallel supporting
hyperplanes.
\item Deduce from the hyperplanes certain new constraints of the forms $x_{i}=0,x_{j}=1,x_{i}=x_{j}$
or $x_{i}\leq x_{j}$ by lifting separating hyperplanes back to $\mathbb{R}^{n}$
\item Consolidate $A$ and $f$. If some $x_{i}\leq x_{j}$ added, replace
$A$ by its transitive closure.
\item Repeat Step 1 with updated $A$ and $f$. (Any previously found separating
hyperplanes are discarded.)
\end{enumerate}
The minimizer can be constructed by unraveling the recursion.

First of all, to be able to run Cutting Plane on $P^{B}$ we must
come up with a polyhedral description of $P^{B}$ which consists of
just the constraints involving $B$. This is shown in the next lemma.
\begin{lem}
Let $P^{B}=\{\vx\in\mathbb{R}^{B}:\exists\vx'\in\mathbb{R}^{\bar{B}}\text{ s.t. }(\vx,\vx')\text{ satisfies }\eqref{eq:ring}\}$.
Then
\[
P^{B}=\{\vx\in\mathbb{R}^{B}:0\leq\vx\leq1,x_{i}\leq x_{j}\forall(i,j)\in A\cap(B\times B)\}
\]
\end{lem}
\begin{proof}
It is clear that $P^{B}\subseteq\{\vx\in\mathbb{R}^{B}:0\leq\vx\leq1,x_{i}\leq x_{j}\forall(i,j)\in A\cap(B\times B)\}$
as the constraints $0\leq x\leq1,x_{i}\leq x_{j}\forall(i,j)\in A\cap(B\times B)$
all appear in \eqref{eq:ring}.

Conversely, for any $\vx\in\mathbb{R}^{B}$ satisfying $0\leq\vx\leq1,x_{i}\leq x_{j}\forall(i,j)\in A\cap(B\times B)$,
we know there is some completion $x^{\mathcal{C}}$ of $\vx$ by Lemma
\ref{lem:completion} as $A$ is complete. Now $x^{\mathcal{C}}$
satisfies \eqref{eq:ring} by definition, and hence $\vx\in P^{B}$.
\end{proof}
The only place where we have really changed the algorithm is Step
(3).

\subsubsection{Deducing New Constraints $x_{i}=0$, $x_{j}=1$, $x_{i}=x_{j}$ or
$x_{i}\leq x_{j}$\label{sub:submod:n4:new-constraints}}

Our method will deduce one of the following:
\begin{itemize}
\item $x_{i}=0$, $x_{j}=1$ or $x_{i}=x_{j}$
\item for each $p\in B_{1}\cup\cdots\cup B_{k}$, $x_{p}\leq x_{q}$ for
some $q\notin R(p)$ or $x_{p}\ge x_{q}$ for some $q\notin Q(p)$
\end{itemize}
Our argument is very similar to the last section's. Roughly speaking,
it is the same argument but with ``noise'' introduced by $i\notin B$.
We use extensively the notations from the last section.

Our main tool is again Theorem \ref{thm:cuttingplanerestated}. Note
that $n$ should be replaced by $b$ in the Theorem statement. We
invoke it with $\tau=k\log_{b}n=O(\log^{2}n)$ (using $k\leq\log n$)
to get a width of $1/b^{\Theta(\tau)}=1/n^{\Theta(k)}$. This takes
time at most $O(bn\log^{2}n\cdot\text{EO}+bn^{2}\log^{O(1)}n)$. Again,
this is intuitively clear as we run it for $O(kb\log n)$ iterations,
each of which takes time $O(n\cdot\text{EO}+n^{2}\log^{O(1)}n)$.

After each phase of (roughly $O(kb\log n)$ iterations) of Cutting
Plane, $P^{B}$ is sandwiched between a pair of approximately parallel
supporting hyperplanes $F$ and $F'$ which have width $1/n^{20k}$.
Let $F$ and $F'$ be

\[
\vc^{T}\vx_{B}=\sum_{i\in B}c_{i}x_{i}\leq M,\quad\vc'^{T}\vx_{B}=\sum_{i\in B}c_{i}'x_{i}\leq M',
\]
such that 
\[
|M+M'|,\:||\vc+\vc'||_{2}\leq\gap,\quad\text{where }\gap=\frac{1}{n^{20k}}\min\{||\vc||_{2},||\vc'||_{2}\}.
\]
The rest of this section presents an execution of the ideas discussed
above. All of our work is basically geared towards bringing the amortized
cost for identifying a valid arc down to $\widetilde{O}(n\cdot\text{EO}+n^{2})$.
Again, we can write these two constraints as a nonnegative combination.
Here $\bar{x}_{h}^{\mathcal{C}}$ is the completion of the point $\bar{x}_{h}$
used to construct $\vh_{B}^{T}\vx_{B}\leq\vh_{B}^{T}\left(\bar{x}_{h}^{\mathcal{C}}\right)_{B}$.
(Recall that $\left(\bar{x}_{h}^{\mathcal{C}}\right)_{B}$ is the
restriction of $\bar{x}_{h}^{\mathcal{C}}$ to $B$.)

\[
\vc^{T}\vx_{B}=-\sum_{i\in B}\alpha_{i}x_{i}+\sum_{j\in B}\beta_{j}x_{j}+\sum_{(i,j)\in A\cap B^{2}}\gamma_{ij}(x_{i}-x_{j})+\sum_{h\in H}\lambda_{h}\vh_{B}^{T}\vx_{B}\,\,\,\,\text{ and }\,\,\,\, M=\sum_{j\in B}\beta_{j}+\sum_{h\in H}\lambda_{h}\vh_{B}^{T}\left(\bar{x}_{h}^{\mathcal{C}}\right)_{B}.
\]

\[
\vc'^{T}\vx_{B}=-\sum_{i\in B}\alpha_{i}'x_{i}+\sum_{j\in B}\beta_{j}'x_{j}+\sum_{(i,j)\in A\cap B^{2}}\gamma_{ij}'(x_{i}-x_{j})+\sum_{h\in H}\lambda_{h}'\vh_{B}^{T}\vx_{B}\,\,\,\,\text{ and }\,\,\,\, M'=\sum_{j\in B}\beta_{j}'+\sum_{h\in H}\lambda_{h}'\vh_{B}^{T}\left(\bar{x}_{h}^{\mathcal{C}}\right)_{B}.
\]

As we have discussed, the problem is that the separating hyperplanes
$\vh_{B}^{T}\vx_{B}\leq\vh_{B}^{T}\left(\bar{x}_{h}^{\mathcal{C}}\right)_{B}$
are not actually valid. We can, however, recover their valid counterpart
by lifting them back to $\vh^{T}\vx\leq\vh^{T}\bar{x}_{h}^{\mathcal{C}}$.
The hope is that $\vh_{B}^{T}\vx_{B}\leq\vh_{B}^{T}\left(\bar{x}_{h}^{\mathcal{C}}\right)_{B}$
and $\vh^{T}\vx\leq\vh^{T}\bar{x}_{h}^{\mathcal{C}}$ are not too
different so that the arguments will still go through. We show that
this is indeed the case.

Again, we scale $c,c',\alpha,\alpha',\beta,\beta',\gamma,\gamma',\lambda,\lambda'$
so that
\[
\sum_{h\in H}(\lambda_{h}+\lambda_{h}')=1.
\]
By adding all the constituent separating hyperplane inequalities,
we get

\[
\sum_{h\in H}\lambda_{h}\vh^{T}\vx+\sum_{h\in H}\lambda_{h}'\vh^{T}\vx\leq\sum_{h\in H}\lambda_{h}\vh^{T}\bar{x}_{h}^{\mathcal{C}}+\sum_{h\in H}\lambda_{h}'\vh^{T}\bar{x}_{h}^{\mathcal{C}}
\]
Let 
\[
LHS\defeq\sum\alpha_{i}x_{i}+\sum\alpha_{i}'x_{i}-\sum\beta_{j}x_{j}-\sum\beta_{j}'x_{j}+\sum\gamma_{ij}(x_{j}-x_{i})+\sum\gamma_{ij}'(x_{j}-x_{i}).
\]
Here we know that
\[
\sum_{h\in H}\lambda_{h}\vh^{T}\vx+\sum_{h\in H}\lambda_{h}'\vh^{T}\vx=LHS+(\vc+\vc')^{T}\vx_{B}+\sum_{h\in H}\lambda_{h}\vh_{\bar{B}}^{T}\vx_{\bar{B}}+\sum_{h\in H}\lambda_{h}'\vh_{\bar{B}}^{T}\vx_{\bar{B}}
\]
\[
\sum_{h\in H}\lambda_{h}\vh^{T}\bar{x}_{h}^{\mathcal{C}}+\sum_{h\in H}\lambda_{h}'\vh^{T}\bar{x}_{h}^{\mathcal{C}}=(M+M')+\sum_{h\in H}\lambda_{h}\vh_{\bar{B}}^{T}\left(\bar{x}_{h}^{\mathcal{C}}\right)_{\bar{B}}+\sum_{h\in H}\lambda_{h}'\vh_{\bar{B}}^{T}\left(\bar{x}_{h}^{\mathcal{C}}\right)_{\bar{B}}-\sum\beta_{j}-\sum\beta_{j}'
\]

Combining all yields
\[
LHS+(\vc+\vc')^{T}\vx_{B}+\sum_{h\in H}\lambda_{h}\vh_{\bar{B}}^{T}\vx_{\bar{B}}+\sum_{h\in H}\lambda_{h}'\vh_{\bar{B}}^{T}\vx_{\bar{B}}\leq(M+M')+\sum_{h\in H}\lambda_{h}\vh_{\bar{B}}^{T}\left(\bar{x}_{h}^{\mathcal{C}}\right)_{\bar{B}}+\sum_{h\in H}\lambda_{h}'\vh_{\bar{B}}^{T}\left(\bar{x}_{h}^{\mathcal{C}}\right)_{\bar{B}}-\sum\beta_{j}-\sum\beta_{j}'
\]

Here $(\vc+\vc')^{T}\vx_{B}$ can be bounded as before: $(\vc+\vc')^{T}\vx_{B}\geq-\sqrt{n}||\vc+\vc'||_{2}\geq-\sqrt{n}\gap$.
Since $M+M'\leq\gap$, We then obtain
\[
LHS+\sum_{h\in H}\lambda_{h}\vh_{\bar{B}}^{T}\vx_{\bar{B}}+\sum_{h\in H}\lambda_{h}'\vh_{\bar{B}}^{T}\vx_{\bar{B}}\leq2\sqrt{n}\gap+\sum_{h\in H}\lambda_{h}\vh_{\bar{B}}^{T}\left(\bar{x}_{h}^{\mathcal{C}}\right)_{\bar{B}}+\sum_{h\in H}\lambda_{h}'\vh_{\bar{B}}^{T}\left(\bar{x}_{h}^{\mathcal{C}}\right)_{\bar{B}}-\sum\beta_{j}-\sum\beta_{j}'
\]

We should expect the contribution from $\vh_{\bar{B}}$ to be small
as $h_{i}$ for $i\notin B$ is small compared to $B_{1}\cup\ldots\cup B_{k}$.
We formalize our argument in the next two lemmas.
\begin{lem}
We have $\sum_{h\in H}\lambda_{h}\vh_{\bar{B}}^{T}\left(\bar{x}_{h}^{\mathcal{C}}\right)_{\bar{B}}+\sum_{h\in H}\lambda_{h}'\vh_{\bar{B}}^{T}\left(\bar{x}_{h}^{\mathcal{C}}\right)_{\bar{B}}\leq N/n^{10(k+1)-1}$.\end{lem}
\begin{proof}
We bound each component of $\sum_{h\in H}\lambda_{h}\vh_{\bar{B}}^{T}\left(\bar{x}_{h}^{\mathcal{C}}\right)_{\bar{B}}+\sum_{h\in H}\lambda_{h}'\vh_{\bar{B}}^{T}\left(\bar{x}_{h}^{\mathcal{C}}\right)_{\bar{B}}$.
For $i\in\bar{B}$, we have $\upi\leq N/n^{10(k+1)}$. By Lemma \ref{lem:boundonsmall-1}
$h_{i}\le\upi$. Therefore, 
\begin{eqnarray*}
\sum_{h\in H}\lambda_{h}\vh_{i}^{T}\left(\bar{x}_{h}^{\mathcal{C}}\right)_{i}+\sum_{h\in H}\lambda_{h}'\vh_{i}^{T}\left(\bar{x}_{h}^{\mathcal{C}}\right)_{i} & \le & \left(\sum_{h\in H}\lambda_{h}+\sum_{h\in H}\lambda_{h}'\right)N/n^{10(k+1)}=N/n^{10(k+1)}.
\end{eqnarray*}

Our result then follows since 
\[
\sum_{h\in H}\lambda_{h}\vh_{\bar{B}}^{T}\left(\bar{x}_{h}^{\mathcal{C}}\right)_{\bar{B}}+\sum_{h\in H}\lambda_{h}'\vh_{\bar{B}}^{T}\left(\bar{x}_{h}^{\mathcal{C}}\right)_{\bar{B}}=\sum_{i\in\bar{B}}\left(\sum_{h\in H}\lambda_{h}\vh_{i}^{T}\left(\bar{x}_{h}^{\mathcal{C}}\right)_{i}+\sum_{h\in H}\lambda_{h}'\vh_{i}^{T}\left(\bar{x}_{h}^{\mathcal{C}}\right)_{i}\right).
\]
\end{proof}
\begin{lem}
We have $\sum_{h\in H}\lambda_{h}\vh_{\bar{B}}^{T}\vx_{\bar{B}}+\sum_{h\in H}\lambda_{h}'\vh_{\bar{B}}^{T}\vx_{\bar{B}}\geq-N/n^{10(k+1)-1}$.\end{lem}
\begin{proof}
The proof is almost identical to the last lemma except that we use
$h_{i}\ge\lowi$ instead of $h_{i}\leq\upi$, and $\lowi\geq-N/n^{10(k+1)}$.
\end{proof}
The two lemmas above imply that
\[
LHS\leq2\sqrt{n}\gap-\sum\beta_{j}-\sum\beta_{j}+2N/n^{10(k+1)-1}=\gap'-\sum\beta_{j}-\sum\beta_{j}
\]
where $\gap'=2\sqrt{n}\gap+2N/n^{10(k+1)-1}$.
\begin{lem}
\label{lem:dimcut-new} Suppose $x$ satisfies \eqref{eq:ring} and
$LHS\leq\gap'-\sum\beta_{j}-\sum\beta_{j}'$ with $\alpha_{i},\beta_{j},\gamma_{ij},\alpha_{i}',\beta_{j}',\gamma_{ij}'\geq0$.
\begin{enumerate}
\item If $\alpha_{i}>\gap'$ or $\alpha_{i}'>\gap'$, then $x_{i}<1$.
\item If $\beta_{j}>\gap'$ or $\beta_{j}'>\gap'$, then $x_{j}>0$.
\item If $\gamma_{ij}>\gap'$ or $\gamma_{ij}'>\gap'$, then $0\leq x_{j}-x_{i}<1$.
\end{enumerate}
\end{lem}
\begin{proof}
The proof is exactly the same as Lemma \ref{lem:dimcut} with $2\sqrt{n}\gap$
replaced by $\gap'$.
\end{proof}
From now on we may assume that
\begin{equation}
\max\{\alpha_{i},\alpha_{i}',\beta_{j},\beta_{j}',\gamma_{ij},\gamma_{ij}'\}\leq\gap'.\label{eq:small-1}
\end{equation}

\begin{lem}
\label{lem:newarcfound-new}Let $\vy\defeq\sum_{h\in H}\lambda_{h}\vh$
and $\vy'\defeq\sum_{h\in H}\lambda'_{h}\vh$ and let $p\in\argmax_{l\in B}\{\max\{|y_{l}|,|y'_{l}|\}$
then
\[
N\geq n^{10k+6}\normInf{\vy_{B}+\vy'_{B}}
\]
assuming \eqref{eq:small-1}.\end{lem}
\begin{proof}
Recall that $\norm{\vc+\vc'}_{2}\leq\gap<\gap'$ where $\gap=\frac{1}{n^{20k}}\min\{||\vc||_{2},||\vc'||_{2}\}$
and $\gap'=2\sqrt{n}\gap+2N/n^{10(k+1)-1}$. Now there are two cases.

\textbf{Case 1}: $2\sqrt{n}\gap\geq2N/n^{10(k+1)-1}$. Then $\gap'\leq4\sqrt{n}\gap$
and we follow the same proof of Lemma \ref{lem:newarcfound}. We have

\[
\vc=\vy_{B}-\sum_{i}\alpha_{i}\indicVec i+\sum_{j}\beta_{j}\indicVec j+\sum_{(i,j)}\gamma_{ij}(\indicVec i-\indicVec j)\enspace\text{ and }\vc'=\vy'_{B}-\sum_{i}\alpha_{i}'\indicVec i+\sum_{j}\beta_{j}'\indicVec j+\sum_{(i,j)}\gamma_{ij}'(\indicVec i-\indicVec j)\,.
\]
By \eqref{eq:small-1} we know that $\norm{\vc-\vy_{B}}_{2}\leq4n^{2}\gap'\leq\frac{1}{n^{17k}}\norm{\vc}_{2}$
and $\norm{\vc'-\vy'_{B}}_{2}\leq4n^{2}\gap'\leq\frac{1}{n^{17k}}\norm{\vc}_{2}$.
Consequently, by the triangle inequality we have that
\begin{align*}
\norm{\vy_{B}+\vy_{B}'}_{2} & \leq\norm{\vc+\vc'}_{2}+\norm{\vc-\vy_{B}}_{2}+\norm{\vc'-\vy_{B}'}_{2}\leq9n^{2}\gap'
\end{align*}
and
\[
\norm{\vc}_{2}\leq\norm{\vc-\vy_{B}}_{2}+\norm{\vy_{B}}_{2}\leq\frac{1}{n^{17k}}\norm{\vc}_{2}+\norm{\vy_{B}}_{2}\enspace\Rightarrow\enspace\norm{\vc}_{2}\leq2\norm{\vy_{B}}_{2}
\]
Similarly, we have that $\norm{\vc'}_{2}\leq2\norm{\vy_{B}'}_{2}$.
Consequently since $\gap'\leq\frac{1}{n^{19k}}\min\{||\vc||_{2},||\vc'||_{2}\}$,
we have that 
\[
\norm{\vy_{B}+\vy_{B}'}_{2}\leq\frac{18}{n^{17k}}\min\left\{ \norm{\vy_{B}}_{2},\norm{\vy_{B}'}_{2}\right\} 
\]
and thus, invoking Lemma~\ref{lem:submod:sandwhich-stronger} yields
$N\geq\upp\geq n^{16k}\normInf{\vy_{B}+\vy'_{B}}$, as desired.

\textbf{Case 2}: $2\sqrt{n}\gap<2N/n^{10(k+1)-1}$. Then for any $i\in B$,
$|c_{i}+c_{i}'|\leq||\vc+\vc'||_{2}\leq\gap<2N/n^{10(k+1)-1}$. Since
\[
\vy_{B}+\vy'_{B}=(\vc+\vc')+\sum_{i}\alpha_{i}\indicVec i-\sum_{j}\beta_{j}\indicVec j-\sum_{(i,j)}\gamma_{ij}(\indicVec i-\indicVec j)+\sum_{i}\alpha_{i}'\indicVec i-\sum_{j}\beta_{j}'\indicVec j-\sum_{(i,j)}\gamma_{ij}'(\indicVec i-\indicVec j)
\]
we have
\[
\normInf{\vy_{B}+\vy'_{B}}\leq2N/n^{10(k+1)-1}+2n^{1.5}\gap'\leq N/n^{10k+7}.
\]

\end{proof}

\begin{cor}
\label{cor:final} Let $P$ be the feasible region after running Cutting
Plane on \eqref{eq:ring} with the projected separation oracle. Then
one of the following holds:
\begin{enumerate}
\item We found a BFS $\vh$ with $\vh_{B}=0$.
\item The integral points of $P$ all lie on some hyperplane $x_{i}=0,x_{j}=1$
or $x_{i}=x_{j}$.
\item Let $H$ be the collection of BFS's $\vh$ used to construct our separating
hyperplanes for $P$. Then there is a convex combination $\vy$ of
$H$ such that for $p\in B_{1}\cup\cdots\cup B_{k}$, we have $n^{4}|y_{i}|<\upp$
or $\lowp<-n^{4}|y_{i}|$ for all $i$.
\end{enumerate}
\end{cor}
\begin{proof}
As mentioned before, (1) happens if some separating hyperplane satisfies
$\vh_{B}=0$ when running cutting plane on the non-negligible coordinates.
We have (2) if some condition in Lemma \ref{lem:dimcut-new} holds.
Otherwise, we claim $y=\sum_{\boldsymbol{h}}\lambda_{\boldsymbol{h}}\boldsymbol{h}+\sum_{\boldsymbol{h}}\lambda_{\boldsymbol{h}}'\boldsymbol{h}$
is a candidate for Case 3. $y$ is a convex combination of BFS and
by Lemma \ref{lem:newarcfound-new}, for the big elements $i\in B$
we have
\[
|y_{i}|\leq N/n^{10k+6}\leq\frac{1}{n^{4}}\max\{\upp,-\lowp\}.
\]
where the last inequality holds since for $p\in B_{1}\cup\cdots\cup B_{k},$
$\max\{\upp,-\lowp\}\geq N/n^{10k}$.

On the other hand, for the small elements $i\notin B$, $|y_{i}|\leq N/n^{10(k+1)}\leq\frac{1}{n^{4}}\max\{\upp,-\lowp\}$
as desired.
\end{proof}
The gap is then smaller enough to add an arc for each $p\in B_{1}\cup\cdots\cup B_{k}$
by Lemmas \ref{lem:newarc} and \ref{lem:newarc-new-1}. Therefore
we can add a total of $|B_{1}\cup\cdots\cup B_{k}|/2\geq b/4$ arcs
with roughly $O(kb\log n)=\widetilde{O}(b)$ iterations of Cutting
Plane, each of which takes $\widetilde{O}(n\cdot\text{EO}+n^{2})$.
That is, the amortized cost for each arc is $\widetilde{O}(n\cdot\text{EO}+n^{2})$.
We give a more formal time analysis in below but it should be somewhat
clear why we have the desired time complexity.
\begin{lem}
\label{lem:manyarcs}Suppose there is a convex combination $\vy$
of $H$ such that for $p\in B_{1}\cup\cdots\cup B_{k}$, we have $n^{4}|y_{i}|<\upp$
or $\lowp<-n^{4}|y_{i}|$ for all $i$. Then we can identify at least
$b/4$ new valid arcs.\end{lem}
\begin{proof}
We have $|H|=O(n)$ since $H$ is the set of BFS's used for the constraints
of $P$ which has $O(n)$ constraints. By Lemmas \ref{lem:newarc}
and \ref{lem:newarc-new-1}, for $p\in B_{1}\cup\cdots\cup B_{k}$
we can add a new valid arc $(p,q)$ or $(q,p)$. However note that
a new arc $(p_{1},p_{2})$ may added twice by both $p_{1}$ and $p_{2}$.
Therefore the total number of new arcs is only at least $|B_{1}\cup\cdots\cup B_{k}|/2\geq b/4$.
\end{proof}

\subsubsection{Running Time\label{sub:submod:n4:runtime}}

Not much changes to the previous runtime analysis are needed. To avoid
repetition, various details already present in the corresponding part
of the last section are omitted. Recall $k\leq\log n$, and of course,
$b\leq n$.

For each (roughly) $O(kb\log n)$ iterations of Cutting Plane we either
get $x_{i}=0$,$x_{i}=1$,$x_{i}=x_{j}$ or $b/4$ $x_{i}\leq x_{j}$'s.
The former can happen at most $n$ times while in the latter case,
the amortized cost of each arc is $O(k\log n)$ iterations of Cutting
Plane. In the worst case the overall number of iterations required
is $\widetilde{O}(n^{2})$. Thus our algorithm has a runtime of $\widetilde{O}(n^{3}\cdot\text{EO}+n^{4})$
since each iteration is $\widetilde{O}(n\cdot\text{EO}+n^{2})$ as
shown below.
\begin{thm}
Our algorithm runs in time $O(n^{3}\log^{2}n\cdot\text{EO}+n^{4}\log^{O(1)}n)$.\end{thm}
\begin{proof}
We use Corollary \ref{cor:final}. First we note that Case 1 can actually
be integrated into Case 3 since $\max\{\upp,-\lowp\}\geq N/n^{10k}=n^{10}N/n^{10(k+1)}\geq h_{i}$
for $i\notin B$.

As we have argued in the beginning of the last section, Theorem \ref{thm:cuttingplanerestated}
with $\tau=k\log_{b}n$ implies that the runtime for each phase is
$O(bn\log^{2}n\cdot\text{EO}+bn^{2}\log^{O(1)}n)$. In each phase
we either get $x_{i}=0$, $x_{i}=1$, $x_{i}=x_{j}$ (Case 2) or $b/4$
$x_{i}\leq x_{j}$'s (Case 3), the latter of which follows from Corollary
\ref{cor:final} and Lemma \ref{lem:manyarcs}.

Case 2 can only happen $n$ times. Thus the total cost is at most
$O(n^{3}\log^{2}n\cdot\text{EO}+n^{4}\log^{O(1)}n)$. The overhead
cost is also small. Similar to before, given $F$ and $F'$ represented
as a nonnegative combination of facets, we can check for the conditions
in Lemma \ref{lem:dimcut-new} in $O(n)$ time as there are only this
many facets of $P$. This settles Case 2.

For case 3 the amortized cost for each arc is $O(n\log^{2}n\cdot\text{EO}+n^{2}\log^{O(1)}n)$.
Our desired runtime follows since there are only $O(n^{2})$ arcs
to add. Unlike Case 2 some extra care is needed to handle the overhead
cost. The time needed to deduce a new arc (applying Lemmas \ref{lem:newarc}
and \ref{lem:newarc-new-1} to $\vy$ and $p\in B_{1}\cup\cdots\cup B_{k}$)
is still $O(n\cdot\text{EO}+n^{2})$. But as soon as we get a new
arc, we must update $A$ to be its transitive closure so that it is
still complete. Given $A$ complete and a new arc $(p,q)\notin A$,
we can simply add the arcs from the ancestors of $p$ to $q$ and
from $p$ to the descendants of $q$. There are at most $O(n)$ arcs
to add so this takes time $O(n^{2})$ per arc, which is okay.
\end{proof}

\section{Discussion and Comparison with Previous Algorithms\label{sec:submodular_discussion}}

We compare and contrast our algorithms with the previous ones. We
focus primarily on strongly polynomial time algorithms.

\medskip

\noindent \textbf{Convex combination of BFS's}

All of the previous algorithms maintain a convex combination of BFS's
and iteratively improve over it to get a better primal solution. In
particular, the new BFS's used are typically obtained by making local
changes to existing ones. Our algorithms, on the other hand, considers
the geometry of the existing BFS's. The weighted ``influences''%
\footnote{In the terminology of Part \ref{part:Ellipsoid}, these weighted influences
are the leverage scores.%
} then aggregately govern the choice of the next BFS. We believe that
this is the main driving force for the speedup of our algorithms.

\medskip

\noindent \textbf{Scaling schemes}

Many algorithms for combinatorial problems are explicitly or implicitly
scaling a potential function or a parameter. In this paper, our algorithms
in some sense aim to minimize the volume of the feasible region. Scaling
schemes for different potential functions and parameters were also
designed in previous works \cite{iwata2001combinatorial,iwata2003faster,iwata2009simple,iwata2002fully}.
All of these functions and parameters have an explict form. On the
contrary, our potential function is somewhat unusual in the sense
that it has no closed form.

\medskip

\noindent \textbf{Deducing new constraints}

As mentioned in the main text, our algorithms share the same skeleton
and tools for deducing new constraints with \cite{iwata2001combinatorial,iwata2003faster,iwata2009simple,iwata2002fully}.
Nevertheless, there are differences in the way these tools are employed.
Our algorithms proceed by invoking them in a geometric manner, whereas
previous algorithms were mostly combinatorial.

\medskip

\noindent \textbf{Big elements and bucketing}

Our bucketing idea has roots in Iwata-Orlin's algorithm \cite{iwata2009simple}
but is much more sophisticated. For instance, it is sufficient for
their algorithm to consider only big elements, i.e. $\upi\geq N/n^{O(1)}$.
Our algorithm, on the other hand, must carefully group elements by
the size of both $\upi$ and $\lowi$. The speedup appears impossible
without these new ideas. We do however note that it is unfair to expect
such a sophisticated scheme in Iwata-Orlin's algorithm as it would
not lead to a speedup. In other words, their method is fully sufficient
for their purposes, and the simplicity in their case is a virtue rather
than a shortcoming.

\subsection{Open Problems}

One natural open problem is improving our weakly polynomial algorithm
to $O(n^{2}\log M\cdot\text{EO}+n^{3}\log^{O(1)}n\cdot\log M)$ time.
Our application of center of mass to SFM demonstrates that it should
be possible.

For strongly polynomial algorithms, the existential result of Theorem
\ref{thm:n3logn} shows that SFM can be solved with $O(n^{3}\log n\cdot\text{EO})$
oracle calls. Unfortunately, our algorithm incurs an overhead of $\log n$
as there can be as many as $\log n$ buckets each time. One may try
to remove this $\log n$ overhead by designing a better bucketing
scheme or arguing that more arcs can be added.

The other $\log n$ overhead seem much trickier to remove. Our method
currently makes crucial use of the tools developed by \cite{iwata2001combinatorial},
where the $\log n$ factors in the runtime seem inevitable. We suspect
that our algorithm may have an analogue similar to \cite{schrijver2000combinatorial,orlin2009faster},
which do not carry any $\log n$ overhead in the running time.

Perhaps an even more interesting open problem is whether our algorithm
is optimal (up to polylogarithmic factors). There are grounds for
optimism. So far the best way of \textit{certifying} the optimality
of a given solution $S\subseteq V$ is to employ duality and express
some optimal solution to the base polyhedron as a convex combination
of $n+1$ BFS's. This already takes $n^{2}$ oracle calls as each
BFS requires $n$. Thus one would expect the optimal number of oracle
calls needed for SFM to be at least $n^{2}$. Our bound is not too
far off from it, and anything strictly between $n^{2}$ and $n^{3}$
seems instinctively unnatural.

\section*{Acknowledgments}

We thank Matt Weinberg for insightful comments about submodular minimization
and minimizing the intersection of convex sets that were deeply influential
to our work. We thank Yan Kit Chim, Stefanie Jegelka, Jonathan A.
Kelner, Robert Kleinberg, Pak-Hin Lee, Christos Papadimitriou, and
Chit Yu Ng for many helpful conversations. We thank Chien-Chung Huang
for pointing out a typo in an earlier draft of this paper. This work
was partially supported by NSF awards 0843915 and 1111109, NSF grants
CCF0964033 and CCF1408635, Templeton Foundation grant 3966, NSF Graduate
Research Fellowship (grant no. 1122374). Part of this work was done
while the first two authors were visiting the Simons Institute for
the Theory of Computing, UC Berkeley. Lastly, we would like to thank
Vaidya for his beautiful work on his cutting plane method. 

\bibliographystyle{plain}
\bibliography{main,main_app}

\end{document}